\documentclass[11pt]{book}
\usepackage{fullpage}

\usepackage{graphicx, amsmath,amsfonts,amssymb,xcolor,enumerate, mathrsfs, amsthm,url, mathtools, booktabs, array,paralist,verbatim, algorithmic,algorithm,tikz,standalone,tcolorbox, hyperref} 
\usepackage{subcaption}
\usetikzlibrary{arrows, patterns}

\usepackage{tocloft}

\numberwithin{equation}{chapter}
\newtheorem{claim}{Claim}[chapter]
\newtheorem{lemma}[claim]{Lemma}

\newtheorem{theorem}{Theorem}[chapter]
\newtheorem{proposition}[claim]{Proposition}
\newtheorem{corollary}[claim]{Corollary}
\newtheorem{definition}[claim]{Definition}
\newtheorem{remark}{Remark}[chapter]

\usepackage{sectsty}
\allsectionsfont{\sffamily\mdseries\upshape} 

\newcommand{\un}{\underline}
\newcommand{\be}{\begin{equation}}
\newcommand{\ee}{\end{equation}}
\newcommand{\ben}{\begin{equation*}}
\newcommand{\een}{\end{equation*}}
\newcommand{\bi}{\begin{itemize}}
\newcommand{\ei}{\end{itemize}}
\newcommand{\mc}{\mathcal}
\newcommand{\mbf}{\mathbf}
\newcommand{\msf}{\tsf}
\newcommand{\tsf}{\mathsf}
\newcommand{\snr}{\textsf{snr}}

\newcommand{\e}{\epsilon}

\newcommand{\statt}{\textsf{stat}^t}
\newcommand{\mcb}{\mathcal{B}_{M,L}}
\newcommand{\mcs}{\mathcal{S}}
\newcommand{\pr}{\mathbb{P}}
\newcommand{\re}{\mathbb{R}}
\newcommand{\bth}{\hat{\beta}}
\newcommand{\bfs}{s}  
\newcommand{\bfsh}{\hat{s}}  
\newcommand{\tbfs}{\tilde{\bfs}}

\newcommand{\abs}[1]{\left \lvert #1\right \rvert}
\newcommand{\norm}[1]{ \lVert #1 \rVert}
\newcommand{\normsq}[1]{ \lVert #1 \rVert^2}
\newcommand{\expec}{\mathbb{E}}
\newcommand{\bks}{\backslash}

\newcommand{\mscrs}{\mathscr{S}}

\newcommand{\sfi}{\mathsf{I}}
\newcommand{\Ah}{A_{\mathsf{H}}}
\newcommand{\Esec}{\mc{E}_{\textsf{sec}}}
\DeclareMathOperator*{\argmin}{arg\,min}

\newcommand{\reals}{\mathbb{R}}

\newcommand{\xtl}{x_{\textsf{L}}(\tau)}
\newcommand{\algorithmicbreak}{\textbf{break}}
\newcommand{\BREAK}{\STATE \algorithmicbreak}
\newcommand{\wh}[1]{\widehat{#1}}
\newcommand{\normal}{\mc{N}}
\newcommand{\bmin}{b_{\text{min}}}

 \newcommand{\highlightgr}[1]{%
  \colorbox{black!15}{$\displaystyle#1$}}

\newcommand{\M}{M}
\newcommand{\Lc}{L_C}
\newcommand{\Lr}{L_R}
\newcommand{\Mc}{M_C}
\newcommand{\Mr}{M_R}

\newcommand{\sfb}{\textsf{b}}
\newcommand{\Wricj}{W_{\textsf{r}(i)\textsf{c}(j)}}
\newcommand{\sfr}{\textsf{r}}

\newcommand{\sfc}{\textsf{c}}

\newcommand{\stdnorm}{\mathcal{N}(0,1)}
\newcommand{\Smat}{S}

\title{\Huge \sf Sparse Regression Codes \vspace{1in}}

\author{
{\sf Ramji Venkataramanan, \emph{University of Cambridge} \hfill \vspace{7pt}} \\ 
{\sf Sekhar Tatikonda, \emph{Yale University} \hfill \vspace{7pt} }  \\ 
{\sf Andrew Barron, \emph{Yale University} \hfill  } 
}

\date{}

\newcommand\summaryname{\sf \large Abstract}
\newenvironment{Abstract}%
    {\begin{center}%
    \bfseries{\summaryname} \end{center}}

\newcommand\ackname{\sf \large Acknowledgements}
\newenvironment{ack}%
    {\begin{center}%
    \bfseries{\ackname} \end{center}}

\begin{document}

\frontmatter
\maketitle
\tableofcontents

\newpage

\begin{Abstract}
\begin{changemargin}{1cm}{1cm}
Developing computationally-efficient codes that approach the Shannon-theoretic limits for communication and compression has long been one of the major goals of information and coding theory.  There have been significant advances towards this goal in the last couple of decades, with the emergence of turbo codes,  sparse-graph codes, and polar codes. These codes are  designed primarily  for discrete-alphabet channels and sources. For Gaussian channels and sources, where the alphabet is inherently continuous, \emph{Sparse Superposition Codes} or \emph{Sparse Regression Codes} (SPARCs) are a promising class of codes for achieving the Shannon limits.

This monograph provides a unified and comprehensive over-view of sparse regression codes,  covering theory, algorithms, and practical implementation aspects. The first part of the monograph focuses on SPARCs for AWGN channel coding, and the second part on  SPARCs for lossy compression (with  squared error distortion criterion). In the third part, SPARCs are used to construct codes for Gaussian multi-terminal channel and source coding models such as broadcast channels, multiple-access channels,  and source and channel coding with side information. The monograph concludes with a discussion of open problems and directions for future work.
\end{changemargin}
\end{Abstract}

\mainmatter

\chapter{Introduction} \label{chap:intro} 
%

Developing computationally-efficient codes that approach the Shannon-theoretic limits for communication and compression has long been one of the major goals of information and coding theory.  There have been significant advances towards this goal in the last couple of decades, with the emergence of turbo and sparse-graph codes in the '90s \cite{Berrou96,costForney2007,RichUBook},  and more recently polar codes and spatially-coupled LDPC codes \cite{Polarcc, Polarrd,spatialCoup13}. These codes are primarily designed for  channels with discrete input alphabet, and for discrete-alphabet sources. 

There are many channels and sources of practical interest where the alphabet is inherently continuous, e.g., additive white Gaussian noise (AWGN) channels, and Gaussian sources. This monograph discusses a class of codes for such Gaussian models called   \emph{Sparse Superposition Codes} or \emph{Sparse Regression Codes} (SPARCs). These codes were introduced by Barron and Joseph \cite{AntonyML,AntonyFast} for efficient communication over AWGN channels, but have since also been used  for lossy compression \cite{RVGaussianML,RVGaussFeasible} and multi-terminal communication \cite{RVAller12}. Our goal in this monograph is to provide a unified and comprehensive view of SPARCs, covering theory, algorithms, as well as practical implementation aspects. 

To motivate the construction of SPARCs, let us begin with the standard AWGN channel. The goal is to construct codes with computationally efficient encoding and decoding that \emph{provably} achieve the channel capacity $\mc{C} =\frac{1}{2}\log_2(1 + \snr)$ bits/transmission, where $\snr$ denotes the signal-to-noise ratio.   In particular, we are interested in codes whose encoding and decoding complexity grows no faster than a low-order polynomial in the block length $n$.

 Though it is well known that rates approaching $\mc{C}$ can be achieved with Gaussian codebooks, this has been largely avoided in practice because of the high decoding complexity of unstructured Gaussian codes.   Instead, the popular approach has been to separate the design of the coding scheme into two steps: \emph{coding} and \emph{modulation}.  State-of-the-art coding schemes for the AWGN channel such as coded modulation \cite{forney98,bicmAlbert,bocherer15}  use this two-step design, and combine  binary error-correcting codes such as LDPC and turbo codes with standard modulation schemes such as Quadrature Amplitude Modulation (QAM). Though such schemes have good empirical performance,  they have not been proven to be capacity-achieving for the AWGN channel.    With sparse regression codes, we step back from the coding/modulation divide and instead use a structured codebook to construct low-complexity, capacity-achieving schemes tailored to the AWGN channel. 
 
There have been several lattice based schemes \cite{ErezZ04,zamir14lattice} proposed for communication over the AWGN channel, including low density lattice codes \cite{LDLC08}  and polar lattices \cite{yanLing14,abbe2011polar}. The reader is referred to the cited works for details of the performance vs. complexity trade-offs of these codes.

In the rest of this chapter, we describe the sparse regression codebook, and give a brief overview of the topics covered in the later chapters. First, we lay down some notation that will be used throughout the monograph.
 
 \paragraph{Notation}  The Gaussian distribution with mean $\mu$ and variance $\sigma^2$ is denoted by $\mc{N}(\mu,\sigma^2)$. For a positive integer $L$, we use $[L]$ to denote the set $\{1, \ldots, L\}$. The Euclidean norm of a vector $x$ is denoted by $\norm{x}$.  The indicator function of an event $\mc{E}$ is denoted by $\mbf{1}\{ \mc{E} \}$.   The transpose of a matrix $A$ is denoted by $A^*$. The $n \times n$ identity matrix is denoted by $\mathsf{I}_n$, with the subscript dropped when it is clear from context.

 Both $\log$ and $\ln$  are used to denote the natural logarithm. Logarithms to the base $2$ are denoted by $\log_2$. For most of the theoretical analysis, we will find it convenient to use natural logarithms. Therefore, rate is measured in \emph{nats}, unless otherwise specified. Throughout, we use $n$ for the block length of the code. 
 
 For random vectors $X,Y$ defined on the same probability space, we write $X \stackrel{d}{=} Y$ to indicate that $X$ and $Y$ have the same distribution.

\section{The Sparse Regression Codebook} \label{sec:sparc_const}

\begin{figure}
\centering
\includegraphics[width = 4.5in]{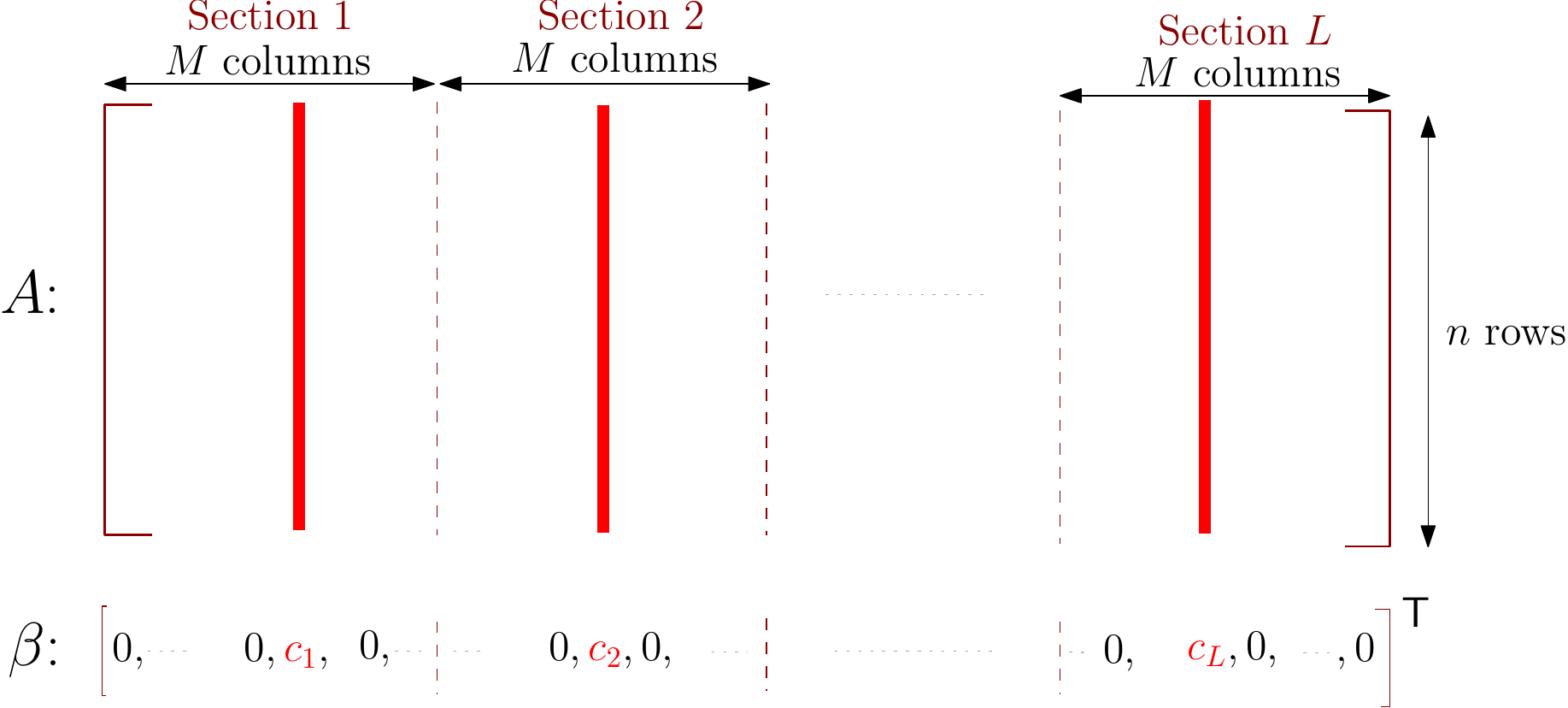}
\caption{\small{A Gaussian sparse regression codebook of block length $n$: ${A}$ is a design matrix with independent Gaussian entries, and $\beta$ is a sparse vector with one non-zero in each of $L$ sections. Codewords are of the form ${A}\beta$, i.e., linear combinations of the columns corresponding to the non-zeros in $\beta$. The message is indexed by the \emph{locations} of the non-zeros, and the values $c_1, \ldots, c_L$ are fixed a priori.}}
\vspace{-5pt}
\label{fig:sparc_codebook}
\end{figure}

As shown in  Fig. \ref{fig:sparc_codebook}, a SPARC is defined  in terms of a `dictionary' or design matrix $A$ of dimension $n \times ML$,  whose entries are chosen i.i.d. $\sim$ $\mc{N}(0,\tfrac{1}{n})$. Here $n$ is the block length, and $M, L$ are integers whose values are specified below in terms of $n$ and the rate $R$.  We think of the matrix $A$ as being composed of $L$ sections with $M$ columns each.  The variance of the entries ensures that the lengths of the columns of $A$ are close to $1$ for large $n$. \footnote{In some papers, the entries of  $A$ are assumed to be $\sim_{ i.i.d} \mc{N}(0,1)$. For consistency,  throughout this monograph we will assume that the entries are $\sim_{i.i.d.} \mc{N}(0,1/n)$.}

Each codeword is a linear combination of $L$ columns, with exactly one column chosen per section. Formally, a codeword can be expressed as  $A \beta$, where $\beta = (\beta_1, \ldots, \beta_{ML})^*$ is  a length $ML$ message vector with the following property:  there is exactly one non-zero $\beta_j$ for  $1 \leq j \leq M$, one non-zero $\beta_j$ for $M+1 \leq j \leq 2M$, and so forth.  We denote the set of valid message vectors by $\mcb$. Since each of the $L$ sections contains $M$ columns, the size of this set is 
\be \abs{\mcb}=M^L. \ee

The non-zero value of $\beta$ in section $\ell \in [L]$  is set to $c_\ell$, where the coefficients $\{ c_\ell \}$ are specified a priori. Since the entries of $A$ are i.i.d.\ $\mc{N}(0,\tfrac{1}{n})$, the entries of the codeword $A\beta$ are  i.i.d.  $\mc{N}(0,\tfrac{1}{n} \sum_{\ell=1}^L c^2_\ell)$.  In the case of AWGN channel coding, the variance $\frac{1}{n} \sum_{\ell=1}^L c^2_\ell$ is equal to the average symbol power.

\emph{Rate}: Since each of the $L$ sections contains $M$ columns, the total number of codewords is $M^L$. To obtain a  rate $R$ code, we need
\be
M^L = e^{nR} \quad \text{ or } \quad L \log M = nR.
\label{eq:ml_nR}
\ee
 There are several choices for the pair $(M,L)$ which satisfy \eqref{eq:ml_nR}. For example, $L=1$ and $M=e^{nR}$ recovers the Shannon-style random codebook in which the number of columns in $A$ is $e^{nR}$. For most of our constructions, we will often choose $M$ equal to $L^{\textsf{a}}$, for some constant $\textsf{a} >0$. In this case, \eqref{eq:ml_nR} becomes
 \be
 \textsf{a} L \log L = nR.
 \label{eq:llogl_nR}
 \ee
  Thus $L= \Theta(\tfrac{n}{\log n})$, and the size of the design matrix $A$ (given by $n \times ML = n \times L^{\textsf{a}+1}$)  grows polynomially in $n$.  In our numerical simulations, typical values for $L$ are $512$ or $1024$.

We note that the SPARC is a non-linear code with pairwise dependent codewords. Indeed, two codewords $A\beta$ and $A\beta'$ are dependent whenever the underlying  message vectors $\beta, \beta'$ share one or more common  non-zero entries.

\paragraph{Subset superposition coding} The SPARC described above has a partitioned structure, i.e., the message vector contains exactly one non-zero in each of the $L$ sections, with each section having $M$ entries. One could also define a non-partitioned SPARC,  where a message can be indexed by \emph{any} subset of $L$ entries of the length-$ML$ vector $\beta$. The number of codewords in this case would be ${ML \choose L}$, compared to  $M^L$ for the partitioned case. For a given pair $(M,L)$, the non-partitioned SPARC has a larger number of codewords. However,  using Stirling's formula we find that 
\[ \frac{\log {ML \choose L}}{ \log M^L} =1 + O\left( \frac{1}{\log M}\right). \]
Hence the ratio of the rates tends to $1$ as $M$ grows large.  
Though subset based (non-partitioned)  superposition codes have a small rate advantage for finite $M$, we focus on the partitioned structure in this monograph as it facilitates the design and analysis of efficient coding algorithms.

\section{Organization of the monograph} \label{sec:sparc_org}

\textbf{In Part I}, we focus on communication over the AWGN channel.  The performance of SPARCs with optimal (least-squares) decoding is analyzed in Chapter \ref{chap:AWGN_opt}. Though optimal decoding is infeasible, its performance provides a benchmark for  the computationally efficient decoders described in the next chapter.  It is shown that SPARCs with optimal encoding achieve the  AWGN capacity with an error exponent of the same order as Shannon's random coding ensemble. Similar results are also obtained for SPARCs defined via Bernoulli dictionaries rather than Gaussian ones.

\begin{figure}[t]
\centering
\includegraphics[width=0.49\textwidth, height=2in]{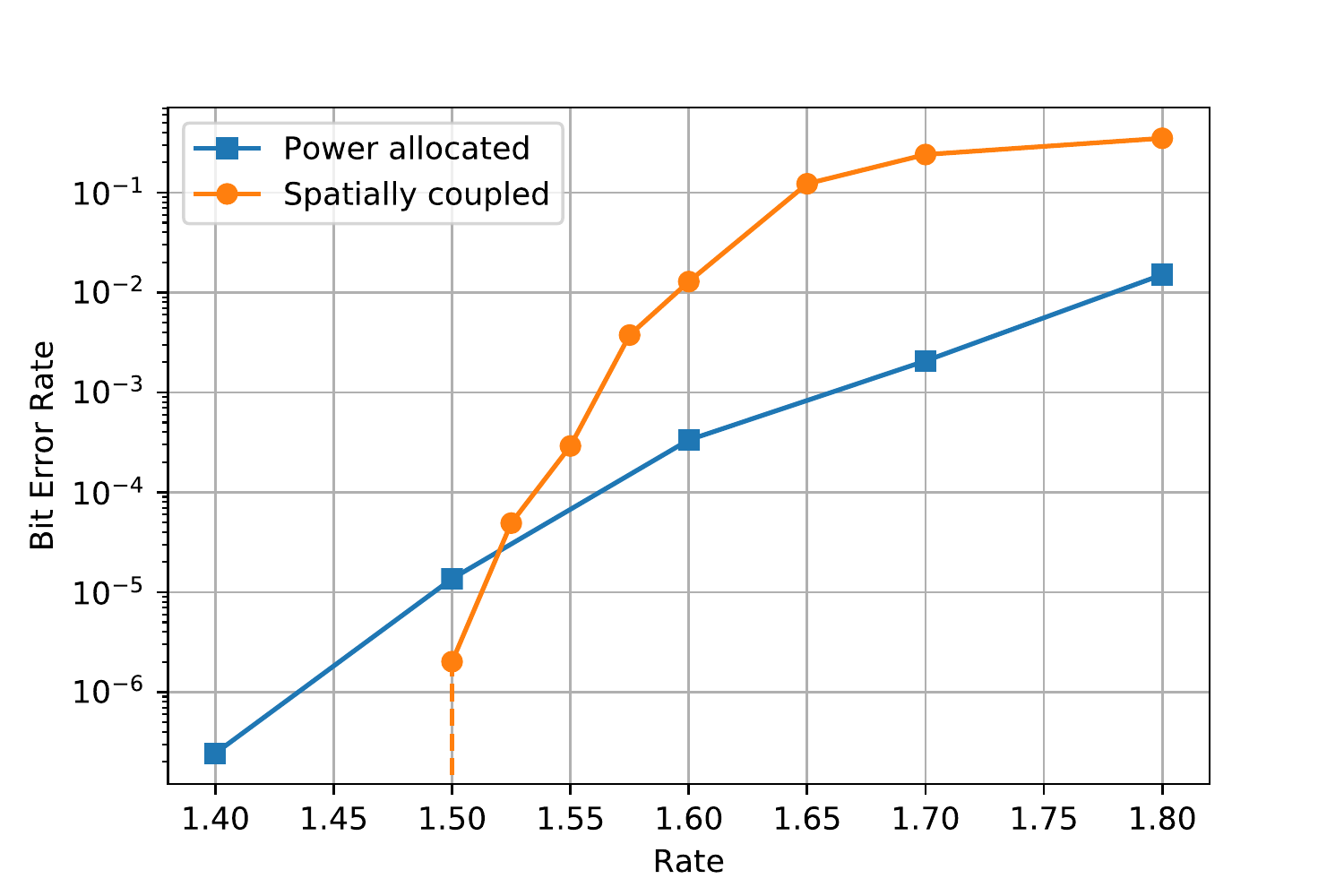}
\includegraphics[width=0.49\textwidth, height=2in]{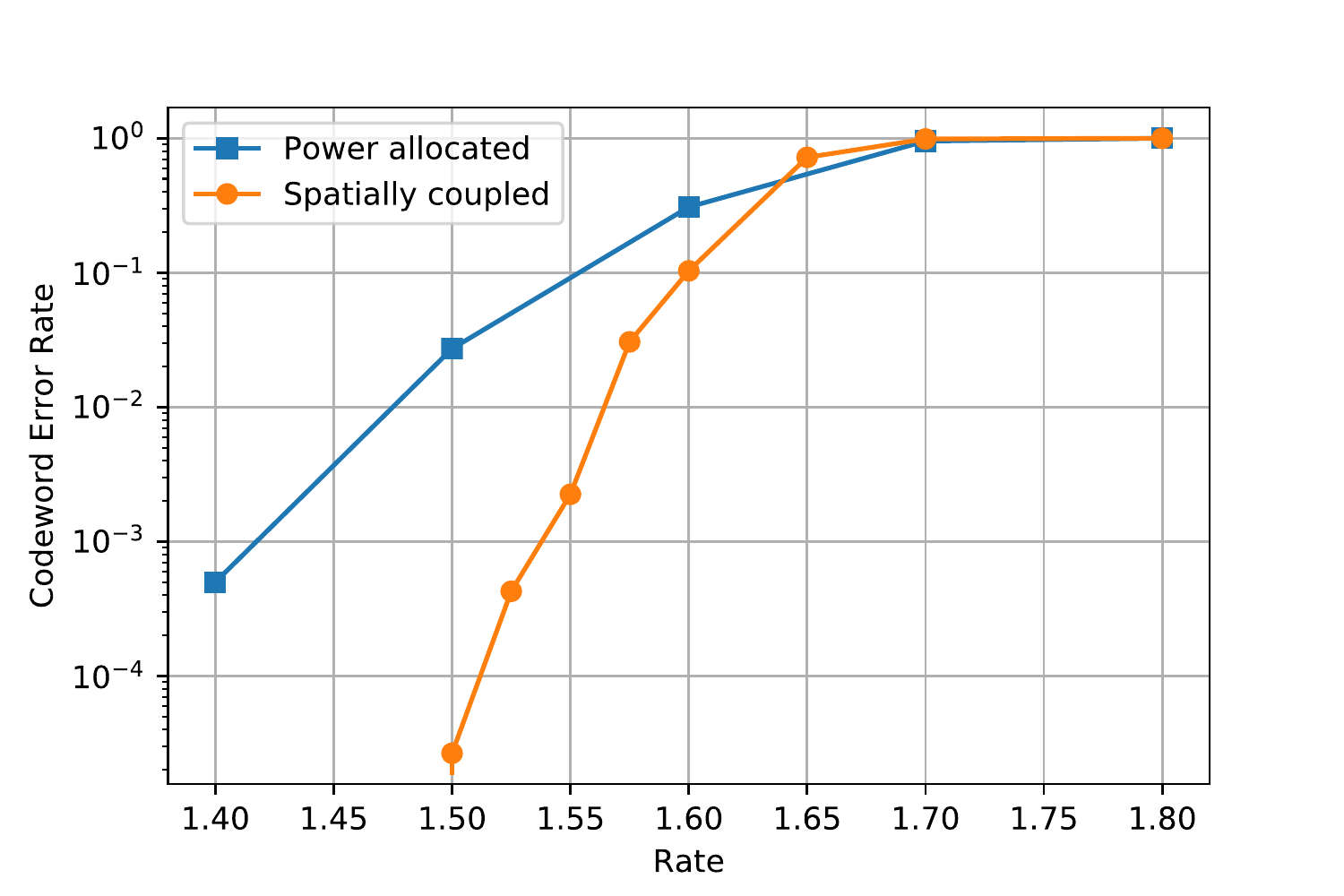}
\caption{ \small  Average bit error rate (left) and codeword error rate (right) vs. rate for SPARC over an AWGN channel with $\text{snr}=15$, $\mc{C}=2$ bits. The SPARC parameters are $\M=512$, $L=1024$, $n\in[5100,7700]$. Curves are shown for for power allocated SPARC (Chapter \ref{chap:emp_perf}) and spatially coupled SPARC (Chapter \ref{chap:sc_sparcs}). The different ways of measuring error rate performance in a SPARC are discussed in Chapter \ref{chap:AWGN_opt} (p.\pageref{eq:SERdef}).  The SPARC is decoded using the Approximate Message Passing (AMP) algorithm described in Chapters \ref{chap:AWGN_eff} and \ref{chap:sc_sparcs}.  
}
\label{fig:PASCsparcs}
\vspace{-5pt}
\end{figure}

 In Chapter \ref{chap:AWGN_eff}, we describe three efficient iterative decoders. These decoders generate an updated estimate of the message vector in each iteration based on a test statistic. The first decoder makes hard decisions,  decoding a few sections of the message vector $\beta$  in each iteration. The other two  decoders are based on soft-decisions, and generate new estimates of the whole message vector in each iteration. All three efficient decoders are asymptotically capacity-achieving, but the soft-decision decoders have better finite length error performance.  

In Chapter \ref{chap:emp_perf}, we turn our attention to techniques for improving the decoding performance at moderate block lengths. We observe that the power allocation (choice of the non-zero coefficients $\{c_\ell\}$) has a crucial effect on the finite length error performance.  We describe an  algorithm to determine a good power allocation, provide guidelines on choosing the parameters of the design matrix, and compare the empirical performance with coded modulation using LDPC codes from the WiMAX standard.  In Chapter \ref{chap:sc_sparcs}, we discuss spatially coupled SPARCs, which consist of several smaller SPARCs chained together in a band-diagonal structure.  An attractive feature of spatially coupled SPARCs is that  they are  asymptotically capacity-achieving and have good finite length performance without requiring a tailored power allocation.  Figure \ref{fig:PASCsparcs} shows the finite length error rate performance of power allocated SPARCs and spatially coupled SPARCs over an AWGN channel. The figure is discussed in detail in Sec. \ref{sec:sims}.

\textbf{In Part II} of the monograph, we use  SPARCs for lossy compression with the squared error distortion criterion. In Chapter  \ref{chap:comp_opt}, we analyze compression with optimal (least-squares) encoding, and show that SPARCs attain the optimal rate-distortion  function and the optimal excess-distortion exponent for i.i.d.  Gaussian sources.  We then describe an efficient successive cancellation encoder  in Chapter \ref{chap:comp_eff_enc}, and show that it achieves the optimal Gaussian rate-distortion function, with the probability of excess distortion decaying exponentially in the block length.

\textbf{In Part III}, we design rate optimal coding schemes using SPARCs for a few canonical models in multiuser information theory. In Chapter \ref{chap:bcmac}, we show how SPARCs designed for point-to-point AWGN channels can be combined to construct rate-optimal superposition coding schemes for the AWGN broadcast and multiple-access channels. In Chapter 
\ref{chap:sideinfo}, we show how to implement random binning using SPARCs. Using this, we can nest the channel coding and source coding SPARCs constructed in Parts I and II to construct rate-optimal schemes for a variety of problems in multiuser information theory.   We conclude in Chapter \ref{chap:summary} with a discussion of open problems and directions for future work. 

Proofs or proof sketches for the main results in a chapter are given at the end of the chapter.   The proofs of some intermediate lemmas are omitted, with pointers to the relevant references. The goal is to describe the key technical ideas in the proofs, while not impeding the flow within the chapter.

\part{AWGN Channel Coding with SPARCs}
\chapter{Optimal Decoding} \label{chap:AWGN_opt}
%

In this chapter, we consider sparse regression codes for the additive white Gaussian noise (AWGN) channel, and analyze the performance under optimal (maximum-likelihood) decoding.  Though the optimal decoder has computational complexity that grows exponentially with $n$,  its decoding performance sets a benchmark for the efficient decoders discussed in the next chapter. The results in this chapter show that the SPARC error probability with optimal decoding decays exponentially with $n$ for any rate $R$ less than the AWGN channel capacity. In particular, we will see that the error probability bound for SPARCs has the same form as the bound for a Shannon-style random codebook consisting of independent Gaussian codewords \cite{gallagerBook,ppv10}, but with a weaker constant.  We note that a Shannon-style random codebook is infeasible except for very short code lengths as the complexity of encoding and decoding grow exponentially with $n$. 

\section{Problem set-up} \label{sec:prob_setup}

\paragraph{Channel Model} The discrete-time AWGN channel is described by the model
\be
y_i  = x_i + w_i, \quad i=1,\ldots,n.
\ee
That is, the channel output $y_i$ at time instant $i$ is the sum of the channel input $x_i$ and the Gaussian noise variable $w_i$. The random variables $w_i, \ 1 \leq i \leq n$, are i.i.d. $\sim \mc{N}(0, \sigma^2)$.   There is an average power constraint $P$ on the input: the codeword $x=(x_1,\ldots, x_n)$ should satisfy $\frac{1}{n} \sum_{i=1}x_i^2 \leq P$.  The signal-to-noise ratio $\frac{P}{\sigma^2}$ is denoted by $\snr$.

We wish to use the sparse regression codebook described in Section \ref{sec:sparc_const} to communicate reliably at any rate $R < \mc{C}$, where the channel capacity $\mc{C} = \frac{1}{2} \log (1+ \snr) $.

\paragraph{Power Allocation} We need to specify the non-zero coefficients $c_1, \ldots, c_L$  in the message vector   so as to satisfy the power constraint. Recall that the entries of each codeword $A\beta$ are  i.i.d.  $\mc{N}(0,\tfrac{1}{n} \sum_{\ell=1}^L c^2_\ell)$. In this chapter, we consider the flat power allocation with
\ben
c_1 = c_2 \ldots = c_L= \sqrt{\frac{nP}{L}}.
\een
This choice ensures that for a message $\beta \in \mcb$, the expected codeword power, given by $ \expec \norm{A\beta}^2/n$, equals $P$.  Using standard large deviations techniques, it can be shown that the distribution of the average codeword power $\norm{A \beta}^2/n$ is tightly concentrated around $P$ \cite[Appendix B]{AntonyMLarxiv}.

In the next  two chapters, we will consider different power allocations where the coefficients $c_1, \ldots, c_L$ are not equal to one another. One example is the exponentially decaying allocation, where $c_\ell \propto e^{- \mc{C} \ell/L}$,  for $\ell \in [L]$.  As we will see, such power allocations facilitate computationally feasible decoders that are reliable at  rates close to capacity.

\paragraph{Encoding} The encoder splits its stream of input bits into segments of $\log M$ bits each. A length $ML$  message  vector $\beta$ is indexed by $L$ such segments --- the decimal equivalent of segment $\ell$ determines the position of the non-zero coefficient in section $\ell$ of $\beta$. The input codeword is then computed as $x=A\beta$. Note that computing $x$ simply involves adding  $L$ columns of $A$, weighted by the appropriate coefficients.

\paragraph{Maximum Likelihood Decoding}  Assuming that the messages are equally likely is equivalent to assuming  a uniform prior over $\mcb$ for the message vector $\beta$. Then the  decoder that minimizes the probability of message decoding error is  the maximum likelihood decoder. We will refer to the maximum likelihood decoder as the optimal decoder. Given the channel output sequence $y=(y_1, \ldots, y_n)$, the optimal decoder produces
  \be \hat{\beta}_{\textsf{opt}} = \argmin_{\hat{\beta} \in \mcb} \, \norm{y-A\hat{\beta}}^2. \label{eq:opt_decoder} \ee

  \paragraph{Probability of Decoding Error}  A natural performance metric for a SPARC decoder is  the \emph{section error rate}, which is the fraction of sections decoded incorrectly. If the true message vector is $\beta$ and  the decoded message vector is $\hat{\beta}$, the section error rate is defined as 
  \be \label{eq:SERdef}
  \Esec = \frac{1}{L} \sum_{\ell =1}^{L}  \mathbf{1}\{ \hat{\beta}_\ell \neq \beta_{\ell} \}, 
  \ee where  $\beta_{\ell}, \hat{\beta}_{\ell} \in \mathbb{R}^M$ denote the $\ell$th  section of $\beta, \hat{\beta}$, respectively. We will first aim to bound the probability of excess section error rate, i.e., the probability of the event 
 $\{ \Esec  \geq \e \}$, for  $\e >0$.

 Assuming that the mapping that determines  the  non-zero location  within a section for each segment of $\log M$ input bits is generated uniformly at random, a section error will,  on average, lead to half the bits corresponding to the section being decoded wrongly. Therefore, when a large number of segments are transmitted, the \emph{bit error rate} of a SPARC decoder will be close to half its section error rate. 
 
 Finally, one may also wish to minimize the probability of codeword (or message) error, i.e.,  $\pr(\bth \neq \beta )$. For this,  one can use a concatenated code with the SPARC as the inner code and  an outer Reed-Solomon (RS) code.   Later in this chapter (see p. \pageref{eq:dRS}), we describe how an RS code of rate $(1-2\e)$  can be used to ensure that $\bth =\beta$ whenever the section error rate $\Esec < \e$, for any $\e >0$. With a SPARC of rate $R$, such a concatenated code has rate $R(1-2\e)$ and its  probability of codeword error is bounded by $\pr(\Esec \geq \e)$.  The main result of this chapter, Theorem \ref{thm:MLresult}, shows that  $\pr(\Esec \geq \e)$ decays exponentially in $n$ for any $R < \mc{C}$. 
 
 \section{Performance of the optimal decoder}
 
The goal is to  obtain bounds on the probability of excess section error rate, averaged over all messages and over the space of design matrices. More precisely,   for any $\e >0$, we wish to bound
 \be
 \begin{split}
 \pr(\Esec \geq \e ) &= \expec_{A,\beta} \left[ \pr( \Esec \geq \e \mid A, \beta) \right] \\
 & =  \frac{1}{M^L} \sum_{\beta \in \mcb} \expec_A \left[ \pr( \Esec \geq \e \mid A, \beta) \right],
 \end{split}
 \label{eq:Perrbeta}
 \ee
where the subscripts indicate the random variable(s) the expectation is computed over.  In \eqref{eq:Perrbeta}, we note that the probability measure  on $A$ is that induced by its i.i.d. $\mc{N}(0,1)$ entries. By symmetry, $\expec_A \left[ \pr( \Esec \geq \e \mid A, \beta) \right]$ is the same for all $\beta \in \mcb$. Therefore we shall obtain bounds for 
\be \pr_{\beta_0}(\Esec \geq \e) = \expec_A \left[ \pr( \Esec \geq \e \mid A, \beta_0) \right], \label{eq:excess_ser} \ee 
for a fixed message vector $\beta_0 \in \mcb$.

\paragraph{Preliminaries} We list some facts and definitions that will be used in the bounds.  

If $Z, \tilde{Z}$ are jointly Gaussian random variables with means equal to 0, variances equal to 1, and correlation coefficient $\rho$, then we have the following Chernoff bound for the difference of their squares. For any $\Delta >0$, 
\be
\pr \left( \frac{1}{2}(Z^2 -\tilde{Z}^2) > \Delta \right) \leq \exp\left(-D(\Delta, 1-\rho^2) \right),
\label{eq:chi_sq_diff}
\ee
where the Cram{\'e}r-Chernoff large deviation exponent  is
\be \label{eq:Ddef}
D(\Delta, 1-\rho^2) = \max_{\lambda \geq 0} \left\{ \lambda \Delta + \frac{1}{2} \log\left( 1- \lambda^2(1-\rho^2) \right) \right\}.
\ee
We also define 
\be \label{eq:D1def}
D_1(\Delta, 1-\rho^2) = \max_{0 \leq \lambda \leq 1} \left\{ \lambda \Delta + \frac{1}{2} \log\left( 1- \lambda^2(1-\rho^2) \right) \right\}.
\ee
Finally, for $0 \leq \alpha \leq 1$, let 
\be
\label{eq:calph_def}
\mc{C}_{\alpha} = \frac{1}{2}\log(1 + \alpha  \, \snr).
\ee
Recalling that the capacity $\mc{C} = \mc{C}_{1}$, we note that $\mc{C}_{\alpha} - \alpha \mc{C}$ is a concave function equal to $0$ when $\alpha$ is $0$ or $1$, and strictly positive in between. 

\paragraph{Error probability bounds} The first result is a non-asymptotic bound on the probability of excess section error rate defined in \eqref{eq:excess_ser}.

\begin{proposition} \cite[Eq. (24)]{AntonyML}
 \label{prop:nonasympML}
For any $\beta_0 \in \mcb$ and $\e >0$, 
\be
\pr_{\beta_0}(\Esec \geq \e)  \leq \sum_{\ell = \e L}^{L} \min \left\{ \textsf{err}_1(\ell/L),  \  \textsf{err}_2(\ell/L) \right\},
\label{eq:Pbeta0_bnd}
\ee
where for $0 <  \alpha \leq 1$, the functions $\textsf{err}_1(\alpha)$ and $\textsf{err}_2(\alpha)$ are defined as follows.
\begin{align}
\textsf{err}_1(\alpha) &  = {L \choose \alpha L} \exp \left(  - n D_1\left( \mc{C}_{\alpha} - \alpha R, \, \frac{\alpha \, \snr}{1+ \alpha \, \snr} \right) \right), \\
\textsf{err}_2(\alpha) & = \min_{t_\alpha  \in [0, C_\alpha - \alpha R]} \ \textsf{err}_2(\alpha, t_{\alpha}),
\end{align}
where
\begin{align}
\textsf{err}_2(\alpha, t_\alpha) =  &  {L \choose \alpha L} \exp \left(  - n D_1\left( \mc{C}_{\alpha} - \alpha R - t_\alpha, \, \frac{\alpha(1-\alpha) \, \snr}{1+ \alpha \, \snr} \right) \right) \nonumber \\
& + \exp \left(  - n D\left( t_\alpha, \, \frac{\alpha^2\, \snr}{1+ \alpha^2 \, \snr} \right) \right). \label{eq:nonasympML}
\end{align}
\end{proposition}

The proof of the proposition is given in Section \ref{subsec:proof_main_prop}. 

The bound in \eqref{eq:Pbeta0_bnd} can be computed numerically  given the rate $R$ and the SPARC parameters $(M,L)$.  The function $\textsf{err}_1(\alpha)$ gives the better bound for $\alpha$ close to $0$, while $\textsf{err}_2(\alpha)$ is better for $\alpha$ close to $1$.

The next result simplifies the non-asymptotic bound and shows that the probability of excess section error rate decays exponentially in $n$, with the exponent depending on the gap from capacity $\Delta= \mc{C} -R$. First, a few definitions that are needed to state the result.  

For $x>0$, let
\be
\label{eq:gx_def}
g(x)=\sqrt{1 + 4x^2} -1.
\ee
It follows that
\be
g(x) \geq \min\{ \sqrt{2}x, x^2 \} \  \text{ for all } \ x \geq 0.
\label{eq:gx_lb}
\ee
Next, let 
\be
w(\snr) = \frac{\snr}{2 (1 + \snr)^2 \sqrt{4+ \snr^3/(1 + \snr) }}.
\label{eq:wsnr_def}
\ee
Finally, let  $\msf{a}^*_L(\snr)$ be defined as 
\be
\msf{a}^*_L(\snr) = \max_{\alpha \in \{\frac{1}{L}, \ldots, 1-\frac{1}{L} \}} 
\frac{ R \log  {L \choose L\alpha}}{ D_1\left(\mc{C}_\alpha - \alpha \mc{C}, \frac{\alpha(1-\alpha) \snr}{1 + \alpha \snr} \right) L \log L},
\label{eq:astar_def}
\ee
where the $D_1$ is defined in \eqref{eq:D1def}. The behavior of $\msf{a}^*_L(\snr)$ as $L \to \infty$ is described in Remark \ref{rem:al}.

\begin{theorem} \cite{AntonyML}
Assume that $M=L^{\msf{a}}$, where $\msf{a} \geq \msf{a}^*_L(\snr)$, and that the  gap from capacity $\Delta=(\mc{C}-R)$ is strictly positive.  Then, for any $\e >0$, the section error rate $\Esec$ of the optimal decoder satisfies
\be
\label{eq:P_SER_ML_bnd}
\pr\left( \Esec \geq \e \right) = e^{-nE(\e, R)},
\ee
with
\be
\label{eq:heD}
E(\e , R) \geq h(\e, \Delta) - \frac{\log 2L}{n},
\ee
where
\be
h(\e, \Delta)= \min \left\{ \e  \Delta w(\snr), \ \frac{1}{4} g\left( \frac{\Delta}{2 \sqrt{\snr}} \right) \right\}.
\label{eq:hedel_def}
\ee
\label{thm:MLresult}
\end{theorem}
The proofs of the theorem is based on Proposition \ref{prop:nonasympML}, and is given in Section \ref{subsec:proof_GaussML}.  

\begin{remark}
The lower bound on $g(x)$ in \eqref{eq:gx_lb} implies that the function $h(\e, \Delta)$ can be bounded from below as
\[ h(\e, \Delta) \geq \min \left \{ \e \Delta w(\snr), \, \frac{\Delta}{\sqrt{32 \snr}}, \, \frac{\Delta^2}{16 \snr} \right \}, \]
revealing that the exponent is, up to a constant, of the form $\min\{ \e \Delta, \Delta^2 \}$.

An improved lower bound on the exponent $E(\e, R)$ is obtained in \cite[Appendix C]{AntonyML}. This lower bound replaces the function $h(\e, \Delta)$ with a larger function $\tilde{h}(\e, \Delta)$, and shows that the exponent is of the form $\min\{ \e, \Delta^2 \}$.
\end{remark}

\begin{remark} The parameter $\msf{a}^*_L(\snr)$ approaches the following limiting value as $L \to \infty$. Let $v^*$ near 15.8 be the solution to the equation
$ (1+ v^*)  \log(1 + v^*)=3v^*$.
Then \cite[Lemma 5]{AntonyML} shows that
\be
\lim_{L \to \infty} \msf{a}^*_L(\snr) =
\left\{ 
\begin{array}{cc}
\frac{8R \, \snr \, (1+ \snr)}{ [ (1+ \snr) \log(1+\snr)- \snr ]^2} & \text{ for }\snr < v^*, \\
\frac{2R (1+\snr)}{[ (1+\snr) \log(1+\snr) -2 \, \snr ]} & \text{ for }\snr \geq v^*.
\end{array}
\right.
\ee
Taking the upper bound of $\mc{C}$ for $R$, it can be shown that the above limit is  approximately $16/\snr^2$ for small values of $\snr$, and $1$ for large $\snr$.
\label{rem:al}
\end{remark}

\paragraph{Probability of message error} Using a suitable outer code,  the bound on the probability of excess section error  in Theorem \ref{thm:MLresult} can be translated into a bound on the probability of  message error, i.e.,  $\pr(\hat{\beta} \neq\beta)$.  

Consider a concatenated code with a SPARC of  rate $R$ as the inner code, and a Reed-Solomon (RS) outer code, chosen as follows.
 For simplicity,  assume that $M=2^m$.  We consider a systematic $(n_{out}, k_{out})$ RS code with symbols in $GF(2^m)$. From the theory of RS codes \cite{blahut03algebraic,linCost04}, we can take 
 \be n_{out} =M, \quad k_{out}= \lceil (1-\e) M \rceil. \label{eq:nk_outer} \ee
 to obtain an RS code with minimum distance
 \be
 d_{\textsf{RS}}=M -  \lceil (1-\e)  M \rceil +1 \ \text{ symbols}.
 \label{eq:dRS}
 \ee
 The information bits are encoded into the SPARC codeword as follows.  First consider the case where $L=M$. Here, the RS encoder maps $k_{out} m$ information bits  ($k_{out}$ symbols) into a length $L$ RS codeword.   Since each SPARC section has $M$ columns, each symbol of the RS encoder represents the index for one section of the SPARC.  For the case where $L < M$, we can use the same procedure by setting the first $(M-L)$ symbols of the systematic RS codeword to 0.
 
 From \eqref{eq:dRS}, the number of symbol errors that this code is guaranteed to correct in a length $n_{out}$ codeword is
\[ \left \lfloor   \frac{d_{\textsf{RS}}}{2}  \right \rfloor  \geq \lfloor \e M \rfloor. \]
Therefore, the decoded message $\hat{\beta}$ equals  the transmitted one $\beta$ whenever the optimal SPARC decoder makes no more than $\lfloor \e M \rfloor$ section errors.  Therefore, the probability of message error for the concatenated code is bounded by the RHS of \eqref{eq:P_SER_ML_bnd}. From \eqref{eq:nk_outer}, the rate of the concatenated code is at least $R(1-2\e)$, where $R$ is the rate of the  SPARC.

We therefore have the following result. 
\begin{proposition} \cite{AntonyML}
Consider a SPARC with rate $R < \mc{C}$, with parameters $(M,L)$ satisfying the assumptions of Theorem \ref{thm:MLresult}. Then for any $\e >0$, through concatenation with an outer RS code, one obtains a code of rate $R(1-2\e)$ with message error probability bounded by  $e^{-nE(\e,R)}$. Here $E(\e,R)$ is the exponent from Theorem \ref{thm:MLresult}, which can be bounded from below as in \eqref{eq:heD}.
\label{thm:msg_error}
\end{proposition}

\begin{remark}
Consider the regime where the SPARC rate $R$ is made to approach $\mc{C}$  as $R= C-\Delta_n$. Let $\Delta_n$ tend to zero at a rate slower than $1/\sqrt{n}$, e.g., $\frac{1}{n^{1/4}}$ or $\frac{1}{\log n}$). Then choosing $\e = \Delta_n$, Proposition \ref{thm:msg_error} shows that we have a code whose overall rate is $(\mc{C}- \Delta_n)(1-\Delta_n)$ and whose probability of message error decays as $\exp(-\kappa n \Delta_n^2)$, where $\kappa$ is a universal positive constant. 
\label{rem:ML_gap_to_cap}
\end{remark}

\section{Performance with i.i.d. Bernoulli dictionaries} \label{sec:bern_dict}

SPARCs defined via an i.i.d. Gaussian design matrix are not suitable for practical implementation, especially for large code lengths. Large Gaussian design matrix have prohibitive storage complexity as the entries will span a large range of real numbers which need to be stored with high precision.  To reduce the storage requirement, one could define the SPARC via a Bernoulli design matrix entries are chosen uniformly at random from the set $\{1, -1 \}$.  As before, the set of valid message vectors is $\mcb$, i.e.,  $\beta$ which have one non-zero in each of the $L$ sections. In this section, we consider Bernoulli-defined SPARCs with equal power allocation, i.e., each non-zero entry of $\beta$ equals $\sqrt{P/L}$. Each entry of the codeword is therefore a sum of $L$ i.i.d. random variables, each drawn uniformly from $\left\{\sqrt{{P}/{L}}, -\sqrt{{P}/{L}} \right\}$. Therefore, by the central limit theorem, each codeword entry converges in distribution to an i.i.d. $\mc{N}(0, P)$ random variable.

The performance of Bernoulli dictionaries with optimal decoding was analyzed by Takeishi et al. \cite{takeishi14,takeishi2016improved}. The main result, stated below, gives an error probability bound that is almost identical to the one in Theorem \ref{thm:MLresult} for the Gaussian case, except for a slightly weaker exponent. 

\begin{theorem} \cite{takeishi2016improved}
With the same assumptions and notation as in Theorem \ref{thm:MLresult}, the section error rate of a SPARC defined with a Bernoulli dictionary satisfies 
\be
\pr\left( \Esec \geq \e \right) = e^{-nE(\e, R)},
\ee
with
\be
\label{eq:heD_Bern}
E(\e , R) \geq h(\e, \Delta) - \iota(L) ,
\ee
where $\iota(L)=O(1/\sqrt{L})$. 
\label{thm:Bern_ML}
\end{theorem}
We note that the only difference from result for the Gaussian case is that in the lower bound for $E(\e, R)$, the $\log (2L)/n$ term is now replaced with $\iota(L)=O(1/\sqrt{L})$.  A proof sketch of the proposition is given in Section \ref{subsec:bern_proof}.

\section{Proofs} \label{sec:proofML}

\subsection{Proof of Proposition \ref{prop:nonasympML} }\label{subsec:proof_main_prop}

We obtain \eqref{eq:Pbeta0_bnd} by proving that for $\ell \in \{1, \ldots, L\}$, 
\begin{align}
 \pr_{\beta_0}(\Esec = \ell/L)  & \leq  \textsf{err}_1(\ell/L), \label{eq:err1_ell} \\
  \pr_{\beta_0}(\Esec = \ell/L)  & \leq  \textsf{err}_2(\ell/L),  \label{eq:err2_ell}
\end{align}
where $\textsf{err}_1(\cdot )$ and $\textsf{err}_2(\cdot )$ are defined in the statement of the proposition.

For any $\beta \in \mcb$, let 
$\mc{S}(\beta)= \{ j: \, \beta_j=1 \}$ denote the set of non-zero indices. Let $\mc{S}^* =\mc{S}(\beta_0)$ denote the set of non-zero indices for the true message vector $\beta_0$, and let $\mc{S}$ denote the set of non-zero indices in the decoded message vector  When there are $\ell$ section errors, the set $\mc{S}$ differs from  $\mc{S}^*$ in exactly $\ell$ elements. Letting $X_{\mc{S}} = A \beta_{\mc{S}}$ and $X_{\mc{S}^*} = A \beta_{\mc{S}^*} =  A \beta_{0}$, the ML decoder decodes $\mc{S}$ only when the received vector $Y$ satisfies $\norm{Y- X_{\mc{S}}}^2  \leq  \norm{Y- X_{\mc{S}^* }}^2$, or equivalently, when $T(\mc{S}) \leq 0$, where
\be
T(\mc{S}) = \frac{1}{2n} \left[ \frac{\normsq{Y- X_{\mc{S}}} }{\sigma^2} -   \frac{\normsq{Y- X_{\mc{S}^*}} }{\sigma^2} \right].
\label{eq:TSdef}
\ee

The analysis proceeds by  obtaining a  bound for  $\pr_{\beta_0}(T(\mcs) \leq 0 )$ that holds for each choice of $\mcs$.   Noting that there are ${L \choose \ell } M^\ell$ choices for $\mc{S}$, the natural way to combine these is via a union bound:
\[   \pr_{\beta_0}(\Esec = \ell/L)  = {L \choose \ell } M^\ell \,  \pr_{\beta_0}(T(\mcs) \leq 0 ).   \]
However, such a union bound gives a result weaker than that of Proposition \ref{prop:nonasympML}.  Therefore, we obtain \eqref{eq:err1_ell} and  \eqref{eq:err2_ell} by decomposing $T(\mcs)$ in two different ways and using a modified union bound.

\paragraph{Proof of \eqref{eq:err1_ell}  \cite[Lemma 3]{AntonyML}:}   Let $\mcs_1 = \mcs \cap \mcs^*$ and $\mcs_2= \mcs- \mcs_1$  denote, respectively,  the intersection and difference between sets $\mcs$ and $\mcs^*$. With $\alpha = \ell /L$, the sizes of $\mcs_2$ and $\mcs_1$ are $\alpha L$ and $L - \alpha L$, respectively.  With this notation, the probability of the event $\{ \Esec = \ell/L \}$ can be bounded as follows. For any $\lambda >0$, the indicator of the event satisfies
\be
\mbf{1}\{ \Esec = \ell/L   \} \leq \sum_{\mcs_1} \Big( \sum_{\mcs_2} e^{-nT(\mcs)} \Big)^\lambda.
\label{eq:TSlamb}
\ee
We decompose the test statistic $T(\mc{S})$ as $T_1 + T_2$, where
\be
T_1= \frac{1}{2n} \left[ \frac{\normsq{Y- X_{\mc{S}_1}} }{\sigma^2 + \alpha P} -   \frac{\normsq{Y- X_{\mc{S}^*}} }{\sigma^2} \right],
\label{eq:T1_def}
\ee
and
\be
T_2 = \frac{1}{2n} \left[ \frac{\normsq{Y- X_{\mc{S}}} }{\sigma^2} -   \frac{\normsq{Y- X_{\mc{S}_1}} }{\sigma^2 + \alpha P} \right].
\label{eq:T2_def}
\ee
Observing that $T_1 $ depends only on the indices in $\mcs^*$ (and not on those in $\mc{S}_2$), we take  expectations on both sides of  \eqref{eq:TSlamb} to write
\begin{align}
 \pr_{\beta_0}(\Esec = \ell/L)  & \leq \sum_{\mcs_1} \expec e^{-n\lambda T_1(\mcs_1)} \, 
 \expec_{ X_{\mcs_2} | Y, X_{\mcs_1}, X_{\mcs^*} } \Big( \sum_{\mcs_2} e^{-nT_2(\mcs)} \Big)^\lambda \nonumber \\
 & \stackrel{(a)}{\leq}   \sum_{\mcs_1} \expec e^{-n\lambda T_1(\mcs_1)} \Bigg[ \sum_{\mcs_2}  \expec_{X_{\mcs_2}}  e^{-nT_2(\mcs)} \Bigg]^\lambda \nonumber \\
& \stackrel{(b)}{=} \sum_{\mcs_1} \expec e^{-n\lambda T_1(\mcs_1)}  \Bigg[ \sum_{\mcs_2}  \left( 1 + \frac{\alpha P}{\sigma^2} \right)^{-\frac{n}{2}} \Bigg]^{\lambda},  \nonumber \\
& \stackrel{(c)}{\leq} \sum_{\mcs_1} \expec e^{-n\lambda T_1(\mcs_1)} \, e^{-n\lambda(\mc{C}_\alpha - \alpha R)}.
\label{eq:EnlamT}
\end{align}
where $(a)$ is obtained using Jensen's inequality (arranging for $\lambda$ to be not more than 1), and noting that  $X_{\mcs_2}$ is  independent of 
$(Y, X_{\mcs_1}, X_{\mcs^*})$. Step $(b)$ is obtained by writing
\[ 
e^{-nT_2(\mcs)} = \exp\left( -\frac{1}{2} \Bigg[ \frac{\normsq{Y- X_{\mc{S}_1} -  X_{\mc{S}_2}  } }{\sigma^2} -   \frac{\normsq{Y- X_{\mc{S}_1}} }{\sigma^2 + \alpha P}  \Bigg] \right), 
\]
and evaluating the expectation with respect to $X_{\mcs_2}$, which is  i.i.d. $\sim \mc{N}(0, \alpha P)$. For step $(c)$, we observe that the sum over $\mcs_2$ involves at most $M^\ell = e^{nR \alpha}$ terms, and use the definition of $\mc{C}_\alpha$ from \eqref{eq:calph_def}.

We note from \eqref{eq:T1_def} that $T_1(\mcs_1)$ is distributed as 
\[ T_1(\mcs_1) \stackrel{d}{=}  \frac{1}{2n} \sum_{i=1}^n\left( Z_i^2 - \tilde{Z}_i^2 \right),  \]
where  each pair  $(Z_i, \tilde{Z}_i)$ is bivariate Gaussian with squared correlation equal to $1/(1+ \alpha \snr)$. The pairs are i.i.d. for
$1\leq i \leq n$. Using this, the expectation in  \eqref{eq:EnlamT} is found to be
\[ 
\expec e^{-n\lambda T_1(\mcs_1)} = (1- \lambda^2 \alpha \snr/(1+\alpha \snr))^{-n/2}.
\]
The proof is completed by using this in \eqref{eq:EnlamT},  noting that the sum over $\mcs_1$ has ${L \choose L\alpha}$ terms, and optimizing the bound over $\lambda \in [0,1]$.


\paragraph{Proof of \eqref{eq:err2_ell}  \cite[Lemma 4]{AntonyML}:}  For any $\mcs$ which differs from $\mcs^*$ in $\ell$ sections, we decompose the test statistic  in \eqref{eq:TSdef} as $T(\mcs) = \tilde{T}(\mcs) +T^*$, where
\begin{align}
\tilde{T}(\mcs)  &=  \frac{1}{2n} \left[ \frac{\normsq{Y- X_{\mc{S}}} }{\sigma^2} -   \frac{\normsq{Y- (1-\alpha)X_{\mc{S}^*}} }{\sigma^2 + \alpha^2 P} \right], \label{eq:Ttil_def}\\
T^* & = \frac{1}{2n} \left[ \frac{\normsq{Y- (1-\alpha)X_{\mc{S}^*}} }{\sigma^2 + \alpha^2 P} - \frac{\normsq{Y- X_{\mc{S}^*}}}{\sigma^2}  \right]. \label{eq:Tstar_def}
\end{align}
Let $t_\alpha \in [0, \mc{C}_\alpha-\alpha R]$. Then, 
\be 
\pr_{\beta_0} (\Esec = \ell/L )  \leq  \pr_{\beta_0}(\exists \, \mcs :  \tilde{T}(\mcs)  \leq t_\alpha) +   \pr_{\beta_0}( T^* \leq -t_\alpha). 
\label{eq:tilT_Tstar_decomp}
\ee
We note that $T^*$ does not depend on $\mcs$: it is a mean zero average of the difference of squared Gaussian random variables, with squared correlation $1/(1+ \alpha^2 \snr)$. The second term on the RHS of \eqref{eq:tilT_Tstar_decomp} can therefore be bounded via a Chernoff bound (as in \eqref{eq:chi_sq_diff}) to obtain the second term in \eqref{eq:nonasympML}.

The analysis of the first term in \eqref{eq:tilT_Tstar_decomp} is very similar to the proof of \eqref{eq:err1_ell} above. We write $\tilde{T}(\mcs) = \tilde{T}_1 + \tilde{T}_2$, where 
\begin{align}
\tilde{T}_1   & = \frac{1}{2n} \left[ \frac{\normsq{Y- X_{\mc{S}_1}} }{\sigma^2 + \alpha P} -   \frac{\normsq{Y- (1-\alpha) X_{\mc{S}^*}} }{\sigma^2 + \alpha P} \right],  \nonumber \\
\tilde{T}_2 & = \frac{1}{2n} \left[ \frac{\normsq{Y- X_{\mc{S}}} }{\sigma^2} -   \frac{\normsq{Y- X_{\mc{S}_1}} }{\sigma^2 + \alpha P} \right].
\label{eq:t1_til}
\end{align}
The key difference between $\tilde{T}_1$  and $T_1$ (defined in \eqref{eq:T1_def}) is that  the two standard normals have a higher correlation coefficient in $\tilde{T}_1$. Indeed, the squared correlation coefficient between the two standard normals in \eqref{eq:t1_til} is $\rho_\alpha^2 = (1+\alpha^2 \snr)/(1+ \alpha \snr)$, and the moment generating function is found to be
\[ 
\expec e^{-n\lambda \tilde{T}_1(\mcs_1)} = (1- \lambda^2 \alpha (1-\alpha) \snr/(1+\alpha \snr))^{-n/2}.  
\]
Following steps similar to \eqref{eq:EnlamT} yields the second term in \eqref{eq:nonasympML}.

This proves \eqref{eq:err2_ell}, and therefore  Proposition \ref{prop:nonasympML}. \qed


\subsection{Proof of Theorem \ref{thm:MLresult}} \label{subsec:proof_GaussML} 
To prove Theorem \ref{thm:MLresult}, we use the following weaker bound implied by Proposition \ref{prop:nonasympML}: $\pr_{\beta_0}(\Esec \geq \e)  \leq \sum_{\ell = \e L}^{L} \, \textsf{err}_2(\ell/L)$, where $\textsf{err}_2(\cdot)$ is defined in \eqref{eq:nonasympML}.   Let
\begin{align} 
\Delta_\alpha=  \mc{C}_{\alpha} - \alpha R - t_\alpha, \quad \tilde{\Delta}_\alpha=\mc{C}_\alpha - \alpha \mc{C}, \quad
1- \rho_\alpha^2 =\frac{\alpha(1-\alpha) \snr}{1 + \alpha \snr}.
\label{eq:Del_defs}
\end{align}
Noting that $\Delta_\alpha = \tilde{\Delta}_\alpha+ \alpha(\mc{C} - R) - t_\alpha$, the idea is to cancel the combinatorial coefficient ${ L \choose L\alpha }$ using $\exp(-n D_1(\tilde{\Delta}_\alpha,  1- \rho_\alpha^2))$, and produce an exponentially small error probability using $\exp(-n [ D_1(\Delta_\alpha, 1- \rho_\alpha^2) -  D_1(\tilde{\Delta}_\alpha, 1- \rho_\alpha^2) ])$.  

The derivative of $D_1(\Delta,  1- \rho_\alpha^2)$ with respect to  $\Delta$, denoted by $D_1^\prime(\Delta)$, is equal to 
\be
D_1^\prime(\Delta) = 
\begin{cases}
 \frac{2 \Delta}{(1- \rho_\alpha^2) (1+ \sqrt{1 + 4 \Delta^2/(1- \rho_\alpha^2)})}   & \text{ if } \Delta < \frac{1- \rho_\alpha^2}{\rho_\alpha^2}, \\
 1 & \text{ otherwise}.
 \end{cases}
 \label{eq:D1_def}
 \ee
 Since the derivative  is non-decreasing in $\Delta$, using a first-order Taylor expansion we deduce
\be
D_1( \Delta_\alpha,  1- \rho_\alpha^2) \geq D_1(\tilde{\Delta}_\alpha,  1- \rho_\alpha^2) + (\Delta_\alpha - \tilde{\Delta}_\alpha) 
D_1^\prime(\tilde{\Delta}_\alpha).
\ee
Then using the definition of $\textsf{err}_2$ in \eqref{eq:nonasympML}, we have  for $\frac{1}{L} \leq \alpha \leq 1- \frac{1}{L}$, 
\begin{align}
 & \textsf{err}_2(\alpha)   \leq  
 \exp \Big(-n D \Big( t_\alpha, \, \frac{\alpha^2\, \snr}{1+ \alpha^2 \, \snr} \Big) \Big) 
    \nonumber \\
& \qquad +  {L \choose \alpha L} \exp(-n D_1(\tilde{\Delta}_\alpha,  1- \rho_\alpha^2))\exp(-n ( \alpha(\mc{C} -R) - t_\alpha )  D_1^\prime(\tilde{\Delta}_\alpha))  \nonumber \\
& \leq \exp \Big(-n D \Big( t_\alpha, \, \frac{\alpha^2\, \snr}{1+ \alpha^2 \, \snr} \Big) \Big)   + \exp(-n ( \alpha(\mc{C} -R) - t_\alpha )  D_1^\prime(\tilde{\Delta}_\alpha))
\end{align}
where the last inequality is obtained by using the relation $nR = L \log M =  \textsf{a} L \log L$, and the fact that 
$\msf{a} \geq \msf{a}^*_L(\snr)$ (see \eqref{eq:astar_def}). Choosing $t_\alpha =  \alpha(\mc{C} -R)/2$, we obtain
\begin{align}
 & \textsf{err}_2(\alpha)    \leq   \exp \Bigg(-n D \Bigg(\frac{\alpha(\mc{C} -R)}{2}, \, \frac{\alpha^2\, \snr}{1+ \alpha^2 \, \snr} \Bigg) \Bigg)     +   \exp\Bigg(\frac{-n\alpha(\mc{C} -R)  D_1^\prime(\tilde{\Delta}_\alpha)}{2}  \Bigg) \nonumber  \\
 & \stackrel{(a)}{\leq} \exp \Bigg(-n D \Bigg(\frac{\alpha(\mc{C} -R)}{2}, \, \frac{\alpha^2\, \snr}{1+ \alpha^2 \, \snr} \Bigg) \Bigg)  
 +  \exp\big(- n\alpha(\mc{C} -R) w(\snr)  \big) \nonumber  \\
 & \stackrel{(b)}{\leq}   \exp \Bigg(\frac{-n}{4}  g\Bigg( \frac{(\mc{C}-R) \sqrt{1+ \alpha^2 \snr}}{2 \sqrt{\snr}} \Bigg) \Bigg)  +  \exp\big(- n\alpha(\mc{C} -R) w(\snr)  \big)
 \label{eq:err2_bnd2}
\end{align}
where the function $g$ is defined in \eqref{eq:gx_def}.  In the above, inequality $(a)$ is obtained using the lower bound $D_1^\prime(\tilde{\Delta}_\alpha) \geq 2 w(\snr)$, with $w(\snr)$ being defined in \eqref{eq:wsnr_def}. This lower bound is obtained by using the definition of $1-\rho_\alpha^2$ from \eqref{eq:Del_defs} and the lower bound $\tilde{\Delta}_\alpha \geq \frac{\snr}{4(1+\snr)^2} \alpha (1-\alpha)$ in the expression for $D_1^\prime(\tilde{\Delta}_\alpha)$  in \eqref{eq:D1_def}.  Inequality $(b)$ is obtained using the following lower bound \cite[Lemma 6]{AntonyML}:
\ben
D(x, 1-\rho^2) \geq \frac{1}{4} g\left( \sqrt{1+ \frac{4x^2}{1-\rho^2}} -1 \right).
\een

Finally, using \eqref{eq:err2_bnd2} in \eqref{eq:Pbeta0_bnd} and noting that  there are at most $L$ terms in the sum, we deduce
\ben
\begin{split}
\pr_{\beta_0}(\Esec \geq \e) & \leq 2 L  \exp \left( -n \min \Big\{ \e  \Delta w(\snr), \, \frac{1}{4}g\Big(\frac{\Delta}{2 \sqrt{\snr}} \Big) \Big\} \right)  \\
& = \exp \left(-h(\e,\Delta) - \frac{\log 2L}{n} \right),
\end{split}
\een
 where $\Delta= (\mc{C} - R)$, and $h(\e,\Delta)$ is defined in \eqref{eq:hedel_def}.  This completes the proof of Theorem \ref{thm:MLresult}.
 
 
 \subsection{Proof sketch of  Theorem \ref{thm:Bern_ML}} \label{subsec:bern_proof}
 
 We will prove the theorem via the following bound similar to Proposition \ref{prop:nonasympML}: \be \pr_{\beta_0}(\Esec \geq \e)  \leq \sum_{\ell = \e L}^{L} \, \textsf{err}'_2(\ell/L)\label{eq:pr_beta0} \ee where $\textsf{err}'_2(\cdot)$ is defined as follows.  For $0 < \alpha \leq 1$,
 \be
 \textsf{err}'_2(\alpha)  = \min_{t_\alpha  \in [0, C_\alpha - \alpha R]} \ \textsf{err}'_2(\alpha, t_{\alpha}), 
 \label{eq:err2_bnd_bern}
 \ee
where
 \begin{align}
 \textsf{err}'_2(\alpha, t_\alpha) =  &  {L \choose \alpha L} \exp \left(  - n D_1\left( \mc{C}_{\alpha} - \alpha R - t_\alpha, \, \frac{\alpha(1-\alpha) \, \snr}{1+ \alpha \, \snr} \right) - \iota_1 \right) \nonumber \\
& + \exp \left(  - n D\left( t_\alpha, \, \frac{\alpha^2\, \snr}{1+ \alpha^2 \, \snr} -\iota_2 \right) \right). \label{eq:nonasympML_bern}
 \end{align}
 Here $\iota_1=O(1/\sqrt{L})$ and $\iota_2=O(1/L)$.
 
The only difference between  $\textsf{err}'_2(\cdot)$ and $\textsf{err}_2(\cdot)$ defined in \eqref{eq:nonasympML} is the  presence of $\iota_1$ and $\iota_2$ in the former. Using the bound in \eqref{eq:pr_beta0}, Theorem \ref{thm:Bern_ML} can be established using steps similar to those used for Theorem \ref{thm:MLresult} in Sec. \ref{subsec:proof_GaussML}. 
 
 We now sketch the proof of \eqref{eq:pr_beta0}. The proof hinges on two key lemmas. The first uniformly bounds the ratio between a binomial pmf and a Gaussian with the same mean and variance.  
 
 \begin{lemma} \cite{takeishi14}
Let $\phi(x; \, \mu,  \sigma^2)$ denote the normal density with mean $\mu$ and variance $\sigma^2$. Then for any $\ell \in \mathbb{N}$, 
\ben
\max_{k \in \{0,1,\ldots, \ell \}} \ \frac{{\ell \choose k} \, 2^{-\ell}}{\phi(k; \, \ell/2, \ell/4)} \leq \exp(\varphi(\ell)),
\een
where $\varphi(\ell) \leq 5/\ell$ for $\ell \geq 1000$.
\label{lem:bin_normal_ratio}
 \end{lemma}
 
The next two lemmas give bounds on the ratio of certain Reimann sums to the corresponding integrals. 

\begin{lemma} \cite{takeishi2016improved}
For $n \in \mathbb{N}$, let $h=2/\sqrt{n}$ and $x_k= -\sqrt{n} + \frac{2k}{\sqrt{n}}$ for $k=0,1,\ldots, n$.For $\mu \in \reals$ and $s >0$, define
\begin{align*}
I_d & = h \, \sum_{k=0}^n \exp\left\{ -\frac{s^2}{2} (x_k- \mu)^2\right\}, \\
I_c & =  \int_{-\infty}^{\infty} \exp\left\{ -\frac{s^2}{2}(x - \mu)^2 \right\} \, dx. 
\end{align*}
Then $$I_d \leq \left( 1  + \frac{\eta s^2}{n} \right)I_c,$$
where $\eta = 3 / \sqrt{8 \pi e} \leq 0.37$.
\label{lem:pmf_bound}
\end{lemma}

\begin{lemma}
For $n, n'\in \mathbb{N}$, let $h=2/\sqrt{n}$, $\mc{X}= \{  -\sqrt{n} + \frac{2k}{\sqrt{n}} \mid k=0,1,\ldots, n \}$, and 
let $h'=2/\sqrt{n'}$, $\mc{X}= \{  -\sqrt{n'} + \frac{2k}{\sqrt{n'}} \mid k=0,1,\ldots, n \}$.
\begin{enumerate}[(a)]
\item  For a two-dimensional vector $\mbf{x} =[x_1, x_2]^T \in \reals^2$ and a $2 \times 2$ positive definite matrix $B$, define
\begin{align*}
I_d & = \int_{\reals} h \, \sum_{x_1\in \mc{X}} \exp\left\{ -\mbf{x}^T B \mbf{x}/2 \right\} \, dx_2, \\
I_c & =  \int_{\reals^2} \exp\left\{ -\mbf{x}^T B \mbf{x}/2 \right\} \, d\mbf{x}. 
\end{align*}
Then $I_d \leq \left( 1  + \frac{\eta B_{11}}{n} \right)I_c$, where $\eta$ is defined in Lemma  \ref{lem:pmf_bound}, and $B_{ij}$ denotes the $(i,j)$th element of  the matrix $B$.

\item For a three-dimensional vector $\mbf{x} =[x_1, x_2, x_3]^T \in \reals^3$ and a $3 \times 3$ positive definite matrix $B$, define
\begin{align*}
I_d & = \int_{\reals} h \, h'  \, \sum_{x_1\in \mc{X}_1} \sum_{x_1\in \mc{X}_2}  \exp\left\{ -\mbf{x}^T B \mbf{x}/2 \right\} \, dx_3, \\
I_c & =  \int_{\reals^3} \exp\left\{ -\mbf{x}^T B \mbf{x}/2 \right\} \, d \mbf{x}.
\end{align*}
Then $I_d \leq \left( 1  + \frac{\eta B_{11}}{n} \right)\left( 1  + \frac{\eta B_{22}}{n'} \right) I_c$.
\end{enumerate}
\label{lem:reimann_int_bnd}
\end{lemma}

The proof is along the lines of that of \eqref{eq:err2_ell} on p.\pageref{eq:tilT_Tstar_decomp}. We have
\be 
\pr_{\beta_0} (\Esec = \ell/L )  \leq  \pr_{\beta_0}(\exists \, \mcs :  \tilde{T}(\mcs)  \leq t_\alpha) +   \pr_{\beta_0}( T^* \leq -t_\alpha). 
\label{eq:tilT_Tstar_decomp_bern}
\ee
 with $\tilde{T}(\mc{S})$ and $T^*$ defined as in \eqref{eq:Ttil_def}-\eqref{eq:Tstar_def}.  Using a Chernoff bound, the second term can be bounded as 
 \be
  \pr_{\beta_0}( T^* \leq -t_\alpha) \leq  e^{-n\lambda t_\alpha} \expec_{Y, X_{\mc{S}_*}}e^{-n\lambda T^*}. 
  \label{eq:prTstar}
 \ee
The moment generating function can be written as 
\be
\expec_{Y, X_{\mc{S}_*}}e^{-n\lambda T^*} = \left[ \expec_{Z_1, \tilde{Z}_1} e^{\lambda(Z_1^2 - \tilde{Z}_1^2)/2} \right]^n,
\label{eq:mgf_exp0}
\ee
 where 
 $$Z_1 \sim \mc{N}(0,1), \qquad \tilde{Z}_1 = \frac{(\sigma Z + \alpha \sqrt{P} \, W_1)}{\sqrt{\sigma^2 + \alpha^2 P}}.$$
 Here $W_1$ (independent of $Z_1$) is the sum of $L$  independent equiprobable $\pm 1$ random variables, normalized to have unit variance. If $W_1$ was Gaussian, then the moment generating function in \eqref{eq:mgf_exp0} would be exactly equal to $(1- \lambda^2 \alpha^2 \snr/(1 + \alpha^2 \snr))^{-n/2}$. For the $W_1$ arising from a Bernoulli dictionary, using Lemmas \ref{lem:bin_normal_ratio}, \ref{lem:pmf_bound} and \ref{lem:reimann_int_bnd}, it can be shown that 
 \[  \expec_{Y, X_{\mc{S}_*}}e^{-n\lambda T^*}  \leq \left( \frac{e^{\iota_2}}{1- \lambda^2 \alpha^2 \snr/(1 + \alpha^2 \snr) } \right)^n. \]
 Using this in  \eqref{eq:prTstar} yields the second term in \eqref{eq:nonasympML_bern}.
 
 For the first term in \eqref{eq:tilT_Tstar_decomp_bern}, we write $\tilde{T}(\mc{S}) = \tilde{T}_1 + \tilde{T}_2$ where $\tilde{T}_1$ and $\tilde{T}_2$ are defined in \eqref{eq:t1_til}.  Then, using steps similar to \eqref{eq:EnlamT}, we obtain 
 \be
 \pr_{\beta_0}(\exists \, \mcs :  \tilde{T}(\mcs)  \leq t_\alpha) \leq e^{n t_\alpha} \sum_{\mcs_1} \expec_{Y, X_{\mcs^*}} e^{-n\lambda T_1(\mcs_1)} \Bigg[ \sum_{\mcs_2}  \expec_{X_{\mcs_2}}  e^{-nT_2(\mcs)} \Bigg]^\lambda.
 \label{eq:Ttil_T2}
 \ee
 Again, using Lemmas  \ref{lem:bin_normal_ratio}, \ref{lem:pmf_bound} and \ref{lem:reimann_int_bnd}, the two moment generating functions in \eqref{eq:Ttil_T2} can be bounded to yield the first term in \eqref{eq:nonasympML_bern}.  The details of the computation can be found in \cite[Section III.C]{takeishi14} and \cite{takeishi2016improved}.

%
%
\chapter{Computationally Efficient Decoding}  \label{chap:AWGN_eff}

In this chapter, we will discuss computationally efficient decoders for SPARCs over the AWGN channel. The goal is to design and analyze  feasible capacity-achieving decoders whose complexity is polynomial in the code length $n$, in contrast to the infeasible maximum-likelihood decoder.

The channel model and the encoding procedure are as described in Sec. \ref{sec:prob_setup}. The first idea for designing a  efficient decoder is to use a decaying power allocation across sections. As shown in Fig.  \ref{fig:sparc_nPl}, the non-zero coefficients in the message vector $\beta$ are 
\[ 
c_1= \sqrt{nP_1}, \ c_2= \sqrt{nP_2}, \ldots, c_L=\sqrt{nP_L}.
\]
Without loss of generality, we assume that the power allocation is non-increasing across sections, i.e., $P_1 \geq  P_2 \ldots \geq P_L$. 
\begin{figure}
\centering
\includegraphics[width = 4.5in]{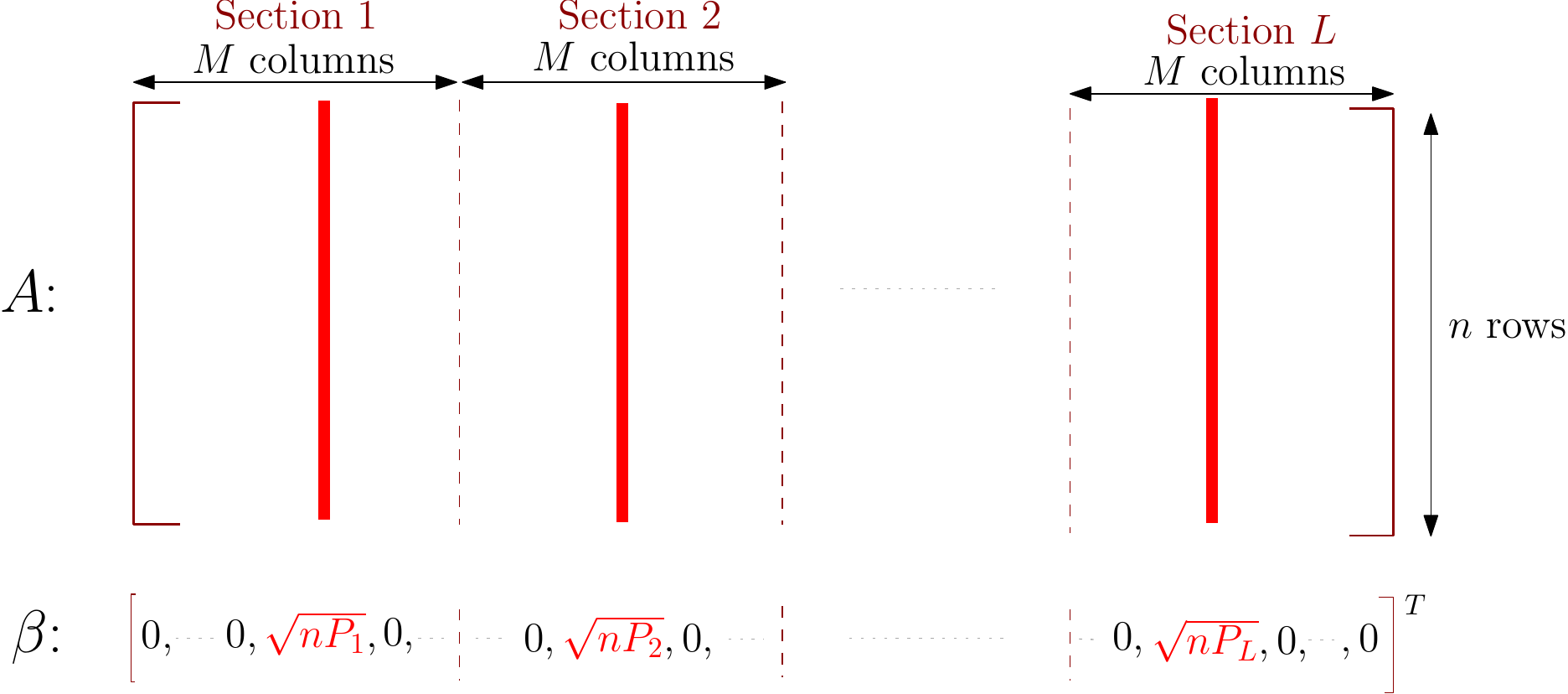}
\caption{\small{A Gaussian sparse regression codebook with power allocation. The power allocation coefficients $P_1, \ldots, P_L$ are of order $\frac{1}{L}$ and satisfy $P_1+ \ldots + P_L=P$.}}
\vspace{-5pt}
\label{fig:sparc_nPl}
\end{figure}
Denoting  the column of $A$ corresponding to the $\ell$th non-zero entry of $\beta$ by $A_{i_{\ell}}$, for $\ell \in [L]$, the received sequence $y \in \re^n$  is 
\be y = \sqrt{nP_1} A_{i_{1}} +  \sqrt{nP_2} A_{i_{2}} + \ldots +  \sqrt{nP_L} A_{i_{L}} + w. \label{eq:rx_cwd} \ee

The decoding task is to recover the non-zero locations $i_1, \ldots, i_L$.
The idea of power allocation is to facilitate an iterative decoder that first decodes (either exactly or approximately) the sections with the highest power, then the sections with the next highest power etc.  Correctly recovering a subset of the indices, say allows the decoder to cancel their contribution from $y$, thereby making the decoding of the remaining sections easier.  

In the next section, we discuss an adaptive `hard-decision' decoder based on the successive cancellation idea above. In Sections \ref{sec:adap_soft_dec} and \ref{sec:AMP_dec}, we discuss two `soft-decision' versions of the iterative decoder. All three decoders are asymptotically capacity-achieving, but the soft-decision decoders have better finite length error performance.   In the next chapter, we will discuss how to design alternative power allocations to optimize finite length error performance.

We emphasize that the decoders above do not  pre-specify an order in which the sections are decoded. Rather, the power allocation makes it likely that sections with higher power are decoded before those with lower power.  This is similar in spirit to how algorithms such as Orthogonal Matching Pursuit for  recovering sparse vectors  can be significantly more powerful when the magnitudes of the non-zero coefficients have a decaying profile \cite{herzet2016relaxed}.

\section{Adaptive successive hard-decision decoding}

For our theoretical results we will use the following exponentially decaying allocation, with the power in section $\ell$ proportional to $e^{-2\mc{C} \ell /L}$: 
 \be  
 P_\ell = P \cdot  \frac{e^{2\mc{C}/L} -1}{1- e^{-2\mc{C}}} \cdot e^{-2\mc{C}\ell/L}, \quad  \ell \in [L].  \label{eq:exp_power_alloc} 
 \ee
Recalling that $\mc{C} = \frac{1}{2} \log (1+ \snr)$, we note that $1-e^{-2\mc{C}} =\snr/(1+\snr)$. 

This power allocation is motivated by thinking of the $L$ sections of the SPARC as corresponding to  $L$ users of a Gaussian multiple-access channel (MAC) with total power constraint $P$.   Indeed, consider  the equal-rate point  on the capacity region of a $L$-user Gaussian MAC where each user gets rate $\mc{C}/{L}$. It is well-known \cite{coverT12,elGKBook} that this rate point can be achieved  via successive cancellation decoding, where user $1$ is first decoded, then user $2$ is decoded after subtracting the codeword of user $1$, and so on. For this successive cancellation scheme, the power allocation for the $L$ users is determined by the following set of equations:
\be
\frac{1}{2} \log \left( 1+  \frac{P_\ell}{\sigma^2 + P_{\ell+1} \ldots + P_L} \right) =\frac{\mc{C}}{L}, \quad  \ell \in [L].
\label{eq:rec_PA_form}
\ee
Sequentially solving the set of equations in \eqref{eq:rec_PA_form}, starting from $\ell=L$, yields the exponentially decaying  power allocation in \eqref{eq:exp_power_alloc}. 

Continuing the analogy with an $L$-user MAC,  we ask: can the above successive cancellation scheme  be used for SPARC decoding to achieve rates close to $\mc{C}$?  Unfortunately, successive cancellation performs poorly for SPARC decoding. This is because  $L$, the number of sections (`users') in the codebook grows with $n$. Indeed, for the choice $M= L^{\textsf{a}}$, $L$ grows as  $n/\log n$, while $M$, the number of codewords per user, only grows polynomially in $n$.   An error in decoding one section affects the decoding of future sections,  leading to a large number of section errors after $L$ steps. We note that in the standard MAC set-up, the number of users $L$ remains constant as the code length $n$ grows; hence the rate per user is also of constant order. 

The first feasible SPARC decoder, proposed in \cite{AntonyFast}, controls the accumulation of section errors using  \emph{adaptive} successive decoding. The idea is to \emph{not} pre-specify the order  in which sections are decoded, but to look across all the undecoded sections in each step, and adaptively decode columns which have a large inner product with the residual.  The adaptive successive decoding algorithm proceeds as follows.

Given $y = A\beta + w$, start with estimate $\beta^0=0$. 
\paragraph{Initial step [$t=1$]}
  \begin{enumerate}
  \item Compute the inner product of $\sqrt{n} \,  {y}/{\norm{y}}$ with each column of $A$.
  \item Pick the columns corresponding to inner products that cross a {threshold} $\sqrt{2 \log M} +a$ to form $\beta^1$, for a fixed constant $a >0$.
  \item Form the  initial fit as weighted sum of columns: $\textsf{Fit}^1=A \beta^1$.
  \end{enumerate}
  
\paragraph{Iterate [step $t+1, \ t \geq 1$]}
\begin{enumerate}
 \item Compute the normalized residual $\textsf{Res}^t =  \sqrt{n}\, (y- \textsf{Fit}^{t})/{\norm{y -  \textsf{Fit}^{t}}}$.
  \item Compute the inner product of $\textsf{Res}^t$ with each  \emph{remaining} column of $A$.
  \item Pick  the columns that cross the threshold $\sqrt{2\log M}+a$  to form $\beta^{t+1}$.
  \item Compute the new fit $\textsf{Fit}^{t+1}=A \beta^{t+1}$.
  \end{enumerate}

\paragraph{Stop}  if there are no additional inner products above threshold, or after $(\snr)\log M$ steps.

\subsection{Intuition and analysis}

\paragraph{First step} The key observation is that  in Step $1$, the columns of $A$ that are not sent (i.e., correspond to a zero entries in $\beta$) will produce  normalized inner products whose joint distribution is close to i.i.d. $\mc{N}(0,1)$.  On the other hand, the column that was sent in section $\ell$ will produce an inner product that is close to a standard normal plus a shift of size $\sqrt{nP_\ell/(P+\sigma^2)}$.  This is made precise in the following lemma.
\begin{lemma} \cite[Lemma 3]{AntonyFast}
For $j \in [ML]$, let $A_j$ denote the $j$th column of $A$, and let $\mc{Z}_{1,j} = \sqrt{n} A_j^* y/\norm{y}$. Then for $j \in$ section $\ell$, $\ell \in [L]$  we have
 \be
 \mc{Z}_{1,j} \stackrel{d}{=} \sqrt{\frac{nP_\ell}{P+\sigma^2}} \,  \frac{\chi_n}{\sqrt{n}} \,   \mbf{1}\{ j \, \textsf{sent} \} + N_{1,j},
 \ee
 where $N_1=(N_{1,j}: 1 \leq j \leq ML)$ is multivariate normal with zero mean and covariance matrix $\textsf{I} - \frac{\beta \beta^*}{n(P+\sigma^2)}$. Furthermore, $\chi_n^2 = \norm{y}^2/(P+\sigma^2)$ is a Chi-square $n$ random variable that is independent of $N_1$.
 \label{lem:step1_dist}
\end{lemma}
\begin{proof}
Recall from \eqref{eq:rx_cwd} that 
$
y = \sqrt{nP_1} A_{i_{1}}  + \ldots +  \sqrt{nP_L} A_{i_{L}} + w,
$
where $i_1,\ldots, i_L$ denote the indices of the sent terms. Also recall that $w \sim \mc{N}(0,  \sigma^2 \msf{I})$ and $A_j \sim \mc{N}(0, \tfrac{1}{n} \msf{I})$ are i.i.d. for $1 \leq j \leq ML$.  Using this we find that the conditional distribution of $A_j$ given $y$, for $j$ in section $\ell$, is:
\be
A_j \mid y  \sim 
\begin{cases}
\mc{N}(0, \tfrac{1}{n} \msf{I}) & \text{if } j \neq i_\ell, \\
\mc{N}\left( y \,\frac{\sqrt{nP_\ell}}{n(P+\sigma^2)}, \,  \tfrac{1}{n} (1-\tfrac{P_\ell}{P+\sigma^2})\msf{I} \right) & \text{if } j = i_\ell.
\end{cases}
\ee
Hence the conditional distribution of $A_j$ given $y$ may be expressed as 
\be
A_j = \frac{1}{\sqrt{n}}\left( \frac{\beta_j}{P+\sigma^2} \, \frac{y}{\sqrt{n}} \, + \, U_j \right)
\label{eq:Aj_dist}
\ee
where $U_j \sim \mc{N}(0, (1-\tfrac{\beta_j^2}{n(P+\sigma^2)}))$ is independent of $y$. Moreover, for a given row index $i$, since $\expec[A_{j,i} A_{k,i}] = \tfrac{1}{n} \mbf{1}\{ j=k\}$,  we have 
$\expec[U_{j,i} U_{k,i}] = 
\tfrac{-\beta_j \beta_k}{n(P+\sigma^2)}$ for  $j\neq k$. Therefore, for any $i \in [n]$ the random vector $(U_{1,i}, \ldots, U_{ML,i})$ has distribution 
$\mc{N}(0, (1-\tfrac{\beta \beta^*}{n(P+\sigma^2)})\msf{I})$. 

From \eqref{eq:Aj_dist} we have
\be
\mc{Z}_{1,j} = \sqrt{n} A_j^* \frac{y}{\norm{y}} = \frac{\beta_j}{\sqrt{P+\sigma^2}} \frac{\norm{y}}{\sqrt{n(P+\sigma^2)}} +  \frac{U_j^*  y}{\norm{y}}.
\ee
Letting $N_{1,j}=\frac{U_j^*  y}{\norm{y}}$ and $N_1=(N_{1,j}: 1 \leq j \leq ML)$, to complete the proof we need to show that  $N_1$ is a multivariate normal that is independent of $y$ with covariance matrix 
$\msf{I} - \tfrac{\beta \beta^*}{n(P+\sigma^2)}$. Indeed, conditioning on any (non-zero) realization of $y$ it is seen that   $N_1$ is a $\mc{N}(0, (1-\tfrac{\beta \beta^*}{n(P+\sigma^2)})\msf{I})$ random vector. This completes the proof. 
\end{proof}

In Lemma \ref{lem:step1_dist}, since $\chi_n/\sqrt{n}$ is close to $1$ for large $n$, the shift  in the inner product corresponding to the sent term in section in $\ell$ is
\be
\sqrt{\frac{nP_\ell}{P+\sigma^2}} = \sqrt{\frac{LP_\ell}{R(P+\sigma^2)} \log M} \stackrel{(a)}{=} \sqrt{\frac{\mc{C}}{R}(1 +  O\big( \tfrac{1}{L} \big)) e^{-2 \mc{C} \ell/L} } \sqrt{2 \log M},
\label{eq:shift_step1}
\ee
where $(a)$ is obtained using the exponential power allocation in \eqref{eq:exp_power_alloc}, and the fact that $e^{2\mc{C}/L} -1 = \tfrac{2 \mc{C}}{L}(1+ O(\tfrac{1}{L}))$. Since $R< \mc{C}$ we observe from \eqref{eq:shift_step1} that the shift will be larger than $\sqrt{2 \log M}$  for $1 \leq \ell \leq \ell_0$, where $\ell_0$ is determined by $\mc{C}/R$.

On the other hand, for any column $j$ that is \emph{not} sent in section $\ell$, the shift is zero, and the test statistic  $\mc{Z}_{1,j}$ normal. Recalling that each section has $M$ columns, we note that the maximum of $M$ standard normals concentrates near $\sqrt{2 \log M}$ for large $M$ \cite{hall79rate}. Therefore, if the constant $a$ defining the threshold is chosen to be small compared to $\sqrt{2\log M}$, then the true columns in sections $1 \leq \ell \leq \ell_0$  are likely to have inner products that exceed the threshold $\sqrt{2\log M} +a$. On the hand, $a >0$ ensures that the probability of inner product of a wrong column crossing the threshold is small. It is evident that the value of $a$ determines the trade-off between the probabilities of false alarm and missed detection.

\paragraph{Subsequent steps}  Let $\tsf{dec}_t$ denote the set of sections decoded up to the end of step $t$. Then the residual $\tsf{Res}_t$ removes the contribution of the sections in  $\tsf{dec}_t$ from $y$. Assuming that no mistakes were made until step $t$, by analogy with the Step 1 analysis above we expect the shift for the sent term in  (a yet to be decoded) section $\ell$   to be close to 
\be  \sqrt{ \frac{n P_\ell}{ \sigma^2 + P(1-x_t)}}, \label{eq:later_shifts} \ee
where $x_t =  \frac{1}{P} \sum_{k \in \tsf{dec}_t} P_k$ is the fraction of power that has  already been decoded.  Thus as decoding successfully progresses, $x_t$ increases with $t$, making the shift in \eqref{eq:later_shifts} larger and facilitating the decoding of sections with lower power. 

However, establishing a result analogous to Lemma \ref{lem:step1_dist}  for $t>1$ is challenging. This is because the dependence between the residual $\tsf{Res}_t$ and the matrix $A$ cannot be easily characterized. Indeed, recall that $\tsf{Res}_t$ has been generated via decisions based on inner products with columns on $A$ computed  in previous steps. 

To address this, Barron and Joseph consider a slightly modified version of the decoder, where at  the end of each step $t$, we compute $G_t$, the part of 
$\tsf{Fit}_t$ that is orthogonal to $y, \tsf{Fit}_1, \ldots, \tsf{Fit}_{t-1}$. That is, with $G_0 \triangleq y$, the collection
\[
\frac{G_0}{\norm{G_0}},   \frac{G_1}{\norm{G_1}}, \ldots, \frac{G_t}{\norm{G_t}} 
\label{eq:g1g2def}
\]
forms an orthonormal basis for $\tsf{Fit}_1, \ldots, \tsf{Fit}_t$. Then in step $(t+1)$, instead of residual-based inner products, we compute the following test statistic for each column $j$ in an undecoded section:
\be
\mc{Z}^{\tsf{comb}}_{t,j} = (A_j)^* \left[ \lambda_0 \frac{G_0}{\norm{G_0}} + \ldots + \lambda_t \frac{G_t}{\norm{G_t}}  \right]
\label{eq:comb_test_stat}
\ee
where $\lambda_0, \ldots, \lambda_t$ are deterministic positive constants such that $\sum_{k=0}^t \lambda_k^2 =1$. For an appropriate choice of  $\lambda_k$'s, the test statistic in \eqref{eq:comb_test_stat} closely mimics the residual based statistic. Essentially,   $\lambda_k$  is chosen to be a deterministic proxy for the inner product
\[ \frac{ (\tsf{Res}_t)^* G_k}{\norm{\tsf{Res}_t} \norm{G_k}}. \]

With this choice, the test statistic $\mc{Z}^{\tsf{comb}}_{t,j}$ can be shown to have a distributional representation that is approximately a shifted normal. That is, 
\be
\mc{Z}^{\tsf{comb}}_{t,j}  \stackrel{(a)}{\approx}   \sqrt{ \frac{n P_\ell}{ \sigma^2 + P(1-x_t)}} \mbf{1}\{ j \, \tsf{sent} \} + N_{t,j},
\ee
where $N_{t,j}$ is normal zero-mean random variable with variance near $1$. The parameter $x_t$ quantifies the expected success rate, and can be interpreted as the expected fraction of power in the sections decoded by the end of iteration $t$.  It can be recursively computed as follows, starting from $x_0=0$. With $\tau= \sqrt{2\log M} +a$ denoting the threshold used in each step  and $\Phi$ denoting the standard normal distribution function, we have
 \begin{align}
& x_{t+1}   = \sum_{\ell=1}^L \frac{P_\ell}{P}  \, {\Phi}\left( \sqrt{ \frac{n P_\ell}{ \sigma^2 + P(1-x_t)}} - \tau \right) 
\label{eq:xt_hard_update} \\
& = \sum_{\ell=1}^L \frac{P_\ell}{P} \, {\Phi}\Bigg( \sqrt{2 \log M}
 \Bigg(  \sqrt{\frac{\mc{C}}{R} \cdot  \frac{\sigma^2 +P}{ \sigma^2 + P(1-x_t)}  e^{-2\mc{C} \ell/L}} -1 \Bigg) - a  + O(\tfrac{1}{L}) \Bigg)
 \label{eq:hard_update_fn2}
\end{align}
where \eqref{eq:hard_update_fn2} is obtained using the expression for $\sqrt{nP_\ell}$ from \eqref{eq:shift_step1}. 

Figure \ref{fig:SE_hard_dec_example} shows the progression of $(1-x_t)$ with $t$, for three different values of the threshold $\sqrt{2\log M}+a$, with $a=1,0.8, 0$. The parameter $(1-x_t)$ quantifies the expected fraction of power in the undecoded sections 
after iteration $t$. Observe from \eqref{eq:xt_hard_update} that a smaller value of the threshold $\tau$ results in a smaller value of $(1-x_t)$; this is illustrated by the curves shown in Fig. \ref{fig:SE_hard_dec_example}.  However, the  recursive formula \eqref{eq:xt_hard_update} is an idealized prediction: it gives the expected fraction of power in the decoded sections at the end of each iteration $t$, assuming that there are no \emph{false alarms}, i.e.,  no sections have been incorrectly decoded.  For finite block lengths the value of $a >0$ in the threshold $\tau = \sqrt{2\log M}+a$ plays a crucial role in determining the false alarm rate. The larger the value of $a$, the lower the probability of false alarms.

\begin{figure}
\centering
\includegraphics[width=3.4in]{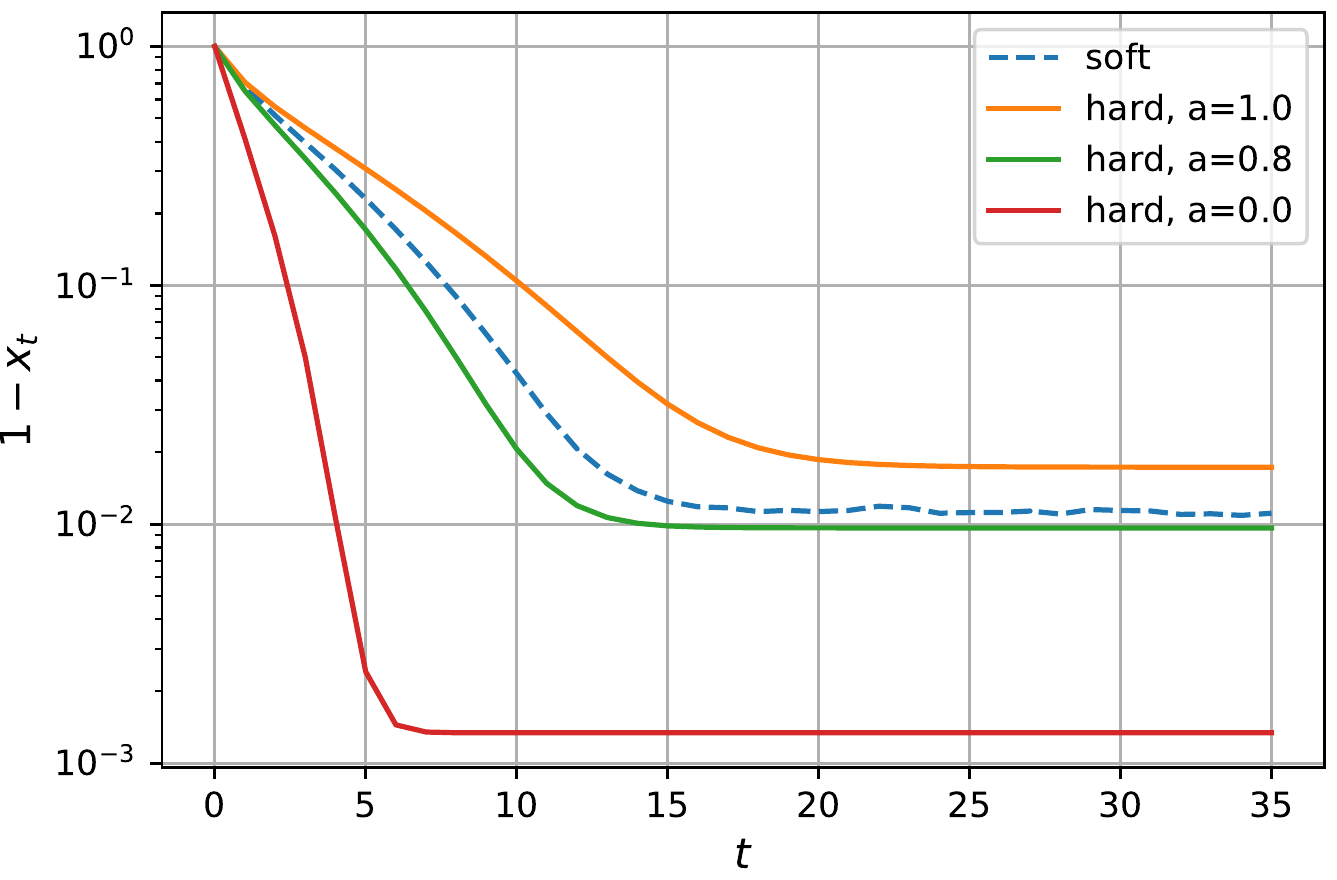}
\caption{\small{Evolution of $(1-x_t)$ with iteration  $t$. The SPARC parameters are $M=512, L=1024, \snr =15, R=0.8 \mc{C}, \, P_\ell \propto e^{-2\mc{C}{\ell}/{L}}$ with $\mc{C}$ in nats. Curves are shown for three different values of the threshold  $\tau=\sqrt{2\log M}+a$, with $a=1,0.8, 0$. The dashed curve shows the evolution of $(1-x_t)$ for the soft-decison decoder discussed in the next section.}}
\vspace{-7pt}
\label{fig:SE_hard_dec_example}
\end{figure}

 A rigorous statement specifying the distributional representation of $\mc{Z}^{\tsf{comb}}_{t,j}$, taking into account the false alarm rate, is given in \cite[Lemma 4]{AntonyFast}. This representation leads to the following performance guarantee for the decoder, which in essence states that rates up to
\be \mc{C}^* := \frac{\mc{C}}{1 + \delta_M} \label{eq:Cstar_def} \ee can be achieved with $O(\delta_M)$ fraction of section errors, where
\be
\delta_M := \frac{1}{\sqrt{\pi \log M}}.
\label{eq:deltaM_def}
\ee
 \begin{theorem} \cite[Theorem 2]{AntonyFast}
Let the rate $R < \mc{C}^*$ be expressed in the form
\be  \frac{\mc{C}^*}{1 + \tfrac{\kappa}{\log M}} \ee with $\kappa > 0$.
Then, with the exponentially decaying power allocation in \eqref{eq:exp_power_alloc}  the adaptive successive decoder has section error rate less than
\be \delta_{err} := \frac{1}{2 \mathcal{C} \sqrt{\pi \log M}} + \frac{3 \kappa + 5}{8 \mc{C} \log M} \ee
with probability at least $1-P_e$, where
\be P_e = \kappa_{1,M} e^{-\kappa_2 L \min\{ \kappa_3 \Delta^2, \ \kappa_4 \Delta \}}. \label{eq:Pe_def} \ee
In \eqref{eq:Pe_def}, $\Delta = \tfrac{\mc{C}^* - R}{\mc{C}^*}$,    $\kappa_{1,M}$ is a polynomial in $M$, and $\kappa_2, \kappa_3$ and $\kappa_4$ are constants that depend on $\snr$.
\label{thm:channel_cod}
\end{theorem}

\begin{remark}
As in Proposition \ref{thm:msg_error}, we can concatenate the SPARC with an outer Reed-Solomon code of rate $(1-2\delta_{err})$ to guarantee that the message error probability is bounded by $P_e$ in \eqref{eq:Pe_def}.  Thus Theorem \ref{thm:channel_cod} tells us that the adaptive successive decoder  can achieve rates of the order of $1/\sqrt{\log M}$ below capacity.

Choosing $L=M^{\tsf{a}}$,  we have  $M$ of order $(n/\log n)^{\tsf{a}}$, and hence the minimum gap from capacity is of order $1/\sqrt{\log n}$. This gap is much larger than that of the optimal decoder, which can achieve rates up to order $1/n^\alpha$ below capacity with error probability decaying exponentially in $n^{1-2\alpha}$, for any $\alpha \in (0, \tfrac{1}{2})$  (see Remark \ref{rem:ML_gap_to_cap}).
\end{remark}

\begin{remark}
It is shown in \cite[Sec. 4.18]{AntonyThesis}  that the gap from capacity can be improved to $O(\log \log M/\log M)$ using a power allocation that is slightly modified from the one in \eqref{eq:exp_power_alloc}. The power $P_\ell$ is now chosen proportional to
\[ \max\{e^{-\frac{2 \mc{C} \ell}{L}}, \ e^{-2\mc{C}}(1 + \tfrac{c}{\sqrt{2 \log M}}) \}, \]
for a suitably chosen constant $c$. This allocation slightly boosts the power for sections $\ell$ close to $L$. This helps ensure that, even towards the end of the algorithm, there will be sections for which the true terms are expected to have inner product above threshold.
\label{rem:improving_gap}
\end{remark}

\section{Iterative soft-decision decoding} \label{subsec:iter_soft_dec}

 Theorem \ref{thm:channel_cod} shows that the adaptive successive hard-thresholding decoder is asymptotically capacity-achieving. However,   the empirical section error rate  at practically  feasible code lengths is rather high for rates near capacity. We now discuss two  soft-decision decoders,  the adaptive successive soft-decision decoder  and the approximate message passing (AMP) decoder, which have better error performance at finite code lengths. Instead of making hard decisions about which columns to decode in each step, the soft-decision decoders generate  iteratively refined estimates of the message vector in each step. 
 Both soft-decision decoders share a few key underlying principles. We first discuss these principles in this section. The specifics of the two decoding algorithms  are then described in the next two sections.
 
The decoder starts with $\beta^0=0$ (the all-zero vector of length $ML$), and generates an updated estimate of  the message vector in each step; these estimates are denoted by $\beta^1, \beta^2, \ldots$. The key idea in soft-decision decoding  is to form the new estimate in each step by updating  the posterior probabilities of each entry of $\beta$ being the true non-zero in its section. This is done  as follows.

At the end of each step $t$, the decoder produces a test  statistic $\statt \in \reals^{ML}$   that has the form 
 \be
 \statt \approx \beta + \tau_t Z,
 \label{eq:desired_dist}
 \ee 
 where $Z$ is a standard normal random vector independent of $\beta$.  That is, $\statt$ is approximately distributed as the true message vector  plus an independent standard Gaussian vector with known variance $\tau_t^2$.  The test statistic $\statt$ is produced based on  $y, A$ and the previous estimates $\beta^1, \ldots, \beta^{t}$. The details of how $\statt$ is produced to ensure that \eqref{eq:desired_dist} holds depend on the type of soft-decision decoder used. These details are described in Sections \ref{sec:adap_soft_dec} and \ref{sec:AMP_dec}.   
 
 In step $(t+1)$, the decoder  generates an updated estimate $\beta^{t+1}$ based on $\statt$. Assuming that the distributional property in \eqref{eq:desired_dist} exactly holds at the end of step $t$, the Bayes-optimal estimate for $\beta$ that minimizes the expected squared error in the next step $(t+1)$ is 
 \be \beta^{t+1} =\eta^t(\statt) := \expec[\beta \mid \beta + \tau_t Z =\statt]. \ee
 
 The conditional expectation above can be computed as follows using the known prior on $\beta$ in which the location of the non-zero within section is uniformly random. For $\statt =s =(s_1, \ldots, s_{ML})$ and index $i \in \text{sec}(\ell)$, $\ell \in [L]$ we have
\be
\label{eq:cond_exp_beta}
\begin{split}
\eta^t_i(\statt = s) & =  \expec[\beta_i \mid  \beta + \tau_t Z = s ]   = \expec[\beta_i \mid  \{ \beta_j + \tau_t Z_j = s_j \}_{j \in \text{sec}(\ell)} ] \\
& = \sqrt{n P_\ell} \ P(\beta_i = \sqrt{n P_\ell} \mid   \{ \beta_j + \tau_t Z_j = s_j \}_{j \in \text{sec}(\ell)} )\\
&=   \frac{ \sqrt{n P_\ell} \, f( \{s_j \}_{j \in \text{sec}(\ell)} \mid \beta_i = \sqrt{n P_\ell}) \, P(\beta_i = \sqrt{n P_\ell})}
{\sum_{k \in \text{sec}(\ell)}  f( \{s_j \}_{j \in \text{sec}(\ell)} \mid \beta_k = \sqrt{n P_\ell}) \, P(\beta_k = \sqrt{n P_\ell})}
\end{split}
\ee
where we have used Bayes' theorem with  $f(\cdot | \beta_k=\sqrt{nP_\ell})$ denoting the joint density of $\{ \beta_j + \tau_t Z_j \}_{j \in \text{sec}(\ell)}$ conditioned on $\beta_k$ being the non-zero entry in section $\ell$.  Since  $\beta$ and $Z$ are independent with $Z$ having i.i.d.\ $\mc{N}(0,1)$ entries,  for each $k \in \text{sec}(\ell)$ we have
\be
\begin{split}
& f( \{ \beta_j + \tau_t Z_j =s_j \}_{j \in \text{sec}(\ell)} \mid \beta_k = \sqrt{n P_\ell})  \\
& \propto e^{-(s_k - \sqrt{nP_\ell})^2/2 \tau_t^2} \prod_{j \in \text{sec}(\ell), j \neq k } e^{-s_j^2/2 \tau_t^2}.
\end{split}
\label{eq:cond_denf}
\ee
Using  \eqref{eq:cond_denf} in \eqref{eq:cond_exp_beta},  together with the fact that $P(\beta_k = \sqrt{n P_\ell}) = \frac{1}{M}$ for each $k \in \text{sec}(\ell)$, we obtain
\be
 \eta^{t}_i(\statt=s) = \expec[\beta_i \, | \,  \beta + \tau_t Z = s] = \sqrt{nP_\ell}  \frac{e^{s_i \sqrt{n P_\ell}/\tau^2_t}}
{\sum_{j \in \text{sec}(\ell)} \, e^{s_j \sqrt{n P_\ell} / \tau^2_t}}.
\label{eq:etat_def}
\ee

 \subsection{State evolution}
 
 To compute $\beta^{t+1}$ using  \eqref{eq:etat_def} requires the parameter $\tau_t^2$,  which is the  variance of the noise in the desired distributional representation $\statt = \beta + \tau_t Z$. This noise variance  has two components: one  is the channel noise variance $\sigma^2$, and the other is the mean-squared estimation error $\tfrac{1}{n}\expec \norm{\beta-\beta^t}^2$.  
 
Starting  with $\tau_0^2 = \sigma^2 +P$, we recursively compute $\tau_{t+1}^2$ for $t \geq 0$ as follows: 
\be 
\tau_{t+1}^2 =  \sigma^2 +  \frac{1}{n} \, \expec   \norm{\beta -  \expec[ \beta | \beta +  \tau_{t} Z] }^2 =  \sigma^2 +  \frac{1}{n} \, \expec   \norm{\beta -  \eta_t(\beta +  \tau_t Z)] }^2, 
\label{eq:tau_t_rec}
\ee
where the expectation on the right is over $\beta$ and  the independent standard normal vector $Z$. The recursion \eqref{eq:tau_t_rec} to generate $\tau_{t+1}^2$ from $\tau^2_t$ can be   written as
 \be \tau_{t+1}^2 = \sigma^2 + P(1 - x_{t+1})  \label{eq:tau_def} 
 \ee
where $x_{t+1} = x(\tau_t)$, with 
\be
\begin{split}
 x(\tau) := \sum_{\ell=1}^{L} \frac{P_\ell}{P} \, \expec \left[
\frac{\exp\left\{ \frac{\sqrt{n P_\ell}}{\tau} \, \left(U^{\ell}_1  + \frac{\sqrt{n P_\ell}}{\tau}\right)\right\} }{\exp\left\{ \frac{\sqrt{n P_\ell}}{\tau} \, \left(U^{\ell}_1  + \frac{\sqrt{n P_\ell}}{\tau}\right)\right\}  + \sum_{j=2}^M 
\exp\left\{ \frac{\sqrt{n P_\ell}}{\tau}U^{\ell}_j \right\}} \right].
\end{split}
\label{eq:xt_tau_def}
\ee
In \eqref{eq:xt_tau_def}, $\{ U^\ell_j\}$ are i.i.d.\ $\mc{N}(0,1)$ random variables for $j\in [M]$ and $\ell \in [L]$. For consistency, we define $x_0=0$.

The equivalence between  the recursions in  \eqref{eq:tau_t_rec} and \eqref{eq:tau_def} is established by the following proposition. 
\begin{proposition} \cite{RushGV17}
Under the assumption that $\statt= \beta + \tau_t U$, where $U \in \reals^{ML}$ is standard normal  and independent of $\beta$,  the quantity $x_{t+1} = x(\tau_t)$ satisfies
\be
x_{t+1} = \frac{1}{nP} \expec[\beta^* \beta^{t+1}], \quad  1-x_{t+1} = \frac{1}{nP} \expec[ \norm{\beta - \beta^{t+1}}^2],
\label{eq:xt_betatbeta}
\ee
and consequently, \eqref{eq:tau_t_rec} and \eqref{eq:tau_def} are equivalent.
\label{prop:se_cons}
\end{proposition}
\begin{proof}
For convenience of notation, we label the $ML$ components of the standard normal vector $U$ as
$\{U^\ell_j\}_{j \in [M], \ell \in [L]}$.  For any $\ell$, $U^\ell$ denotes the length $M$ vector $\{U^\ell_j\}_{j \in [M]}$. We  have
\be
\begin{split}
& \frac{1}{nP} \expec[\beta^* \beta^{t+1}] = \frac{1}{nP} \expec[\beta^* \,  \eta^t(\beta + \tau_t U)] \\
& \stackrel{(a)}{=}  \  \frac{1}{nP} \sum_{\ell=1}^L \expec[ \sqrt{n P_\ell} \ \eta^t_{\textsf{sent}(\ell)}(\beta_\ell + \tau_t U^\ell)  ] \\
& \stackrel{(b)}{=} \frac{1}{nP} \sum_{\ell =1}^L \expec \left[ \sqrt{nP_\ell} \, 
\frac{\sqrt{nP_\ell} \cdot e^{\sqrt{n P_\ell} (\sqrt{nP_\ell} + \tau_t U^\ell_1)/ \tau^2_t}}
{e^{{\sqrt{n P_\ell} (\sqrt{nP_\ell} + \tau_t U^\ell_1)}/{\tau^2_t}}  + 
\sum_{j =2}^M e^{{\sqrt{n P_\ell}  \tau_t U^\ell_j}/{\tau^2_t}}  } \right] \\
& = \sum_{\ell=1}^{L} \frac{P_\ell}{P} \, \expec \left[
\frac{e^{\frac{\sqrt{n P_\ell}}{\tau_t} \, (U^{\ell}_1  + \frac{\sqrt{n P_\ell}}{\tau_t})} }
{e^{\frac{\sqrt{n P_\ell}}{\tau_t} \, (U^{\ell}_1  + \frac{\sqrt{n P_\ell}}{\tau_t})} + \sum_{j=2}^M e^{\frac{\sqrt{n P_\ell}}{\tau_t}U^{\ell}_j} } \right] = x_{t+1}.
\end{split}
\label{eq:Ebbt}
\ee
In $(a)$ above, the index of the non-zero term in section $\ell$ is denoted by $\textsf{sent}(\ell)$. Step $(b)$ is obtained by assuming that $\textsf{sent}(\ell)$ is the first entry in section $\ell$ --- this assumption is valid because the prior on $\beta$ is uniform over $\mcb(P_1,\ldots,P_L)$. 

Next  consider
\be
\frac{1}{nP} \expec[ \norm{\beta - \beta^{t+1}}^2 ] = 1 + \frac{ \expec[\norm{\beta^{t+1}}^2]- 2 \expec[\beta^* \beta^{t+1}]}{nP}.
\label{eq:betat_beta_sq}
\ee
Under the assumption that $\statt= \beta + \tau_t Z$, recall from Section \ref{subsec:iter_soft_dec} that $\beta^{t+1}$ can be expressed as $\beta^{t+1} = \expec[\beta \mid \statt]$. 
We therefore have
\begin{align}
& \expec[\norm{\beta^{t+1}}^2] = \expec[ \,  \norm{\expec[ \beta | \statt]}^2 \, ]=   \expec[ \, (\expec[ \beta | \statt] - \beta + \beta)^* \expec[ \beta | \statt ]]  \nonumber \\
& \stackrel{(a)}{=} \expec[ \, \beta^* \expec[ \beta | \statt]  \,] = \expec[ \, \beta^* \beta^{t+1}],
\label{eq:betat_sq}
\end{align}
where step $(a)$ follows because $ \expec[ \, (\expec[ \beta | \statt] - \beta)^* \expec[ \beta | \statt] \, ] =0$ due to the orthogonality principle. Substituting \eqref{eq:betat_sq} in \eqref{eq:betat_beta_sq} and using \eqref{eq:Ebbt} yields \[ \frac{1}{nP} \expec[ \norm{\beta - \beta^{t+1}}^2 ] = 1 - \frac{ \expec[ \, \beta^* \beta^{t+1} \, ]}{nP}  = 1 - x_{t+1}. \]
\end{proof}

The parameter $x_t$ can be interpreted as the power-weighted fraction of sections correctly decodable after step $t$: starting from $x_0=0$. The recursion defined by \eqref{eq:xt_tau_def} and \eqref{eq:tau_def} to compute the parameters $(x_t, \tau_t^2)_{t=0,1,\ldots}$ is called \emph{state evolution}. This terminology is due to the similarity with density evolution, the recursion used to predict the performance of LDPC codes \cite{RichUBook}.  

Figure \ref{fig:SE_hard_dec_example}  on page \pageref{fig:SE_hard_dec_example} shows the progression of $(1-x_t)$ for soft-decision decoding in dashed lines, alongside the  solid lines for  hard-decision decoding. For the soft-decision case, $x_t$ is recursively computed using the state evolution recursion in \eqref{eq:xt_tau_def} and \eqref{eq:tau_def}.	
 As we do not make hard decisions on decoded columns until the end, there is no false alarm rate to be controlled in each iteration. If the iterative soft decision decoder is run for $T$ steps, we  wish to ensure that $x_T$ is as close to one as possible, implying that   the expected squared error $\frac{1}{n} \expec \norm{\beta - \beta^{T}}^2 \approx 0$ under the distributional assumption for $\statt$.  

Figure \ref{fig:SE_example} shows the progression of the MSE $\tfrac{1}{n} \| \beta - \beta^t \|^2$ for 200 trials of the AMP decoder (green curves); it is seen that the average is closely tracked by 
$(1-x_t)$ (black curve). The theoretical analysis of the soft-decision decoders discussed in the next two sections shows that the decoding performance of the soft-decision decoders in each step $t$ is closely tracked by the parameter $x_t$ as the SPARC parameters $(L,M,n)$ grow large.

\begin{figure}[t]
\centering
\includegraphics[width=3.4in]{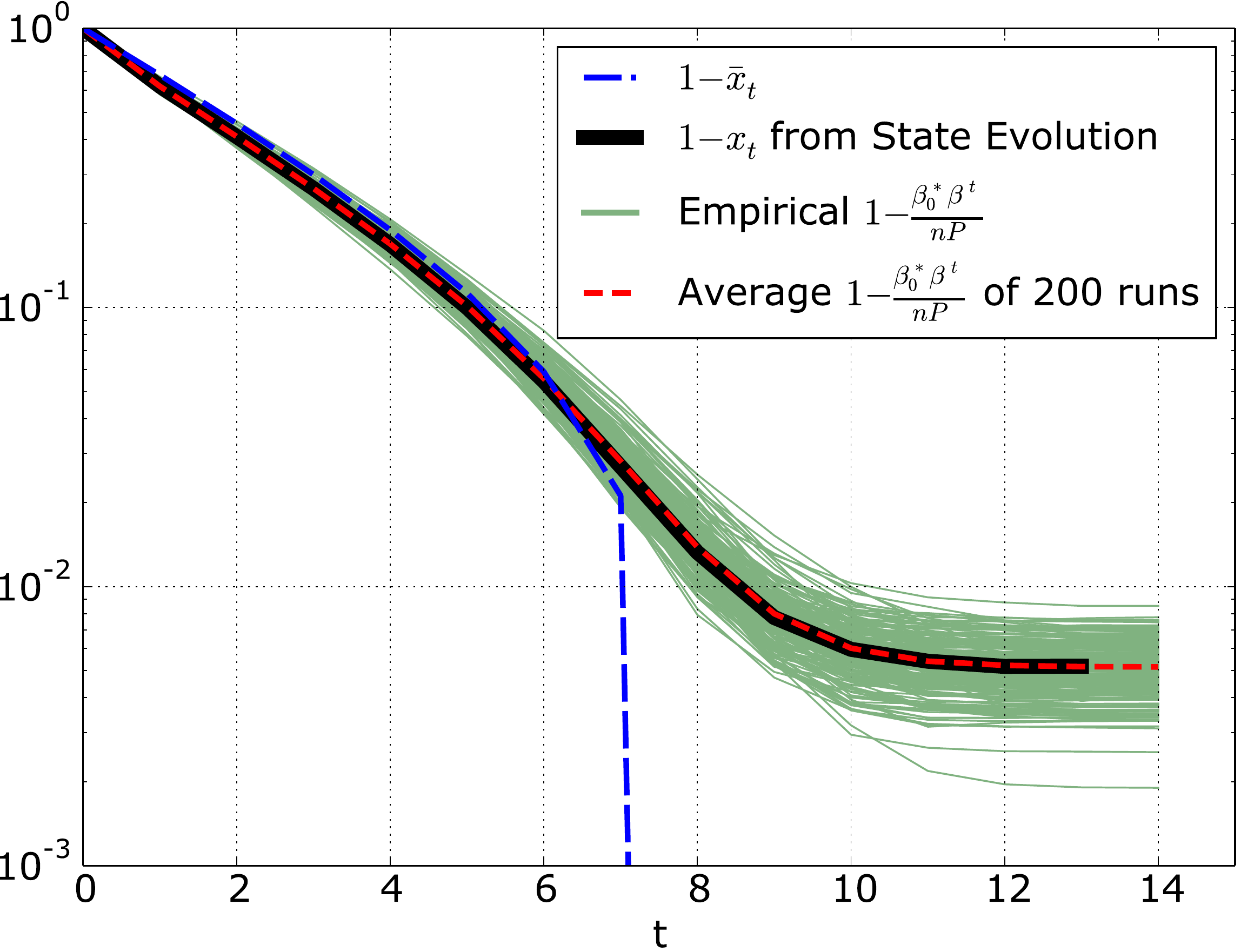}
\caption{\small{Comparison of state evolution predictions with AMP performance. The SPARC parameters are $M=512, L=1024, \snr =15, R=0.7 \mc{C}, \, P_\ell \propto e^{-2\mc{C}{\ell}/{L}}$ with $\mc{C}$ in nats. The average of the $200$ trials (green curves) is the dashed red curve, which is almost indistinguishable from the state evolution prediction (black curve).}}
\vspace{-7pt}
\label{fig:SE_example}
\end{figure}

The following lemma specifies the state evolution recursion in the large system limit, i.e., as $L, M, n \to \infty$ such that $L \log M = nR$. We denote this limit by $\lim$.
\begin{lemma} \cite[Lemma 1]{RushGV17}
For any power allocation $\{ P_\ell \}_{\ell=1, \ldots, L}$ that is non-increasing with $\ell$, we have
\be
\bar{x}(\tau) := \lim x(\tau) = \lim \,  \sum_{\ell=1}^{ \lfloor \xi^*(\tau) L \rfloor} \frac{P_\ell}{P}, 
\label{eq:xt_tau_def1}
\ee
where $\xi^*(\tau)$ is the supremum of all $\xi \in (0,1]$ that satisfy
\[ \lim L P_{ \lfloor \xi L \rfloor} >  2 R \, \tau^2. \]
 If $ \lim L P_{ \lfloor \xi L \rfloor} \leq  2 R \, \tau^2$ for all $\xi > 0$, then $\bar{x}(\tau)=0$. (The rate $R$ is measured in nats.)
\label{lem:conv_expec}
\end{lemma}
\begin{proof} 
In Sec. \ref{subsec:conv_exp_proof}.
\end{proof}

Recalling that ${x}_{t+1}=x(\tau_t)$ is the expected power-weighted fraction of correctly decoded sections after step $(t+1)$,   for any power allocation $\{P_\ell \}$, Lemma \ref{lem:conv_expec}  can be interpreted as follows: in the large system limit,  sections $\ell$ such that $\ell \leq  \lfloor \xi^*(\bar{\tau}_t) L \rfloor$ will be correctly  decodable in step $(t+1)$, i.e., the soft-decision decoder will assign most of the posterior probability mass to the correct  term. Conversely all sections whose power falls below the threshold will not be decodable in this step.

For the  exponentially decaying power allocation in \eqref{eq:exp_power_alloc}, we have for $\xi \in (0,1]$:
\be
\lim L P_{ \lfloor \xi L \rfloor}  =  \sigma^2 (1+\snr)^{1-\xi} \ln(1+\snr).
\label{eq:cell}
\ee
Using this in Lemma \ref{lem:conv_expec} yields the following result. 
\begin{lemma}  \cite[Lemma 2]{RushGV17}
For the power allocation $\{ P_\ell \}$ given in \eqref{eq:exp_power_alloc}, we have for $t=0,1,\ldots$:
\begin{align}
\bar{x}_{t} & := \lim x_{t} = \frac{ (1+ \snr) - (1+ \snr)^{1- \xi_{t-1}}}{\snr} \label{eq:limxt1}, \\
\bar{\tau}^2_{t}&  := \lim \tau^2_{t}  = \sigma^2 + P(1 - \bar{x}_t) = \sigma^2\left( 1 + \snr \right)^{1-\xi_{t-1}} \label{eq:limtaut1}
\end{align}
where $\xi_{-1}=0$, and for $t \geq 0$,
\be
\begin{split}
 \xi_{t} & = \min \left\{ \left(\frac{1}{2\mc{C}}\log\left(\frac{\mc{C} }{R}\right) +  \xi_{t-1}\right), \ 1  \right\}.
\end{split}
\label{eq:lim_alph}
\ee
\label{lem:lim_xt_taut}
\end{lemma}
\begin{proof} 
The result is obtained by applying Lemma \ref{lem:conv_expec} with the exponential power allocation, and using induction on $t$.
\end{proof}
A direct consequence of \eqref{eq:limxt1} and \eqref{eq:lim_alph} is that $\bar{x}_t$ strictly increases with $t$ until it reaches one, and the number of steps $T^*$ until  $\bar{x}_{T^*}=1$  is
$T^* = \left\lceil \frac{2 \mc{C}}{\log(\mc{C}/R)} \right\rceil$.

The constants $\{ \xi_t \}_{t\geq 0}$ have a nice interpretation in the large system limit: at the end of step $t+1$, the first $\xi_t$ fraction of sections in $\beta^{t+1}$ will be correctly decodable with high probability. The other $(1- \xi_t)$ fraction of sections \emph{will not} be correctly decodable from $\beta^{t+1}$ as the power allocated to these sections is not large enough. An additional $\tfrac{1}{2\mc{C}}\log\left(\tfrac{\mc{C} }{R}\right)$ fraction of sections become correctly decodable in each step until step $T^*$, when all the sections are correctly decodable with high probability. 

The discussion in this section --- starting from the way the estimates $(\beta^t)_{t \geq 1}$ are generated, up to the interpretation of the state evolution parameters $(x_t, \tau_t^2)_{t \geq 0}$ --- has been based on the assumption that the decoder has available test statistics of the form $\statt = \beta + \tau_t Z$ at the end of each iteration. In the next two sections, we will describe two decoders which produce $\statt$ of approximately this form when the SPARC parameters $(L,M,n)$ are sufficiently large.

\section{Adaptive successive soft-decision decoder} \label{sec:adap_soft_dec}

The soft-decision decoder proposed by Barron and Cho \cite{BarronC12,choBarron13,choThesis} computes the test statistic $\statt$ at the end of step $t$ as a function of $(A, y, \tsf{Fit}_1, \ldots, \tsf{Fit}_t)$, where we recall $\tsf{Fit}_t = A\beta^t$.   As in hard decision decoding (see p. \pageref{eq:g1g2def}), starting with $G_0 \triangleq y$, let 
 $G_t$ be the part of $\tsf{Fit}_t$ that is orthogonal to $G_0, \ldots, G_{t-1}$.
Then the collection
\[
\frac{G_0}{\norm{G_0}},   \frac{G_1}{\norm{G_1}}, \ldots, \frac{G_t}{\norm{G_t}} 
\]
forms an orthonormal basis for $y, \tsf{Fit}_1, \ldots, \tsf{Fit}_t$.  Also define, for $t \geq 0$: 
\be
\mc{Z}_t = \sqrt{n} \, \frac{A^*G_t}{\norm{G_t}}
\label{eq:Zt_def}
\ee
We  compute a linear combination of $\mc{Z}_0, \ldots, \mc{Z}_t$ given by
\be
\mc{Z}_t^{\tsf{comb}}= \lambda_0 \mc{Z}_0 + \ldots + \lambda_t \mc{Z}_t, 
\label{eq:Ztcomb_def}
\ee
where $\lambda_0, \ldots, \lambda_t$ are coefficients chosen such that $\sum_{k=0}^t \lambda_k^2=1$. 
(These coefficients may depend on $A$ and $y$.)  The adaptive successive soft-decision decoder then computes the statistic 
$\statt = \tau_t \mc{Z}^{\tsf{comb}}_t + \beta^t$, where $\tau_t$ is the state evolution parameter defined in \eqref{eq:tau_def} and $\beta^t$ is the estimate at the end of step $t$.  The new estimate is generated as $\beta^{t+1} =\eta_t(\statt)$, where $\eta_t$ is defined in \eqref{eq:etat_def}. The algorithm is summarized in Fig. \ref{fig:adap_sd_alg}.

The key question  is: how do we choose coefficients $\un{\lambda}_t = (\lambda_0, \ldots, \lambda_t)$ such that $\statt$ 
 has the desired representation $\statt \approx \beta + \tau_t Z$.   To answer this, we use the following lemma which specifies the conditional distribution of the components $\mc{Z}_t$ defined in \eqref{eq:Zt_def}. We need some definitions before stating the result.
 
 \begin{figure}[t]
 \begin{tcolorbox}
\textbf{Step $0$}: Initialize $\beta^0 = 0$ and $G_0 =y$.

\vspace{5pt}
\textbf{Step $t+1$}, for $0 \leq t \leq (T-1)$: 
\begin{enumerate}
\item Compute $\tsf{Fit}_t = A\beta^t$

\item If $t \geq 1$, compute $G_t$, the orthogonal projection of $\tsf{Fit}_t$ onto the space orthogonal to $G_0, \ldots, G_{t-1}$.

\item Compute $\mc{Z}_t  = \sqrt{n} \, {A^*G_t}/{\| G_t \|}$, and 
\begin{align*}
\mc{Z}_t^{\tsf{comb}}  = \lambda_0 \mc{Z}_0 + \ldots + \lambda_t \mc{Z}_t, 
\end{align*}
where $(\lambda_0, \ldots, \lambda_t)$ are given by  \eqref{eq:lambda_det}.

\item Compute  $\statt = \tau_t \mc{Z}^{\tsf{comb}}_t + \beta^t$ where $\tau_t$ is given by \eqref{eq:tau_def}.

\item Generate the updated estimate $\beta^{t+1} =\eta_t(\statt)$, where $\eta_t$ is defined in \eqref{eq:etat_def}.
\end{enumerate}
The number of iterations $T$ is determined using the state evolution recursion, as discussed on p. \pageref{thm:choBar}.
 \end{tcolorbox}
 \caption{Adaptive successive soft-decision decoder with deterministic coefficients of combination.}
 \label{fig:adap_sd_alg}
\end{figure}

Let $b_{0,e},b_{1,e},\ldots,b_{t,e} \in \reals^{ML+1}$  be the successive {orthonormal} components of the length of the   \emph{extended} vectors 
\begin{equation}
\beta_e := \begin{bmatrix} \beta \\ \sqrt{n} \, \sigma \end{bmatrix},
\quad 
\beta^1_e:=\begin{bmatrix} \beta^1 \\ 0   \end{bmatrix},
\quad
\ldots,
\quad
\beta^t_e := \begin{bmatrix}  \beta^t \\ 0  \end{bmatrix}.
\label{eq:ext_defs}
\end{equation}
Let $b_{0},\ldots,b_t \in \reals^{ML}$ be the vectors formed  from the upper $ML$ coordinates of $b_{0,e},\ldots,b_{t,e}$. Let $\Sigma_{t,e}= \sfi -b_{0,e}b_{0,e}^*-b_{1,e}b_{1,e}^* \ldots-b_{t,e} b_{t,e}^*$
denote the $(ML+1) \times (ML+1)$ projection matrix onto the space orthogonal to the vectors in \eqref{eq:ext_defs}.  The upper left  $ML \times ML$ portion of this matrix is denoted $\Sigma_t$.

\begin{lemma} \cite[Lemma 1]{BarronC12}
For $t \geq 0$, let $$\mc{F}_{t-1}=(\mc{Z}_0, \norm{G_0}, \ldots, \mc{Z}_{t-1}, \norm{G_{t-1}}),$$ with $\mc{F}_{-1}$ being the empty set. Then for $t\geq 0$, given $\mc{F}_{t-1}$, the conditional distribution $\pr_{\mc{Z}_t | \mc{F}_{t-1}}$ of $\mc{Z}_t$ is determined by the representation 
\be
\mc{Z}_t = b_t \frac{\norm{G_t}}{\varsigma_t} + Z_t,
\ee
where $Z_t$ has conditional distribution $\mc{N}(0, \Sigma_t)$. Here, $\varsigma_0^2 = \sigma^2 +P$ and for $t \geq 1$ it is $\varsigma_t^2 = {\hat{\beta}^*}_{t}\varsigma_{t-1} \beta^t$. Moreover, $\norm{G_t}^2/\varsigma_t^2$ is distributed as a $\chi_{n-t}^2$ random variable independent of $Z_t$ and $\mc{F}_{t-1}$. 
\label{eq:key_cond_dist}
\label{lem:cond_dist_adap}
\end{lemma}

As number of iterations of the algorithm is small compared to $n$, $\frac{\norm{G_t}}{\varsigma_t}$ is close to $\sqrt{n}$. The lemma tells us that for each $t$, $\mc{Z}$ is approximately equal to $\sqrt{n} b_t$ plus a standard normal vector. We use this property to choose  coefficients $(\lambda_0, \ldots, \lambda_t)$ which lead to $\statt$ having the desired form.

\paragraph{Idealized coefficients.} Consider the choice 
$\un{\lambda}_t^\tsf{id} = (\lambda_0, \ldots, \lambda_t)$ given by 
\be
\un{\lambda}_t^\tsf{id}  = \frac{1}{c^\tsf{id}_t} \left( (\sqrt{n(P+\sigma^2)} - b_0^*\beta^t), \, -b_1^*\beta^t, \ldots, -b_t^*\beta^t \right),
\label{eq:lambda_ideal}
\ee
where $c^\tsf{id}_t$ is a normalizing constant to ensure that $\sum_{k} \lambda_k^2 =1$.
Since $b_0= \beta/\sqrt{n(\sigma^2+P)}$ and the decoder does not know $\beta$, this choice of coefficients cannot be used in practice. We call these idealized coefficients because understanding the test statistic produced by these will help us design good deterministic or observation-based coefficients.   

Recalling that $b_{0,e}, \ldots, b_{t,e}$ form an orthonormal basis, the  normalizing constant  in \eqref{eq:lambda_ideal} is computed as 
\be
\begin{split}
 (c^\tsf{id}_t )^2  &  =  (\sqrt{n(P+\sigma^2)} - b_0^*\beta^t)^2 +  ( -b_1^*\beta^t)^2 + \ldots + (-b_t^*\beta^t)^2 \\
& = n(\sigma^2+P)  +  \norm{\beta^t}^2  - 2 \sqrt{P+\sigma^2} \frac{\beta^*  \beta^t}{\sqrt{P+\sigma^2} }  = n\sigma^2  + \norm{\beta -\beta^t}^2,
\end{split}
\ee
where we have used the fact that $\norm{\beta}^2=nP$.

Let us now examine the distributional properties of the $\mc{Z}_t^{\tsf{comb}}$ generated using these idealized coefficients via \eqref{eq:Ztcomb_def}.   Lemma \ref{lem:conv_expec} tells us that given $\mc{F}_{k-1}$,  $\mc{Z}_k$ is closely approximated by $\sqrt{n} b_k+ Z_k$ with $Z_k$ standard normal, for $0\leq k \leq t$. Therefore, with the idealized coefficients we obtain
\begin{align}
&  \mc{Z}^{\tsf{comb}}_t  = \lambda_0 \mc{Z}_0 + \ldots  + \lambda_t \mc{Z}_t   \nonumber \\
& \stackrel{d}{\approx} \frac{ \sqrt{n(P+\sigma^2)} \, \sqrt{n} b^0 - \left[ (b_0^*\beta^t)\sqrt{n}b_0 + \ldots  + (b_t^*\beta^t) \sqrt{n}b_t \right] }{c^{\tsf{id}}_t} + Z \nonumber \\
& = \frac{\sqrt{n} (\beta - \beta^t)}{c^{\tsf{id}}_t} +  Z = \frac{\beta - \beta^t}{\sqrt{\sigma^2 + \| \beta -\beta^t \|^2/n}} +  Z, 
\label{eq:Zcomb_dist}
\end{align}
where $Z \in \reals^{ML}$ is  standard normal.  If we assume (via an induction hypothesis) that $\tsf{stat}_{t-1}= \beta + \tau_{t-1}Z'$ for a standard normal vector $Z' \in \reals^{ML}$, then Proposition \ref{prop:se_cons} tells us that $\tfrac{1}{n} \expec[\| \beta - \beta^t \|^2] = P(1 - x_t)$. Therefore, for large $n$,  the term 
\ben \sigma^2 + \frac{\| \beta -\beta^t \|^2}{n} \approx \sigma^2+ P(1-x_t) = \tau_t^2. \een

Hence the statistic $\statt = \tau_t \mc{Z}^{\tsf{comb}}_t + \beta^t$ has the following approximate representation:
\be
\statt = \tau_t \mc{Z}^{\tsf{comb}}_t + \beta^t \approx \sqrt{\sigma^2 + \frac{\| \beta -\beta^t \|^2}{n}} \, 
\mc{Z}^{\tsf{comb}}_t + \beta^t  \stackrel{d}{=} \beta + \tau_t Z,
\label{eq:statt_gen}
\ee
where the distributional representation follows from \eqref{eq:Zcomb_dist}.

\paragraph{Deterministic coefficients.} Since $b_0 = \beta/\sqrt{n(\sigma^2+P)}$ is unknown,  the idealized coefficients in \eqref{eq:lambda_ideal} cannot be used to produce $\statt$.  We now specify a deterministic choice for $\un{\lambda}_t$ which mimics the idealized coefficients using deterministic proxies for the inner products $b_0^*\beta^t, b_1^*\beta^t, \ldots, b_t^*\beta^t$.  Recall that the vectors $b_{0,e}, \ldots, b_{1,e}$ are obtained by performing a successive orthonormalization on the extended vectors $(\beta_e, \beta^1_e, \ldots, \beta^t_e)$ defined in \eqref{eq:ext_defs}. We therefore have
\begin{equation}
B:= \left[ \beta_e \  \beta^1_e, \ldots \  \beta^t_e \right] = \left[ b_{0,e}  \quad b_{1,e} \ \ldots \ b_{k,e} \right]
\begin{bmatrix} 
b_{0,e}^* \beta_e & b_{0,e}^* \beta^1_e & \ldots & b_{0,e}^* \beta^t_e \\
0  & b_{1,e}^* \beta^1_e & \ldots &  b_{1,e}^* \beta^t_e  \\
\vdots & \vdots & \ddots  & \quad  \vdots  \quad  \\
0  & 0& \ldots &  b_{t,e}^* \beta^t_e
\end{bmatrix}.
\end{equation}
Noting that the last entry in each of $\beta^1_e, \ldots, \beta^t_e$ is zero, we observe that
\begin{align}
 B^* B = 
\begin{bmatrix}
{\beta_e^* \beta_e} & \beta^* \beta^1 & \ldots & \beta^* \beta^t \\ 
(\beta^1)^* \beta & {(\beta^1)^* \beta^1} & {\stackrel{}{\ldots}} & {(\beta^1)^* \beta^t} \\
\vdots & {\ \quad \vdots \  \quad} & { \ddots  } &  {\ \quad \vdots \ \quad } \\
(\beta^t)^* \beta &  {\stackrel{}{\quad \ldots \quad}}  &  {\stackrel{}{\ldots}}  & { (\beta^t)^* \beta^t}
\end{bmatrix}
= R^* R,
\label{eq:RTR}
\end{align}
where
\be
R = \begin{bmatrix} 
b_{0}^* \beta_e & b_{0}^* \beta^1 & \ldots & \highlightgr{b_{0}^* \beta^t} \\
0  & b_{1}^* \beta^1 & \ldots &  \highlightgr{b_{1}^* \beta^t}  \\
\vdots & \vdots & \ddots  & \quad  \vdots  \quad  \\
0  & 0& \ldots &  \highlightgr{b_{t}^* \beta^t}
\end{bmatrix}.
\label{eq:Rmat_def}
\ee

The high-level idea in obtaining the deterministic coefficients is as follows. Observe that the entries in the last column of $R$ (highlighted) are exactly those that are required to compute the idealized weights of combination in \eqref{eq:lambda_ideal}. These can be estimated by replacing each  entry of the matrix in \eqref{eq:RTR} with its idealized (deterministic) value, and then computing the Cholesky decomposition of this matrix.  The last column of the resulting upper triangular matrix then provides a deterministic proxy
 for the highlighted terms in \eqref{eq:Rmat_def}.

In detail, to obtain the deterministic coefficients we use the following result implied by Proposition \ref{prop:se_cons}:  under the assumption (via an induction hypothesis) that 
$\tsf{stat}^k = \beta + \tau_{k} Z'_k$ for $0\leq k \leq (t-1)$, we have 
\be
\frac{1}{n}\expec[\beta^* \beta^k] = \frac{1}{n}\expec[(\beta^m)^* \beta^k] =P x_k, \quad 1\leq k \leq m \leq t.
\ee
Therefore, replacing each entry of $B^* B/n$ by its expected value, we obtain the matrix
\be
\begin{bmatrix}
\tau_0^2 & x_1 P & \ldots & x_t P \\ 
x_1 P  & x_1 P & {\stackrel{}{\ldots}} & x_1 P \\
\vdots & {\ \quad \vdots \  \quad} & { \ddots  } &  {\ \quad \vdots \ \quad } \\
x_t P &  {\stackrel{}{\quad \ldots \quad}}  &  {\stackrel{}{\ldots}}  & x_t P
\end{bmatrix}
= \widehat{R}^*\widehat{R},
\ee
where $\widehat{R}$ is  the upper triangular matrix obtained via the Cholesky decomposition. This is found to be
\be
\widehat{R} = \begin{bmatrix} 
\tau_0 & \tau_0 - \tau_1^2 \sqrt{\omega_0} & \ldots & \highlightgr{\tau_0 - \tau_t^2 \sqrt{\omega_0}} \\
0  & \tau_1^2 \sqrt{\omega_1} & \ldots &  \highlightgr{ \tau_t^2 \sqrt{\omega_1} }  \\
\vdots & \vdots & \ddots  & \highlightgr{\quad  \vdots  \quad } \\
0  & 0& \ldots &    \highlightgr{ \tau_t^2 \sqrt{\omega_t} }, 
\end{bmatrix}
\ee
where
\be 
\omega_0 = \frac{1}{\tau_0^2}, \qquad \omega_k = \frac{1}{\tau_k^2} - \frac{1}{\tau_{k-1}^2}, \quad k \geq 1.
\ee 
The last column of $\widehat{R} $ (highlighted) is a deterministic estimate for $(b_0^* \beta^t/\sqrt{n}, \ldots, b_t^*\beta^t/\sqrt{n})$.   This is used to replace  the idealized coefficients in \eqref{eq:lambda_ideal}, yielding the  following deterministic choice for $\un{\lambda}_t$:
\be
\un{\lambda}_t^{\tsf{det}} = (\tau_t \sqrt{\omega_0}, \ - \tau_t \sqrt{\omega_1},  \ \ldots, \ ,- \tau_t \sqrt{\omega_t}), \quad t \geq 0.
\label{eq:lambda_det}
\ee

At the end of each iteration $t \geq 1$, these coefficients are used to first produce $\mc{Z}^{\tsf{comb}}_t$, which is then used to compute $\statt = \tau_t \mc{Z}^{\tsf{comb}}_t + \beta^t$. The updated estimate of the message vector in step $(t+1)$ is $\beta^{t+1} =\eta_t(\statt)$, where $\eta_t$ is defined in \eqref{eq:etat_def}. 

The performance of the adaptive successive decoder with deterministic coefficients in \eqref{eq:lambda_det} is given by the following theorem.

\begin{theorem} \cite[Lemma 11]{choThesis} \cite[Lemma 7]{choBarron13}
Consider a SPARC with a rate $R <\mc{C}$, parameters $(n,L, M)$ chosen according to \eqref{eq:ml_nR}, and power allocation $P_\ell \propto e^{-2\mc{C} \ell /L}$. For $t \geq 1$, let
$$\mc{A}_t : = \left\{ \abs{ \frac{1}{nP} \beta^* \beta^t-x_t }> \e  \right\} \cup \left\{ \abs{ \frac{1}{nP} \| \beta^{t} \|^2-x_t)} > \e \right\},$$
where $x_t$ is defined by the state evolution recursion in \eqref{eq:tau_def} and \eqref{eq:tau_def}. Then,
\be
\mathbb{P}\{ \cup_{k=1}^t \mc{A}_{k} \} \leq 
\sum_{k =1}^{t} 6(k+1) \exp \left(  \frac{-2L \e^2}{ c^2 (\log M /R)^{2k-1}} \right),
\label{eq:adap_succ}
\ee
where $c^2 = \max_\ell L P_\ell /P$, which is a constant close to $2 \mc{C} (1+ \snr)/\snr$ for large $L$.
\label{thm:choBar}
\end{theorem}

We run the decoder for $T$ steps, where $T$ can be determined using the SE recursion in \eqref{eq:tau_def}  as  the minimum number of steps after which $1-x_T$ is below a specified small value $\delta$. Or, using the asymptotic SE characterization in Lemma \ref{lem:lim_xt_taut}, we can take 
$T= T^* = \left\lceil \frac{2 \mc{C}}{\log(\mc{C}/R)} \right\rceil$. The large deviations bound for $ \frac{1}{nP} \beta^* \beta^T$ and $ \frac{1}{nP} \| \beta^{T} \|^2$ in  \eqref{eq:adap_succ} can then be translated into a bound on the excess section error rate. We defer the explanation of how this is done to the next section   where we analyze the AMP decoder.   (See Eq. \eqref{eq:sec_error_implies} and the surrounding discussion.)

As an alternative to the deterministic coefficients of combination, Cho and Barron  \cite{choBarron13,choThesis}  propose another method of  choosing coefficients based on the Cholesky decomposition of $B^*B=RR^*$ in \eqref{eq:RTR}. This method uses the known values of $(\beta^k)^*\beta^m$, $1 \leq k \leq m \leq t$, in the matrix $B^*B$ and estimates based on Lemma \ref{lem:cond_dist_adap} for the diagonal entries of $R$ to recursively solve for the 
$(b_{0}^* \beta^t, \ldots, b_{t-1}^* \beta^t)$. These resulting values are then used to generate 
$\mc{Z}^{\tsf{comb}}_t$  via \eqref{eq:lambda_ideal}.  The reader is referred to \cite[Sec. 4.3]{choBarron13} or \cite[Sec. 4.3]{choThesis} for details of the Cholesky decomposition based estimates and the corresponding performance analysis.

\section{Approximate Message Passing (AMP) decoder} \label{sec:AMP_dec}

Approximate message passing (AMP) refers to a class of algorithms \cite{DonMalMont09, MontChap11, BayMont11, BayLASSO, krz12,Rangan11,DonSpatialC13}
that are Gaussian or quadratic approximations of  loopy belief propagation algorithms (e.g., min-sum, sum-product) on dense factor graphs. In its basic form \cite{DonMalMont09,BayLASSO},  AMP gives a fast iterative algorithm to solve the LASSO \cite{tibshirani1996regression}  under certain conditions on the design matrix. The LASSO is the following convex optimization problem. Given a  matrix $A \in \reals^{n \times N}$,   an observation vector $y \in \reals^n$, and a scalar $\lambda >0$,  compute
\be  
 \argmin_{\hat \beta \in \reals^N} \ \norm{y-A \hat \beta}^2_2 + \lambda \norm{\hat \beta}_1,
 \label{eq:lasso_def}
 \ee
 where the $\ell_1$-norm is defined as $\norm{\hat \beta}_1 = \sum_{i=1}^N {\hat \beta}_i$.
The $\ell_1$ penalty added to the least-squares term promotes sparsity in the solution. The LASSO has been widely used  in applications such as compressed sensing and sparse linear regression; see, e.g., \cite{tibshirani2015statistical}.

When $A$ has i.i.d. entries drawn from a Gaussian or sub-Gaussian distribution, AMP  has been found to converge to the LASSO solution \eqref{eq:lasso_def}  faster than the best competing solvers (based on first-order convex optimization). This is because AMP takes advantage of the distribution of the matrix $A$, unlike generic convex optimization methods. The AMP also yields sharp results for the asymptotic risk of LASSO with Gaussian matrices \cite{BayLASSO}.

For SPARCs, recall from  \eqref{eq:opt_decoder}  that the optimal decoder solves the optimization problem
\be \hat{\beta}_{\textsf{opt}} = \argmin_{\hat{\beta} \in \mcb} \, \norm{y-A\hat{\beta}}^2.  \ee
One cannot directly use the LASSO-AMP of  \cite{DonMalMont09,BayLASSO} for SPARC  decoding as it does not use the prior knowledge about $\beta$, i.e., the knowledge that $\beta$ has exactly one non-zero value in each section, with the values of the non-zeros also being known.

We start with the factor graph for the  model
$y = A\beta + w$,
where $\beta \in \mcb(P_1, \ldots, P_L)$. Each row of $A$ corresponds to a constraint (factor) node, while each column corresponds to a variable node. We use the indices $a,b$ to denote factor nodes, and indices $i,j$ to denote variable nodes. The AMP algorithm is obtained via  a first-order approximation to the following message passing updates that iteratively computes estimates of $\beta$ from $y$.  For $i \in [N]$, $a \in [n]$, set $\beta^0_{ j \to a}=0$,  and compute the following for $t \geq 0$:
\begin{align}
z^t_{a \to i}  & = y_a - \sum_{j \in [N] \bks i} A_{aj}  \beta^t_{j \to a},  \label{eq:z_update} \\
\beta^{t+1}_{i \to a} & = \eta_i^t \left( \tsf{stat}_{a \to i}  \right),
\label{eq:beta_update}
\end{align}
where $\eta_i^t (\cdot)$ is the estimation function defined in \eqref{eq:etat_def}, and for $i \in \text{sec}(\ell)$, the entries of the test statistic ${\tsf{stat}}_{i \to a} \in \mathbb{R}^M$ are defined as
\be
\begin{split}
(\tsf{stat}_{a \to i})_i &  = \sum_{b \in [n] \bks a}  A_{bi} z^t_{b \to i},  \\
(\tsf{stat}_{a \to i})_j & =  \sum_{b \in [n] }  A_{bj} z^t_{b \to j}, \quad j \in \text{sec}(\ell) \bks i.
\end{split}
\ee
As before, the update function \eqref{eq:beta_update} is based on the assumption that $\tsf{stat}_{a \to i} \approx \beta + \tau_t Z$.  In \eqref{eq:z_update}, note that the dependence of $z^t_{a \to i}$ on $i$ is only due to the  term $A_{ai} \beta^t_{i \to a}$ being excluded from the sum. Similarly, in \eqref{eq:beta_update} the dependence of $\beta^t_{i \to a}$ on $a$ is  due to excluding the term $A_{ai} z^t_{a \to i}$ from the argument. 
We therefore write
\be
z^{t}_{a \to i} = z^t_a + \delta z^t_{a \to i}, \quad \text{ and } \quad \beta^{t+1}_{i \to a} = \beta^{t+1}_{i} + \delta\beta^{t+1}_{i \to a}.
\label{eq:delta_defs}
\ee
Using a first-order Taylor approximation for the updates \eqref{eq:z_update} and \eqref{eq:beta_update}  around the terms $z^t_a$ and $\beta^{t+1}_{i}$ and simplifying yields the  AMP decoding algorithm which produces iterates $(z^t, \beta^{t+1})$ in each iteration. The AMP algorithm is described in Fig. \ref{fig:AMP_alg}. (See \cite[Appendix A]{RushGV17} for details of the derivation.)

\begin{figure}[t]
 \begin{tcolorbox}
\textbf{Step $0$}: Initialize $\beta^0 = 0$ and $z^{-1} =0$.

\vspace{5pt}
\textbf{Step $t+1$}, for $0 \leq t \leq (T-1)$: 
Compute 
\begin{align}
& z^t = y - A \beta^t + \frac{z^{t-1}}{\tau^2_{t-1}}\left( P - \frac{\norm{\beta^t}^2}{n} \right),  \label{eq:amp1} \\ 
& \statt = A^* z^t + \beta^t,   \label{eq:amp12}  \\
& \beta^{t+1} = \eta_t(\statt).   \label{eq:amp2}
\end{align}
where the constants $(\tau_{t}^2)_{t \geq 0}$ required in  \eqref{eq:amp1}  and \eqref{eq:amp2} are given by the SE recursion described  in \eqref{eq:tau_def}. Instead of pre-computing $\tau_{t}^2$, it can also be estimated online as $\| z^t \|^2/n$ (see Sec. \ref{subsec:online_tau}).
The number of iterations $T$ is  determined either using the state evolution recursion as discussed on p. \pageref{thm:choBar}, or using the termination criterion in Sec. \ref{subsec:online_tau}. 
 \end{tcolorbox}
 \caption{Approximate Message Passing (AMP) Decoder.}
 \label{fig:AMP_alg}
\end{figure}

The vector $z^t$ in \eqref{eq:amp1} is a modified residual: it consists  of the standard residual $y-A\beta^t$, plus an extra term $\frac{z^{t-1}}{\tau^2_{t-1}} ( P - \frac{\norm{\beta^t}^2}{n} )$. This extra `Onsager' term is crucial to ensuring that $\statt$ has the desired distributional property. To get some intuition about the role of the Onsager term, we express $\statt$ as  
\begin{align}
\statt  & = A^* z^t + \beta^t  = A^* (y-A \beta^t) + \beta^t +  \frac{A^* z^{t-1}}{\tau^2_{t-1}}\left( P - \frac{\norm{\beta^t}^2}{n} \right) \nonumber  \\
& = \beta + A^*w + (\tsf{I} - A^*A)(\beta^t - \beta) +     \frac{A^* z^{t-1}}{\tau^2_{t-1}}\left( P - \frac{\norm{\beta^t}^2}{n} \right) 
\label{eq:statt_expand}
\end{align}

We can interpret  the second and third terms on the RHS of \eqref{eq:statt_expand} as noise terms added to $\beta$.  The term $A^*w$ is a random vector independent of $\beta$ with i.i.d $\mc{N}(0, \sigma^2)$ entries. For the next term, the entries of the symmetric matrix
$(\tsf{I} - A^*A)$ can be shown to be approximately $\mc{N}(0, \tfrac{1}{n})$, with distinct entries being approximately pairwise independent.  Therefore, \emph{if} the $(\beta^t - \beta)$ were independent of $A$, then the vector $(\tsf{I} - A^*A)(\beta^t - \beta)$ would be approximately i.i.d.$\sim \mc{N}(0, \tfrac{\| \beta^t -\beta \|^2}{n})$; consequently the second and third terms of \eqref{eq:statt_expand} combined would be close to standard normal with variance
$\sigma^2 +  \frac{\| \beta^t -\beta \|^2}{n} \approx \tau_t^2$. However,   $(\beta^t - \beta)$ is not independent of $A$, since $A$ is used to generate $\beta^1, \ldots, \beta^t$.  The role of the last term in \eqref{eq:statt_expand} is to asymptotically cancel the correlation between $A$ and $(\beta^t-\beta)$, so that $\statt$ is well approximated as  $\beta + \tau_t Z$.  This intuition is made precise in the analysis of the AMP decoder in the next subsection.

\subsection{Analysis of the AMP decoder} We now obtain a non-asymptotic bound on error performance of the AMP decoder.  To do this, we first need a lower bound on how much the state evolution parameter $x_t$ increases in each iteration of the algorithm. 

\begin{lemma}\cite{RushV17ErrExp}
Let $\delta \in (0,  \min\{ \Delta_R, \frac{1}{2} \}]$, where $\Delta_R := (\mc C - R)/\mc C$. Let $f(M) := \frac{M^{-\kappa_2 \delta^2}}{\delta \sqrt{\log M}}$, where $\kappa_2$ is the universal constant in Lemma \ref{lem:conv_expec}(b). Consider the sequence of state evolution parameters $x_0=0, x_1,  \ldots$ computed according to \eqref{eq:tau_def} --\eqref{eq:xt_tau_def} with the exponentially decaying power allocation in \eqref{eq:exp_power_alloc}. For sufficiently large $L,M$, we have:
\be
x_1 \geq  \chi_1 := \left(1 - f(M)\right)  \frac{P+\sigma^2}{P} \left( 1 - \frac{(1+ \delta/2) R}{\mc C} - \frac{5}{L}  \right),
\label{eq:chi1}
\ee
and for $t > 1$:
\begin{align}
   & x_t - x_{t-1}   \nonumber  \\
   & \geq  \chi := \left(1 - f(M)\right)  \left[  \frac{\sigma^2}{P} \left( 1 - \frac{(1+ \delta/2) R}{\mc C} \right) - f(M) \frac{(1+ \delta/2) R}{\mc C}  \right]  \nonumber \\
   & \qquad \qquad  - \frac{5(1 +\sigma^2/P)}{L},    \label{eq:chi}
\end{align} 
until $x_t$ reaches (or exceeds) $(1-f(M))$. 
\label{lem:xt_taut_lb}
\end{lemma}
\begin{proof}
In Section \ref{app:proof_lem_xt_taut_lb}. 
\end{proof}

\paragraph{Number of iterations and the gap from capacity} \label{subsec:del_vs_CR}

We want the lower bounds $\chi_1$ and $\chi$ in \eqref{eq:chi1} and \eqref{eq:chi} to be strictly positive and depend only on the gap from capacity $\Delta_R=(\mc C - R)/\mc C$ as $M,L \to \infty$. 
For all $\delta \in (0, \Delta_R]$, we have
\be
\left(1- \frac{(1+ \delta/2)R}{\mc{C}} \right) \geq  \left(1- \left(1+ \frac{\Delta_R}{2}\right) (1-\Delta_R)  \right) =\frac{\Delta_R +  \Delta_R^2}{2}.
\label{eq:key_term_lb}
\ee
Therefore, the quantities on the RHS of \eqref{eq:chi1} and \eqref{eq:chi}  can be bounded from below as
\begin{align}
\chi_1 &  \geq (1-f(M)) \frac{P+ \sigma^2}{P}\left( \frac{\Delta_R +  \Delta_R^2}{2}- \frac{5}{L} \right),  \label{eq:chi1_lb1} \\
 \chi &  \geq (1-f(M)) \left[ \frac{\sigma^2}{P} \left(  \frac{\Delta_R +  \Delta_R^2}{2}  \right)  - f(M)\right] -  \frac{5(1 +\sigma^2/P)}{L}.
 \label{eq:chi_lb1}
\end{align}
We take $\delta=\Delta_R$, which\footnote{As Lemma \ref{lem:lim_xt_taut} assumes that $\delta \in (0, \min\{\frac{1}{2}, \Delta_R \}]$, by taking $\delta=\Delta_R$ we have assumed that $\Delta_R \leq \frac{1}{2}$, i.e., $R \geq \mc{C}/2$. This assumption can be made without loss of generality --- as the probability of error increases with rate, the large deviations bound of Theorem \ref{thm:main_amp_perf}  evaluated for $\Delta_R = \frac{1}{2}$ applies for all $R$ such that $\Delta_R < \frac{1}{2}$.
}  
gives the smallest value for $f(M)$.  We denote this value by
 \be f_R(M) := \frac{M^{-\kappa_2 \Delta_R^2}}{\Delta_R \sqrt{\log M}}. \label{eq:f_RM} \ee  
From \eqref{eq:chi_lb1}, if  $f_R(M)/\Delta_R \to 0$ as $M \to \infty$, then  $\frac{\sigma^2}{P}\left(\frac{\Delta_R +  \Delta_R^2}{2}\right)$  will be the dominant term in $\chi$ for large enough $L,M$.  The condition $f_R(M)/\Delta_R \to 0$  will be satisfied if we choose $ \Delta_R$ such that 
\be
\Delta_R \geq \sqrt{\frac{\log \log M}{\kappa_2 \log M}},
\label{eq:DelRbnd}
\ee
where $\kappa_2$ is the universal constant of Lemma \ref{lem:lim_xt_taut}.  From here on, we assume that $\Delta_R$ satisfies \eqref{eq:DelRbnd}.

Let $T$ be the number of iterations until $x_t$ exceeds $(1-f_R(M))$. We run the AMP decoder for $T$ iterations, where
\begin{align}
T  := \min_t \, \{ t: \,  x_{t} \geq 1- f_R(M) \} 
&  \stackrel{(a)}{\leq} \frac{(1-f_R(M))}{\chi}   \nonumber \\
 & \stackrel{(b)}{=} \frac{P/\sigma^2}{(\Delta_R + \Delta_R^2)/2}(1 + o(1)),  \label{eq:Tdef}
\end{align}
where $o(1) \to 0$ as $M,L \to \infty$.  In \eqref{eq:DelRbnd}, inequality $(a)$ holds for sufficiently large $L,M$ due to Lemma \ref{lem:xt_taut_lb}, which shows for large enough $L,M$, the $x_t$ value increases by at least $\chi$ in each iteration. The equality $(b)$ follows from the lower bound on $\chi$ in \eqref{eq:chi_lb1}, and because $f(M)/\Delta_R = o(1)$.

After running the decoder for $T$ iterations, the decoded message $\hat{\beta}$ is obtained  by setting the maximum of $\beta^{T}$  in each section $\ell \in [L]$ to $\sqrt{nP_{\ell}}$ and the remaining entries to $0$.    From \eqref{eq:Tdef}, we see that the number of iterations  $T$ increases as $R$ approaches $\mc{C}$.   The definition of $T$ guarantees that
  $x_T \geq (1-f_R(M))$. Therefore, using $\tau_T^2=\sigma+P(1-x_T)$ we have
    \be \sigma^2 \leq \tau^{2}_T  \leq \sigma^2 +  P f_R(M).
  \label{eq:tauT_bound} \ee 

\paragraph{Performance of the AMP decoder}

The main result is a bound on the probability of the section error rate exceeding any fixed $\e > 0$. 

\begin{theorem}\cite{RushV17ErrExp}
Fix any rate $R < \mc{C}$. Consider a rate $R$ SPARC $\mc{S}_n$  with block length $n$, design matrix parameters $L$ and $M$ determined according to  \eqref{eq:ml_nR}, and an exponentially decaying power allocation given by \eqref{eq:exp_power_alloc}. Furthermore, assume that $M$ is large enough that 
\[  \Delta_R \geq \sqrt{\frac{\log \log M}{\kappa_2 \log M}}, \]
where $\kappa_2$ is the universal constant used in Lemmas \ref{lem:conv_expec}(b) and \ref{lem:lim_xt_taut}. Fix any $\e > \frac{2 \snr}{\mc{C}} f_R(M)$, where $f_R(M) := \frac{M^{-\kappa_2 \Delta_R^2}}{\Delta_R \sqrt{\log M}}$.

Then, for sufficiently large $L,M$, the section error rate of the AMP decoder satisfies
\be
\begin{split}
& P \left( \mc{E}_{sec}(\mc{S}_n)  >  \e \right) \leq K_{T} \exp\left\{\frac{-\kappa_{T} L}{(\log M)^{2T-1}}   
\left(\frac{\e  \sigma^2 \mc{C}}{2} - P f_R(M) \right)^2 \right\},
 \label{eq:pezero}
\end{split}  
\ee
where $T$ is defined in \eqref{eq:Tdef}. The constants $\kappa_T$ and $K_T$ in \eqref{eq:pezero} are given by $\kappa_T = [c^{2T} (T!)^{17}]^{-1}$ and $K_T = C^{2T} (T!)^{11}$ where $c, C > 0$ are universal constants (not depending on AMP parameters $L, M, n, $ or $\e$) but are not explicitly specified.  
\label{thm:main_amp_perf}
\end{theorem}

\begin{remark}
The dependence of the constants $K_T, \kappa_T$ on $T!$ is due to the induction-based proof of a key concentration result used in the proof (Lemma \ref{lem:main_lem}). These constants have not been optimized, but we believe that the dependence of  these constants on $T!$ is inevitable in any induction-based proof of the result.
\end{remark}

\subsection{Error exponent and gap from capacity with AMP decoding}

In this subsection we consider the behavior of the bound in Theorem \ref{thm:main_amp_perf} in two different regimes. The first is where $R < \mc{C}$ is held constant as $L, M \to \infty$ (with $n= L \log M /R$) --- this is  the so-called ``error exponent" regime. In this case, $\Delta_R$ is of constant order, so $f_R(M)$ in \eqref{eq:f_RM} decays polynomially with growing $M$.  The other regime is where $R$ approaches $\mc{C}$ as $L,M \to \infty$  (equivalently, $\Delta_R$ shrinks to $0$), while ensuring that the error probability remains small or goes to $0$. Here, \eqref{eq:DelRbnd} specifies that $\Delta_R$ should be of order at least $\sqrt{\frac{\log \log M}{ \log M}}$.

\paragraph{Error exponent} \label{subsec:errexp}
 For any ensemble of codes, the error exponent specifies how the codeword error probability decays with growing code length $n$ for a fixed $R < \mc{C}$ \cite{gallagerBook}.  In the SPARC setting,  we wish to understand how the bound on the probability of excess section error rate in Theorem \ref{thm:main_amp_perf}  decays with $n$ for fixed values of $\e >0$ and  $R < \mc{C}$.   (As explained in Remark 5 following Theorem \ref{thm:main_amp_perf}, concatenation using an outer code can be used to extend the result to  the codeword error probability.)  With optimal encoding, it was shown in \cite{AntonyML} that the  probability of excess section error rate decays exponentially in $n  \min\{ \e \Delta, \Delta^2 \}$, where $\Delta=(\mc{C} -R)$.   For the AMP decoder, we consider two choices for $(M,L)$ in terms of $n$ to illustrate the trade-offs involved: 
 \begin{enumerate}
\item  $M=L^{\textsf{a}}$, for some constant $\textsf{a} >0$. Then, \eqref{eq:ml_nR} implies that $L= \Theta(\frac{n}{\log n})$ and $M=  \Theta((\frac{n}{\log n})^\mathsf{a})$. Therefore,  the bound in Theorem \ref{thm:main_amp_perf}  decays exponentially in $n/(\log n)^{2T}$.  

\item   $L =  {\kappa n}/{\log \log n}$, for some constant $\kappa$, which implies $M = \frac{R}{\kappa} \log n$. With this choice the bound in Theorem \ref{thm:main_amp_perf}  decays exponentially in $n/(\log \log n)^{2T}$.
 \end{enumerate}
Note from \eqref{eq:Tdef}  that  for a fixed $R < \mc{C}$, the number of AMP iterations $T$ is an $\Theta(1)$ quantity that does not grow with $L, M,$ or $n$. The  excess section error rate  decays more rapidly with $n$ for the second choice, but this comes at the expense of much smaller $M$ (for a given $n$).  Therefore, the first choice allows for a much smaller target section error rate (due to smaller $f_R(M)$), but has a larger probability of deviation from the target. One can also compare the two cases in terms of decoding complexity, which is $O(nML T
)$ with Gaussian design matrices. The complexity in the first case is $O({n^{2+\mathsf{a}}}/{(\log n)^{1 + \mathsf{a}}})$, while in the second case it is $O({n^2 \log n}/{\log \log n})$.

\paragraph{Gap from capacity} \label{subsec:scaling}

We now consider how fast  $R$ can approach the capacity $\mc{C}$ with growing $n$, so that the probability of excess section error rate still decays to zero. Recall that lower bound on the gap from capacity is already specified by \eqref{eq:DelRbnd}: for the state evolution parameter $x_T$ to converge to $1$ with growing $M$ (predicting reliable decoding), we need $\Delta_R \geq \sqrt{\frac{\log \log M}{\kappa_2 \log M}}$.  When $\Delta_R$ takes this minimum value, the minimum target section error rate $f_{R}(M)$ in Theorem \ref{thm:main_amp_perf} is  
\be \underline{f_{R}}(M) =  \frac{\sqrt{\kappa_2}}{\log M  \sqrt{\log \log M} }. \label{eq:frmin0} \ee

 We evaluate the large deviations bound of Theorem \ref{thm:main_amp_perf} with $\Delta_R$ at the minimum value of $\sqrt{\frac{\log \log M}{\kappa_2 \log M}}$, for $\e > \frac{2 \snr}{\mc C} \underline{f_R}(M)$, with $\underline{f_R}(M)$ given in \eqref{eq:frmin0}. From \eqref{eq:Tdef}, we have the  bound 
 \be 
 T \leq \frac{2\snr}{\Delta_R}  \leq \kappa_4 \sqrt{\frac{\log M}{\log \log M}}  \label{eq:Tbnd0} 
 \ee for large enough $L,M$. Then, using Stirling's approximation to write $\log(T!) =    T\log T - T + O(\log T)$, Theorem \ref{thm:main_amp_perf} yields
 \begin{align}
&  -  \log P  \left( \mc{E}_{sec}(\mc{S}_n) >  \e \right)  \geq   \frac{\kappa_5 L \e^2}{ c^{2T} (T!)^{17} (\log M)^{2T-1}}  - O(T\log T) \nonumber \\
& =   \frac{\kappa_5 L \e^2}{ \exp\{ 2T \log c + 17(T \log T - T) + (2T-1) \log \log  M  + O(\log T)\}}   - O(T\log T) \nonumber \\
& \geq \frac{L \e^2}{\exp\left\{  \kappa_6 \sqrt{(\log M) (\log \log M)} \left(1 + O(\frac{1}{\log \log M}) \right) \right\}} -  O\left(\sqrt{(\log M)(\log \log M)} \right)
\label{eq:logPe_bnd}
 \end{align}
where the last inequality above follows from \eqref{eq:Tbnd0}.

We now evaluate the bound in \eqref{eq:logPe_bnd} for the case $M=L^{\textsf{a}}$  considered in Sec  \ref{subsec:errexp}. We then we have $L = \Theta(\frac{n}{\log n})$ and $M=  \Theta((\frac{n}{\log n})^\mathsf{a})$. Substituting these in \eqref{eq:logPe_bnd}, we obtain  
 \begin{align}
  -  \log P  \left( \mc{E}_{sec}(\mc{S}_n) >  \e \right)  &  \geq \frac{\kappa_7 n \e^2}{(\log n) \exp\{ \kappa_8 \sqrt{(\log n) (\log \log n)}\}}  \nonumber \\
 &  = \kappa_7 \exp \left\{ \log n -  \kappa_8\sqrt{(\log n) (\log \log n)} - \log \log n \right \} \e^2 \nonumber \\
 &  = \kappa n^{1 - O\left(\sqrt{\frac{\log \log n}{\log n}} + \frac{\log \log n}{\log n} \right)} \e^2.
 \end{align}
 Therefore, for the case $M= L^{\textsf{a}}$, we can achieve a probability of  excess section error rate that decays as
 $\exp \left\{-\kappa n^{1 - O\left(\sqrt{{\log \log n}/{\log n}}\right)} \e^2 \right\}$, with a gap from capacity ($\Delta_R$) that is of order $\sqrt{\frac{\log \log n}{\log n}}$. Furthermore, from  \eqref{eq:frmin0} we see that $\e$ must be of  order at least $\frac{1}{\log n  \sqrt{\log \log n}}$.
 
 We note that this gap from capacity is of a much larger order than that for polar codes over binary input, symmetric memoryless channels \cite{GuruXia15polar}. Guruswami and Xia showed in \cite{GuruXia15polar} that for such channels, polar codes  of block length $n$ with gap from capacity of  order $\frac{1}{n^{\mu}}$ can achieve a block error probability decaying as $2^{-n^{0.49}}$ with a decoding algorithm whose complexity scales as $n \cdot  \text{poly}(\log n)$. (Here $0 < \mu < \frac{1}{2}$ is a universal constant.) For AWGN channels, there is no known coding scheme that provably achieves a polynomial gap to capacity with efficient decoding.

 Recall the lower bound on the gap to capacity arises from the condition \eqref{eq:chi_lb1} which is required to ensure that the (deterministic) state evolution sequence $x_1, x_2, \ldots $ is guaranteed to increase by at least an amount proportional to $\Delta_R$ in each iteration. As described in Remark \ref{rem:improving_gap} the capacity gap for the iterative hard-decision decoder can be improved to $O(\frac{\log \log M}{\log M})$ by modifying the exponential power allocation to  flatten the power allocation for a certain number of of sections at the end. We expect such a modification to yield a similar improvement in the capacity gap for the AMP decoder, but we do not detail this analysis as it is involves additional technical details.

\section{Comparison of the decoders }  \label{sec:dec_comparison}
 
All three decoders discussed in this section --  the adaptive successive hard-decision  decoder, the  adaptive successive soft-decision  decoder, and the AMP decoder --- achieve near-exponential decay of  error probability in the regime where $R < \mc {C}$ remains fixed. However, the finite length performance of the two soft-decision decoders is significantly better than that of the hard-decision decoder. This is because of the need to control the proliferation of false alarms in hard-decision decoding. 

 In the regime where $R < \mc {C}$ is fixed, the number of iterations also remains fixed. Consequently, the complexity of  all three decoders is $O(nML)$. The complexity is determined by the matrix-vector products  that need to be computed in each step,  using the design matrix $A \in \reals^{n \times ML}$.  Among the two soft-decision decoders, the AMP decoder has lower per iteration complexity (though still of the same order) as it does not require orthonormalization or Cholesky decomposition  to compute the test statistic. In the next chapter, we describe how replacing the Gaussian design matrix with a Hadamard-based design matrix can lead to significant savings in both running time and memory.

 In the regime where $\Delta_R$ shrinks to $0$ with growing $M$, the decoders discussed in this chapter are no longer efficient as they require $M$ to increase exponentially in $1/\Delta_R$ (cf.\ \eqref{eq:DelRbnd}). An interesting open question is whether SPARCs can achieve a smaller gap from capacity with efficient decoding.  The spatially coupled SPARC discussed in Chapter \ref{chap:sc_sparcs}  is a promising candidate, but a fully rigorous analysis of AMP-decoded spatially coupled SPARCs remains open.

\section{Proofs}

\subsection{Proof of Lemma \ref{lem:conv_expec}}  \label{subsec:conv_exp_proof}

From \eqref{eq:xt_tau_def}, $x(\tau)$ can be written as
\be
x(\tau) := \sum_{\ell=1}^{L} \frac{P_\ell}{P} \, \mc{E}_\ell(\tau),
\ee
where 
\be 
\begin{split}
&\mc{E}_\ell (\tau) =
\expec \left[
\frac{e^{\frac{\sqrt{n P_\ell}}{\tau} \, U^{\ell}_1}}
{ e^{ \frac{\sqrt{n P_\ell}}{\tau} \, U^{\ell}_1}  +  e^{-\frac{n P_\ell}{\tau^2}} \sum_{j=2}^M e^{\frac{\sqrt{n P_\ell}}{\tau}U^{\ell}_j } } \right].
\label{eq:Eell_def}
\end{split}
\ee
The result needs to be proved only for $\xi^* >0$. (For brevity, we supress the dependence of $\xi^*$ on $\tau$.) Since $P_\ell$ is non-increasing with $\ell$,  it is enough\footnote{We can also prove that $\lim \mc{E}_{\lfloor \xi^* L \rfloor} = \tfrac{1}{2}$, but we do not need this for the exponentially decaying power allocation since it will only affect a vanishing fraction of sections as $L$ increases. Since $\mc{E}_\ell \in [0,1]$, these sections do not affect the value of $\lim x(\tau)$ in \eqref{eq:Eell_def}.} to prove that for  $\xi \in  ( 0, 1]$,  
\be
\lim \mc{E}_{\lfloor \xi L \rfloor} (\tau) = \left\{
\begin{array}{ll}
1, & \text{ if } \xi <  \xi^*, \\
0, & \text{ if } \xi > \xi^*.
\end{array}
\right.
\label{eq:Eell_def0}
\ee
Using the relation $nR =  {L \ln M}$,
we can write
\ben
\frac{n P_{\lfloor \xi L \rfloor} }{\tau^2} = \nu_{\lfloor \xi L \rfloor} \ln M, \quad \text{ where } \quad \nu_{\lfloor \xi L \rfloor} = \frac{ L P_{\lfloor \xi L \rfloor} }{R \tau^2 }.
\een
From the definition of $\xi^*$ in the lemma statement and the non-increasing power-allocation, we see that $\lim  \nu_{\lfloor \xi L \rfloor} >2$ for $\xi <\xi^*$, and $\lim  \nu_{\lfloor \xi L \rfloor} < 2$ for $\xi > \xi^*$.

For brevity, in what follows we drop the superscripts on  $U_j^{\lfloor \xi L \rfloor}$, and denote it by $U_j$ for $j \in [M]$. From  \eqref{eq:Eell_def},  $\mc{E}_{\lfloor \xi L \rfloor}(\tau) $  can be written as
\begin{small}
\begin{align}
 \mc{E}_{\lfloor \xi L \rfloor}(\tau) &    = \expec \left[  \frac{e^{\sqrt{ \nu_{\lfloor \xi L \rfloor} \ln M} \, U_1}}
{e^{\sqrt{ \nu_{\lfloor \xi L \rfloor} \ln M} \, U_1 }  +   M^{- \nu_{\lfloor \xi L \rfloor}} \sum_{j=2}^M e^{\sqrt{ \nu_{\lfloor \xi L \rfloor} \ln M} \, U_j} }  \right] \nonumber  \\
& = \expec \, \expec \left[  \frac{e^{\sqrt{ \nu_{\lfloor \xi L \rfloor} \ln M} \, U_1 }}
{e^{\sqrt{ \nu_{\lfloor \xi L \rfloor} \ln M} \, U_1}  +   M^{- \nu_{\lfloor \xi L \rfloor}} \sum_{j=2}^M e^{\sqrt{ \nu_{\lfloor \xi L \rfloor} \ln M} \, U_j } }   \Big{|} U_1\right].
\label{eq:Eell_iter}
\end{align}
\end{small}
The inner expectation in \eqref{eq:Eell_iter} is of the form
\ben
\begin{split}
&\expec \left[  \frac{e^{\sqrt{ \nu_{\lfloor \xi L \rfloor} \ln M} \, U_1}}
{e^{\sqrt{ \nu_{\lfloor \xi L \rfloor} \ln M} \, U_1 }  +   M^{- \nu_{\lfloor \xi L \rfloor}} \sum_{j=2}^M e^{\sqrt{ \nu_{\lfloor \xi L \rfloor} \ln M} \, U_j} }   \Big{|} U_1\right]  = \expec_X \left[ \frac{c}{c + X} \right], 
\label{eq:inner_exp0} 
\end{split}
\een
where $c = \exp\left(\sqrt{ \nu_{\lfloor \xi L \rfloor} \ln M} \, U_1 \right)$ is treated as a positive constant, and the expectation is with respect to the random variable
\be
X := M^{- \nu_{\lfloor \xi L \rfloor}} \sum_{j=2}^M \exp\left(\sqrt{ \nu_{\lfloor \xi L \rfloor} \ln M} \, U_j \right).
\label{eq:Xrv_def}
\ee

\textbf{Case $1$: $\xi < \xi^*$}. Here we have $ \lim \nu_{\lfloor \xi L \rfloor} >2$.  Since $\frac{c}{c +X}$ is a convex function of $X$, applying Jensen's inequality we get
$\expec_X [  \frac{c}{c + X} ] \geq \frac{c}{c + \expec X}$.
The expectation of $X$ is
\ben
\begin{split}
\expec X &= M^{- \nu_{\lfloor \xi L \rfloor}} \sum_{j=2}^M \expec \left[e^{\sqrt{ \nu_{\lfloor \xi L \rfloor} \ln M} \, U_j } \right] \\
& \stackrel{(a)}{=} M^{- \nu_{\lfloor \xi L \rfloor}} (M-1)  M^{ \nu_{\lfloor \xi L \rfloor} /2}  \leq M^{1 -  \nu_{\lfloor \xi L \rfloor} /2},
\end{split}
\een
with $(a)$ is obtained from the moment generating function of a Gaussian random variable. Therefore,
\be
\begin{split}
1 \geq \expec_X \left[ \frac{c}{c + X} \right] \geq  \frac{c}{c + \expec X} &\geq \frac{c}{c + M^{1-  \nu_{\lfloor \xi L \rfloor} /2}} \\
&= \frac{1}{1 + c^{-1} \, M^{1-  \nu_{\lfloor \xi L \rfloor} /2}}.
\label{eq:jensen_chain}
\end{split}
\ee
Recalling that $c = \exp\left(\sqrt{ \nu_{\lfloor \xi L \rfloor} \ln M} \, U_1 \right)$, \eqref{eq:jensen_chain} implies that 
\be
\begin{split}
&\expec_X \left[ \frac{e^{\sqrt{ \nu_{\lfloor \xi L \rfloor} \ln M} \, U_1}}{e^{\sqrt{ \nu_{\lfloor \xi L \rfloor} \ln M} \, U_1 } + X}  \ \Big{|} \ U_1 \right]  
\geq \frac{1}{ 1 +   M^{1-  \nu_{\lfloor \xi L \rfloor} /2} \, e^{- \sqrt{ \nu_{\lfloor \xi L \rfloor} \ln M} \, U_1 }}.
\label{eq:UU_lb}
\end{split}
\ee
When $\{ U_1 > - (\ln M)^{1/4} \}$, the RHS of \eqref{eq:UU_lb} is at least $$[1 +  M^{1-  \nu_{\lfloor \xi L \rfloor} /2} \, \exp\left( (\ln M)^{3/4} \sqrt{ \nu_{\lfloor \xi L \rfloor}} \right)]^{-1}.$$ Using this in  \eqref{eq:Eell_iter}, we obtain that
\be
\begin{split}
& 1 \geq \mc{E}_{\lfloor \xi L \rfloor} (\tau) 
\geq  \frac{P(U_1 > - (\ln M)^{1/4} )}{1 +  M^{1-  \nu_{\lfloor \xi L \rfloor} /2} \, e^{ (\ln M)^{3/4} \sqrt{ \nu_{\lfloor \xi L \rfloor}} }} \stackrel{ M \to \infty}{\longrightarrow} \  1,
\end{split}
\ee
since  $\lim  \nu_{\lfloor \xi L \rfloor} > 2$.  Hence $\mc{E}_{\lfloor \xi L \rfloor} \to 1$ when $\lim  \nu_{\lfloor \xi L \rfloor} > 2$.

\textbf{Case $2$: $\xi > \xi^*$.} Here we have $ \lim \nu_{\lfloor \xi L \rfloor} < 2$. The random variable $X$ in \eqref{eq:Xrv_def} can be bounded from below as follows.
\be
\begin{split}
X &\geq  M^{- \nu_{\lfloor \xi L \rfloor}} \max_{j \in \{2, \ldots, M\}} e^{\sqrt{ \nu_{\lfloor \xi L \rfloor} \ln M} \, U_j } \\
&= M^{- \nu_{\lfloor \xi L \rfloor}} e^{ \left[ \max_{j \in \{2, \ldots, M\}} U_j \right] \sqrt{ \nu_{\lfloor \xi L \rfloor} \ln M} }.
\label{eq:X_lb0}
\end{split}
\ee
Using standard bounds for the standard normal distribution, it can be shown that
\be
P\left( \max_{j \in \{2, \ldots, M\}} U_j \ < \sqrt{2 \ln M}(1-\e)\right) \leq e^{-M^{\e(1-\e)}},
\label{eq:pmax_gauss}
\ee
for
$ \e = \omega\left( \frac{\ln \ln M}{\ln M} \right)$.\footnote{Recall that $f(n) = \omega(g(n))$ if for each $k >0$, $\abs{f(n)} / \abs{g(n)} \geq k$ for all sufficiently large $n$. } \label{eq:eps_order}
Combining \eqref{eq:pmax_gauss} and \eqref{eq:X_lb0}, we obtain that
\ben
\begin{split}
&\exp(-M^{\e(1-\e)}) \geq P\left( \max_{j \in \{2, \ldots, M\}} U_j \ < \sqrt{2 \ln M}(1-\e)\right)  \\ 
 & \quad \geq  P\left( X <  M^{- \nu_{\lfloor \xi L \rfloor}} e^{ \sqrt{2 \ln M}(1-\e) \sqrt{ \nu_{\lfloor \xi L \rfloor} \ln M}} \right) \\
 & \quad = P\left( X  <  M^{\sqrt{2  \nu_{\lfloor \xi L \rfloor}} (1-\e) - \nu_{\lfloor \xi L \rfloor}}  \right).
\end{split}
\een
Since $\lim  \nu_{\lfloor \xi L \rfloor} < 2$ and $\e >0$ can be an arbitrarily small constant, there exists a strictly positive constant $\delta$ such that 
$\delta <  \sqrt{2  \nu_{\lfloor \xi L \rfloor}}(1-\e) - \nu_{\lfloor \xi L \rfloor}$ for all sufficiently large $L$. Therefore, for sufficiently large $M$, the expectation in \eqref{eq:inner_exp0} can be bounded as
\be
\begin{split}
\expec_X \left[ \frac{c}{c +X} \right] &\leq P(X < M^\delta) \cdot 1 + P (X \geq M^{\delta})\cdot \frac{c}{c + M^{\delta}} \\
& \leq e^{-M^{\e(1-\e)}} + 1 \cdot \frac{c}{c + M^{\delta}} \leq \frac{2}{1 + c^{-1} M^{\delta}}.
\end{split}
\label{eq:ccx_bound}
\ee
Recalling that $c = \exp\left(\sqrt{ \nu_{\lfloor \xi L \rfloor} \ln M} \, U_1  \right)$, and  using the bound  of \eqref{eq:ccx_bound} in  \eqref{eq:Eell_iter}, we obtain
\be
\begin{split}
&\mc{E}_{\lfloor \xi L \rfloor} (\tau)  \leq \expec \left[  \frac{2}
{1 +   M^{\delta} e^{-\sqrt{ \nu_{\lfloor \xi L \rfloor} \ln M} \, U_1 }}  \right] \\
& \leq P(U_1 > (\ln M)^{1/4}) \cdot 2 + \frac{2P(U_1 \leq (\ln M)^{1/4})}{1 + M^{\delta} e^{-\sqrt{ \nu_{\lfloor \xi L \rfloor}} \, (\ln M)^{3/4}} } \\
&  \stackrel{(a)}{\leq} 2e^{-\tfrac{1}{2} (\ln M)^{1/2}} + \, 1 \cdot \frac{2}{1 + e^{ \delta \ln M -\sqrt{ \nu_{\lfloor \xi L \rfloor}} \, (\ln M)^{3/4}}} \\
& \stackrel{(b)}{\longrightarrow} 0 \text{ as } M \to \infty.
\end{split}
\label{eq:Eell_final}
\ee
In \eqref{eq:Eell_final}, $(a)$ is obtained using the bound  $\Phi(x) < \exp(-x^2/2)$ for $x \geq 0$, where $\Phi(\cdot)$ is the Gaussian cdf; $(b)$ holds since $\delta$ and $\lim  \nu_{\lfloor \xi L \rfloor}$ are both positive constants.

This proves that $\mc{E}_{\lfloor \xi L \rfloor}(\tau) \to 0$ when $\lim \nu_{\lfloor \xi L \rfloor} < 2$.  The proof of the lemma is complete since we have proved both statements in \eqref{eq:Eell_def0}.

\subsection{Proof of Lemma \ref{lem:xt_taut_lb}} \label{app:proof_lem_xt_taut_lb}

We will use the following lower bound on the function $x(\tau)$ in \eqref{eq:xt_tau_def}.   
\begin{lemma}\cite[Lemma 2.1]{RushV17ErrExp}
Consider the exponential power allocation power allocation in \eqref{eq:exp_power_alloc}, and let $\nu_\ell :=  {L P_\ell}/ (R \tau^2)$. Then $x(\tau) \geq \xtl$, where for sufficiently large $M$ and any $\delta \in (0, \tfrac{1}{2})$,   
 \begin{align}
 \xtl & \geq  \left(1- \frac{M^{-\kappa_1 \delta^2}}{\delta \sqrt{\log M}} \right) \sum_{\ell=1}^{L} \frac{P_\ell}{P} \,
\mathbf{1}\left\{ \nu_\ell > 2 + \delta \right\}  \\
& \qquad + \frac{1}{4}\sum_{\ell=1}^{L} \frac{P_\ell}{P} \mathbf{1}\left\{ 2\left(1- \frac{\kappa_2}{\sqrt{\log M}} \right) \leq \nu_\ell \leq  2 + \delta  \right\},  \label{eq:xlb_asym}
 \end{align}
 where $\kappa_1, \kappa_2$ are universal positive constants.
 \label{lem:xtl_lemma}
 \end{lemma}
 
 Let $x_{t-1} =x < (1-f(M))$.  We only need to consider the case where $\nu_L < (2+\delta)$, because otherwise all the $\{ \nu_\ell \}_{\ell \in [L]}$ values are at  least $(2+\delta)$, and \eqref{eq:xlb_asym} guarantees that $x_t \geq (1-f(M))$.  

With $x_{t-1} =x$, we have $\tau_{t-1}^2 = \sigma^2 +P(1-x)$. Therefore, from \eqref{eq:cell} we have
\be 
\nu_\ell = \frac{L P_\ell}{R \tau_{t-1}^2} =  \frac{ \tau_0^2}{R \tau_{t-1}^2} L((1 + \snr)^{1/L}-1)  \left( 1 +\snr \right)^{-\ell/L}, \quad \ell \in [L].  \label{eq:nuellx}\ee 
Using this in \eqref{eq:xlb_asym}, we have
\begin{align}
x_t & \geq (1- f(M)) \sum_{\ell=1}^L   \frac{P_\ell}{P} \, \mathbf{1}\left\{ \nu_\ell > 2 + \delta \right\}  \nonumber \\
& \stackrel{(a)}{=} \left(1 - f(M)\right)  \sum_{\ell=1}^L \frac{P_\ell}{P} \, \mathbf{1}\left\{\frac{\ell}{L} < \frac{1}{2 \mc{C}} \log \left(\frac{L((1+\snr)^{1/L} - 1) \tau_0^2}{(2+\delta)  R \tau_{t-1}^2}\right) \right\} \nonumber \\
& \stackrel{(b)}{\geq}   \left(1 - f(M)\right)  \sum_{\ell=1}^L \frac{P_\ell}{P} \, \mathbf{1}\left\{\frac{\ell}{L} \leq \frac{1}{2 \mc{C}} \log 
\left(\frac{2 \mc{C} \tau_0^2}{(2+\delta)  R \tau_{t-1}^2}\right) \right\} \nonumber \\
& \stackrel{(c)}{\geq}  \left(1 - f(M)\right)  \frac{P+\sigma^2}{P}\left[ 1 - \exp\left\{- \log\left( \frac{2 \mc{C} \tau_0^2}{(2+\delta)  R \tau_{t-1}^2} \right) + \frac{2 \mc C}{L}   \right\} \right] \nonumber \\
& \stackrel{(d)}{\geq}  \left(1 - f(M)\right)  \frac{P+\sigma^2}{P}\left[ 1- \frac{(2 + \delta) R \tau_{t-1}^2}{2 \mc C \tau_0^2} - \frac{5}{L} \right]. \label{eq:xtlb0}
\end{align}
In the above, $(a)$ is obtained using the expression for $\nu_\ell$ in \eqref{eq:nuellx}, while $(b)$ by noting that
$ L((1+\snr)^{1/L} - 1) = L(e^{2 \mc C/L} - 1) \geq 2 \mc{C}$.
Inequality $(c)$ is obtained by  using the geometric series formula: for any $\xi \in (0,1)$, we have
\[ \sum_{\ell=1}^{\lfloor \xi L \rfloor} P_{\ell} = (P+\sigma^2)(1- e^{-2\mc{C}\lfloor \xi L \rfloor/L}) \geq 
(P+\sigma^2)(1- e^{-2\mc{C}\xi} e^{2 \mc C/L}).    \]
Inequality $(d)$ uses $e^{2\mc{C}/L} \leq 1 + 4 \mc{C}/L$ for large enough $L$. Substituting $\tau_{t-1}^2 = \sigma^2 +P(1-x)$, \eqref{eq:xtlb0} implies
\begin{align}
x_t - x & \geq   \left(1 - f(M)\right)  \frac{P+\sigma^2}{P} \left(1 - \frac{5}{L} \right)  \nonumber \\
& \qquad - (1-f(M)) \frac{(1+ \delta/2) R}{\mc C} \left( \frac{P+ \sigma^2}{P} -x \right) - x  \nonumber \\
& =  \left(1 - f(M)\right)  \frac{P+\sigma^2}{P} \left( 1 - \frac{(1+ \delta/2) R}{\mc C} - \frac{5}{L}  \right) \nonumber \\
& \qquad - x \left(1 -  \left(1 - f(M)\right) \frac{(1+ \delta/2) R}{\mc C}   \right).
\label{eq:xtlb1}
\end{align}
Since $\delta < (\mc C -  R)/\mc C$,  the term $\frac{(1+ \delta/2) R}{\mc C}$ is strictly less than $1$, and the RHS of \eqref{eq:xtlb1} is strictly decreasing in $x$. Using the upper bound of $x < (1- f(M))$ in \eqref{eq:xtlb1} and simplifying, we obtain
\begin{align}
x_t - x & \geq  \left(1 - f(M)\right)  \frac{\sigma^2}{P} \left( 1 - \frac{(1+ \delta/2) R}{\mc C} \right) \nonumber  \\
& \qquad - f(M) (1- f(M))\frac{(1+ \delta/2) R}{\mc C}  - \frac{5(1 +\sigma^2/P)}{L}.
\end{align}
This completes the proof for $t > 1$. For $t=1$, we start with $x=0$, and we get the slightly stronger lower bound of $\chi_1$ by substituting $x=0$ in \eqref{eq:xtlb1}.

\subsection{Proof Sketch of Theorem \ref{thm:main_amp_perf}} \label{subsec:amp_proof}

The main ingredients in the proof of Theorem \ref{thm:main_amp_perf} are two technical lemmas (Lemma \ref{lem:hb_cond} and Lemma \ref{lem:main_lem}).  After laying down some definitions and notation that will be used in the proof, we state the two lemmas and use them to prove Theorem \ref{thm:main_amp_perf}.

\paragraph{Definitions and notation for the proof.}
For consistency with earlier analyses of AMP, we use notation similar to  \cite{BayMont11,RushGV17}.  Define the following column vectors recursively for $t\geq 0$, starting with $\beta^0=0$ and $z^0=y$.
\begin{equation}
\begin{split}
h^{t+1}  := \beta_0 - (A^*z^t + \beta^t), \qquad &  q^t  :=\beta^t - \beta_0, \\
b^t := w-z^t,\qquad & m^t :=-z^t.
\end{split}
\label{eq:hqbm_def}
\end{equation}
Recall that $\beta_0$ is the message vector chosen by the transmitter.  The vector $h^{t+1}$ is the noise in the effective observation $A^*z^t + \beta^t$, while $q^t$ is the error in the estimate $\beta^t$.  A key ingredient of the proof is showing that $h^{t+1}$ and $m^t$ are approximately i.i.d.\ $\mc{N}(0, \tau_t^2)$, while $b^t$ is approximately i.i.d.\ $\mc{N}(0, \tau_t^2 - \sigma^2)$.

Define  $\mathscr{S}_{t_1, t_2}$ to be the sigma-algebra generated by
\[ b^0, ..., b^{t_1 -1}, m^0, ..., m^{t_1 - 1}, h^1, ..., h^{t_2}, q^0, ..., q^{t_2}, \text{ and }  \beta_0, w. \]
Lemma \ref{lem:hb_cond} iteratively computes the conditional distributions $b^t |_{ \mscrs_{t, t}}$ and $h^{t+1} |_{ \mscrs_{t+1, t}}$. Lemma \ref{lem:main_lem} then uses this conditional distributions to show the concentration of the mean squared error $\| q^t \|^2/n$.

For $t \geq 1$, let
\be
\lambda_t := \frac{-1}{\tau^2_{t-1}}\left( P - \frac{\norm{\beta^t}^2}{n} \right).
\label{eq:lambda_t_def}
\ee
We then have
\begin{equation}
b^{t} + \lambda_t m^{t-1} = A q^t, \quad \text{ and } \quad h^{t+1} + q^t = A^* m^t,
\label{eq:bmq}
\end{equation}
which follows from \eqref{eq:amp1} and \eqref{eq:hqbm_def}. From \eqref{eq:bmq}, we have the matrix equations 
 \be  B_t +  [0 | M_{t-1}] \Lambda_t = A Q_t  \qquad \text{ and } \qquad  H_t + Q_{t} = A^* M_t, \label{eq:XtYt_rel} \ee
where for $t \geq 1$,
\be
\begin{split}
&M_t  := [m^0 \mid \ldots \mid m^{t-1} ],  \qquad Q_t  :=  [q^0 \mid \ldots \mid q^{t-1} ] \\ 
& B_t := [b^0 | \ldots | b^{t-1}], \quad  H_t = [h^1 | \ldots | h^{t}], \quad  \Lambda_t := \text{diag}(\lambda_0, \ldots, \lambda_{t-1}). 
\end{split}
\label{eq:XYMQt}
\ee
The notation $[c_1 \mid c_2 \mid \ldots \mid c_k]$ is used to denote a matrix with columns $c_1, \ldots, c_k$.  We define $M_0, Q_0, B_0$, $H_0$, and $\Lambda_0$ to be all-zero vectors.  

 We use  $m^t_{\|}$ and $q^t_{\|}$ to denote the projection of $m^t$ and $q^t$ onto the column space of $M_t$ and $Q_t$, respectively. Let
 $\alpha_t := (\alpha^t_0, \ldots, \alpha^t_{t-1})^*$ and $\gamma_t :=  (\gamma^t_0, \ldots, \gamma^t_{t-1})^*$ be the coefficient vectors of these projections, i.e.,
 \be
 m^t_{\| } = \sum_{i=0}^{t-1} \alpha^t_i m^i, \quad  q^t_{\|} = \sum_{i=0}^{t-1} \gamma^t_i q^i.
 \label{eq:mtqt_par}
 \ee
 The projections of $m^t$ and $q^t$ onto the orthogonal complements of $M^t$ and $Q^t$, respectively,  are denoted by
 \be
 m^t_{\perp} := m^t - m^t_{\|}, \quad  q^t_{\perp} := q^t - q^t_{\|}
  \label{eq:mtqt_perp}
 \ee
 
The proof of Lemma \ref{lem:main_lem}  shows that for large $n$, the entries of $\alpha_t$ and $\gamma_t$ concentrate around constants. We now specify these constants. With $\tau^2_t$  and $x_t$ as defined in \eqref{eq:tau_def} and \eqref{eq:xt_tau_def}, for $t \geq 0$ define
\be
\sigma^2_t : = \tau_t^2 - \sigma^2 = P(1-x_t).
\label{eq:sigt_def}
 \ee
The concentrating values for $\gamma^t$ and $\alpha^t$ are
\be
\begin{split}
\hat{\gamma}^{t} &:=  (0,\ldots, 0, \sigma_t^2/\sigma_{t-1}^2)^* \in \mathbb{R}^t, \\
\hat{\alpha}^{t} &:= (0,\ldots, 0, \tau_t^2/\tau_{t-1}^2)^*  \in \mathbb{R}^t.
\label{eq:hatalph_hatgam_def}
\end{split}
\ee
Let $(\sigma^{\perp}_0)^2 := \sigma_0^2$ and $(\tau^{\perp}_0)^2 := \tau_0^2$, and for $t > 0$ define 
\be
\begin{split}
& (\sigma_{t}^{\perp})^2 := \sigma_{t}^2 \left(1 - \frac{ \sigma_{t}^2 }{\sigma_{t-1}^2 }\right), \quad \text{ and } \quad (\tau^{\perp}_{t})^2 := \tau_{t}^2 \left(1 - \frac{\tau_{t}^2}{\tau_{t-1}^2}\right).
\label{eq:sigperp_defs}
\end{split}
\ee

\begin{lemma} [Conditional distribution lemma {\cite[Lemma 4]{RushGV17}}] 
For the vectors $h^{t+1}$ and $b^t$ defined in \eqref{eq:hqbm_def}, the following hold for $1 \leq t \leq T$, provided $n >T$, and $M_t$ and $Q_t$ have full column rank. (We recall that the number of iterations $T$ is defined in \eqref{eq:Tdef}.)
\begin{align}
h^{1} \lvert_{\mscrs_{1, 0}} \stackrel{d}{=} \tau_0 Z_0 + \Delta_{1,0}, \quad &\text{ and } \quad h^{t+1} \lvert_{\mscrs_{t+1, t}} \stackrel{d}{=} \frac{\tau_t^2}{\tau_{t-1}^2} h^{t} + \tau_{t}^{\perp} \, Z_t + \Delta_{t+1,t}, \label{eq:Ha_dist} \\ 
b^{0} \lvert_{\mscrs_{0, 0}} \stackrel{d}{=} \sigma_0 Z'_0, \quad &\text{ and } \quad b^{t} \lvert_{\mscrs_{t, t}}\stackrel{d}{=} \frac{\sigma_t^2}{\sigma_{t-1}^2} b^{t-1} +  \sigma_{t}^{\perp} \, Z'_t + \Delta_{t,t}. \label{eq:Ba_dist}
\end{align}
where $Z_0, Z_t \in \mathbb{R}^N$ and $Z'_0, Z'_t \in \mathbb{R}^n$ are i.i.d.\ standard Gaussian random vectors that are independent of the corresponding conditioning sigma algebras. The deviation terms are $\Delta_{0,0}=0$,
\begin{align}
\Delta_{1,0} &= \left[ \left(\frac{\norm{m^0}}{\sqrt{n}}  - \tau_0\right)\mathsf{I} -\frac{\norm{m^0}}{\sqrt{n}} \mathsf{P}_{q^0}\right] Z_0 \nonumber \\
& \qquad + q^0 \left(\frac{\norm{q^0}^2}{n}\right)^{-1} \left(\frac{(b^0)^*m_0}{n} - \frac{\norm{q^0}^2}{n}\right), \label{eq:D10}
\end{align}
and for $t >0$,
\begin{align}
& \Delta_{t,t} =   \sum_{r=0}^{t-2} \gamma^t_r b^r  + \left( \gamma^t_{t-1} - \frac{\sigma^2_t}{\sigma^2_{t-1}} \right)b^{t-1}  +  \Bigg[  \Bigg(\frac{\norm{q^t_{\perp}}}{\sqrt{n}} - \sigma_{t}^{\perp}\Bigg) \mathsf{I}  - \frac{\norm{q^t_{\perp}} }{\sqrt{n}} \mathsf{P}_{M_t}\Bigg]Z'_t  \nonumber \\
& \quad + M_t\left(\frac{M_{t}^* M_{t}}{n}\right)^{-1} \left(\frac{H_t q^t_{\perp}}{n} - \frac{M_t}{n}^*\left[\lambda_t m^{t-1} - \sum_{r=1}^{t-1} \lambda_{r} \gamma^t_{r} m^{r-1}\right]\right),\label{eq:Dtt} 
\end{align}
\begin{align}
& \Delta_{t+1,t}  =   \sum_{r=0}^{t-2} \alpha^t_r h^{r+1}+ \left( \alpha^t_{t-1} - \frac{\tau^2_t}{\tau^2_{t-1}} \right) h^{t} \nonumber  \\
 & \ + \left[\left(\frac{\norm{m^t_{\perp}}}{\sqrt{n}} - \tau_{t}^{\perp}\right)  \mathsf{I} -\frac{\norm{m^t_{\perp}}}{\sqrt{n}} 
 \mathsf{P}_{Q_{t+1}}\right]Z_t \nonumber \\
& \ + Q_{t+1} \left(\frac{Q_{t+1}^* Q_{t+1}}{n}\right)^{-1} \left(\frac{B^*_{t+1} m^t_{\perp}}{n} - \frac{Q_{t+1}^*}{n}\left[q^t - \sum_{i=0}^{t-1} \alpha^t_i q^i\right]\right).\label{eq:Dt1t}  
\end{align} 
\label{lem:hb_cond}
\end{lemma}

The next lemma uses the representation in Lemma \ref{lem:hb_cond} to show that for each $t \geq 0$, $h^{t+1}$ is the sum of an i.i.d.\ $\mc{N}(0, \tau_t^2)$ random vector plus a deviation term.  Similarly $b^t$ is the sum of an i.i.d.\ $\mc{N}(0, \sigma_t^2)$ random vector and a deviation term.  
\begin{lemma}
For $t \geq 0$, the conditional distributions in Lemma \ref{lem:hb_cond} can be expressed as 
\be
h^{t+1} \lvert_{\mscrs_{t+1, t}} \stackrel{d}{=} \tilde{h}^{t+1} + \tilde{\Delta}_{t+1} , \qquad  b^{t} \lvert_{\mscrs_{t, t}}\stackrel{d}{=} \breve{b}^{t} + \breve{\Delta}_{t},
\label{eq:htil_rep}
\ee
where
\begin{align}
& \tilde{h}^{t+1} := \tau_{t}^2 \sum_{i=0}^{t} \left(\frac{\tau^{\perp}_i}{\tau_{i}^2}\right) Z_i, 
\qquad \tilde{\Delta}_{t+1} := \tau_{t}^2 \sum_{i=0}^{t} \left(\frac{1}{\tau_{i}^2}\right) \Delta_{i+1, i},  \label{eq:htilde_def} \\
& \breve{b}^{t}:= \sigma_{t}^2 \sum_{i=0}^{t} \left(\frac{\sigma^{\perp}_i}{\sigma_{i}^2}\right) Z'_i,
\qquad \breve{\Delta}_{t} := \sigma_{t}^2 \sum_{i=0}^{t} \left(\frac{1}{\sigma_{i}^2}\right) \Delta_{i, i}.
\label{eq:btilde_def}
\end{align}
\label{lem:ideal_cond_dist}
Here $Z_i \in \reals^N$, $Z'_i \in \reals^n$ are the independent standard Gaussian vectors defined in Lemma \ref{lem:hb_cond}.

Consequently, $\tilde{h}^{t+1}  \stackrel{d}{=} \tau_t \tilde{Z}_t$, and $\breve{b}^{t}  \stackrel{d}{=} \sigma_t \breve{Z}_t$, where $\tilde{Z}_t \in \reals^N$ and $ \breve{Z}_t \in \reals^n$ are standard  Gaussian random vectors such that for any $j  \in [N]$ and  $i \in [n]$,  the vectors $(\tilde{Z}_{0,j}, \ldots, \tilde{Z}_{t,j})$ and  $(\breve{Z}_{0,i}, \ldots, \breve{Z}_{t,i})$ are each jointly Gaussian with 
\be \expec[\tilde{Z}_{r,j} \tilde{Z}_{s,j}] = \frac{\tau_s}{\tau_r}, \qquad  \expec[ \breve{Z}_{r,i} \breve{Z}_{s,i}] = \frac{\sigma_s}{\sigma_r} \qquad \text{ for }0 \leq r \leq s \leq t. \ee
\end{lemma}
\begin{proof}
We give the proof for the distributional representation of $h^{t+1}$, with the proof for $b^t$ being similar. The representation in  \eqref{eq:htil_rep} can be directly obtained by using Lemma \ref{lem:hb_cond} Eq.\ \eqref{eq:Ha_dist} to recursively write $h^{t}$ in terms of $(h^{t-1}, Z_{t-1}, \Delta_{t, t-1})$, then $h^{t-1}$ in terms of $(h^{t-2}, Z_{t-2}, \Delta_{t-1, t-2})$, and so on. 

Using \eqref{eq:htilde_def}, we write $\tilde{h}^{t+1}= \tau_t  \tilde{Z}_t$, where $\tilde{Z}_t = \tau_{t} \sum_{i=0}^{t} \left(\frac{\tau^{\perp}_i}{\tau_{i}^2}\right) Z_i$ is n Gaussian random vector with i.i.d.\ entries, with zero mean and variance equal to
\begin{align}
 \tau_{t}^2\sum_{i=0}^{t} \frac{ (\tau^{\perp}_i)^2}{\tau_{i}^4}  = \frac{ \tau_{t}^2}{\tau_0^2} + \sum_{i=1}^{t} \left(\frac{ \tau_{t}^2}{\tau_{i}^2} \right) \left(1 - \frac{\tau_i^2}{\tau_{i-1}^2}\right) &=  \frac{\tau_{t}^2}{\tau_0^2} + \sum_{i=1}^{t} \left(\frac{\tau_{t}^2}{\tau_{i}^2} - \frac{\tau_{t}^2}{\tau_{i-1}^2}\right)  \nonumber \\
 & = 1.
\label{eq:barhvar}
\end{align}
For $j \in [N]$ the covariance between the $j$th entries of $\tilde{Z}_r$ and $\tilde{Z}_s$, for $0 \leq r \leq s \leq t$, is 
\begin{align}
 \expec[\tilde{Z}_{r,j} \tilde{Z}_{s,j}] 
 =  \tau_{r} \tau_{s} \sum_{u=0}^{r} \sum_{v=0}^{s} \left(\frac{\tau^{\perp}_u}{\tau_{u}^2} \right) \left(\frac{\tau^{\perp}_v}{\tau_{v}^2} \right)  
\mathbb{E}\left\{Z_{u_j} Z_{v_j} \right\} & \stackrel{(a)}{=} \tau_{r} \tau_{s} \sum_{u=0}^{r} \frac{(\tau^{\perp}_u)^2}{\tau_{u}^4} \nonumber   \\
& \stackrel{(b)}{=} \frac{\tau_s}{\tau_r},
\end{align}
where $(a)$ follows from the independence of $Z_{u_j}$ and $Z_{v_j}$ and $(b)$ from the calculation in \eqref{eq:barhvar}.
\end{proof}


The next lemma shows that the deviation terms in Lemma \ref{lem:hb_cond} are small, in the sense that their section-wise maximum absolute value and norm concentrate around $0$. It also shows that the mean-squared error 
$\| q^t \|^/n= \| \beta - \beta^t \|^2/n$ concentrates around $\sigma_t^2$ for  $0\leq t \leq T$.

 \begin{lemma} \cite{RushV17ErrExp}
With $C, K, c, \kappa$ denoting generic positive universal constants, the following large deviations inequalities hold for $0  \leq t < T$: 
 \begin{align}
& \hspace{-10pt} P\left( \left[ \frac{1}{L} \sum_{\ell = 1}^L \max_{j \in sec_{\ell}} \abs{[\Delta_{t+1,t}]_{j}} \right]^2 \geq \epsilon \right)  
 \leq  P \left(\frac{1}{L} \sum_{\ell = 1}^L \max_{j \in sec_{\ell}} ([\Delta_{t+1,t}]_{j})^2 \geq \epsilon \right)  
 \nonumber \\%
&  \qquad \leq   K C^{2t} (t!)^{11} \exp\left\{-  \frac{ \kappa   L \e}{(c\log M)^{2t} (t!)^{17}}\right\}  \label{eq:Ha1}, \\
\end{align}
\begin{align}
 P\left( \frac{1}{n}\norm{\Delta_{t,t}}^2 \geq \e \right) 
\leq   K C^{2t} (t!)^{11} \exp\left\{-  \frac{ \kappa   L \e^2}{(c\log M)^{2t-1} (t!)^{17}}\right\},
\label{eq:Ba}
\end{align}
 \be
 P\left( \abs{ \frac{\norm{q^{t+1}}^2}{n} - \sigma_{t+1}^2} \geq  \e \right) 
 \leq 
 K C^{2t} (t!)^{11} \exp\left\{-  \frac{ \kappa   L \e^2}{(c\log M)^{2t+1} (t!)^{17}}\right\}.
 \label{eq:Hc}
 \ee
\label{lem:main_lem}
\end{lemma}

The proof of Lemma \ref{lem:main_lem} can be found in \cite[Sec. 5]{RushV17ErrExp}. The proof is inductive. To prove Theorem \ref{thm:main_amp_perf}, we only need the concentration result for the squared error $\| q^t \|^2/n$ in \eqref{eq:Hc}. But the proof of this result requires concentration results for various inner products and functions involving $\{h^{t+1}, q^t, b^t, m^t\}$, which are proved inductively. 

The dependence on $t$ of the probability bounds in Lemma \ref{lem:main_lem} is determined by the induction used in the proof:  the concentration results for step $t$ depend on those corresponding to all the previous steps. The $t!$ terms in the constants arise due to quantities that can be expressed as a sum of $t$ terms with step indices $1,\ldots, t$, e.g., $\Delta_{t,t}$ and $\Delta_{t+1,t}$ in \eqref{eq:Dtt} and \eqref{eq:Dt1t}.  The concentration results for such quantities have $1/t$ and $t$ multiplying the exponent and pre-factor, respectively, in each step $t$, which results in the $t!$ terms in the bound. Similarly, the $C^{2t}$ and $c^{2t}$ terms arise due to quantities  that are the \emph{product} of two terms, for each of which we have a concentration result available from the induction hypothesis.

\paragraph{Proof of  Theorem \ref{thm:main_amp_perf}.} \label{subsec:proof_thm1}
The event that the section error rate exceeds $\e$ is
$ \{ \mc{E}_{sec}(\mc{S}_n)  > \e \} =  \left\{ \sum_{\ell }  \mathbf{1} \{ \hat{\beta}_\ell \neq \beta_{0_\ell} \} > L \e \right\}$. 
Recall that the largest entry within each section of $\beta^T$ is chosen to produce $\hat{\beta}$. Therefore, 
when a section $\ell$ is decoded in error, the correct non-zero entry has no more than half the total mass of section
$\ell$ at the termination step $T$. That is, $\beta^{T}_{\textsf{sent}(\ell)} \leq  \frac{1}{2} \sqrt{n P_ \ell}$
where $\textsf{sent}(\ell)$ is the index of the non-zero entry in section $\ell$ of the true message $\beta_0$.  Since 
$\beta_{0_\textsf{sent}(\ell)} = \sqrt{n P_\ell}$, we have
\be
\mathbf{1} \{ \hat{\beta}_\ell \neq \beta_{0_\ell} \} \ \  \Rightarrow \ \ \norm{ \beta^{T}_\ell - \beta_{0_\ell}}^2 \geq  \frac{n P_\ell}{4}, \quad \ell \in [L].
\label{eq:sec_error_implies}
\ee
Hence when  $\{ \mc{E}_{sec}(\mc{S}_n)  > \e \}$, we  have
\be
\begin{split}
\norm{\beta^{T} - \beta_0}^2  = \sum_{\ell=1}^L \, \norm{\beta^{T}_\ell - \beta_{0_\ell}}^2 & \stackrel{(a)}{\geq}
\sum_{\ell=1}^L   \mathbf{1} \{ \hat{\beta}_\ell \neq \beta_{0_\ell} \}  \frac{nP_\ell}{4} \, \\
& \stackrel{(b)}{\geq} \, L\e \frac{nP_L}{4}  \stackrel{(c)}{\geq} \,
\frac{n \, \e \, \sigma^2 \ln(1 + \snr)}{4} =  \frac{n \e \sigma^2 \mc{C}}{2},
\label{eq:sec_error_chain}
\end{split}
\ee
where $(a)$ follows from   \eqref{eq:sec_error_implies}; $(b)$ is obtained using  the fact that
$P_\ell > P_L$ for $\ell \in [L-1]$ for the exponentially decaying power allocation in \eqref{eq:exp_power_alloc}; $(c)$ is obtained using the first-order Taylor series lower bound
$L P_L \geq \sigma^2 \ln(1+\tfrac{P}{\sigma^2})$. We therefore conclude that
\be
\{ \mc{E}_{sec}(\mc{S}_n)  > \e \} \ \Rightarrow \  \left\{ \frac{ \| \beta^{T} - \beta_0 \|^2 }{n}  \geq \frac{\e  \sigma^2 \mc{C}}{2} \right\},
\label{eq:sec_error_exp}
\ee
where $\beta^T$ is the AMP estimate at the termination step $T$. 

Now, from \eqref{eq:Hc} of Lemma \ref{lem:main_lem}, we know that for any $\tilde{\e} \in (0,1)$:
\be
\begin{split}
P\left( \frac{\| \beta^{T} - \beta_0 \|^2}{n} \geq \sigma_{T}^2 + \tilde{\e} \right) & = P\left( \frac{\| q^{T} \|^2}{n} \geq \sigma_{T}^2 + \tilde{\e} \right)  \\
& \leq  K_{T}\exp\left\{- \frac{\kappa_{T}L \tilde{\e}^2}{(\log M)^{2T-1}}\right\}.
\label{eq:betaTst_conc}
\end{split}
\ee
From  the definition of $T$ and \eqref{eq:tauT_bound}, we have  $\sigma_{T}^2 = \tau_{T}^2 - \sigma^2 \leq 
P f_R(M)$. 
Hence, \eqref{eq:betaTst_conc} implies
\begin{align}
  P\left( \frac{\| \beta^{T} - \beta_0 \|^2}{n}  \geq P f_R(M) + \tilde{\e} \right)   & \leq P\left(\frac{\| \beta^{T} - \beta_0 \|^2}{n} \geq \sigma_{T}^2 + \tilde{\e} \right) \nonumber \\
&    \leq  K_{T} \exp\left\{- \frac{\kappa_{T} L \tilde{\e}^2}{(\log M)^{2T-1}}\right\}.
\label{eq:norm_bound}
\end{align}
Now take  $\tilde{\e} = \frac{\e  \sigma^2 \mc{C}}{2} - Pf_R(M)$, noting that this $\tilde{\e}$ is strictly positive whenever $\e > 2 \snr f_R(M)/\mc{C}$, the condition specified in the theorem statement.   Finally, combining \eqref{eq:sec_error_exp} and \eqref{eq:norm_bound} we obtain
\ben
P\left( \mc{E}_{sec}(\mc{S}_n)  >  \e \right) \leq  K_{T} \exp\left\{- \frac{\kappa_{T} L}{(\log M)^{2T-1}}  \left(\frac{\e  \sigma^2 \mc{C}}{2} - P f_R(M) \right)^2\right\}. 
\een
\qed

  \chapter{Finite Length Decoding Performance}  \label{chap:emp_perf}

In this chapter, we investigate the empirical error performance of SPARCs with AMP decoding at finite block lengths. In Section \ref{sec:dec_comp}, we describe how decoding complexity can be reduced by using Hadamard-based design matrices, and how a key parameter of the AMP decoder can be estimated online.
In Section \ref{sec:pow_alloc}, we show  that the choice of power allocation can have a significant impact on  decoding performance, and describe a simple algorithm to design a good allocation for a given rate and $\snr$. Section \ref{sec:lvsm} discusses how the choice of the  code parameters $L,M$ influences finite length error performance. Finally, in Section \ref{sec:ldpc-outer} we show how partial outer codes can be used in conjunction with AMP decoding to obtain a steep waterfall in the error rate curves.  We compare the error rates of AMP-decoded sparse superposition codes  with coded modulation using LDPC codes from the WiMAX standard.

\section{Reducing AMP decoding complexity} \label{sec:dec_comp}

\subsection{Hadamard-based design matrices}  \label{subsec:fwht}
In the sparse regression codes described and analyzed thus far, the design matrix $A$ is chosen to have zero-mean i.i.d. entries, either Gaussian $\sim \mc{N}(0,\frac{1}{n})$ or Bernoulli entries drawn uniformly from  $\pm  \frac{1}{\sqrt{n}}$ as in Sec. \ref{sec:bern_dict}. As discussed in Sec. \ref{sec:dec_comparison}, with such matrices the computational complexity of the AMP decoder in \eqref{eq:amp1}--\eqref{eq:amp2}  is $O(LMn)$ when the matrix-vector multiplications $A\beta$ and $A^*z^t$ are performed in the usual way. Additionally, storing $A$ requires $O(LMn)$ memory, which is prohibitive for reasonable code lengths. For
example, $L=1024$, $M=512$, $n=9216$ ($R=1$ bit) requires 18 gigabytes of memory using a double-precision (4-byte) floating point representation, all of which must be
accessed twice per iteration. 

To reduce decoding complexity, we replace the i.i.d. design matrix  with a structured Hadamard-based design matrix, which we denote in this section by  $\Ah$.  With $\Ah$, the key matrix-vector multiplications can be performed  via a fast Walsh-Hadamard Transform (FWHT)\cite{Shanks1969}. Moreover, $\Ah$ can be implicitly defined which greatly reduces the memory required.

We denote the  Hadamard matrix of size $2^k \times 2^k$ by $H_k$. We recall that  $H_k$ is a  square matrix  with $\pm 1$ entries and mutually orthogonal rows, recursively defined as follows. Starting with $H_0=1$, for $k \geq 1$,
\ben
H_{k} = \begin{pmatrix} H_{k-1} & H_{k-1} \\ H_{k-1} & -H_{k-1}   \end{pmatrix}.
\een

To construct the design matrix $\Ah \in \reals^{n \times ML}$, one option is to take $k=\lceil\log_2(LM)\rceil$ and select $n$ rows uniformly at random from the Hadamard matrix $H_{k}$.  In this case, the matrix-vector multiplications are performed by embedding the vectors into $\mathbb{R}^{ML}$, and then multiplying by $H_k$ using a FWHT. A more efficient way is to construct each $n \times M$ section of $\Ah$ independently from a smaller Hadamard matrix.  This is done as follows.

Take $k=\lceil\log_2(\max(n+1, M+1))\rceil$. Each section of $\Ah$ is constructed independently by choosing a permutation of $n$ distinct rows  from $H_k$ uniformly at random.\footnote{To obtain the desired statistical properties for $A$, we do not pick the first row of $H_k$ as it is all-ones. The $n+1$ in the definition of $k$ ensures that we still have enough rows left to pick $n$ at random after removing the first,
all-one, row; the $M+1$ ensures that we can always have one leading $0$ when embedding
$\beta$ so that the first, all-one, column is also never picked.}
The multiplications $\Ah \beta^t$ and $\Ah^*z^t$ are performed by computing $A_{\mathsf{H} \ell}\beta_\ell^t$ and $A_{\mathsf{H} \ell}^* z^t$, for $\ell \in [L]$, where the $n \times M$ matrix $A_{\mathsf{H} \ell}$  is the $\ell$th section of $\Ah$, and $\beta_\ell^t \in \mathbb{R}^M$ is the $\ell$th section of $\beta^t$.
To compute $A_{\mathsf{H}\ell}\beta_\ell^t$, zero-prepend $\beta_\ell$ to length $2^k$, perform the FWHT, then choose $n$ entries corresponding to the rows in $A_{\mathsf{H} \ell}$.
Sum the $n$-length result from each section to obtain $\Ah\beta^t$. Note that we
prepend with $0$ because  the first column of $H_k$ must always be ignored as it is always all-ones.
For $A_{\mathsf{H}\ell}^* z^t$, embed entries from $z^t$ into a $2^k$ long
vector again corresponding to the rows in $A_{\mathsf{H}}$, with all other entries set to zero, perform the FWHT, and return the last $M$ entries. Concatenate
the result from each section to form $\Ah^*z^t$.

The  empirical error performance of the AMP decoder with $\Ah$ constructed as above is indistinguishable
from that of  a full i.i.d. matrix. The computational complexity  of the decoder is reduced to $O(Ln\log n)$ (in the common case where $n>M$, otherwise it is $O(LM\log M)$). The memory requirements are reduced to $O(LM)$, typically a few megabytes.   In comparison, for i.i.d. design matrices, the complexity and memory requirements scale as 
$O(LMn)$. For reasonable code lengths, this represents around a thousandfold improvement in both time
and memory.
Furthermore, the easily parallelized structure would enable a hardware
implementation to trade off between a slower and smaller series implementation
and a faster though larger parallel implementation, potentially leading to
significant practical speedups.

\subsection{Online computation of $\tau_t^2$ and early termination} \label{subsec:online_tau}
Recall that  these coefficients  $(\tau_t^2)_{t \geq 1}$ are required for the AMP update steps \eqref{eq:amp1} and \eqref{eq:amp2}. In the standard implementation, these are recursively computed  in advance via the SE equations \eqref{eq:tau_def} and \eqref{eq:xt_tau_def}.    The total number of iterations  $T$  is also determined in advance by computing the number of iterations required  the SE to converge to its fixed point  (to within a specified tolerance). This advance computation is slow as each
of the  $L$ expectations in \eqref{eq:xt_tau_def} needs to be computed
numerically via Monte-Carlo simulation, for each  $t$.

A simple way to estimate $\tau_t^2$ online during the decoding process is as follows. 
In each step $t$, after producing $z^t$ as in \eqref{eq:z_update}, we
estimate \be \widehat{\tau}_t^2 = \frac{\norm{z^t}^2}{n} = \frac{1}{n} \sum_{i=1}^n
z_i^2. \ee The justification for this estimate comes from the analysis of the
AMP decoder in \cite{RushV17ErrExp}, which  provides a concentration inequality that shows that for large $n$,
$\wh{\tau}_t^2$ is close to $\tau_t^2$  with high
probability.   We note that such a similar online estimate has been used previously in various AMP and GAMP algorithms \cite{barbKrzISIT14,BarbSK2015,BarbKrz15,Rangan11}. 

In addition to being fast, the online estimator permits an interpretation as a measure of SPARC decoding progress and  provides a flexible termination criterion for the decoder.  Recall from the previous chapter
(cf. Section \ref{subsec:iter_soft_dec}) that in each step we have 
\ben 
\statt = \beta^t + A^* z^t
\approx \beta + \tau_t Z, 
\een 
where $Z$ is a standard normal random vector independent of $\beta$.   The online estimator $\wh{\tau}_t^2$ is found to  track $\text{Var}(\statt - \beta)=\norm{\statt - \beta}^2/n$ very accurately, even when this variance deviates significantly from 
$\tau_t^2$. This indicates that we can use the final value $\wh{\tau}_T^2$ to
 accurately estimate the power of the undecoded sections --- and thus the
 number of sections decoded correctly --- at runtime. Indeed, $(\wh{\tau}^2_T -
 \sigma^2)$ is an accurate estimate of the total power in the incorrectly
 decoded sections. This, combined with the fact that the power allocation is
 non-increasing, allows the decoder to estimate the number of incorrectly decoded sections. 

Furthermore, we can use the change in $\wh{\tau}_t^2$ between iterations to terminate the decoder early. If the value $\wh{\tau}_t^2$ has not changed between successive iterations, or the change is within some small threshold, then the decoder has stalled and no further iterations are worthwhile. Empirically we find that a stopping criterion with a small threshold (e.g., stop when $\abs{\wh{\tau}^2_t - \wh{\tau}^2_{t-1}} < P_L$)  leads to no additional errors compared to running the decoder for the full iteration count, while giving a significant speedup in most trials.
Allowing a larger threshold for the stopping criterion  gives even better running time improvements.

All the simulation results reported in this chapter are obtained using Hadamard-based design matrices, the online estimate  $\wh{\tau_t^2}$, and a corresponding early termination criterion.

\section{Power allocation} \label{sec:pow_alloc}

Theorem \ref{thm:main_amp_perf} shows that for any fixed $R < \mc{C}$ and an exponentially decaying power allocation  $P_\ell \propto e^{-2 \mc{C}\ell/L}, \ \ell \in [L]$,  the probability of section error of the AMP decoder can be made arbitrarily small for sufficiently large values of the code parameters $(n,M,L)$.  However, the error rate of the exponentially decaying allocation is rather high at practical block lengths. This is illustrated in Fig. \ref{fig:amp_sec_err_rate0}. The black curve at the top shows the average section error rate with the exponentially decaying allocation for various rates $R$ with $\mc{C} = 2$ bits. The blue curve in the middle shows the average section error rate with two different power allocation schemes, with the code parameters $(n,M,L)$  at each rate.

\begin{figure}[t] \centering
    \includegraphics[width=4.5in, height=2.7in]{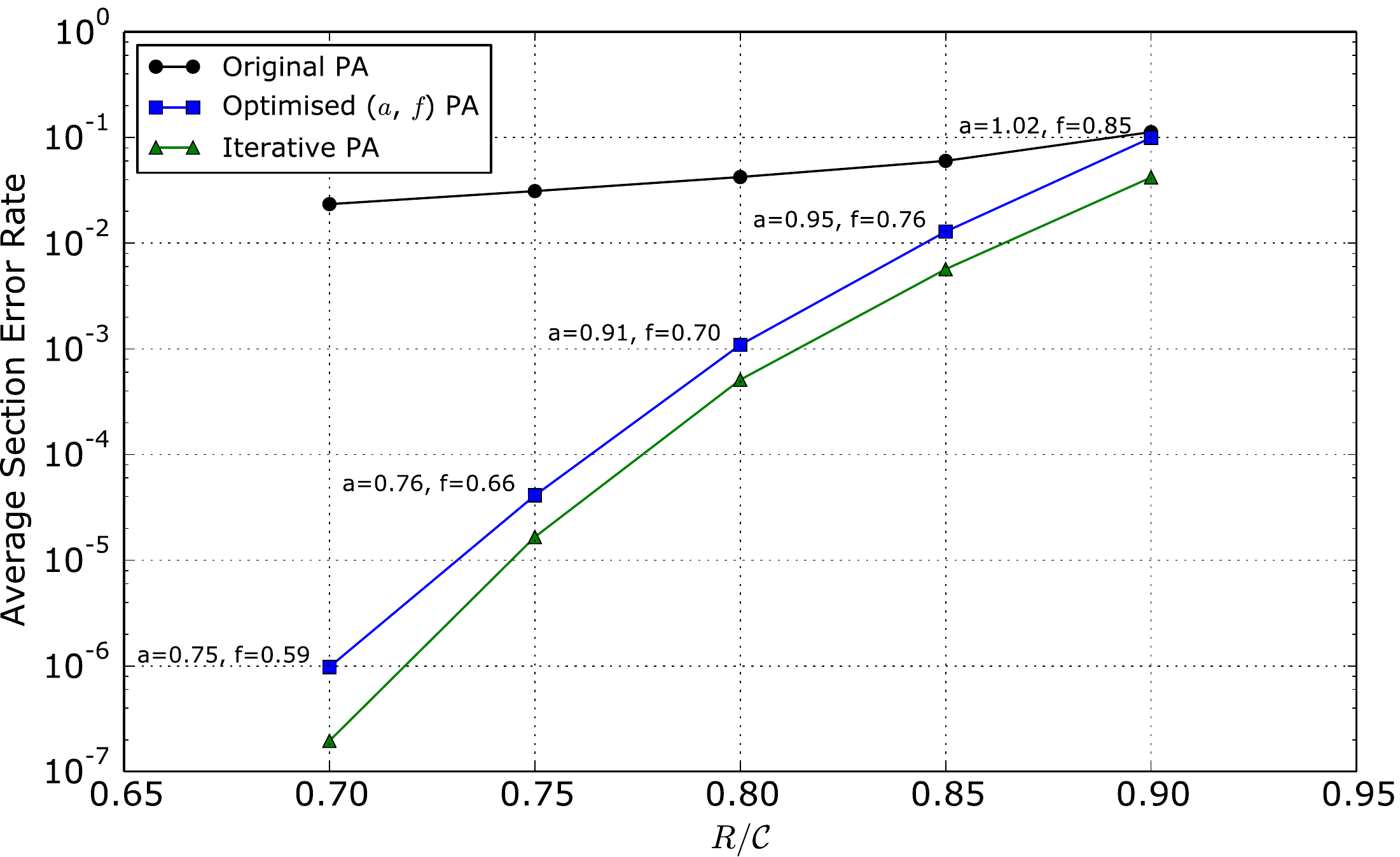}
    \caption{\small{Section error rate vs $R/\mc{C}$ at $\snr=15, \mc{C}=2$ bits.  The SPARC parameters for all the curves are $M=512, L=1024$.
        The top  curve shows the average section error rate of the AMP
        over $1000$ trials with $P_\ell \propto 2^{-2\mc{C}\ell/L}$ (with $\mc{C}$ in bits).
        The  curve in the middle shows the section error rate using the
        power allocation in \eqref{eq:mixed_power_alloc} with the $(a,f)$  values shown.
        The bottom curve shows the section error rate with the iterative power allocation scheme described in Section \ref{sec:pa:iterative}.
        }}
    \vspace{-7pt}
\label{fig:amp_sec_err_rate0} 
\end{figure}

It is evident that as we back off from capacity, the power allocation can crucially determine error performance.  The reason for the relative poor performance of the exponential allocation at lower rates such as $R=0.6\mc{C}$ and $0.7 \mc{C}$ is that that it allocates too much power to the initial sections, leaving too little for the final sections to decode reliably.  This motivates the following \emph{modified} exponential allocation  characterized by two parameters $a,f$.
For $f \in [0,1] $, let
\be
P_\ell = \begin{cases}
        \kappa \cdot 2^{-2a\mc{C} \ell/L}, & 1 \leq \ell \leq fL\\
        \kappa \cdot 2^{-2a\mc{C} f},   & fL+1 \leq \ell \leq L
\end{cases}
\label{eq:mixed_power_alloc}\ee
where 
the normalizing constant $\kappa$ ensures  that the total power across sections
is $P$.  For intuition, first assume that $f=1$. Then
\eqref{eq:mixed_power_alloc} implies that $P_\ell \propto 2^{-a2\mc{C} \ell/L}$
for $\ell \in [L]$. Setting $a=1$ recovers the original power allocation of
\eqref{eq:exp_power_alloc}, while $a=0$ allocates $\frac{P}{L}$ to each section.
Increasing $a$ increases the power allocated to the initial sections which
makes them more likely to decode correctly, which in turn helps by decreasing
the effective noise variance $\bar{\tau}^2_t$ in subsequent AMP iterations.
However, if $a$ is too large, the final sections may have too little power to
decode correctly.  
 
Hence we want the parameter $a$ to be large enough to ensure that the AMP gets
started on the right track, but not much larger.  This intuition can be made
precise in the large system limit using Lemma~\ref{lem:conv_expec}:  recall that for a section $\ell$ to be correctly decoded in step
$(t+1)$,  the limit of  $L P_\ell$ must exceed a threshold proportional to $R\bar{\tau}^2_t$. For rates close to $\mc{C}$, we need $a$ to be close to $1$ for the initial sections to cross this threshold and get decoding started
correctly. On the other hand, for rates such as $R=0.6 \mc{C}$,  $a=1$ allocates more power than necessary to the initial sections, leading to poor error performance in the final sections.  

In addition, we found that the section error rate can be further improved by \emph{flattening} the power allocation in the final sections.  For a given $a$, \eqref{eq:mixed_power_alloc} has an exponential power allocation until section $fL$, and constant power for the remaining $(1-f)L$ sections.  The allocation in \eqref{eq:mixed_power_alloc}  is continuous, i.e.\ each section in the flat
part is allocated the same power as the final section in the exponential part. Flattening boosts the power given to the final sections compared to an exponentially decaying allocation. The two parameters $(a,f)$ let us trade-off between the conflicting objectives of assigning enough power to the initial sections  and ensuring that the final sections have enough power to be decoded
correctly.

The middle  curve (blue) in Figure \ref{fig:amp_sec_err_rate0} shows the error performance with this modified allocation.  While this allocation improves the section error rate by a few orders of
magnitude, it requires costly
numerical optimization of $a$ and $f$. A good starting point is to use
$a=f=R/\mc{C}$, but further optimization is generally necessary. This motivates
the need for a fast power allocation algorithm with fewer tuning parameters.

\subsection{Iterative power allocation}
\label{sec:pa:iterative}
We now describe a simple  iterative algorithm to design a power allocation.   The starting point for our power allocation design is the asymptotic expression for the state evolution parameter $x(\tau)$ in Lemma \ref{lem:conv_expec} (see also the non-asymptotic lower bound in Lemma \eqref{lem:xtl_lemma}). Here, assuming $(L,M, n)$ are sufficiently large, we use the following approximation:
\be
x(\tau) \approx \,   \sum_{\ell=1}^{L} \frac{P_\ell}{P} \,
\mathbf{1}\left\{ LP_\ell >  2R \tau^2   \right\}.
\label{eq:lemma1b}
\ee
 We note that $R$ in \eqref{eq:lemma1b} is measured in nats. If the effective noise variance after step $t$ is $\tau_t^2$, then  \eqref{eq:lemma1b} says that any section $\ell$ whose normalized power $L P_\ell$ is larger than the threshold  $2 R \tau^2_t $ is likely to be decodable in step $(t+1)$, i.e., in  $\beta^{t+1}$, the probability mass within the section will be concentrated on the correct non-zero entry. 


The $L$ sections of the SPARC  are divided into $B$ \emph{blocks} of $L/B$ sections each. Each section within a block is allocated the same power.  For example, with $L=512$ and $B=32$, there are $32$ blocks  with $16$ sections per block. The algorithm sequentially allocates power to each of the $B$ blocks as follows. Allocate the minimum power to the first block of sections so that they can be decoded in the first iteration when $\tau_0^2=\sigma^2+P$. Using \eqref{eq:lemma1b}, we set the  power in each section of the first block to
\[ P_\ell = \frac{2 R\tau_0^2}{L}, \quad 1\le\ell\le\frac{L}{B}.\]  Using \eqref{eq:lemma1b}, we estimate $x_1= x(\tau_0)=BP_1$, and hence
$\tau_1^2=\sigma^2+(P -  B P_{1})$. Using this value, allocate the minimum
required power for the second block of sections to decode, i.e.,
$P_\ell=2 R\tau_1^2/L$ for $\frac{L}{B}+1\le\ell \le \frac{2L}{B}$.
If we sequentially allocate power in this manner to  each of the $B$ blocks, then the total power allocated by this scheme will be strictly less than $P$ whenever $R < \mc{C}$.  We  therefore modify the scheme as follows.  

\begin{figure}[t]
    \centering
    \includegraphics[width=0.3\columnwidth]{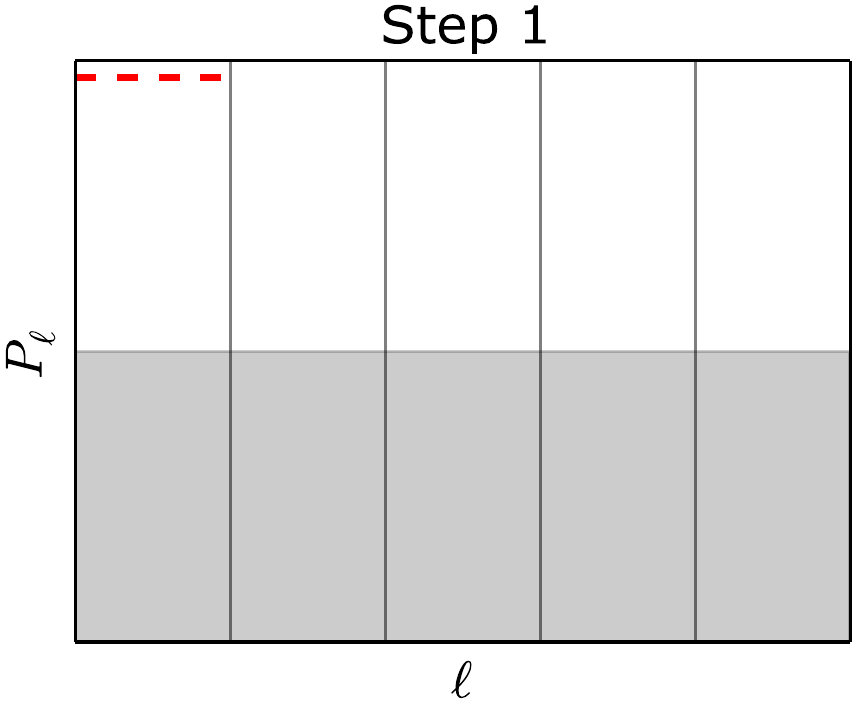}
    \includegraphics[width=0.3\columnwidth]{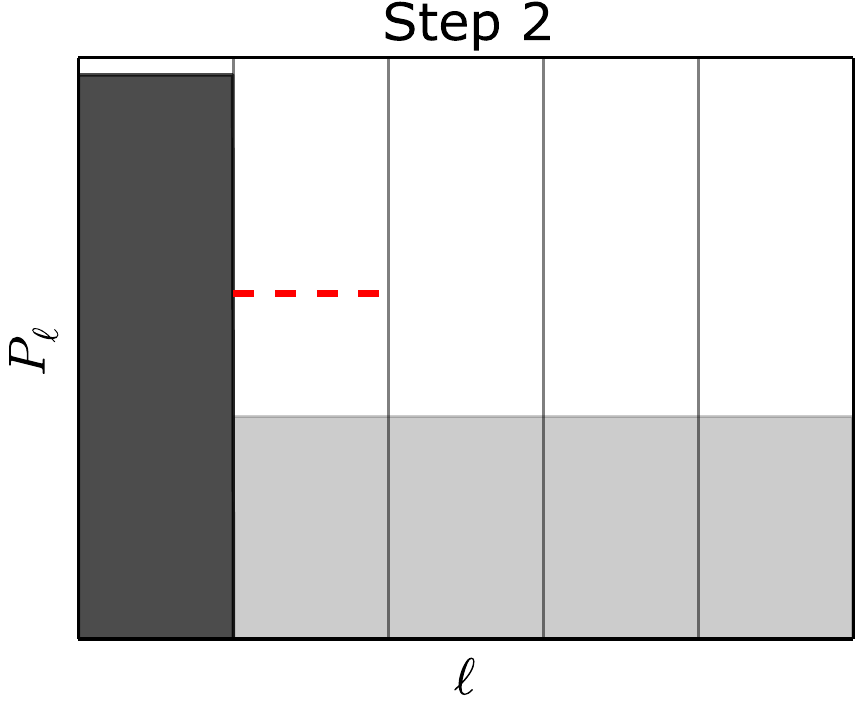}
    \includegraphics[width=0.3\columnwidth]{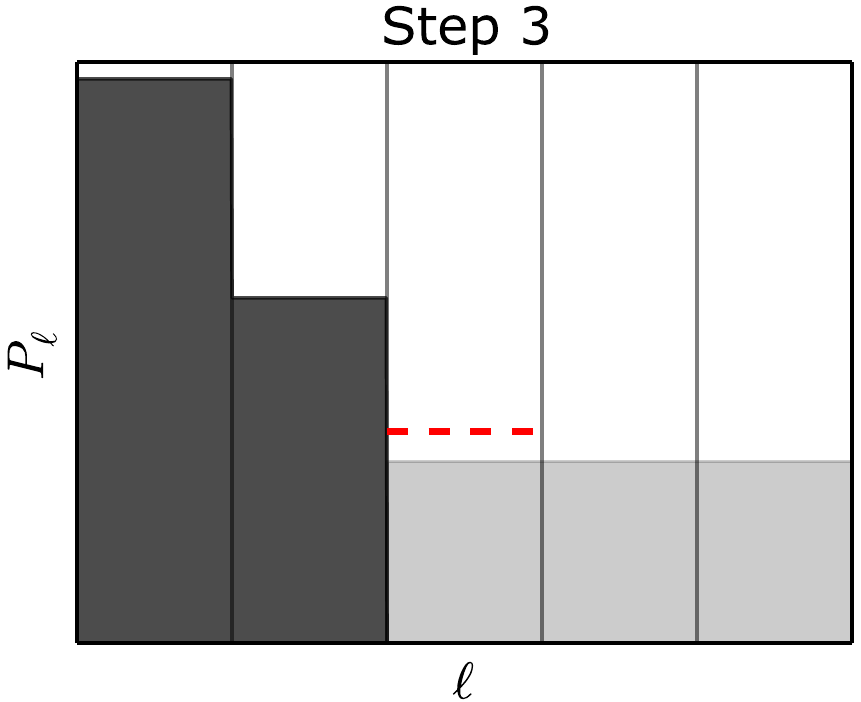}
    \includegraphics[width=0.3\columnwidth]{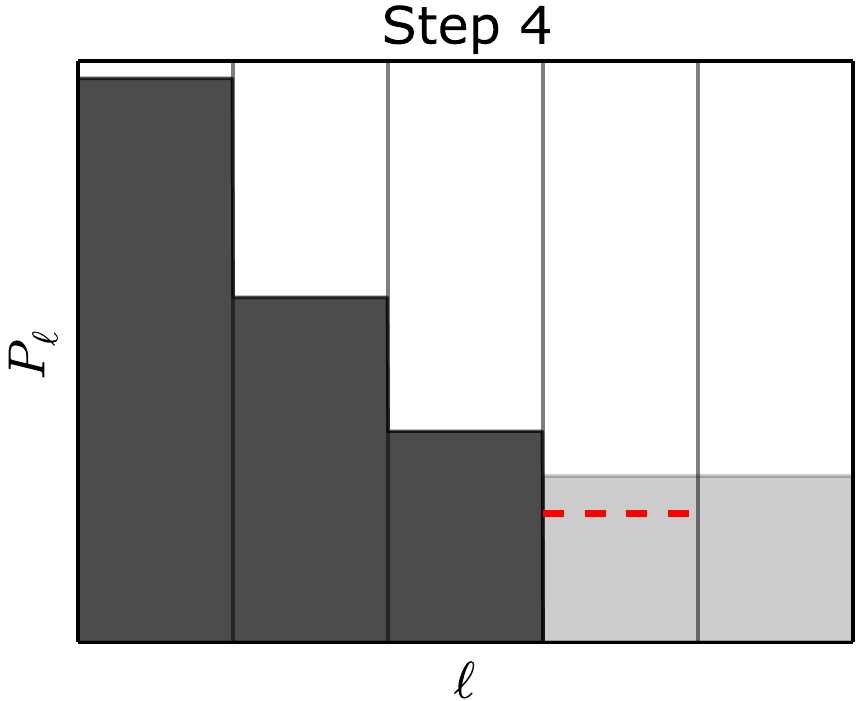}
    \includegraphics[width=0.3\columnwidth]{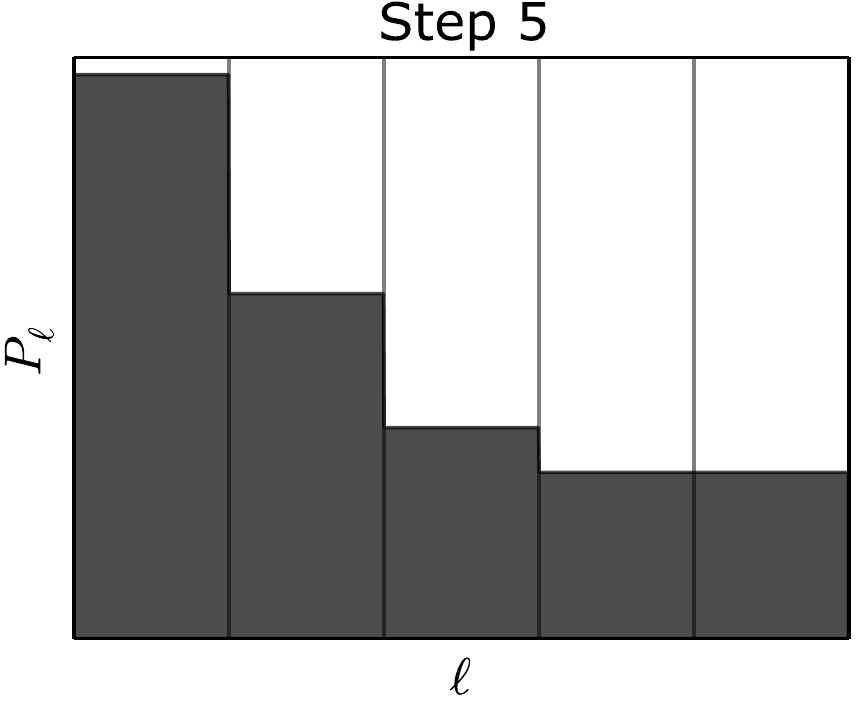}
    \caption{\small Example illustrating the iterative power allocation algorithm with $B=5$. In each step, 
        the height of the light gray region represents the allocation that distributes the remaining power equally
        over all the remaining sections. The dashed red line indicates the minimum power required for decoding the current block of sections.
       The dark gray bars represent the power that has been allocated at the beginning of the current step. }
    \label{fig:pa_lemma1_animation}
\end{figure}

\begin{algorithm}[t]
\caption{Iterative power allocation routine}%
\label{alg:iterative-pa}
\begin{algorithmic}
    \REQUIRE $L$, $B$, $\sigma^2$, $P$, $R$ such that $B$ divides $L$.
    \STATE Initialise $k \leftarrow  \frac{L}{B}$
    \FOR{$b=0$ to $B-1$}
        \STATE $P_{\text{remain}} \leftarrow P - \sum_{\ell=1}^{bk}P_\ell $
        \STATE $\tau^2 \leftarrow \sigma^2 + P_{\text{remain}}$
        \STATE $P_{\text{block}} \leftarrow 2  R \tau^2 / L$
        \IF{$P_{\text{remain}}/(L-bk) > P_{\text{block}}$}
            \STATE $P_{bk+1},\ldots,P_L \leftarrow P_{\text{remain}}/(L-bk)$
            \BREAK
        \ELSE
        \STATE $P_{bk+1},\ldots,P_{(b+1)k} \leftarrow P_{\text{block}}$
        \ENDIF
    \ENDFOR
    \RETURN $P_1,\ldots,P_L$
\end{algorithmic}
\end{algorithm}


For $1 \leq b \leq B$, to allocate power to the $b$th block of sections assuming that the first $(b-1)$ blocks have been allocated, we compare the two options and choose the one that allocates higher power to the block: i) allocating the minimum required power (computed as above) for the $b$th block of sections to decode; ii) allocating the remaining available power equally to sections in blocks $b, \ldots, B$, and terminating the algorithm. This gives a flattening in the final
blocks similar to the allocation in \eqref{eq:mixed_power_alloc}, but without requiring a specific parameter that determines where the flattening begins.  The iterative power allocation routine is described in Algorithm~\ref{alg:iterative-pa}. Figure~\ref{fig:pa_lemma1_animation} shows a toy example building up the power allocation for $B=5$, where flattening is seen to occur in step 4.

\emph{Choosing $B$}: By construction, the iterative power allocation scheme specifies the number of iterations of the AMP decoder in the large system limit.  This is given by the number of blocks with distinct powers; in particular the number of iterations (in the large system limit) is of the order of  $B$.  For finite code lengths, we find that it is better to use the  termination criterion described in \ref{subsec:online_tau} based on the online estimates $\wh{\tau_t^2}$. This  termination criterion allows us to choose the number of blocks $B$ to be as large as $L$. We found that choosing $B=L$, together with the termination criterion consistently gives a small improvement in error performance (compared to other choices of $B$), with no additional time
or memory cost.  

Additionally, with $B=L$, it is possible to quickly determine a pair
$(a,f)$ for the modified exponential allocation in \eqref{eq:mixed_power_alloc} which gives a nearly identical
allocation to the iterative algorithm. This is done by first setting $f$ to obtain the same
flattening point found in the iterative allocation, and then searching for an
$a$ which matches the first allocation coefficient $P_1$ between the iterative
and the modified exponential allocations. Consequently, any simulation results
obtained for the iterative power allocation could also be obtained using a
suitable $(a,f)$ with the modified exponential allocation, without having to
first perform a costly numerical optimization over $(a,f)$. 
\begin{figure}[t]
    \centering
    \includegraphics[width=0.85\columnwidth]{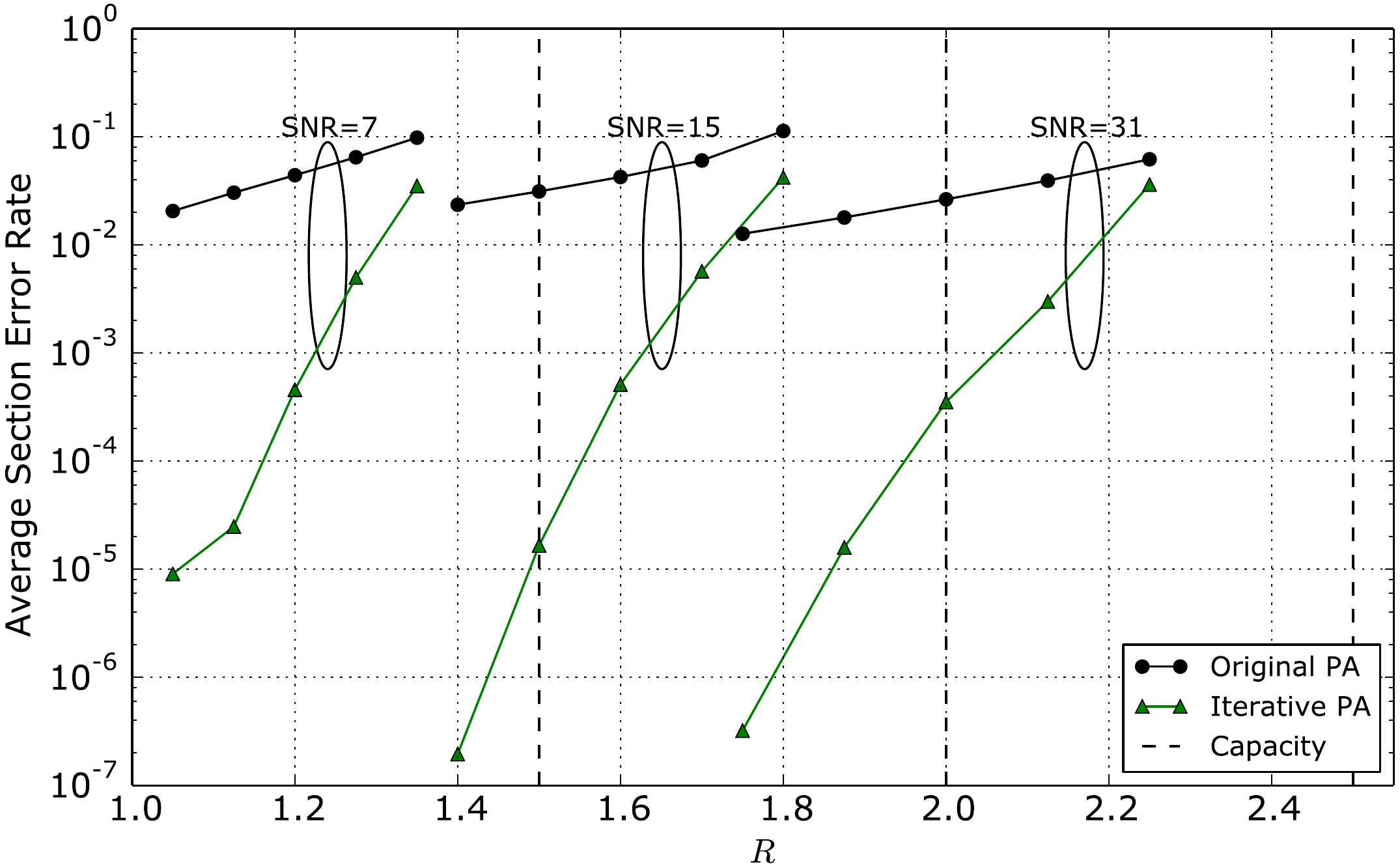}
    \caption{\small AMP section error rate  vs $R$ (in bits) at $\snr=7,15,31$,
    corresponding to $\mc{C}=1.5,2,2.5$ bits (shown with dashed vertical lines).
    At each $\snr$, the section error rate is reported for rates $R/\mc{C}=0.70,0.75,0.80,0.85,0.90$.
    The SPARC parameters are $M=512,L=1024$. The top black curve shows 
    the error rate with the exponential allocation $P_\ell \propto 2^{-2\mc{C}\ell/L}$ (with $\mc{C}$ in bits). The lower green curve
    shows the error rate with iterative power allocation, with $B=L$.}
    \label{fig:pa_perf_comparison}
    \vspace{-8pt}
\end{figure}

Figure~\ref{fig:pa_perf_comparison} compares the error performance of the exponential and iterative power allocation schemes discussed above for different values of $R$ at $\snr=7,15,31$. Compared to the original exponential power allocation, the iterative allocation has significantly improved error performance for rates away from capacity. It also generally  outperforms the modified exponential allocation results, as seen Figure \ref{fig:amp_sec_err_rate0}, where the bottom curve (green) shows the  error performance of the iterative allocation.

For the experiments in Figure~\ref{fig:pa_perf_comparison}, the value for $R$ used in constructing the iterative allocation (denoted by $R_{PA}$) was optimized numerically. Constructing an iterative allocation with  $R=R_{PA}$ yields good results, but due to finite length concentration effects,  the $R_{PA}$  yielding the smallest average error rate may be slightly different from the communication rate $R$. The effect of $R_{PA}$ on the  concentration of error rates  is discussed in Section~\ref{subsec:conc-pa}. We emphasize that this optimization over $R_{PA}$ is simpler than numerically optimizing the pair $(a,f)$ for the modified exponential allocation. Furthermore, guidelines for choosing $R_{PA}$ as a function of $R$ are given in Section~\ref{subsec:conc-pa}.

\section{Code parameter choices at finite code lengths} \label{sec:lvsm}

In this section, we discuss how the choice of SPARC design parameters  can influence finite length error performance with the AMP decoder. We will see that  the parameters $(L,M)$ and the power allocation both inducee a  trade-off between the `typical'  value of section error rate predicted by state evolution,  and concentration of actual error rates around the typical values.  

If the termination step is $T$, then we expect the test statistic in the final iteration to be $\text{stat}^T \approx \beta + \tau_T Z$, where $\tau_T$ is determined from the SE equations.  (For reliable decoding, we expect $\tau_T^2$ to be close to $\sigma^2$.) This leads to the following SE-based prediction for the section error rate \cite[Proposition 1]{GreigV18}:
\be
\bar{\mc{E}}_{\sec} = 1 -  \frac{1}{L} \sum_{\ell=1}^L  \expec_U \left[  \Phi\left(\frac{\sqrt{nP_\ell}}{\sigma} + U \right)
\right]^{M-1}. \label{eq:est_ser}
\ee

\subsection{Effect of $L$ and $M$ on concentration} 

\begin{figure}[t]
    \centering
    \includegraphics[width=0.8\columnwidth]{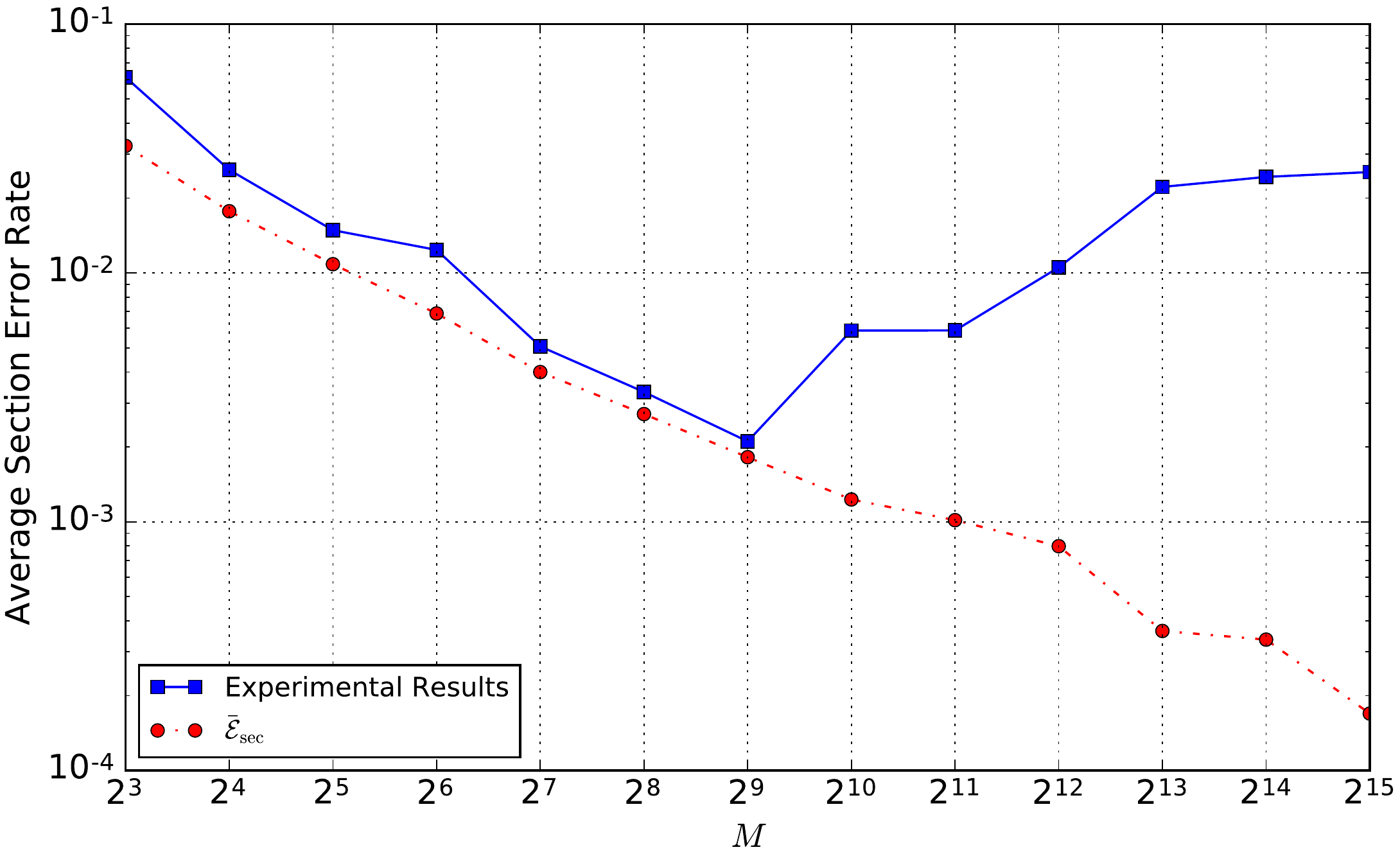}
    \caption{\small AMP error performance with increasing $M$, for
        $L=1024$, $R=1.5$ bits, and $\frac{E_b}{N_0}=5.7$ dB (2 dB from Shannon limit). 
    }
\label{fig:lvsm_ser}
\vspace{-8pt}
\end{figure}
\begin{figure}[h]
    \centering
    \includegraphics[width=0.95\columnwidth]{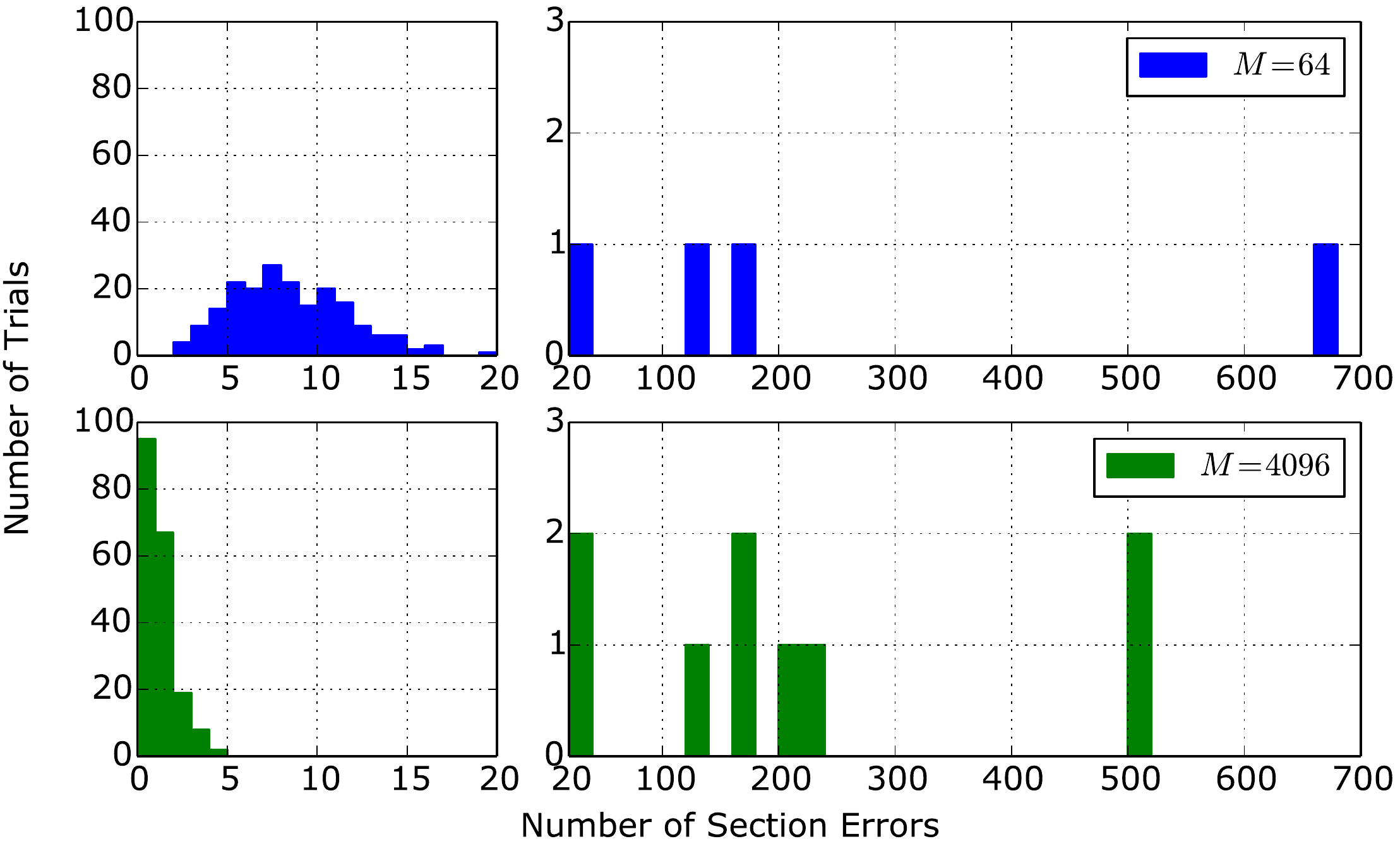}
    \caption{\small Histogram of AMP section errors over $200$ trials $M=64$ (top) and $M=4096$
        (bottom), with $L=1024$, $R=1.5$ bits, $\frac{E_b}{N_0}=5.7$dB. The left panels highlight
        distribution of errors around low section error counts, while the right
        panels show the distribution around high-error-count events. As shown
        in Figure~\ref{fig:lvsm_ser}, both cases have an average section error
        rate of around $10^{-2}$.}%
\label{fig:lvsm_hist}
\end{figure}

To understand the effect of  increasing $M$, consider Figure \ref{fig:lvsm_ser} which shows the error performance of a SPARC with $R=1.5, L=1024$, as we increase the value of $M$. Since $n= L \log M/R$, the code length
$n$ increases logarithmically with $M$ for a fixed $L$. We observe that the  section error rate (averaged over $200$ trials)
decreases with $M$ up to $M=2^9$, and then starts increasing. This is in sharp
contrast to the SE prediction \eqref{eq:est_ser} (dashed line in Figure \ref{fig:lvsm_ser}) which keeps decreasing as   $M$ is increased.  

This divergence between the actual section error rate and the SE prediction for large $M$  is
due to large fluctuations  in the number of section errors across trials. Theorem \ref{thm:main_amp_perf} shows how the concentration of section error rates near the SE prediction depends on $L$ and $M$.  Since the  probability bound in \eqref{eq:pezero} depends on the ratio $L / (\log M)^{2T-1}$, for a given $L$  the probability of large deviations from the SE
prediction  increases with $M$. 

 This leads to the situation shown in Figure~\ref{fig:lvsm_ser}, which shows that the SE
prediction  $\mc{E}^{\text{SE}}_{\sec}$ continues to decrease with $M$, but
beyond a certain value of $M$, the observed average section error rate becomes
progressively worse due to loss of concentration. This is caused by a small
number of trials with a very large number of section errors, even as the
majority of trials experience lower and lower error rates as $M$ is increased.
This effect can be clearly seen in  Figure~\ref{fig:lvsm_hist}, which 
compares the histogram of section error rates over $200$ trials for $M=64$ and
$M=4096$. The distribution of errors is clearly different, but both cases have
the same average section error rate due to the poorer concentration for
$M=4096$.

Therefore,  given $R, \snr$, and $L$, there is an optimal $M$ that minimizes
the empirical section error rate. Beyond this value of $M$, the benefit from
any further increase is outweighed by the loss of concentration. For a given
$R$, values of $M$ close to $L$ are a good starting point for optimizing the
empirical section error rate, but obtaining closed-form estimates of the
optimal $M$  for a given $L$ is still an open question.

\subsection{Effect of power allocation on concentration} \label{subsec:conc-pa}

The non-asymptotic result of Lemma \ref{lem:xt_taut_lb} indicates that at finite lengths, the minimum power required for a section $\ell$ to decode  in an iteration may be slightly different   than that indicated by the asymptotic approximation  in \eqref{eq:lemma1b}.  Recall that the iterative power allocation algorithm in Section \ref{sec:pa:iterative} was designed based on \eqref{eq:lemma1b}. We can compensate for the difference between the approximation  and the actual value of $x(\tau)$ by running the iterative power allocation in Algorithm \ref{alg:iterative-pa} using a modified rate $R_{\text{PA}}$ which may be slightly  different from the  communication rate $R$.

If we run the power allocation algorithm with $R_\text{PA}> R$,  from
\eqref{eq:lemma1b} we see that additional power is allocated to the initial
blocks, at the cost of less power for the final blocks (where the allocation is
flat). Conversely, choosing $R_\text{PA} < R$ allocates less power to the initial
blocks, and increases the power in the final  sections which have a flat
allocation. This increases the likelihood of the initial section being decoded
in error; in a trial when this happens, there will be a large number of section
errors. However, if the initial sections are decoded correctly, the additional
power in the final sections increases the probability of the trial being
completely error-free.  Thus choosing  $R_{PA}< R$  makes completely error-free
trials more likely, but also increases the likelihood of having trials with a
large number of sections in error.

To summarize, the larger the  $R_{PA}$, the better the concentration of section error rates of individual trials around the overall average.  However, increasing $R_{PA}$ beyond a point just increases the average section error rate because of too little power being allocated to the final sections. 

Through numerical experiments, we found that the value of $\frac{R_{PA}}{R}$ that minimizes the average section error rate increases with $R$. In
particular, the optimal  $\frac{R_{PA}}{R}$ was  $0$ for $R \leq 1$ bit; the
optimal $\frac{R_{PA}}{R}$ for $R=1.5$ bits  was close to 1, and for $R=2$ bits, the
optimal $\frac{R_{PA}}{R}$ was between $1.05$ and $1.1$. Though this provides a useful design
guideline, a deeper theoretical analysis of the role of $R_{PA}$ in optimizing
the finite length error performance  is an open question.
 


\section{Comparison with coded modulation}
\label{sec:coded_mod_comp}

\begin{figure}
    \centering
    \includegraphics[width=0.9\columnwidth]{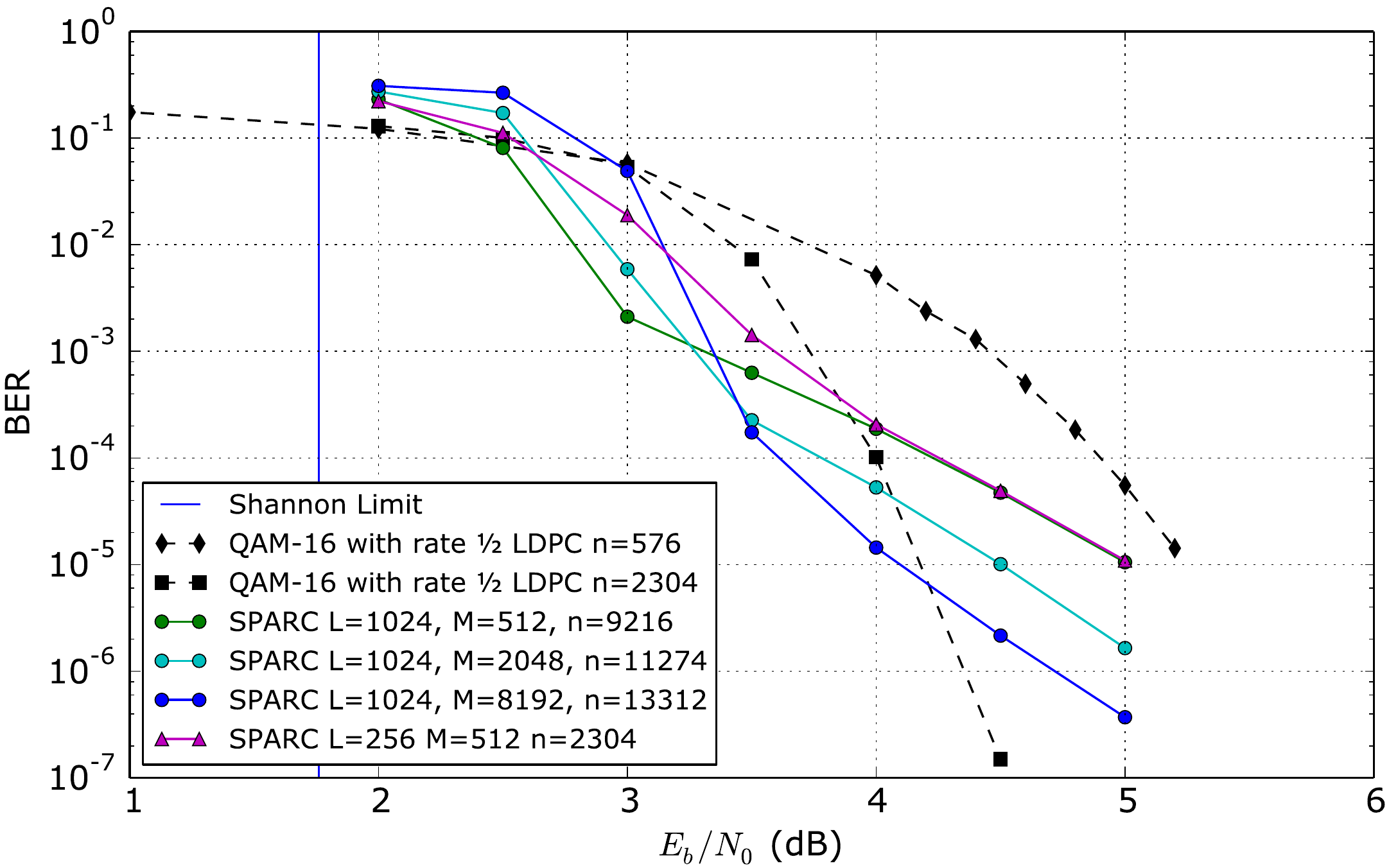}
    \caption{\small Comparison with LDPC coded modulation at $R=1$ bit}
    \label{fig:ldpc-comparison-r10}
\end{figure}
\begin{figure}
    \centering
    \includegraphics[width=0.9\columnwidth]{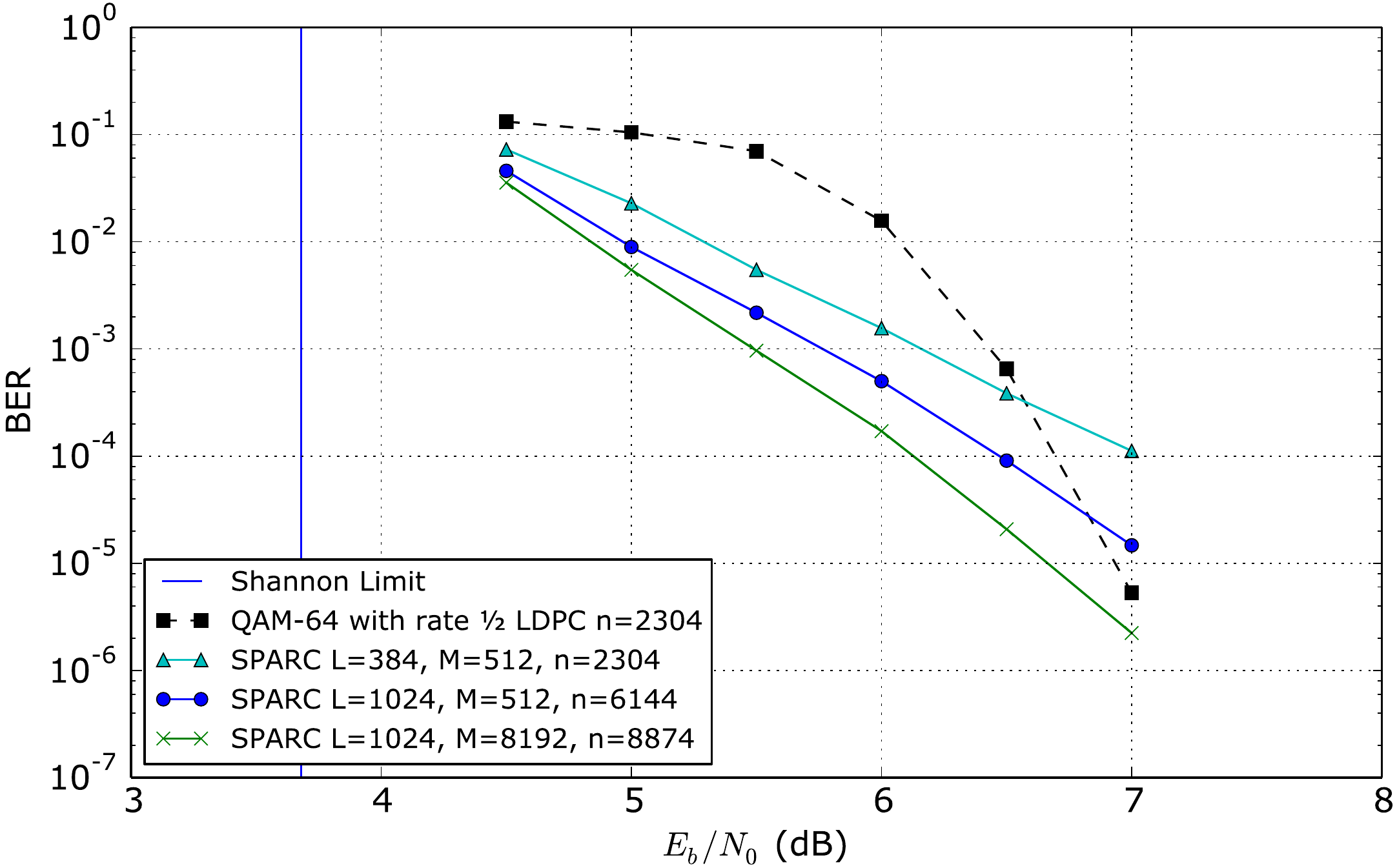}
    \caption{\small Comparison with LDPC coded modulation at $R=1.5$ bits}
    \label{fig:ldpc-comparison-r15}
\end{figure}

We compare the error performance of AMP-decoded SPARCs with coded modulation with LDPC codes from the WiMax standard IEEE 802.16e. For the latter, we consider: 1)  A $16$-QAM constellation with a rate $\frac{1}{2}$ LDPC code  for an overall rate  $R= 1$ bit/channel use/real dimension, and 2) A $64$-QAM constellation with a rate $\frac{1}{2}$ LDPC code  for an overall rate $R= 1.5$ bits/channel use/real dimension. (The spectral efficiency is $2R$ bits/s/Hz.)
The coded modulation results, shown in dashed lines in Figures \ref{fig:ldpc-comparison-r10} and \ref{fig:ldpc-comparison-r15}, are obtained using the CML toolkit \cite{cml} with LDPC code lengths $n=576$ and $n=2304$.

Throughout this section and the next,  rate will be measured in \emph{bits}.

Each figure compares the bit error rates (BER) of the coded modulation schemes with various SPARCs of the same rate, including a SPARC with a matching code length of $n=2304$.  Using $P=E_b R$ and $\sigma^2 = \frac{N_0}{2}$,
 the signal-to-noise ratio of the SPARC can be expressed as  $\frac{P}{\sigma^2} = \frac{2R E_b}{N_0}$.   The SPARCs are implemented using Hadamard-based design matrices,  power allocation designed using the iterative algorithm in Sec. \ref{sec:pa:iterative} with $B=L$, parameters $\widehat{\tau}_t^2$ computed online, and the early termination criterion described in \ref{subsec:online_tau}.  A Jupyter notebook detailing the SPARC implementation in Python is available at \cite{adamSPARC}.

 \section{AMP with partial outer codes}
\label{sec:ldpc-outer}

Figures~\ref{fig:ldpc-comparison-r10} and~\ref{fig:ldpc-comparison-r15} show
that for block lengths of the order of a few thousands, AMP-decoded SPARCs
do not exhibit a steep waterfall in section error rate.  Even at high $E_b/N_0$ values, it is still common to observe a small
number of section errors. If these could be corrected, we could hope to obtain
a sharp waterfall behavior similar to the LDPC codes.

 In the simulations of the AMP decoder described above, when $M$
and $R_\text{PA}$ are chosen such that the average error rates are well-concentrated around the state evolution prediction, the number of section errors observed is similar across
 trials.  Furthermore, we observe that the majority
of sections decoded incorrectly are those in the flat region of the power
allocation, i.e., those with the lowest allocated power. This suggests we could use a
high-rate outer code to protect just these sections, sacrificing some rate, but
less than if we na\"{\i}vely protected all sections. We call the sections covered by the outer code \emph{protected} sections, and conversely the earlier sections which are not covered by the outer code  are  \emph{unprotected}.   In \cite{AntonyML}, it was shown that  a Reed-Solomon outer code (that covered all the sections)  could be used to obtain a bound the probability of codeword error from  a bound on the probability of excess section error rate.

Encoding with  an outer code (e.g., LDPC or Reed-Solomon code) is straightforward:  just replace the message bits corresponding to the protected sections with coded bits generated using the usual encoder for the chosen outer code. To decode, we would like to obtain bit-wise posterior probabilities for each  codeword bit of the outer
code, and use them as inputs to a soft-information decoder, such as a
sum-product or min-sum decoder for LDPC codes. The output of the AMP decoding
algorithm permits this: it yields $\beta^T$, which contains weighted \emph{section-wise} posterior
probabilities; we can directly transform these into \emph{bit-wise}
posterior probabilities. See
Algorithm~\ref{alg:pp2bp} for details.

Moreover, in addition to correcting AMP decoding errors in the protected sections, successfully decoding the outer code also provides a way to correct  remaining errors in the unprotected sections of the SPARC codeword. Indeed, after decoding the outer code we can subtract the contribution of the protected sections from the channel output sequence 
$y$, and  re-run the AMP decoder on just the unprotected sections. The key point  is that subtracting the contribution of the later (protected) sections eliminates the interference due to these sections; then running  the AMP decoder on the unprotected sections is akin to operating at a much lower rate.

Thus the decoding procedure has three stages:  i)  first round of AMP  decoding, ii)  decoding the outer code using soft outputs from the AMP, and  iii)  subtracting the contribution of the sections protected by the outer code, and running the AMP decoder again for the unprotected sections. We find that the final stage, i.e., running the AMP decoder again after the outer code recovers errors in the protected sections of the SPARC, provides a significant advantage over a standard application of
an outer code, i.e., decoding the final codeword after the second stage.

We describe this combination of SPARCs with outer codes below, using an
LDPC outer code. The resulting error rate curves exhibit sharp waterfalls in
final error rates, even when the LDPC code
only covers a minority of the SPARC sections.

\begin{figure}[t]
    \centering
    \includegraphics[width=\columnwidth]{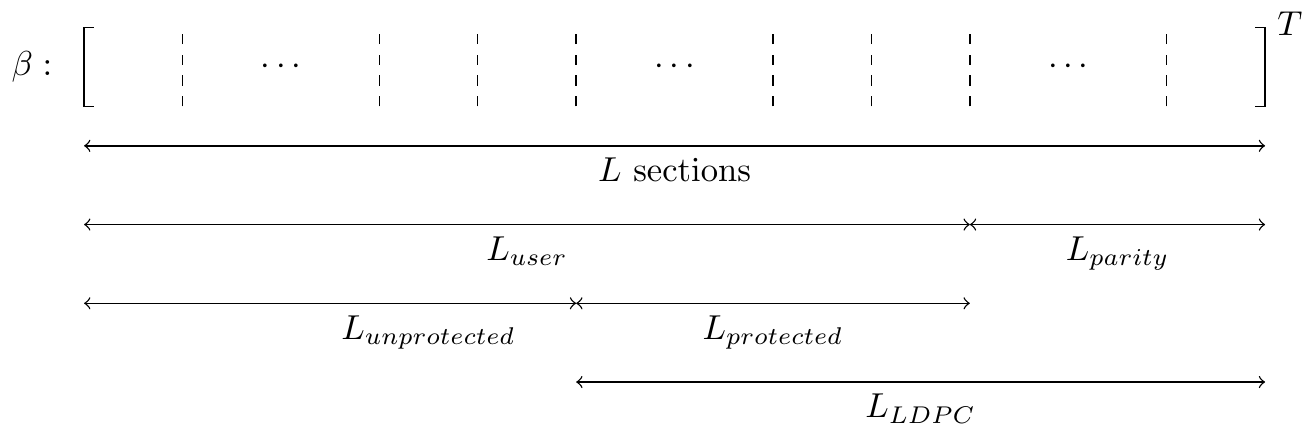}
    \caption{\small Division of the $L$ sections of $\beta$ for an outer LDPC code}
    \label{fig:ldpc-outer-code-division}
    \vspace{-10pt}
\end{figure}

We use a binary LDPC outer code with rate  $R_{LDPC}$, block
length $n_{LDPC}$ and code dimension $k_{LDPC}$, so that
$k_{LDPC}/n_{LDPC}=R_{LDPC}$. For clarity of exposition we assume that both
$n_{LDPC}$ and $k_{LDPC}$ are multiples of $\log M$ (and consequently that $M$
is a power of two). As each section of the SPARC corresponds to $\log M$ bits, if $\log M$ is an integer, then $n_{LDPC}$ and $k_{LDPC}$ bits represent an
integer number of SPARC sections, denoted by  $$L_{LDPC} = \frac{n_{LDPC}}{\log M}  \quad \text{and} \quad L_{protected} = \frac{k_{LDPC}}{\log M},$$ respectively. The assumption that $k_{LDPC}$ and  $n_{LDPC}$ are multiples of $\log M$ is not necessary in practice; the general case is discussed at the end of the next subsection.

We partition the $L$  sections of the SPARC codeword as shown in Fig~\ref{fig:ldpc-outer-code-division}. There are $L_{user}$ sections corresponding to the user (information) bits; these sections are divided into
\emph{unprotected} and \emph{protected} sections, with only the latter being covered by the outer LDPC code.  The parity bits of the LDPC codeword index the last $L_{parity}$ sections of the SPARC. For convenience, the \emph{protected} sections and the  \emph{parity} sections  together are referred to as the  \emph{LDPC} sections.

For a numerical example, consider the case where $L=1024$, $M=256$. There are
$\log M=8$ bits per SPARC section. For a $(5120, 4096)$ LDPC code ($R_{LDPC}=4/5$)
we obtain the following relationships between the number of the sections of each kind:
\begin{align*}
  &   L_{parity}=\frac{n_{LDPC}-k_{LDPC}}{\log M}=\frac{(5120-4096)}{8}=128,\\
    &     L_{user} =L-L_{parity}=1024-128=896, \\
  &   L_{protected}=\frac{k_{LDPC}}{\log M}= \frac{4096}{8}=512,\\
   & L_{LDPC} =L_{protected} +L_{parity}=512+128=640, \\
      &     L_{unprotected} = L_{user} -  L_{protected} = L-L_{LDPC}=384. 
\end{align*}
There are $L_{user}\log M=7168$ user bits, of which the final $k_{LDPC}=4096$
are encoded to a systematic $n_{LDPC}=5120$-bit LDPC codeword. The resulting
$L\log M=8192$ bits (including both the user bits and the LDPC parity bits) are
encoded to a SPARC codeword using the  SPARC encoder and power allocation described in previous sections.

We continue to use $R$ to denote the overall user rate, and $n$ to denote the SPARC code length so that  $nR=L_{user}\log M$.
The  underlying SPARC rate (including the overhead due to the outer code) is denoted by $R_{SPARC}$. We note that 
$nR_{SPARC}=L\log M$, hence $R_{SPARC} > R$. For example,  with $R=1$ and $L,M$ and the outer code parameters as chosen above, $n=L_{user}(\log M)/R=7168$, so $R_{SPARC}=1.143$.

\begin{algorithm}[t]
    \caption{Weighted position posteriors $\beta_\ell$ to bit posteriors $p_0,\ldots,p_{\log M - 1}$ for section $\ell\in [L]$}
    \label{alg:pp2bp}
    \begin{algorithmic}
        \REQUIRE $\beta_\ell=[\beta_{\ell,1},\ldots,\beta_{\ell,M}]$, for $M$ a power of 2
        \STATE Initialise bit posteriors $p_0,\ldots,p_{\log M -1} \leftarrow 0$
        \STATE Initialise normalization constant $c \leftarrow \sum_{i=1}^{M}\beta_{\ell,i}$
        \FOR{$\log i=0,1,\ldots,\log M - 1$}
            \STATE $b \leftarrow \log M - \log i - 1$
            \STATE $k \leftarrow i$
            \WHILE{$k < M$}
                \FOR{$j=k+1, k+2, \ldots, k+i$}
                    \STATE $p_b \leftarrow p_b + \beta_{\ell,j}/c$
                \ENDFOR
                \STATE $k \leftarrow k + 2i$
            \ENDWHILE
        \ENDFOR
        \RETURN $p_0,\ldots,p_{\log M -1}$
    \end{algorithmic}
\end{algorithm}

\subsection{Decoding SPARCs with LDPC outer codes}

At the receiver, we decode as follows:

\begin{enumerate}
    \item Run the AMP decoder to obtain $\beta^T$. Recall that entry $j$ within section $\ell$ of $\beta^T$
is proportional to  the posterior probability of the column $j$ being the
transmitted one for section $\ell$. Thus the AMP decoder gives section-wise posterior probabilities for each section $\ell \in [L]$. 

    \item Convert the section-wise posterior probabilities to bit-wise posterior probabilities using
Algorithm~\ref{alg:pp2bp}, for each of the $L_{LDPC}$ sections.  This requires
$O(L_{LDPC}M \log M)$ time complexity, of the same order as one iteration of
AMP.

    \item Run the LDPC decoder using the bit-wise posterior probabilities obtained in Step 2 as inputs.

    \item If the LDPC decoder fails to produce a valid LDPC codeword, terminate decoding here,
using $\beta^T$  to produce $\hat{\beta}$ by selecting the maximum value in each
section (as per usual AMP decoding).

    \item If the LDPC decoder succeeds in finding a  valid codeword, we use it to re-run AMP decoding on
        the unprotected sections. For this, first convert the LDPC codeword bits to a partial
        $\hat{\beta}_{LDPC}$ as follows, using a method similar to the original SPARC encoding:
        \begin{enumerate}
            \item Set the first $L_{unprotected}M$ entries of  $\hat{\beta}_{LDPC}$  to zero,
            \item The remaining $L_{LDPC}$ sections (with $M$ entries per section) of $\hat{\beta}_{LDPC}$ will have exactly one-non zero entry per section, with the LDPC codeword   determining the location of the non-zero in each section.     Indeed, noting that $n_{LDPC}= L_{LDPC} \log M$, we consider the LDPC codeword as a concatenation of $L_{LDPC}$ blocks of $\log M$ bits each, so that each block of bits indexes the location of the non-zero entry in one section of  $\hat{\beta}_{LDPC}$.  The value of the non-zero in section $\ell$ is set to  $\sqrt{n P_\ell}$, as per the power allocation.
        \end{enumerate}
        Now subtract the codeword corresponding to $\hat{\beta}_{LDPC}$ from
        the original channel output $y$, to obtain $y'=y-A\hat{\beta}_{LDPC}$.

    \item Run the AMP decoder again, with input $y'$, and operating only
        over the first $L_{unprotected}$ sections. As this operation is effectively at a much
        lower rate than the first decoder (since the interference contribution from
       all the protected sections is removed), it is more likely that the
        unprotected bits are decoded correctly than in the first AMP decoder.
        
     We note that instead of generating $y'$, one could run the AMP decoder directly on $y$, but enforcing that in each AMP iteration, each  of the $L_{LDPC}$ sections has all its non-zero mass on the entry determined by $\hat{\beta}_{LDPC}$, i.e., consistent with Step 5.b).

    \item Finish decoding, using the output of the final AMP decoder to find
        the first $L_{unprotected}M$ elements of $\hat{\beta}$, and using
        $\hat{\beta}_{LDPC}$ for the remaining $L_{LDPC}M$ elements.

\end{enumerate}


\subsection{Simulation results}

The combined AMP and outer LDPC setup described above was simulated using the
(5120, 4096) LDPC code ($R_{LDPC}=4/5$) specified in \cite{CCSDS131} with a
min-sum decoder. Bit error rates were measured only over the user bits, ignoring any bit
errors in the LDPC parity bits.

Figure~\ref{fig:ldpc-outer-r08} plots results at overall rate $R=\frac{4}{5}$, where
the underlying LDPC code (modulated with BPSK) can be compared to the 
SPARC with LDPC outer code, and to a plain SPARC  with rate $\frac{4}{5}$. In this case
$R_{PA}=0$, giving a flat power allocation. Figure~\ref{fig:ldpc-outer-r15} plots results at overall rate $R=1.5$, where
we can compare to the QAM-64 WiMAX LDPC code, and to the plain SPARC with rate 1.5  of
Figure~\ref{fig:ldpc-comparison-r15}.

The plots show that  protecting a fraction of sections with an outer code does provide a steep waterfall above a threshold value of $\frac{E_b}{N_0}$. Below this threshold, the combined SPARC + outer code has worse error performance than the plain rate $R$ SPARC without the outer code. This  can be explained as follows. The combined code has a higher SPARC rate
$R_{SPARC}>R$, which leads to  a  larger section error rate for the first AMP decoder, and consequently,  to worse bit-wise posteriors at the input of the LDPC decoder. For $\frac{E_b}{N_0}$ below the threshold, the noise level at the input of the LDPC decoder is beyond than the error-correcting capability of the LDPC code, so the LDPC code effectively does not correct any section errors. Therefore the overall error rate is worse with the outer code.

Above the threshold, we observe that the second AMP decoder (after subtracting the contribution of the LDPC-protected sections) is successful at
decoding the unprotected sections that were initially decoded incorrectly. This
is especially apparent in the  $R=\frac{4}{5}$ case (Figure \ref{fig:ldpc-outer-r08}), where the section errors are
uniformly distributed over all sections due to the flat power allocation;
errors are just as likely in the unprotected sections as in the protected
sections.

\begin{figure}[t]
    \centering
    \includegraphics[width=0.8\columnwidth]{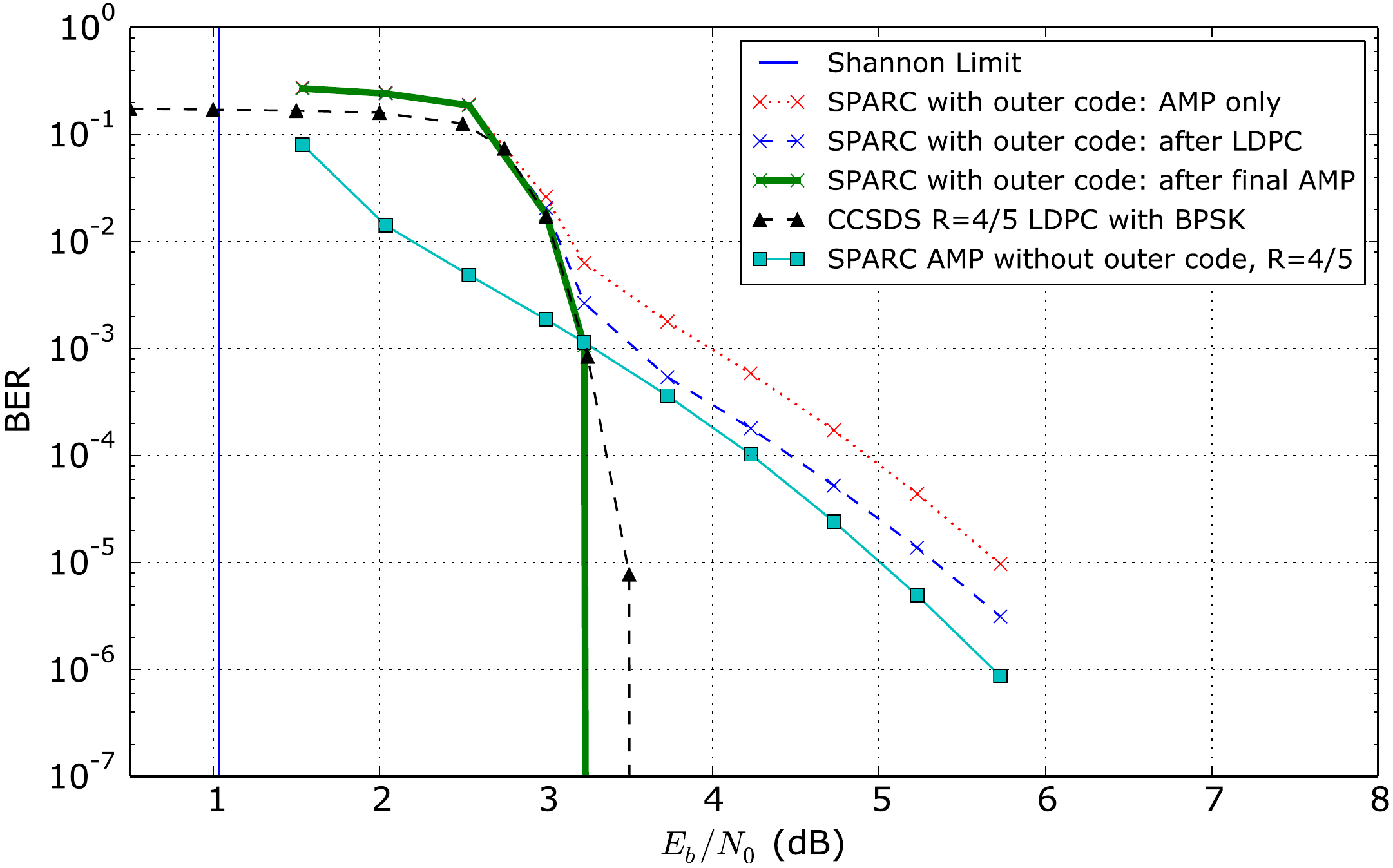}
    \caption{\small Comparison to plain AMP and to BPSK-modulated LDPC at overall rate $R=0.8$.
    The SPARCs are both $L=768$, $M=512$. The underlying SPARC rate when the outer code is included is $R_{SPARC}=0.94$. The BPSK-modulated LDPC is the same CCSDS LDPC
code \cite{CCSDS131} used for the outer code. For this configuration, $L_{user}=654.2$, $L_{parity}=113.8$,
$L_{unprotected}=199.1$, $L_{protected}=455.1$, and $L_{LDPC}=568.9$.}
    \label{fig:ldpc-outer-r08}
    \vspace{-8pt}
\end{figure}
\begin{figure}[h]
    \centering
    \includegraphics[width=0.8\columnwidth]{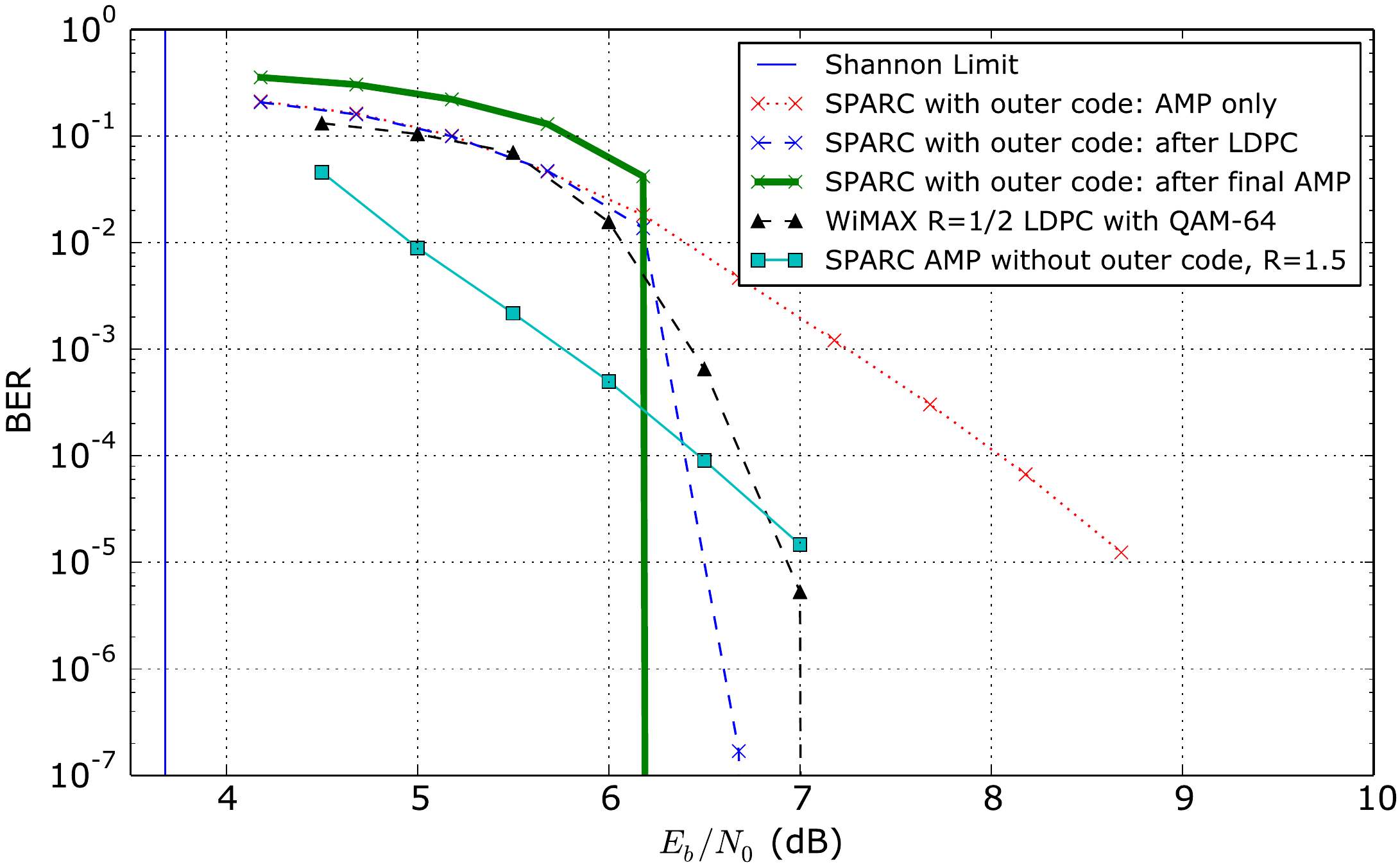}
    \caption{\small Comparison to plain AMP and to the QAM-64 WiMAX LDPC of Section~\ref{sec:coded_mod_comp} at overall rate $R=1.5$
    The SPARCs are both $L=1024$, $M=512$. The underlying SPARC rate including the
outer code is $R_{SPARC}=1.69$. For this configuration, $L_{user}=910.2$,
$L_{parity}=113.8$, $L_{unprotected}=455.1$, $L_{protected}=455.1$, and $L_{LDPC}=455.1$.}
    \label{fig:ldpc-outer-r15}
    \vspace{-8pt}
\end{figure}

\subsection{Outer code design choices}

The error performance with an outer code is sensitive to what fraction of sections are
protected by the outer code. When more sections are protected by the outer code, the overhead of using the outer code is also
higher, driving $R_{SPARC}$ higher for the same overall user rate $R$. This
leads to worse error performance in the initial AMP decoder, which has to operate
at the higher rate $R_{SPARC}$. As discussed above, if $R_{SPARC}$ is increased too much, the bit-wise 
posteriors input to the LDPC decoder are degraded beyond its ability to successfully decode,
giving poor overall error rates.

Since the number of sections covered by the outer code depends on both $\log M$ and $n_{LDPC}$,
various trade-offs are possible. For example,  given $n_{LDPC}$, choosing a larger value of $\log M$ corresponds to  fewer sections being covered by the outer code. This results in smaller rate overhead, but increasing $\log M$ may also affect concentration of the error rates around the SE predictions, as discussed in Section~\ref{sec:lvsm}. We conclude with two remarks about the choice  of parameters for the SPARC and the outer code. 
\begin{enumerate}
\item When using an outer code, it is highly beneficial to have good concentration of
the section error rates for the initial AMP decoder. This is because a small
number of errors in a single trial can usually be fully corrected by the outer code, while
occasional trials with a very large number of errors cannot.

\item Due to the second AMP decoder operation, it is not necessary for all sections
with low power to be protected by the outer code. For
example,  in Figure~\ref{fig:ldpc-outer-r08}, all sections have equal power, and around 30\% are not protected by the outer code. Consequently, these sections
are  often not decoded correctly by the first decoder. Only once
the protected sections are removed is the second decoder  able to correctly
decode these unprotected sections. In general the aim should be to cover all or most of the sections in the flat
region of the power allocation, but experimentation is necessary to determine
the best trade-off.
\end{enumerate}

{An interesting direction for future work would be to develop an EXIT chart  analysis \cite{ten1999convergence,ashikhmin2004extrinsic, RichUBook} to jointly optimize the design of the SPARC and the outer LDPC code.}

\chapter{Spatially Coupled SPARCs} \label{chap:sc_sparcs}

The efficient capacity-achieving decoders discussed in Chapter \ref{chap:AWGN_eff} all relied on power allocation across the sections. The design matrix was chosen with independent, identically distributed Gaussian entries, while the values of the non-zero coefficients varied across sections of the codeword. Equivalently, one can define the power allocation by changing the variance of the Gaussian entries in each section of the design matrix, while the non-zero coefficients of the codeword all have the same value. \emph{Spatial coupling} is a generalization of the latter view of power allocation. 

In a spatially coupled SPARC (SC-SPARC),  the design matrix is divided into multiple blocks, each with independent zero-mean Gaussian entries of a specified variance. The variance of the entries may vary across blocks, while the values of the non-zero entries in the message vector are all equal. Within this general framework, we will consider a simple construction with a band-diagonal spatially coupled design matrix, and show that it can achieve the AWGN capacity with AMP decoding without the need for power allocation. Furthermore, at finite code lengths,  numerical simulations indicate that SC-SPARCs have a much steeper decay of error rate than power allocated SPARCs  as we as we back off from the Shannon limit.  

The idea of spatial coupling was introduced in the context of LDPC codes \cite{felstrom1999time,lentmaier2010iterative, kudekar2011threshold,spatialCoup13,mitchell2015}, and later applied to compressed sensing in \cite{kudekar2010effect,krz12,DonSpatialC13}. Subsequently, spatially coupled SPARCs were studied by Barbier et al. in 
\cite{BarbKrz15,barbISIT16,barbITW16,BarbKrz17, BarbierDM17}. The discussion in this chapter is largely based on the spatially coupled SPARC construction and analysis presented in \cite{hsiehISIT18,rushITW18}.

\section{Spatially coupled SPARC construction} \label{sec:sc_AMP}

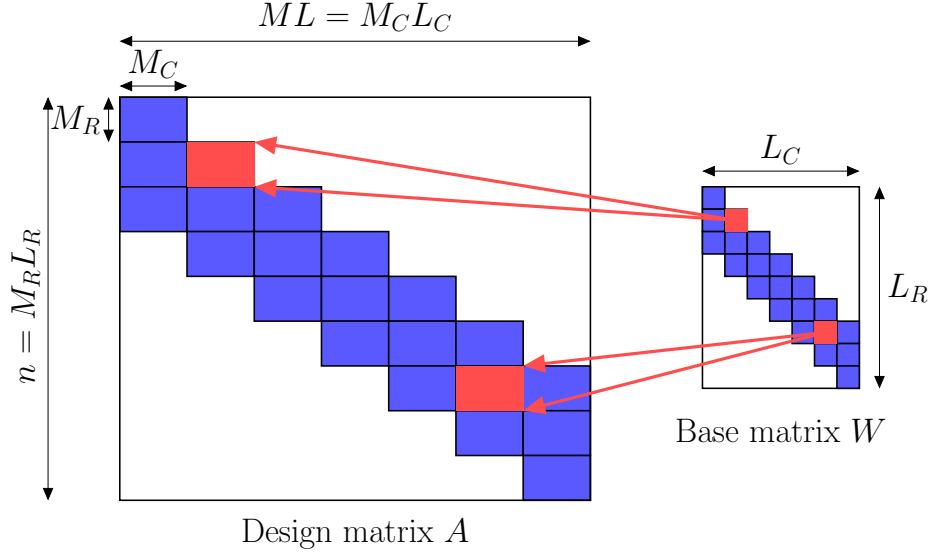
\begin{figure}[t]
\centering
    \begin{tikzpicture}[scale=0.4, font=\LARGE]
	
	\foreach \x in {0,1, 2, 3, 4, 5, 6}{
		\draw[fill=blue!65, thick] (10+3*\x, {-2*\x}) rectangle (13+3*\x, {-2-2*\x});
		\draw[fill=blue!65, thick] (10+3*\x, {-2-2*\x}) rectangle (13+3*\x, {-4-2*\x});
		\draw[fill=blue!65, thick] (10+3*\x, {-4-2*\x}) rectangle (13+3*\x, {-6-2*\x});
	}

        \draw[thick] (10, 0) -- (31, 0) -- (31, -18) -- (10, -18) -- cycle;
	\draw (20.5, -18.5) node[anchor=north] {Design matrix $A$};
	
        \draw[triangle 45-triangle 45] (9.5, 0) -- (9.5, -2);
        \draw (9.5, -1) node[anchor=east] {$\Mr$};
        \draw[triangle 45-triangle 45] (6.8, 0) -- (6.8, -18);
        \draw (5.8, -5.5) node[anchor=east, rotate=90] {$n = \Mr\Lr$};
        \draw[triangle 45-triangle 45] (10, 0.5) -- (13, 0.5);
        \draw (11.5, 0.5) node[anchor=south] {$\Mc$};
        \draw[triangle 45-triangle 45] (10, 2.5) -- (31, 2.5);
        \draw (20.5, 2.6) node[anchor=south] {$\M L = \Mc\Lc$};  	
                     
        \draw[thick] (36, -4) -- (43, -4) -- (43, -13) -- (36, -13) -- cycle; 
        \draw (39.5, -14) node[anchor=north] {Base matrix $W$};
                
        \draw[triangle 45-triangle 45] (43.9, -4) -- (43.9, -13);
        \draw (43.9, -8.5) node[anchor=west] {$\Lr$};
        \draw[triangle 45-triangle 45] (36, -3.3) -- (43, -3.3);
        \draw (39.5, -3.3) node[anchor=south] {$\Lc$};
                
			
	\foreach \x in {0,1, 2, 3, 4, 5, 6}{
		\draw[fill=blue!65, thick] (36+\x, {-4-\x}) rectangle (37+\x, {-5-\x});
		\draw[fill=blue!65, thick] (36+\x, {-5-\x}) rectangle (37+\x, {-6-\x});
		\draw[fill=blue!65, thick] (36+\x, {-6-\x}) rectangle (37+\x, {-7-\x});
	}
	
        	\fill[red!70] (37,-5) rectangle (38,-6);
        	\fill[red!70] (13,-2) rectangle (16,-4);        
	\draw[-triangle 45, red!70, ultra thick] (37.5, -5.5) -- (16, -2);
	\draw[-triangle 45, red!70, ultra thick] (37.5, -5.5) -- (16, -4);
	
	\fill[red!70] (36+5, {-5-5}) rectangle (37+5, {-6-5});
        	\fill[red!70] (25,-12) rectangle (28,-14);
	\draw[-triangle 45, red!70, ultra thick] (41.5, -10.5) -- (28, -12);
	\draw[-triangle 45, red!70, ultra thick] (41.5, -10.5) -- (28, -14);
    \end{tikzpicture}
\caption{\small A spatially coupled design matrix $A$ is divided into blocks of size $\Mr\times \Mc$. There are $\Lr$ and $\Lc$ blocks in each column and row respectively. The independent matrix entries are normally distributed, $A_{ij} \sim \mathcal{N}(0,\frac{1}{L}\Wricj)$, where $W$ is the base matrix. The base matrix shown here is an $(\omega, \Lambda)$ base matrix with parameters $\omega=3$ and $\Lambda=7$. The white parts of $A$ and $W$ correspond to zeros.}
\label{fig:sparc_scmatrix}
\vspace{-6pt}
\end{figure}

As in the standard construction, a spatially coupled (SC) SPARC  is defined by a design matrix $A$ of dimension  $n\times \M L$, where $n$ is the code length.  The codeword is $x=A\beta$, where $\beta$ has one non-zero entry in each of the $L$ sections. 

 In an SC-SPARC, the matrix $A$ consists of independent zero-mean normally distributed entries whose variances are specified by a \emph{base matrix} $W$ of dimension $\Lr \times \Lc$. The design matrix $A$ is obtained from the base matrix $W$ by replacing each entry $W_{rc}$, for $r\in[\Lr]$, $c\in[\Lc]$, by an $\Mr\times\Mc$ block with i.i.d. entries $\sim \mc{N}(0, W_{rc}/{L} )$.  This is analogous to the ``graph lifting'' procedure in constructing SC-LDPC codes from protographs \cite{mitchell2015}. See Fig. \ref{fig:sparc_scmatrix} for an example, and note that $n= \Lr \Mr$ and $ML = \Lc \Mc$. 

From the construction, the design matrix has independent normal entries
\be\label{eq:construct_Aij}
A_{ij} \sim \mc{N}\left(0,\frac{1}{L} \Wricj \right) \ \forall \ i \in [n], \ j\in[\M L].
\ee
The operators $\sfr(\cdot):[n]\rightarrow[\Lr]$ and $\sfc(\cdot):[\M L]\rightarrow[\Lc]$ in \eqref{eq:construct_Aij} map a particular row or column index to its corresponding \emph{row block} or \emph{column block} index. 
We require  $\Lc$ to divide $L$, resulting in $\frac{L}{\Lc}$ sections per column block.

The non-zero coefficients of $\beta$  are all set to $1$. Then,  to satisfy the power constraint $\frac{1}{n}\norm{x}^2=P$, it can be shown that the entries of the base matrix $W$ must satisfy 
\be
\label{eq:W_power_contraint}
\frac{1}{\Lr \Lc}\sum_{r=1}^{\Lr}  \sum_{c=1}^{\Lc} W_{rc} = P.
\ee

The trivial  base matrix with $\Lr=\Lc=1$  corresponds to a standard (non-SC) SPARC  with flat power allocation  (discussed in Chapter \ref{chap:AWGN_opt}), while a single-row base matrix $\Lr=1$, $\Lc=L$   is equivalent to a standard SPARC with power allocation  (Chapters \ref{chap:AWGN_eff} and \ref{chap:emp_perf}). Without loss of generality, we will assume that $\frac{1}{\Lc}  \sum_{c=1}^{\Lc} W_{rc}$ and $\frac{1}{\Lr}  \sum_{r=1}^{\Lr} W_{rc}$
are bounded above and below by strictly positive constants. 

Here, we will use the following base matrix  inspired by the coupling structure of SC-LDPC codes constructed from protographs \cite{mitchell2015}.
\begin{definition}
\label{def:ome_lamb}
An $(\omega , \Lambda)$ base matrix $W$ for SC-SPARCs is  described by two parameters: coupling width $\omega\geq1$ and coupling length $\Lambda\geq 2\omega-1$. The matrix has $\Lr=\Lambda+\omega-1$ rows,  $\Lc=\Lambda$ columns, with each column having $\omega$ identical non-zero entries. For an average power constraint $P$, the  $(r,c)$th entry of the base matrix, for  $r \in [\Lr], c\in[\Lc]$, is given by
\begin{equation}\label{eq:W_rc}
W_{rc} =
\begin{cases}
 	\ P \cdot \frac{\Lambda+\omega-1}{\omega} \quad &\text{if} \ c\leq r \leq c+\omega-1,\\
	\ 0 \quad &\text{otherwise}.
\end{cases}
\end{equation}
\end{definition}
For example, the base matrix in Fig. \ref{fig:sparc_scmatrix} has parameters $\omega=3$ and $\Lambda=7$.
This base matrix construction was also used in \cite{liang2017} for SC-SPARCs. Other base matrix constructions can be found in \cite{krz12,DonSpatialC13,barbISIT16,BarbKrz17}. 

Each non-zero entry in an $(\omega , \Lambda)$ base matrix $W$ corresponds to an $\Mr \times (\M L/\Lc)$ block in the design matrix $A$. Each of these blocks can be viewed as a standard (non-SC) SPARC with $\frac{L}{\Lc}$ sections (with $\M$ columns in each section), code length $\Mr$, and rate $R_\text{inner} = \frac{(L/\Lc) \ln \M}{\Mr}$ nats. Since $nR= L \ln M$, the overall rate of the SC-SPARC is related to $R_\text{inner}$ according to
\be\label{eq:R_Rsparc}
R = \frac{\Lambda}{\Lambda + \omega -1} R_\text{inner}.
\ee
The coupling width $\omega$ is usually an integer greater than 1, so $R < R_\text{inner}$. The difference $(R_\text{inner} - R)$ is often referred to as a rate loss. The rate loss depends on the ratio $(\omega-1)/\Lambda$, which becomes negligible when $\Lambda$ is large w.r.t. $\omega$.

\begin{remark}
SC-SPARC constructions generally have a `seed' to jumpstart decoding. In \cite{barbISIT16}, a small fraction of sections of $\beta$ are fixed a priori --- this pinning condition is used to analyze the state evolution equations via the potential function method. Analogously, in the construction in \cite{BarbKrz17}, additional rows are introduced in the design matrix for the blocks corresponding to the first row of the base matrix. In an $(\omega, \Lambda)$ base matrix,  the fact that the number of rows in the base matrix exceeds the number of columns by $(\omega - 1)$ helps decoding start from both ends. The asymptotic state evolution equations  in Sec. \ref{subsec:SE_asymp_analysis} show how AMP decoding progresses in an $(\omega, \Lambda)$ base matrix.
\end{remark}

\section{AMP decoder for spatially coupled SPARCs}\label{sec:AMP}

The decoder wishes to recover the message vector $\beta\in\mathbb{R}^{\M L}$ from the channel output sequence $y\in\mathbb{R}^n$, given by
\begin{equation}\label{eq:linear_model}
y = A\beta + w,
\end{equation}
where the noise vector $w \ \sim_{i.i.d.} \normal{N}(0, \sigma^2)$.

The  procedure to derive an Approximate Message Passing (AMP) decoding algorithm for SC-SPARCs is similar to that for standard SPARCs (Section \ref{sec:AMP_dec}, p. \pageref{eq:beta_update}), with modifications to account for the different variances for the blocks of $A$ specified by the base matrix.  The AMP decoder intitializes  $\beta^0$ to the all-zero vector, and for $t \geq 0$, iteratively computes
\begin{align}
z^t &= y - A\beta^t + \widetilde{\sfb}^t \odot z^{t-1}  \label{eq:scamp_decoder_z} \\
\beta^{t+1} &= \eta^t(\beta^t + (\widetilde{\Smat}^t \odot A)^* z^t).  \label{eq:scamp_decoder_beta}
\end{align}
Here $\odot$ denotes the Hadamard (entry-wise) product.  The denoising function $\eta^t$, and  $\widetilde{\sfb}^t \in \reals^n$, $\widetilde{\Smat}^t \in \reals^{n \times ML}$  are defined below in terms of the state evolution parameters. 

\subsection{State evolution for SC-SPARCs} 

We recall from Section \ref{subsec:iter_soft_dec} that state evolution is a scalar recursion (see \eqref{eq:tau_def}--\eqref{eq:xt_tau_def}) that captures the decoding progression of the AMP decoder for standard SPARCs.
The key difference in SC-SPARCs is that the decoding progression depends on the row block index $r \in [\Lr]$ and  column block  index $c \in [\Lc]$. Consequently,  the  state evolution parameters for SC-SPARCs will also depend on the row block index $r \in [\Lr]$ and the column block index $c \in [\Lc]$. In detail, the state evolution (SE) iteratively computes vectors $\phi^t\in\mathbb{R}^{\Lr}$ and $\psi^t\in\mathbb{R}^{\Lc}$ as follows. Initialize $\psi_c^{0} = 1$ for  $c\in[\Lc]$, and for $t=0,1,\ldots$, compute
\begin{align}
\phi_r^t &= \sigma^2 + \frac{1}{\Lc}\sum_{c=1}^{\Lc}W_{rc}\psi_c^t, \qquad r \in [\Lr], \label{eq:se_phi} \\
\psi_c^{t+1} &= 1 - \mathcal{E}(\tau_c^t), \qquad  \qquad \qquad  \ \ c \in [\Lc],  \label{eq:se_psi}
\end{align}
where \be 
\tau_c^t = \frac{R}{\ln{\M}}\left[\frac{1}{\Lr}\sum_{r} \frac{W_{rc}}{\phi_r^t}\right]^{-1}, 
\label{eq:tau_ct_def}
\ee and 
$\mathcal{E}(\tau_c^t)$ is defined as
\be
\label{eq:tau_asymp}
\mathcal{E}(\tau_c^t) 
= \mathbb{E}\left[ \frac{e^{U_1/\sqrt{\tau_c^t}}}{e^{U_1/\sqrt{\tau_c^t}} + e^{- \frac{1}{\tau_c^t}}\sum_{j=2}^{\M} e^{U_j/\sqrt{\tau_c^t}}}\right],
\ee
%
with $U_1,\ldots,U_{\M} \stackrel{\text{i.i.d.}}{\sim} \stdnorm$. 

We define the entries of the vector $\sfb^t\in\mathbb{R}^{\Lr}$ and the matrix  $\Smat^t \in \mathbb{R}^{\Lr \times \Lc}$ as
\be
\sfb^t_{r} = \frac{(\phi_{r}^{t} - \sigma^2)}{\phi_{r}^{t-1}}, \qquad \Smat^t_{rc} = 
\frac{\tau^t_{{c}}}{ \phi_{r}^t}, \qquad \text{ for } r \in [\Lr], \, c \in [\Lc].
\label{eq:btSmat}
\ee
The vector $\widetilde{\sfb}^t\in\mathbb{R}^{n}$ in \eqref{eq:scamp_decoder_z} is obtained by repeating  $\Mr$ times each entry of $\sfb^t$.  Similarly, $\widetilde{\Smat}^t \in \mathbb{R}^{n \times ML}$ in \eqref{eq:scamp_decoder_beta} is obtained by repeating each entry of $\Smat^t$
in an $\Mr \times \Mc$ matrix.

 The denoising function  $\eta^t= (\eta^t_1, \ldots, \eta^t_{ML}): \,  \mathbb{R}^{\M L}  \to \mathbb{R}^{\M L}$ in \eqref{eq:scamp_decoder_beta} is defined as follows.  For $j \in [\M L]$  such that $j\in \text{sec}(\ell)$ and the section $\ell$ is in the $c$th column block, 
\be\label{eq:eta_function}
\eta^t_j(s)
= \frac{e^{s_j/\tau^t_c}}{\sum_{j'\in \text{sec}(\ell)}e^{s_{j'}/\tau^t_c }}.
\ee
As in the case of standard SPARCs, $\eta^t_j(s)$ depends on all the components of $s$ in the section containing $j$.

\subsection{Interpretation of the AMP decoder}
The input to  $\eta^t(\cdot)$ in \eqref{eq:eta_function} can be viewed as a noisy version of $\beta$. In particular, the $c$th block of $s^t=s$ is approximately distributed as 
$\beta_c + \sqrt{\tau_c^t}Z_c$, where $Z_c \in \reals^{\Mr}$ is a standard normal random vector independent of $\beta$.  (Here $\beta_{c} \in\mathbb{R}^{\Mc}$ is the part of the message vector corresponding to   column block $c$ of the design matrix.)    Under the above distributional assumption, the denoising function $\eta_j$ in \eqref{eq:eta_function} is the  minimum mean squared error (MMSE) estimator  for  $\beta_{j}$, i.e., 
$$ \eta^t_j(s )=\mathbb{E}\left[\beta_j | s = \beta_c + \sqrt{\tau_c^t} \, Z_c \right], \qquad \text{ for } j \in [ML],$$
where the expectation is calculated over $\beta$ and  $Z$, with the location of the non-zero entry in each section of $\beta$ being uniformly distributed within the section. 

The entries of the modified residual $z^t$ in \eqref{eq:scamp_decoder_z} are approximately Gaussian and independent, with the variance determined by the block index. For $r\in[\Lr]$,  the SE parameter $\phi^t_r$ approximates the variance of the $r$th block of the residual $z^t_{r} \in \reals^{\Mr}$.  The `Onsager' term $\widetilde{\sfb}^t \odot z^{t-1}$  in \eqref{eq:scamp_decoder_z} reflects the block-wise structure of $z^t$. Finally, the parameter $\psi_c^t$ approximates the normalized mean-squared error in the estimate of $\beta_c$. This is discussed in the next section.

To summarize, the key difference from the  state evolution parameters for standard SPARCs  is that  the variances of the effective observation and the residual now depend on the column- and row-block indices, respectively.  These variances are captured by $\{ \tau_c^t \}_{c \in [\Lc]}$ and 
$\{ \phi_r^t \}_{r \in [\Lr]}$.

\section{Measuring the performance of the AMP decoder} \label{sec:AMP_perf}

The performance of a SPARC decoder is  measured by the \emph{section error rate},  defined as
\begin{equation}\label{eq:ser_def}
\mathcal{E}_{\text{sec}} := \frac{1}{L}\sum_{\ell=1}^L \mathbf{1}\{\widehat{\beta}_{\text{sec}(\ell)} \neq \beta_{\text{sec}(\ell)} \}.
\end{equation}
The section error rate can be shown to be bounded by the normalized mean squared error (NMSE) as follows.
\be
\mathcal{E}_{\text{sec}}  \leq   \frac{4}{L} \norm{\beta^T - \beta}^2  =4 \left[ \frac{1}{\Lc}\sum_{c=1}^{\Lc} 
\frac{\| \beta_{c}^T-\beta_{c} \|_2^2}{L/\Lc}  \right],
\ee
where in the last expression, we have written the total NMSE as an average over the NMSEs of the $\Lc$ blocks of the message vector.  

Figure  \ref{fig:ser_waveprop} shows that $\psi^t$ closely tracks the NMSE of each block of the message vector, i.e.,  $\psi^t_c \approx \frac{\|\beta_{c}^t-\beta_{c}\|_2^2}{L/\Lc}$ for $c\in[\Lc]$.   
We additionally observe from the figure that as AMP iterates, the NMSE  reduction propagates from the ends towards the center blocks.  This decoding propagation phenomenon can be explained using an asymptotic characterization of  the state evolution equations.

\begin{figure}[t]
\centering
\includegraphics[width=3.6in]{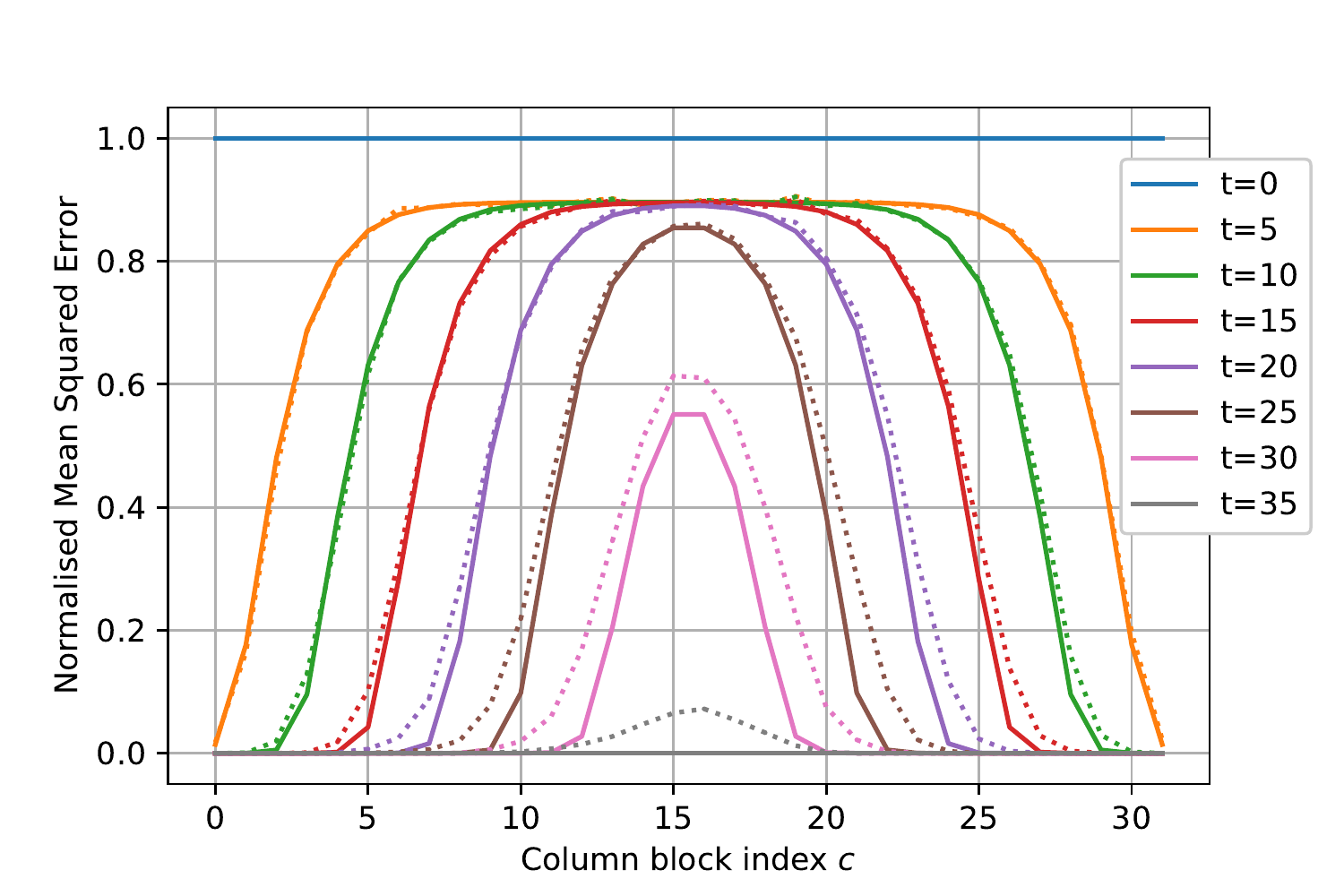}
\caption{ \small NMSE $\frac{\|\beta_{c}^t-\beta_{c}\|_2^2}{L/\Lc}$ vs. column block index $c\in [\Lc]$ for several iteration numbers. The SC-SPARC with an $(\omega,\Lambda)$ base matrix uses parameters: $R=1.5$ bits, $\mc{C}=2$ bits, $\omega=6$, $\Lambda=32$, $M=512$, $L=2048$ and $n=12284$. The solid lines are the SE predictions from \eqref{eq:se_psi}, and the dotted lines are the average NMSE over 100 instances of AMP decoding.}
\label{fig:ser_waveprop}
\vspace{-5pt}
\end{figure}


\subsection{Asymptotic State Evolution analysis} \label{subsec:SE_asymp_analysis}

Note that $\mathcal{E}(\tau_c^t)$ in \eqref{eq:tau_asymp}  takes a value in $[0,1]$. If $\mathcal{E}(\tau_c^t)=1 $, then $\psi^{t+1}_c =0$, which means that  the sections in column block $c$ are expected to decode correctly. If we terminate the AMP decoder at iteration $T$, we want  $\psi^{T}_c =0$, for $ c \in [\Lc]$, so that the entire message vector is expected to decode correctly. The condition under which $\mathcal{E}(\tau_c^t)$ equals 1 in the large system limit is specified by the following lemma.

\begin{lemma}\label{lemma:se_asymp}
In the limit as the section size $\M\to\infty$, the expectation $\mathcal{E}(\tau_c^t)$ in \eqref{eq:tau_asymp} converges to either 1 or 0 as follows.
\begin{equation} \label{eq:lim_Etauc}
\lim_{\M\to\infty} \mathcal{E}(\tau_c^t)
	=\begin{cases}
  		1 \quad \text{if} \ \frac{1}{\Lr}\sum_{r=1}^{\Lr} \frac{W_{rc}}{\phi_r^t}>2R\\
		0 \quad \text{if} \ \frac{1}{\Lr}\sum_{r=1}^{\Lr} \frac{W_{rc}}{\phi_r^t}<2R.
	\end{cases}
\end{equation}
This results in the following asymptotic state evolution recursion. Initialise $\bar{\psi}_c^{0} = 1$, for $c\in[\Lc]$, and for $t=0,1,2,\ldots$,
\begin{align}
\bar{\phi}_r^t &= \sigma^2 + \frac{1}{\Lc}\sum_{c=1}^{\Lc}W_{rc}\bar{\psi}_c^t, \qquad r \in [\Lr], \label{eq:se_asmyp_phi} \\
\bar{\psi}_c^{t+1} &= 1 - \mathbf{1}\left\{ \frac{1}{\Lr}\sum_{r=1}^{\Lr} \frac{W_{rc}}{\bar{\phi}_r^t}>2R \right\},  \qquad c \in [\Lc], \label{eq:se_asmyp_psi}
\end{align}
where $\bar{\phi}, \bar{\psi}$ indicate asymptotic values.
\end{lemma}

\begin{proof}
Recalling the definition of $\tau_c^t$ from \eqref{eq:tau_ct_def},  we  write $\frac{1}{\tau_c^t}=\nu_c^t \ln{\M}$, where
\begin{equation}\label{eq:nu_c}
\nu_c^t= \frac{1}{R \Lr}\sum_{r=1}^{\Lr} \frac{W_{rc}}{\phi_r^t}
\end{equation}
is an order 1 quantity because $\frac{1}{\Lr}\sum_{r=1}^{\Lr} W_{rc}$ is bounded above and below by positive constants. Therefore, 
\begin{equation}
\mathcal{E}(\tau_c^t) = \mathbb{E} \left[\frac{e^{\sqrt{\nu_c^t\ln{\M}}U_1}}{e^{\sqrt{\nu_c^t\ln{\M}}U_1} + {\M}^{-\nu_c^t} \sum_{j=2}^{\M} e^{\sqrt{\nu_c^t\ln{\M}}U_j}}\right],
\end{equation}
which is in the same form as the expectation in \eqref{eq:Eell_iter}. Therefore, following the steps in Section \ref{subsec:conv_exp_proof}, we conclude that
\begin{equation}\label{eq:lemma1_last}
\lim_{\M\to\infty} \mathcal{E}(\tau_c^t)
	=\begin{cases}
  		1 \quad \text{if} \ \nu_c^t>2\\
		0 \quad \text{if} \ \nu_c^t<2.
	\end{cases}
\end{equation}
The proof is completed by substituting the value of  $\nu_c^t$ from \eqref{eq:nu_c} in \eqref{eq:lemma1_last}.
\end{proof}
%

\begin{remark}
Using the definition of $\tau_c^t$ from \eqref{eq:tau_ct_def}, we can also write \eqref{eq:lim_Etauc} as 
\be
\lim_{\M\to\infty} \mathcal{E}(\tau_c^t)
	=\begin{cases}
  		1 \quad \text{if} \ \tau_c^t \ln \M < \frac{1}{2}\\
		0 \quad \text{if} \ \tau_c^t \ln \M > \frac{1}{2}.
	\end{cases}
\ee
\end{remark}

\begin{remark} Lemma \ref{lemma:se_asymp} is a generalization of Lemma \ref{lem:conv_expec}, the asymptotic SE result  for standard SPARCs
The term $\frac{1}{\Lr}\sum_r \frac{W_{rc}}{\phi^t_r}$ in \eqref{eq:lim_Etauc} represents the average signal to effective noise ratio at iteration $t$ for the column index $c$. If this quantity exceeds the prescribed threshold of $2R$, then the $c^{th}$ block of the message vector, $\beta_{c}$, will be decoded at the next iteration in the large system limit, i.e., $\psi^{t+1}_c=0$.
\end{remark}

The asymptotic SE recursion \eqref{eq:se_asmyp_phi}-\eqref{eq:se_asmyp_psi} is given for a general base matrix $W$. We now apply it to the $(\omega,\Lambda)$ base matrix introduced in Definition \ref{def:ome_lamb}.

\begin{lemma}\label{lemma:se_asymp_wLbasematrix}
The asymptotic SE recursion \eqref{eq:se_asmyp_phi}-\eqref{eq:se_asmyp_psi} for an $(\omega,\Lambda)$ base matrix $W$ is as follows. Initialise $\bar{\psi}_c^{0} = 1 \ \forall \ c\in[\Lambda ]$, and for $t=0,1,2,\ldots$,
\begin{align}
\bar{\phi}_r^t &= \sigma^2 \left(1 + \frac{\kappa\cdot \text{snr}}{\omega} \sum_{c= \underline{c}_r }^{\overline{c}_r} \bar{\psi}_c^t\right), \quad r \in [ \Lambda + \omega -1], \label{eq:se_asmyp_flat_phi} \\
\bar{\psi}_c^{t+1} &= 1 - \mathbf{1}\left\{ \frac{P}{\omega}\sum_{r=c}^{c+\omega-1} \frac{1}{\bar{\phi}_r^t}>2R \right\}, \quad c \in [ \Lambda ],  \label{eq:se_asmyp_flat_psi}
\end{align}
where $\kappa = \frac{\Lambda+\omega-1}{\Lambda}$, $\text{snr}=\frac{P}{\sigma^2}$, and
\be\label{eq:c_r}
(\underline{c}_r,\, \overline{c}_r)
	=\begin{cases}
  		(1,\, r) \ &\text{if} \ 1\leq r\leq\omega\\
		(r-\omega+1,\, r) \ &\text{if} \ \omega \leq r \leq \Lambda\\
	        (r-\omega+1,\, \Lambda) \ &\text{if} \ \Lambda \leq r \leq \Lambda + \omega - 1.
	\end{cases}
\ee
\end{lemma}
\begin{proof}
Substitute the value of $W_{rc}$ from \eqref{eq:W_rc}, and $\Lc=\Lambda$, $\Lr=\Lambda+\omega-1$ in \eqref{eq:se_asmyp_phi}-\eqref{eq:se_asmyp_psi}.
\end{proof}

Observe that the $\bar{\phi}_r^t$'s and $\bar{\psi}_c^t$'s are symmetric about the middle indices, i.e. $\bar{\phi}_r^t = \bar{\phi}_{\Lr-r+1}^t$ for $r\leq \lfloor \frac{\Lr}{2} \rfloor$ and $\bar{\psi}_c^t = \bar{\psi}_{\Lc-c+1}^t$ for $c\leq \lfloor \frac{\Lc}{2} \rfloor$.

Lemma \ref{lemma:se_asymp_wLbasematrix} gives insight into the decoding progression for a large SC-SPARC defined using an  $(\omega,\Lambda)$ base matrix. On initialization ($t=0$),  the value of $\bar{\phi}_r^0$  for each $r$ depends on the number of non-zero entries in row $r$ of $W$, which is equal to $\overline{c}_r - \underline{c}_r + 1$, with $\overline{c}_r, \underline{c}_r$ given by \eqref{eq:c_r}. Therefore, $\bar{\phi}_r^0$ increases from  $r=1$ until $r=\omega$, is constant for $\omega\leq r\leq \Lambda$, and then starts decreasing again after $r = \Lambda$. As a result,  $\bar{\psi}_c^{1}$ is smallest for $c$ at either end of the base matrix ($c\in\{1,\Lambda\}$) and increases as $c$ moves towards the middle, since the $\sum_{r=c}^{c+\omega-1} (\bar{\phi}_r^0)^{-1}$ term in \eqref{eq:se_asmyp_flat_psi} is largest  for $c\in\{1,\Lambda\}$, followed by $c\in\{2, \Lambda-1\}$, and so on. Therefore, we expect the blocks of the message vector corresponding to column index $c\in\{1,\Lambda\}$ to be decoded most easily, followed by $c\in\{2, \Lambda-1\}$, and so on. Fig. \ref{fig:ser_waveprop} shows that this is indeed the case. 

The decoding progression for subsequent iterations shown  in Fig. \ref{fig:ser_waveprop} can be explained using Lemma \ref{lemma:se_asymp_wLbasematrix} by tracking the evolution of the $\bar{\phi}_r^t$'s and $\bar{\psi}_c^t$'s. In particular, one finds that if column $c^*$ decodes in iteration $t$, i.e. $\bar{\psi}_{c^*}^t=0$, then columns within a coupling width away (i.e., columns $c\in\{c^*-(\omega-1), \ldots, c^*+ (\omega-1)\}$) will become easier to decode in iteration $(t+1)$.


In the following, with  a slight abuse of terminology, we will use the phrase ``column $c$ is decoded in iteration $t$''  to mean $\bar{\psi}_c^t=0$.
\begin{proposition}
\label{prop:decoding_LSL} \cite{hsiehISIT18}
Consider an SC-SPARC constructed using an $(\omega,\Lambda)$ base matrix with rate $R< \frac{1}{2\kappa}\ln(1+\kappa\cdot\text{snr})$, where $\kappa = \frac{\Lambda+\omega-1}{\Lambda}$. (Note that $\frac{1}{2\kappa}\ln(1+\kappa\cdot\text{snr}) \in [ {\mc{C}}/{\kappa},   \mc{C}]$.) Then, according to the asymptotic state evolution equations in Lemma \ref{lemma:se_asymp_wLbasematrix},  the following statements hold in the large system limit:
\begin{enumerate}
\item The AMP decoder will be able to start decoding if
\be\label{eq:omega_thresh}
\omega > \left( \frac{1}{e^{2R\kappa}-1} - \frac{1}{\kappa\cdot\text{snr}} \right)^{-1}.
\ee
\item If \eqref{eq:omega_thresh} is satisfied, then the  sections in the first and last $c^*$ blocks of the message vector will be decoded in the first iteration (i.e. $\bar{\psi}_c^1=0$ for $c\in\{1,2,\ldots,c^*\}\cup \{\Lambda-c^*+1,\Lambda-c^*+2,\ldots,\Lambda\}$), where $c^*$ is bounded from below as
\begin{align}
 c^* \geq  & \min \Bigg\{  (\omega -1),   \,
\left\lfloor\omega\cdot\frac{1 + \kappa\cdot\text{snr}}{(\kappa\cdot\text{snr})^2} \cdot 
 \left[ \ln\left(1 + \kappa\cdot\text{snr}\right) - 2R\kappa \right]\right\rfloor  \Bigg\}.\label{eq:c_star}
\end{align}

\item At least $2c^*$ additional columns will decode in each subsequent iteration until the message is fully decoded. Therefore, the AMP decoder will fully decode in at most $\left\lceil \frac{\Lambda}{2c^*} \right\rceil$ iterations.
\end{enumerate}
\end{proposition}

\begin{remark}
The proposition implies that for any rate $R < \mc{C}$, AMP decoding is successful in the large system limit, i.e., $\bar{\psi}_c^T=0$ for all $c \in [\Lambda]$. Indeed, consider a rate $R = \mc{C}/\kappa_0$, for any constant $\kappa_0 >1$. Then choose 
$\omega$ to satisfy \eqref{eq:omega_thresh} (with $\kappa$ replaced by $\kappa_0$), and  $\Lambda$ large enough  that $\kappa= \frac{\Lambda + \omega -1}{\Lambda} \leq \kappa_0$. With this choice of $(\omega, \Lambda)$ and rate $R$,  the conditions of the proposition are satisfied, and hence, all the columns decode in the large system limit. 
\end{remark}
\begin{remark}
The proof of the proposition shows that if  $R < \frac{\text{snr}}{ 2(1 + \kappa\cdot\text{snr})}$, then  $\bar{\psi}_c^1=0$, for all $c \in [\Lambda]$, i.e., the entire codeword decodes in the first iteration.
\end{remark}

\begin{remark}
The state evolution recursion was analyzed for a certain class of spatially coupled SPARCs by  Barbier  et al. \cite{barbISIT16} using the potential method introduced in \cite{yedla2014simple,kumar2014threshold,DonSpatialC13}.  It is shown in \cite{barbISIT16} that the fixed points of the state evolution recursion \eqref{eq:se_phi}--\eqref{eq:se_psi} coincide with the stationary points of a suitably defined potential function.  This is then used to show `threshold saturation'  for spatially coupled SPARCs with AMP decoding, i.e., for all rates $R < \mc{C}$,  state evolution predicts vanishing probability of decoding error  in the limit of large section size. In contrast, Proposition \ref{prop:decoding_LSL} establishes threshold saturation by directly characterizing the decoding progression in the large system limit. 
\end{remark}

\begin{remark}
A non-asymptotic version of Proposition \ref{prop:decoding_LSL}, which describes the decoding progression for large but finite $M$, can be found in \cite[Sec. IV]{rushITW18}.
\end{remark}

\begin{remark}
For a fixed rate $R < \mc{C}$, one can establish a bound similar to Theorem \ref{thm:main_amp_perf} on the probability of excess section error rate of an AMP decoded spatially coupled SPARC. This requires two technical ingredients in addition to Proposition \ref{prop:decoding_LSL}. The first is a 
conditional distribution lemma similar to Lemma \ref{lem:hb_cond}, but tailored to the spatially coupled design matrix. In particular, the conditional distributions of the vectors $h^{t+1}$ and $b^t$  now depend on the column block and row block indices, respectively. These conditional distributions are then used to
establish a concentration result similar to Lemma \ref{lem:main_lem}  which shows that the NMSE in each iteration $\frac{1}{L}\| \beta - \beta^t \|^2$ is tracked with high probability by the state evolution quantity $\frac{1}{\Lc}\sum_{c} \psi^t_c$.  Proposition \ref{prop:decoding_LSL}  guarantees that this  quantity is small after $\left\lceil \frac{\Lambda}{2c^*} \right\rceil$ iterations.  The rigorous performance analysis of AMP for spatially coupled SPARCs using the above ingredients will be detailed in a forthcoming paper. 
\end{remark}

\begin{proof}[Proof of Proposition \ref{prop:decoding_LSL}]
Since the  $\bar{\phi}_r^t$'s and $\bar{\psi}_c^t$'s in \eqref{eq:se_asmyp_flat_phi} and \eqref{eq:se_asmyp_flat_psi} are symmetric about the middle indices, we will only consider decoding the first half of the columns, $c\in  \{1, \ldots, \lfloor \frac{\Lambda+1}{2} \rfloor \}$, and the same arguments will apply to the second half by symmetry.

For column $c$  to decode in iteration 1, i.e., for $\bar{\psi}_c^1=0$, we require the argument of the indicator function in \eqref{eq:se_asmyp_flat_psi} to be satisfied for $t=0$, which corresponds to
\begin{align}\label{eq:se_asmyp_flat_t0}
\begin{split}
F_c := \frac{\kappa\cdot\text{snr}}{\omega} \sum_{r=c}^{c+\omega-1} \frac{1}{1 + \frac{\kappa\cdot \text{snr}}{\omega}\cdot(\overline{c}_r- \underline{c}_r+1)} &> 2R\kappa.
\end{split}
\end{align}

1) Since the $F_c$ is largest for column $c=1$, \eqref{eq:se_asmyp_flat_t0} must be satisfied with $c=1$ for \emph{any} column to start decoding. Moreover, using \eqref{eq:c_r}, we find
\begin{align}
F_1 = \frac{\kappa \cdot \text{snr}}{\omega}\sum_{r=1}^{\omega}  \frac{1}{1+ \frac{\kappa \cdot \text{snr}}{\omega}\cdot r}
& \stackrel{(\text i)}{>} \int_{\frac{\kappa \cdot \text{snr}}{\omega}}^{\frac{\kappa \cdot \text{snr}}{\omega}(\omega + 1)} \frac{1}{1 + x} \, dx \nonumber \\
& = \ln\left(1 + \frac{\kappa\cdot\text{snr}}{1+\kappa\cdot\text{snr}\cdot\frac{1}{\omega}}\right), \label{eq:se_asmyp_flat_t0_c1}
\end{align}
where the inequality (i) is obtained by using left Riemann sums on the  decreasing function $\frac{1}{1+x}$. Using \eqref{eq:se_asmyp_flat_t0_c1} in \eqref{eq:se_asmyp_flat_t0}, we conclude that if $\ln\left(1 + \frac{\kappa\cdot\text{snr}}{1+\kappa\cdot\text{snr}/\omega}\right) > 2R\kappa$, then column $c=1$ will decode in the first iteration. 
Rearranging this inequality yields \eqref{eq:omega_thresh}.

2) Given an $(\omega, \Lambda)$ pair that satisfies \eqref{eq:omega_thresh}, we can find a lower bound on the total number of columns that decode in the first iteration. In order to decode column $c$ (and column $\Lambda - c + 1$ by symmetry) in the first iteration, we require \eqref{eq:se_asmyp_flat_t0} to be satisfied. For $c < \omega$, this condition corresponds to
\be\label{eq:se_asmyp_flat_t0_c}
F_{c} = \frac{\kappa\cdot\text{snr}}{\omega} \left[ \left(\sum_{r=c}^{\omega-1} \frac{1}{1 + \frac{\kappa\cdot \text{snr}}{\omega}\cdot r} \right)+ \frac{c}{1 + \kappa\cdot \text{snr}} \right] > 2R\kappa,
\ee
and for columns $c\in\{\omega, \ldots, \Lambda - \omega +1\}$, the condition in \eqref{eq:se_asmyp_flat_t0} becomes
\be\label{eq:se_asmyp_flat_t0_omega}
\frac{\text{snr}}{1 + \kappa\cdot\text{snr}} > 2R,
\ee
where \eqref{eq:c_r} was used to find the values of $\underline{c}_r$ and $\overline{c}_r$ . Since $F_c$  defined in \eqref{eq:se_asmyp_flat_t0} is smallest for columns $c\in\{\omega, \ldots, \Lambda - \omega + 1\}$, all columns decode in the first iteration if \eqref{eq:se_asmyp_flat_t0_omega} is satisfied.

For columns $c < \omega$, we can obtain a lower bound for $F_{c}$:
\begin{align}
F_{c}&=\frac{\kappa\cdot\text{snr}}{\omega} \left[ \left(\sum_{r=c}^{\omega-1} \frac{1}{1 + \frac{\kappa\cdot \text{snr}}{\omega}\cdot r} \right)+ \frac{c}{1 + \kappa\cdot \text{snr}} \right] \nonumber \\
&\stackrel{(\text i)}{>} \int_{\frac{\kappa\cdot\text{snr}}{\omega} c }^{ \frac{\kappa\cdot\text{snr} }{\omega} \omega} \frac{1}{1+ x} \, dx + \frac{c}{\omega}\frac{\kappa\cdot\text{snr}}{(1 + \kappa\cdot\text{snr})} \nonumber \\
&= \ln\left(1 + \kappa\cdot\text{snr}\right)  - \ln\left(1+\kappa\cdot\text{snr}\cdot\frac{c}{\omega} \right) + \frac{c}{\omega}\frac{\kappa\cdot\text{snr}}{(1 + \kappa\cdot\text{snr})} \nonumber \\
&\stackrel{(\text{ii})}{>} \ln\left(1 + \kappa\cdot\text{snr}\right)  - \kappa\cdot\text{snr}\cdot\frac{c}{\omega} + \frac{c}{\omega}\frac{\kappa\cdot\text{snr}}{(1 + \kappa\cdot\text{snr})} \nonumber \\
&= \ln\left(1 + \kappa\cdot\text{snr}\right) - \frac{c}{\omega}\frac{(\kappa\cdot\text{snr})^2}{(1 + \kappa\cdot\text{snr})},
\label{eq:se_asmyp_flat_t0_c_LB} %
\end{align}
where (i) is obtained by using left Riemann sums on the decreasing function $\frac{1}{1+x}$, and (ii) from $\ln (x) \leq x-1$. Therefore, if the RHS of \eqref{eq:se_asmyp_flat_t0_c_LB} is greater than $2R\kappa$ then \eqref{eq:se_asmyp_flat_t0_c} is satisfied, and column $c$ will decode in the first iteration. This inequality corresponds to
\begin{equation}\label{eq:c_condition}
c < \omega\cdot\frac{1 + \kappa\cdot\text{snr}}{(\kappa\cdot\text{snr})^2} \cdot \left[ \ln\left(1 + \kappa\cdot\text{snr}\right) - 2R\kappa \right].
\end{equation}
In other words, all columns $c < \omega$ that also satisfy \eqref{eq:c_condition} will decode in the first iteration. Therefore, the number of columns (in the first half) that decode in the first iteration, denoted $c^*$, can be bounded from below by \eqref{eq:c_star}.

3) 
We want to prove that if the first (and last) $c^*$ columns decode in the first iteration, then at least the first (and last) $tc^*$ columns will decode by iteration $t$, for $t \geq 1$. We look at the $c^*<\omega$ case because all columns would have been decoded in the first iteration if $c^* \geq \omega$. We again only consider the first half of the columns (and rows) due to symmetry.

We prove by induction. The $t=1$ case holds by the previous statement that the first $c^*$ columns decode in the first iteration. From \eqref{eq:se_asmyp_flat_t0_c}, this corresponds to the following inequality being satisfied: 
\be\label{eq:se_asmyp_flat_t0_c_star}
\frac{\text{snr}}{\omega} \left[ \left(\sum_{r=c^*}^{\omega-1} \frac{1}{1 + \frac{\kappa\cdot \text{snr}}{\omega}\cdot r} \right)+ \frac{c^*}{1 + \kappa\cdot \text{snr}} \right] > 2R.
\ee

Assume that the statement holds for some $t \geq 1$, i.e. $\bar{\psi}^{t}_c = 0$ for $c\in[tc^*]$. We assume that 
$tc^* < \lfloor \frac{\Lambda +1}{2} \rfloor$, otherwise all the columns will have already been decoded.
Then, from \eqref{eq:se_asmyp_flat_phi}, we obtain
\begin{align}\label{eq:W_flat_phi_tT}
 \bar{\phi}^t_r 
& 	\leq \begin{cases}  
  		\sigma^2, & 1 \leq r\leq tc^*, \\
        		\sigma^2\left( 1 +  \frac{\kappa\cdot\text{snr}}{\omega}(r-tc^*)\right), &  tc^* < r < tc^*+\omega, \\
        		\sigma^2\left( 1 +  \kappa\cdot\text{snr}\right),  &  tc^*+\omega \leq r \leq \lfloor \frac{\Lambda+\omega-1}{2} \rfloor + 1.
	\end{cases}
\end{align}
(We have a $\leq$ sign in \eqref{eq:W_flat_phi_tT} rather than an equality because indices $r$ near  $\frac{\Lambda+\omega-1}{2}$ may have smaller values in the final iterations, due to  columns from the other half and within $\omega$ indices away having already been decoded.)

We now show that the statement holds for $(t+1)$, i.e., $\psi_c^{t+1} =0$ for columns $c\in [(t+1)c^*]$. In order for columns $c\in\{tc^*+1, \ldots, (t+1)c^*\}$ to decode in iteration $(t+1)$, the inequality in the indicator function in \eqref{eq:se_asmyp_flat_psi} must be satisfied when $c=(t+1)c^*$ (the LHS of the inequality is larger for $c\in\{tc^*+1, \ldots, (t+1)c^*-1\}$). This corresponds to
\be\label{eq:se_asmyp_flat_t0_c_starT}
\frac{\text{snr}}{\omega} \left[\left(\sum_{r=(t+1)c^*}^{tc^*+\omega-1} \frac{1}{1 +  \frac{\kappa\cdot\text{snr}}{\omega}(r-tc^*)} \right) + \frac{c^*}{1 +  \kappa\cdot\text{snr}}\right] > 2R,
\ee
which is equivalent to \eqref{eq:se_asmyp_flat_t0_c_star}, noting that $(t+1)c^* < tc^* +\omega$ since $c^*<\omega$. Therefore, \eqref{eq:se_asmyp_flat_t0_c_starT} holds by the condition, and the statement holds for $(t+1)$. Due to symmetry, the same arguments can be applied to the last $tc^*$ and $(t+1)c^*$ columns. Therefore, at least $c^*$ columns from each half will decode in every iteration.
\end{proof}

\begin{figure}[t]
\centering
\includegraphics[width=0.55\textwidth]{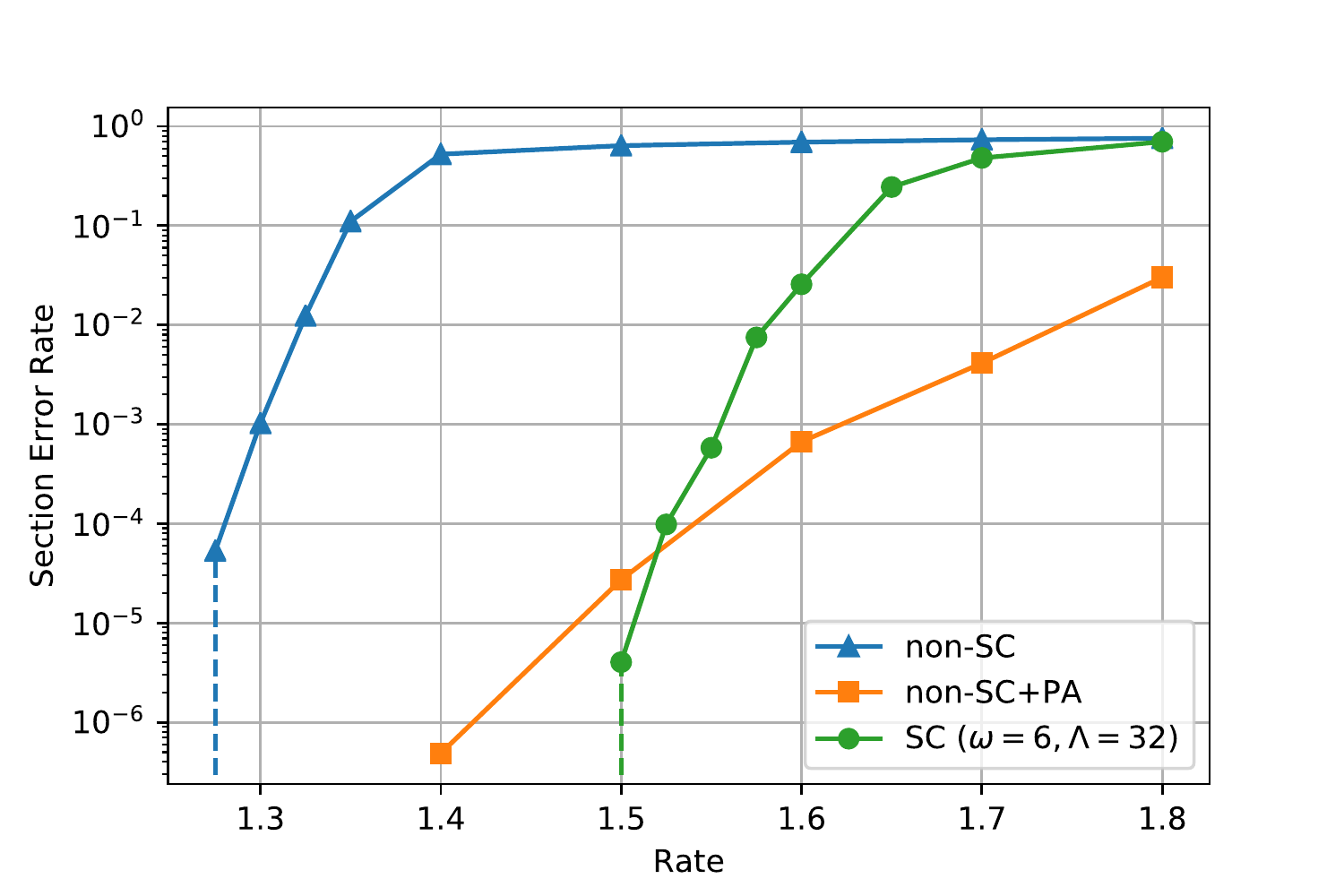}

\includegraphics[width=0.55\textwidth]{cer4sparcs.pdf}
\caption{ \small  Average section error rate (top) and codeword error rate (bottom) vs. rate at $\text{snr}=15$, $\mc{C}=2$ bits. The SPARC parameters are $\M=512$, $L=1024$, $n\in[5100,7700]$.  The graph at the top shows plots for non-SC SPARCs with  and without power allocation, and SC-SPARCs with an $(\omega,\Lambda)$ base matrix with $\omega=6,\Lambda=32$. The code length is the same for the three cases. The dotted vertical lines indicate that no section errors were observed over at least $10^4$ trials at  smaller rates. }
\label{fig:ser4sparcs}
\vspace{-5pt}
\end{figure}

\section{Simulation results} \label{sec:sims}

We evaluate the empirical performance of SC-SPARCs constructed from $(\omega, \Lambda)$ base matrices. As in Chapter \ref{chap:emp_perf}, we use a Hadamard based design matrix (instead of a Gaussian one), which gives  significant reductions in running time and required memory, with very similar error performance to Gaussian design matrices.

\begin{figure}[t]
\centering
\includegraphics[width=3.5in]{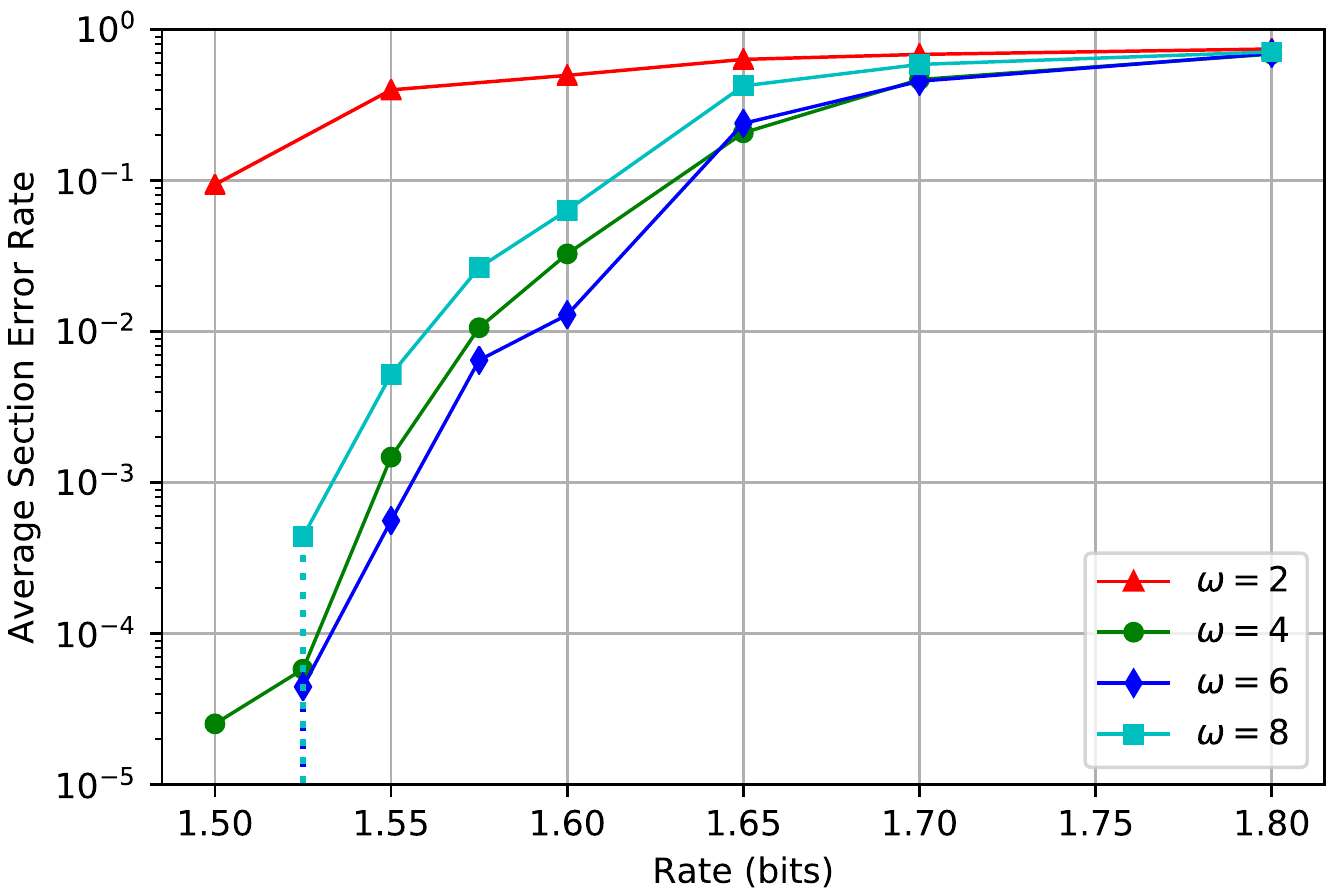}
\caption{\small Average section error rate vs. rate at $\text{snr}=15$, $\mc{C}=2$ bits, $\M=512$, $L=1024$, $n \in [5100, 6200]$. Plots are shown for SC-SPARCs with an $(\omega,\Lambda)$ base matrix with $\Lambda=32$ and $\omega\in\{2,4,6,8\}$. For a given rate, the code length is the same for different $\omega$ values.  The dotted vertical line indicates that for $\omega=6$ and 8, no section errors were observed over $10^4$ trials at $R=1.5$  bits.}
\label{fig:ser_omega}
\vspace{-5pt}
\end{figure}

Figure \ref{fig:ser4sparcs} compares the average section error rate (SER)  and the codeword error rate of spatially coupled SPARCs with  standard (non-SC) SPARCs, both with and without power allocation (PA). The code length is the same for all the codes, and the power allocation was designed using the iterative algorithm described in Section \ref{sec:pa:iterative}.  AMP decoding is used for all the codes. 
Comparing standard SPARCs without PA and SC-SPARCs, we see that spatial coupling significantly improves the error performance: the rate threshold below which the SER drops steeply to a negligible value is higher for SC-SPARCs. We also observe that at rates close to the channel capacity, standard SPARCs with PA have lower SER than SC-SPARCs. However, as the rate decreases, the drop in SER for standard SPARCs with PA is not as steep as that for SC-SPARCs. 

With respect to codeword error rate, we observe that SC-SPARCs significantly outperform non-SC SPARCs with power allocation. This is because power allocated SPARCs tend to have a much larger number of trials with  at least one section error; the number of section errors in such trials is typically small, the errors occur mostly in the sections with low power. 
In contrast, it was observed that the SC-SPARC had many fewer trials with codeword errors, but when a codeword error occurred, it often resulted a large number of sections were in error. 

Next, we examine the effect of changing the coupling width $\omega$. Fig. \ref{fig:ser_omega} compares the average SER of SC-SPARCs with $(\omega,\Lambda)$ base matrices with $\Lambda=32$ and varying $\omega$. For a fixed  $\Lambda$, we observe from \eqref{eq:R_Rsparc} that a larger $\omega$ requires a larger inner SPARC rate $R_\text{inner}$ for the same overall SC-SPARC rate $R$. A larger value of $R_\text{inner}$ makes decoding harder; on the other hand increasing the coupling width $\omega$  helps decoding. Thus  for a given rate $R$, there is a trade-off: as illustrated by Fig. \ref{fig:ser_omega}, increasing $\omega$  improves the SER up to a point, but the performance degrades for larger $\omega$.
In general, $\omega$ should be large enough so that coupling can benefit decoding, but not so large that $R_\text{inner}$ is very close to the channel capacity.  For example, for $R=1.6$ bits and $\Lambda=32$, the inner SPARC rate  $R_\text{inner} =1.65, 1.75, 1.85, 1.95$ bits for $\omega=2,4,6,8$, respectively. With  the capacity  being $\mc{C} =2$ bits,  the figure shows that $\omega=6$  is the best choice for $R=1.6$ bits, with $\omega=8$ being noticeably worse.   This  also indicates that smaller  values $\omega$ would be favored as the rate $R$ gets closer to $\mc{C}$.

%
\part{Lossy Compression with SPARCs}
\chapter{Optimal Encoding}  \label{chap:comp_opt}

In the second part of this monograph, we turn our focus to  SPARCs for lossy compression. Developing practical codes for lossy compression  at rates approaching Shannon's rate-distortion bound has been a long-standing goal in information theory.   A practical compression code requires  a codebook with low storage complexity  as well as encoding and decoding algorithms with low computational complexity.
 The storage complexity of a SPARC is proportional to the size of the size of the design matrix, which is polynomial is the code length $n$.

 In this chapter, we analyze the compression performance of SPARCs with optimal encoding. The performance is measured via the squared error distortion criterion. Though the complexity of the optimal encoder grows exponentially in the  code length, its performance sets a benchmark for efficient SPARC encoders (like the one discussed in the next chapter). 
 
SPARCs were first considered  for lossy compression by Kontoyiannis et al. in \cite{KontSPARC}, where some some preliminary results on their compression performance were presented.  Here we will discuss the analysis in \cite{RVGaussianML} and \cite{RVsparc_ml} which shows  that for i.i.d. Gaussian sources, SPARCs with minimum-distance encoding attain the optimal rate-distortion function and the optimal excess-distortion exponent. 

\section{Problem set-up} \label{sec:comp_setup}

The source sequence is denoted by $s = (s_1, \ldots, s_n)$, and the reconstruction sequence by $\hat{s} = (s_1, \ldots, s_n)$. The distortion is measured by the normalized squared error $\frac{1}{n}\| s -\hat{s} \|^2$.  Throughout this chapter, for any vector $x \in \reals^n$, we will use the notation 
$\abs{x}$ to denote the normalized norm $\| x\|/\sqrt{n}$.

\paragraph{Codebook construction} The sparse regression codebook is as described in Section \ref{sec:sparc_const}. Each codeword is of the form $A\beta$, where the design matrix $A$ has  entries $\sim_{i.i.d.} \normal(0, \frac{1}{n})$. The codeword is determined by the vector $\beta \in \mcb$, which has one non-zero in each section. 

The main difference from the channel coding construction is that the values of the non-zeros in $\beta$ do not have to satisfy a power constraint --- they can be chosen in any way to help the compression encoder. In this chapter, we set all the non-zero values to be equal:
\be
c_1= \ldots = c_L = \sqrt{\frac{n c^2}{L}},
\ee
where the value of $c$ is specified later in \eqref{eq:cval_compr}

As there are $M^L$ codewords, to obtain a compression rate of $R$ nats/sample we need
\be
M^L = e^{nR}.
\label{eq:ml_nR_comp}
\ee
In this chapter, we choose $M=L^b$ for some $b>1$ so that \eqref{eq:ml_nR_comp} implies
\be   L \log L = \frac{nR}{b}. \label{eq:rel_nL} \ee
Thus $L$ is $\Theta\left( n/\log n \right)$, and the number of columns $ML$ in the dictionary $A$ is  
$\Theta\left(\left(n/\log n \right)^{b+1}\right)$, a {polynomial}  in $n$.

\paragraph{Minimum-distance encoder} The optimal encoder for squared-error distortion is the minimum-distance encoder. For the SPARC, it is defined by a mapping $g: \mathbb{R}^n \to \mcb$, which produces the $\beta$ that produces the codeword closest to the source sequence  in Euclidean distance, i.e.,
 \[ \hat{\beta} = g(s) = \underset{\beta \in \mcb}{\operatorname{argmin}} \ \norm{s - A\beta}.\]

\paragraph{Decoder} This is a mapping $h: \mcb \to \mathbb{R}^n$. On receiving $\hat{\beta} \in \mcb$ from the encoder, the decoder produces the reconstruction $h(\hat{\beta}) = A \hat{\beta}$.

\paragraph{Performance measures}
For a rate-distortion code $\mathcal{C}_n$ with code length $n$ and encoder and decoder mappings $g,h$, the probability of excess distortion at distortion level $D$   is
\be P_{e}(\mathcal{C}_n, D) = P\left(\abs{s - h(g(s))}^2 > D\right). \label{eq:pedef} \ee
For a SPARC, the probability measure in \eqref{eq:pedef} is with respect to the random source sequence $s$ and the random design matrix $A$.

\begin{definition}
A rate $R$ is achievable at distortion level $D$  if there exists a sequence of rate $R$ codes $\{\mathcal{C}_n\}_{n=1,2,\ldots}$ such that $\lim_{n \to \infty} P_{e}(\mathcal{C}_n, D) =0$. The infimum of all  rates achievable at distortion level $D$ by any sequence of codes is the Shannon rate-distortion function, denoted by $R^*(D)$. 

A rate $R$ is achievable by SPARCs if there exists a sequence of rate $R$ SPARCs $\{\mathcal{C}_n\}_{n=1,2,\ldots}$, with $\mc{C}_n$ defined by an $n \times L_n M_n$ design matrix whose parameter $L_n$ satisfies \eqref{eq:rel_nL} with a fixed $b$ and $M_n=L_n^b$.
\end{definition}

 For an i.i.d. Gaussian source where $s_1, s_2, \ldots$ are $\sim_{i.i.d.} \normal(0, \sigma^2)$, the Shannon rate-distortion function is \cite{coverT12}
\be
R^*(D) = 
\begin{cases}
\frac{1}{2} \log \frac{\sigma^2}{D} & D < \sigma^2, \\
0 & D \geq \sigma^2.
\end{cases}
\ee

The excess-distortion exponent at distortion-level $D$ of a sequence of rate $R$ codes $\{\mathcal{C}_n\}_{n=1,2,\ldots}$ is given by
\be r(D,R) = - \limsup_{n \to \infty} \frac{1}{n} \log  P_{e}(\mathcal{C}_n,D), \ee
where $P_{e}(\mathcal{C}_n,D)$ is defined in \eqref{eq:pedef}. The optimal excess-distortion exponent for a rate-distortion pair $(R,D)$ is the supremum of the excess-distortion exponents over all sequences of codes with rate  $R$ at distortion-level $D$.

The optimal excess-distortion exponent for discrete memoryless sources was obtained by Marton \cite{MartonRD74}, and for memoryless Gaussian  sources by Ihara and Kubo \cite{IharaKubo00}. 

\begin{theorem}\cite{IharaKubo00}
\label{fact:ihara}
For an i.i.d. Gaussian source distributed as $\mathcal{N}(0, \sigma^2)$ and  squared-error distortion criterion, the optimal excess-distortion exponent at rate $R$ and distortion-level $D$ is
\be
r^*(D,R) = \left\{
\begin{array}{ll}
 \frac{1}{2} \left( \frac{a^2}{\sigma^2} - 1 - \log \frac{a^2}{\sigma^2} \right) & \quad R> R^*(D) \\
 0 & \quad R \leq R^*(D)
\end{array}
\right.
\label{eq:opt_exp}
\ee
where  $a^2 = D e^{2R}$.
\end{theorem}
For $R> R^*(D)$, the exponent in \eqref{eq:opt_exp} is the  Kullback-Leibler divergence between two zero-mean Gaussians, distributed as $\mc{N}(0,a^2)$ and $\mc{N}(0,\sigma^2)$, respectively.

\section{Performance of the optimal decoder}

The key result in this chapter (Theorem \ref{thm:err_exp_rd}) is a large deviations bound on the excess distortion probability of a SPARC. This result is then used to show that  SPARCs attain the optimal rate-distortion function and excess-distortion exponent for i.i.d. Gaussian sources.

For $x >1$, let  
\be
\bmin(x) = \frac{28 R \, x^4}{ \left(1 + \frac{1}{x} \right)^2\left(1 -\frac{1}{x}\right) \left[ -1 + \left( 1 +  
\frac{2 \sqrt{x}}
{(x -1)} \left(R -\frac{1}{2}(1-\frac{1}{x}) \right)\right)^{1/2}\right]^2}
\label{eq:bmin_def}
\ee

\begin{theorem} \cite{RVsparc_ml}
Let the source sequence  $s=(s_1, \ldots,s_n)$ be drawn from an ergodic source with mean zero and variance $\sigma^2$. Let $D \in (0, \sigma^2)$, $R > \frac{1}{2} \log \frac{\sigma^2}{D}$,  and  
$\gamma^2 \in (\sigma^2, De^{2R})$. Let
\be b> \max \left\{ 2, \   \bmin\left( {\gamma^2}/{D} \right) \right\}, \label{eq:bmin_exp} \ee 
where $\bmin(.)$ is defined in \eqref{eq:bmin_def}.
 Let $\mc{C}_n$ be  SPARC  of rate $R$ defined via an  $n  \times L_n M_n$ design matrix with $M_n =L_n^b$ and $L_n$ determined by \eqref{eq:rel_nL}.  Then the probability of excess distortion for $\mc{C}_n$ at distortion level $D$ satisfies 
 \be
P_e(\mc{C}_n, D)  \leq P \left(\frac{\norm{s}^2}{n} \geq \gamma^2 \right) + \exp\left( - \kappa n^{1+c} \right),
\label{eq:ld_bnd}
\ee
where  $\kappa, c$ are strictly positive universal constants.
\label{thm:err_exp_rd}
\end{theorem}

The proof of the theorem is given in Section \ref{sec:proof_err_exp}. 

The first term on the RHS of \eqref{eq:ld_bnd} is the probability that the empirical second moment of the source exceeds $\gamma^2$.  This probability does not depend on the codebook. The second term is a bound on the conditional probability of not finding a SPARC codeword within distortion $D$ given that $\frac{\norm{s}^2}{n} < \gamma^2$. Since the second term decays faster than exponentially in $n$, for large $n$ the excess distortion probability in \eqref{eq:ld_bnd} is dominated by the first term. 

Let us compare the bound in \eqref{eq:ld_bnd} with the excess distortion probability of a Shannon-style random i.i.d. codebook with optimal encoding. The first term remains unchanged as it does not depend on the codebook. The second term, which is the probability of not finding a codeword within distortion $D$ for a source sequence with  $\frac{\norm{s}^2}{n} < \gamma^2$, decays \emph{double exponentially} in $n$ \cite{IharaKubo00} for the random i.i.d. codebook. Though the second term decays much faster for an i.i.d. codebook than for  SPARCs, for large $n$ the excess distortion probability is  still dominated by the first term. We therefore expect the excess-distortion exponent of a SPARC to be the same as that of a random i.i.d. codebook. We also know that a sequence of random i.i.d. codebooks attains the optimal exponent in \eqref{eq:opt_exp}; hence, based on the previous claim a sequence of SPARCs would also attain the optimal exponent. This is made precise in the following corollary.

\begin{corollary}
Let $s$ be drawn from an i.i.d. Gaussian source with mean zero and variance $\sigma^2$. Fix rate $R > \frac{1}{2} \log \frac{\sigma^2}{D}$, and let $a^2=De^{2R}$. Fix any $\e \in (0, a^2 -\sigma^2)$, and 
\be
b > \max \left\{2, \, b_{min}\left(\frac{a^2- \e}{D} \right) \right \}.
\label{eq:b_a2e}
\ee 
There exists a sequence of rate $R$ SPARCs with  parameter $b$ that achieves the excess-distortion exponent 
\[  \frac{1}{2} \left( \frac{a^2 - \e}{\sigma^2} -1 - \log \frac{a^2-\e}{\sigma^2} \right). \]
Consequently:
\begin{enumerate}
\item SPARCs attain the Shannon rate-distortion function  of an i.i.d. Gaussian source.
\item The supremum of excess-distortion exponents  achievable by SPARCs for i.i.d. Gaussian sources sources is equal to the optimal one, given by \eqref{eq:opt_exp}.
\end{enumerate}
\label{corr:err_exp}
\end{corollary}
\begin{proof}
From Theorem \ref{thm:err_exp_rd}, we know that for any $\e \in (0, a^2-\sigma^2)$, there exists a sequence of rate $R$ SPARCs $\{C_n\}$ for which
\be
P_e(\mc{C}_n, D)  \leq P(\abs{s}^2 \geq a^2 - \e) \left(1 + \frac{\exp(-\kappa n^{1+c})}{P(\abs{s}^2 \geq a^2-\e)}\right)
\label{eq:pe_rew}
\ee
for sufficiently large $n$, as long as the  parameter $b$ satisfies \eqref{eq:b_a2e}.
For $s$ that is  i.i.d. $\mc{N}(0, \sigma^2)$, Cram{\'e}r's large deviation theorem \cite{Den2008LD} yields
\be 
\begin{split} 
& \lim_{n \to \infty} - \frac{1}{n} \log P(\abs{s}^2 \geq a^2 - \e ) 
 =  \frac{1}{2} \left( \frac{a^2- \e}{\sigma^2} -1 - \log \frac{a^2-\e}{\sigma^2} \right)
\end{split}
 \label{eq:gaussian_ld} 
\ee
for $(a^2 - \e) > \sigma^2$. Thus $P(\abs{s}^2 \geq a^2- \e)$ decays exponentially with $n$; in comparison $\exp( - \kappa n^{1+c})$ decays \emph{faster} than exponentially with $n$. Therefore, from \eqref{eq:pe_rew},  the excess-distortion exponent satisfies
\be
\begin{split}
& \liminf_{ n \to \infty} \, \frac{-1}{n} \log P_e(\mc{C}_n, D)  \\
& \geq  \liminf_{ n \to \infty}  \frac{-1}{n}\left[ \log P(\abs{s}^2 \geq a^2 - \e)    + \log \left(1 + \frac{\exp(- \kappa n^{1+c})}{P(\abs{s}^2 \geq a^2 - \e)} \right) \right] \\
& = \frac{1}{2} \left( \frac{a^2 - \e}{\sigma^2} -1 - \log \frac{a^2-\e}{\sigma^2} \right).
\end{split}
\ee
Since $\e >0$ can be chosen arbitrarily small, the supremum of all achievable excess-distortion exponents is  $\frac{1}{2} \left( \frac{a^2}{\sigma^2} - 1 - \log \frac{a^2}{\sigma^2} \right)$, which is optimal from Fact \ref{fact:ihara}.
\end{proof}

Theorem \ref{thm:err_exp_rd} and Corollary \ref{corr:err_exp} together show that sparse regression codes are essentially as good as random i.i.d Gaussian codebooks  in terms of rate-distortion function, excess-distortion exponent, and robustness. By robustness, we mean that a SPARC designed to compress an i.i.d Gaussian source with variance $\sigma^2$ to distortion $D$  can compress any ergodic source with variance at most $\sigma^2$ to distortion $D$. This property is also satisfied by random i.i.d Gaussian codebooks \cite{Lapidoth97,SakMismatch1,SakMismatch2}. Moreover, Lapidoth \cite{Lapidoth97} also showed that for any ergodic source,  with an i.i.d. Gaussian random codebook one cannot attain a mean-squared distortion smaller than the distortion-rate function of an i.i.d Gaussian source with the same variance.

To sum up, the sparse regression ensemble has  good covering properties, with the advantage of much smaller codebook storage complexity than the i.i.d random ensemble (polynomial vs. exponential  in block-length).


The remainder of this chapter is devoted to proving Theorem \ref{thm:err_exp_rd}.  The proof involves using the second moment method and Suen's inequality \cite{JansonBook} to show that if $\abs{s}^2 \leq \gamma^2$, then with high probability there exists at least one codeword within distortion $D$ of the source sequence. Proving the result turns out to be significantly easier in the regime where $R > R_0(D)$ where
\be 
R_0(D) := \max \left\{ \frac{1}{2} \log  \frac{\sigma^2}{D}, \, \left(1 - \frac{D}{\sigma^2} \right)\right\}.\label{eq:rsp_def} 
\ee
The  rate $R_0(D)$ in \eqref{eq:rsp_def} is equal to $R^*(D)$ when $\frac{D}{\sigma^2} \leq x^*$, but is strictly larger than $R^*(D)$ when $\frac{D}{\sigma^2} > x^*$, where $x^* \approx 0.203$; see Fig. \ref{fig:rd_ml_perf}. 

The reason for the result being harder to prove for $R \in (R^*(D), R_0(D)]$ is discussed on p.\pageref{eq:x|u1b} after introducing the key elements of the proof. Roughly speaking, at these low rates the probability of the SPARC codebook having an atypically large number of codewords within distortion $D$ of the source sequence is high, and  so a standard application of second moment method fails. %

\begin{figure}[t]
\centering
\includegraphics[width=3.7in]{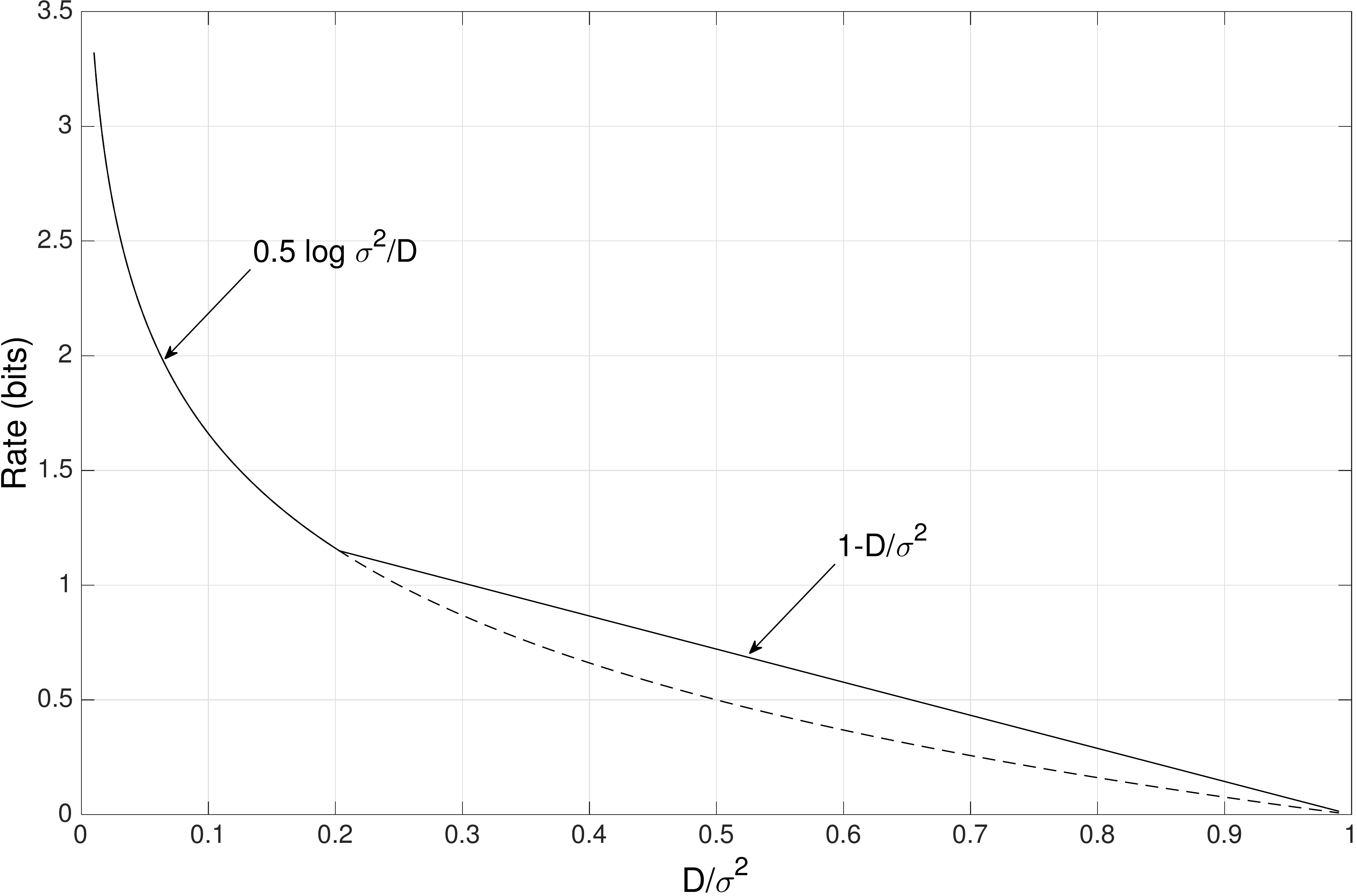}
\caption{\small{The solid line shows the previous achievable rate $R_0(D)$, given  in  \eqref{eq:rsp_def}. The rate-distortion function $R^*(D)$ is shown in dashed lines. It coincides with $R_0(D)$ for $D/\sigma^2 \leq x^*$, where $x^* \approx 0.203$.} }
\label{fig:rd_ml_perf}
\end{figure}

\section{Proof of Theorem \ref{thm:err_exp_rd}} \label{sec:proof_err_exp}

Fix a rate $R > R^*(D)$, and $b$ greater than the minimum value specified by the theorem. Note that $D e^{2R} > \sigma^2$ since $R > \tfrac{1}{2} \log \tfrac{\sigma^2}{D}$. Let $\gamma^2$ be any number such that $\sigma^2 < \gamma^2 < D e^{2R}$.

\emph{Code Construction}:
 Fix block length $n$ and section parameter $b$. Then pick $L$ as specified by \eqref{eq:rel_nL} and $M=L^b$.
Construct an $n \times ML$ design matrix $A$ with entries drawn i.i.d. $\mathcal{N}(0, 1/n)$. The codebook consists of all vectors $A \beta$ such that $\beta \in \mcb$. The non-zero entries of $\beta$ are all set equal to a value specified below.

\emph{Encoding and Decoding}: If the source sequence $\bfs$ is such that $\abs{\bfs}^2 \geq \gamma^2$, then the encoder declares an error. If $\abs{\bfs}^2 \leq D$, then $\bfs$ can be trivially compressed to within distortion $D$ using the all-zero codeword. The addition of this extra codeword to the codebook affects the rate in a negligible way.

 If $\abs{\bfs}^2 \in (D, \gamma^2)$, then $\bfs$ is compressed in two steps. First, quantize $\abs{\bfs}^2$ with an $n$-level uniform scalar quantizer $Q(.)$ with  support in the interval $(D, \gamma^2]$.  Conveying the scalar quantization index to the decoder (with an additional $\log n$ nats)  allows us to adjust the codebook variance according to the norm of the observed source sequence.\footnote{The scalar quantization step is only included to simplify the analysis. In fact, we could use the same codebook variance $(\gamma^2 - D)$ for all $\bfs$ that satisfy 
 $\abs{\bfs}^2 \leq (\gamma^2-D)$, but this would make the forthcoming large deviations analysis quite cumbersome.} The non-zero entries of $\beta$ are each set to $\sqrt{n c^2/L}$, where
 \be 
c^2=  Q(\abs{\bfs}^2)-D.
 \label{eq:cval_compr}
 \ee 
 so that each SPARC codeword has variance $c^2=(Q(\abs{\bfs}^2)-D)$. Define  a `quantized-norm' version of $\bfs$ as
 \be
 \tilde{\bfs} := \sqrt{\frac{Q(\abs{s}^2)}{\abs{s}^2}} \, s.
 \ee
 Note that $\abs{\tilde{\bfs}}^2= Q(\abs{s}^2)$. We  use the SPARC to compress $\tilde{\bfs}$. The encoder finds 
\[ \hat{\beta} := \underset{\beta \in \mcb}{\operatorname{argmin}} \ \norm{\tilde{\bfs} - A\beta}^2. \]
The decoder receives $\hat{\beta}$ and reconstructs $\bfsh = A \hat{\beta}$. Note that  for block length $n$, the total number of bits transmitted by encoder is 
$\log n + L \log M$, yielding an overall rate of $R + \tfrac{\log n}{n}$ nats/sample.

Let $\mc{E}(\tilde{\bfs})$ be the event that the minimum of $\abs{\tilde{\bfs} - A\beta}^2$ over $\beta \in \mcb$ is greater than $D$. The encoder declares an error if  $\mc{E}(\tilde{\bfs})$ occurs. If $\mc{E}(\tilde{\bfs})$  \emph{does not} occur, it can be verified that
the overall distortion  can be bounded as 
\be
\abs{\bfs - A \hat{\beta}}^2 \leq D + \frac{\kappa}{n},
\label{eq:good_dist_bnd}
\ee
 for some positive constant $\kappa$. The overall rate (including that of the scalar quantizer) is $R + \frac{\log n}{n}$.
 
Denoting the probability of excess distortion for this  code by $P_{e,n}$, we have
\be
\begin{split}
P_{e,n}  &  \leq P(\abs{\bfs}^2 \geq \gamma^2) +  \max_{\rho^2 \in (D, \gamma^2)} P(\mc{E}(\tilde{\bfs}) \mid \abs{\tilde{\bfs}}^2 = \rho^2).
\end{split}
\label{eq:err_bound}
\ee

To bound the second term in \eqref{eq:err_bound}, without loss of generality we can assume  that the source sequence 
is  $\tbfs = (\rho, \ldots, \rho)$. This is because the codebook distribution is rotationally invariant, due to the  i.i.d.  $\mc{N}(0,1)$ design matrix 
$A$. For any $\beta$, the entries of $A \beta(i)$ i.i.d. $\mc{N}(0,\rho^2-D)$.
We enumerate the codewords as $A \beta(i)$, where $\beta(i) \in \mcb$ for $i=1,\ldots, e^{nR}$.  

Define the indicator random variables
\be
U_i(\tbfs) = \left\{
\begin{array}{ll}
1 & \text{ if } \abs{A \beta(i) - \tbfs}^2 \leq D,\\
0 & \text{ otherwise}.
\end{array} \right.
\label{eq:ui_def}
\ee
We can then write
\be
P( \mc{E}(\tbfs)  ) = P\left(\sum_{i=1}^{e^{nR}} U_i(\tbfs) =0 \right).
\label{eq:sum_ui}
\ee
For a fixed $\tbfs$,  the $U_i(\tbfs)$'s are dependent.  Indeed,  if ${\beta}(i)$ and ${\beta}(j)$ overlap in $r$ of their non-zero positions, then the column sums forming codewords $\bfsh(i)$ and $\bfsh(j)$ will share $r$ common terms, and consequently $U_i(\tbfs)$ and $U_j(\tbfs)$ will be dependent.

For brevity, we henceforth denote $U_i(\tilde{\bfs})$ by just $U_i$. We also write $X := \sum_{i=1}^{e^{nR}} U_i$. We refer to $\beta_i$ as a \emph{solution} if $U_i=1$. Hence $X$ is the number of solutions.

 To highlight the main ideas in the proof, before obtaining a non-asymptotic bound for the probability in \eqref{eq:sum_ui}, we will first prove the following asymptotic result. 
 \be
P(X > 0) =P\left( \sum_{i=1}^{e^{nR}} U_i >0 \right) \to 1 \text{ as } n \to \infty.
\label{eq:asymp_result}
 \ee
We will first apply the second moment method (second MoM)  to prove \eqref{eq:asymp_result}, and then use Suen's correlation inequality  to  prove the non-asymptotic result in the statement of the theorem.

For any non-negative random variable $X$,  the second MoM bounds the probability of the event $X > 0$ from below as 
\be
P(X > 0) \geq \frac{(\expec X)^2}{\expec [ X^2]}.
\label{eq:2nd_mom}
\ee
The inequality \eqref{eq:2nd_mom} follows from the Cauchy-Schwarz inequality $$ (\expec[X Y])^2 \leq \expec X^2 \, \expec Y^2$$ by substituting $Y=\mathbf{1}_{\{X>0\}}$. To apply it to our setting, we first observe that
\begin{align}
\expec[X^2]  = \expec\Bigg[X \sum_{i=1}^{e^{nR}} U_i \Bigg] = \sum_{i=1}^{e^{nR}} \expec[X U_i] 
&   = \sum_{i=1}^{e^{nR}} P(U_i =1) \expec [X | U_i =1] \nonumber  \\
& = \expec X \cdot  \expec [X | \, U_1 =1]. \label{eq:symm_EX2}
\end{align}

Using \eqref{eq:symm_EX2} in \eqref{eq:2nd_mom}, we obtain 
\be
P(X > 0) \geq \frac{\expec X}{ \expec [X | \, U_1 =1]}.  \label{eq:2mom_cond_expec}
\ee

\subsection{Second moment method computations}  \label{eq:X_asymp}

To compute $\expec X$, we derive a general lemma specifying the probability that a randomly chosen i.i.d $\mc{N}(0, y)$ codeword  is within distortion $z$ of a source sequence $\bfs$ with $\abs{\bfs}^2=x$. This lemma will be used in other parts of the proof as well. 
\begin{lemma}
Let $\bfs$ be a vector with $\abs{\bfs}^2=x$. Let $\bfsh$ be an  i.i.d. $\mc{N}(0, y)$ random vector  that is  independent of $\bfs$. Then for $x,y,z >0$ and sufficiently large $n$, we have
\be 
\frac{\kappa}{\sqrt{n}} e^{-n f(x,y,z)} \leq P \left( \abs{\bfsh - \bfs}^2 \leq  z   \right) \leq  e^{-n f(x,y,z)},
\label{eq:Pfxyz}
\ee
where $\kappa$ is a universal positive constant and for $x,y,z >0$,  the large-deviation rate function $f$ is 
\be
f(x, y, z) = \left\{
\begin{array}{l l}
 \frac{x+z}{2y} - \frac{xz}{A y} - \frac{A}{4y} -\frac{1}{2} \ln\frac{A}{2x} & \text{ if }  z \leq x+y,  \\
 0 & \text{ otherwise}, \\
\end{array}
\right.
\label{eq:fdef}
\ee
and
\be
A = \sqrt{y^2 + 4 x z } - y.
\label{eq:Adef}
\ee
\label{lem:gen_sc}
\end{lemma}

\begin{proof}
We have
\be 
\begin{split}
&P \left( \abs{\bfsh - \bfs}^2 \leq  z  \right) = P\left(\frac{1}{n}\sum_{k=1}^n (\hat{s}_k- s_k)^2 \leq  z  \right)  =  P\left(\frac{1}{n}\sum_{k=1}^n (\hat{s}_k- \sqrt{x} )^2 \leq  z  \right),
\end{split}
\ee
where the last equality is due to the rotational invariance of the distribution of $\bfsh$, i.e.,  $\bfsh$ has the same joint distribution as $O\bfsh$ for any orthogonal (rotation) matrix $O$. In particular,  we choose $O$ to be the matrix that rotates $\bfs$ to the vector $(\sqrt{x}, \ldots, \sqrt{x})$, and note that $ \abs{\bfsh - \bfs}^2 =  \abs{O \bfsh -  O \bfs}^2$. Then, using the strong version of Cram{\'e}r's large deviation theorem due to Bahadur and Rao \cite{Den2008LD, BahadurR60},  we have
\be
 \frac{\kappa}{\sqrt{n}} e^{-nI(x,y,z)}  \leq P\left(\frac{1}{n}\sum_{k=1}^n (\hat{s}_k- x)^2 \leq  z  \right) \leq e^{-nI(x,y,z)},
\ee
where the large-deviation rate function $I$ is given by 
\be
I(x,y,z) = \sup_{\lambda \geq 0 } \left\{ \lambda z  - \log \expec e^{\lambda(\hat{S} -\sqrt{x})^2} \right\}.
\label{eq:Idef}
\ee
The expectation on the RHS of \eqref{eq:Idef} is computed with $\hat{S}  \sim \mc{N}(0,y)$. Using standard calculations, we obtain
\be
 \log \expec e^{\lambda(\hat{S} -\sqrt{x})^2} = \frac{\lambda x}{1-2y\lambda} - \frac{1}{2} \log (1-2y \lambda), \qquad  \lambda <2y.
 \label{eq:logmomgen}
\ee
Substituting the expression in \eqref{eq:logmomgen} in \eqref{eq:Idef} and maximizing over $\lambda \in [0, 2y)$ yields $I(x,y,z) =f(x,y,z)$, where $f$ is given by \eqref{eq:fdef}.
\end{proof}

The expected number of solutions is given by
\be \expec X = e^{nR} P(U_1 =1) =  e^{nR} P\left( \abs{A \beta(1) - \tbfs}^2 \leq D \right).  \label{eq:EX} \ee
Since $\tbfs=(\rho, \rho, \ldots, \rho)$, and $A \beta(1)$ is i.i.d. $\mc{N}(0, \rho^2-D)$,  applying Lemma \ref{lem:gen_sc} we obtain the bounds
\be \frac{\kappa}{\sqrt{n}} e^{nR} e^{-n f(\rho^2, \rho^2-D,D)}  \leq \expec X   \leq e^{nR} e^{-n f(\rho^2, \rho^2-D,D)},  \label{eq:EXub} \ee
Note that \be f(\rho^2, \rho^2-D, D) = \frac{1}{2} \log \frac{\rho^2}{D}. \label{eq:fnice} \ee 

Next consider $ \expec [X | \, U_1 =1]$. If $\beta(i)$ and $\beta(j)$ overlap in $r$ of their non-zero positions, the column sums forming codewords $\bfsh(i)$ and $\bfsh(j)$ will share $r$ common terms. Therefore,
\begin{align}
&  \expec [X | \, U_1 =1]   = \sum_{i=1}^{e^{nR}} P(U_i=1 | \, U_1 =1)  = \sum_{i=1}^{e^{nR}} \frac{P(U_i=1, \,  U_1 =1)}{P(U_1=1)} \nonumber  \\
&  \stackrel{(a)}{=} \sum_{r=0}^{L} {L \choose r} (M-1)^{L-r} \frac{P(U_2=U_1=1 | \, \mc{F}_{12}(r))}{P(U_1=1)}
\label{eq:x|u1a}
\end{align}
where $\mc{F}_{12}(r)$ is the event that the codewords corresponding to ${U}_1$ and ${U}_2$ share $r$ common terms.  In \eqref{eq:x|u1a}, $(a)$ holds because for each codeword $\bfsh(i)$, there are a total of ${L \choose r} (M-1)^{L-r}$  codewords which share exactly $r$ common terms with $\bfsh(i)$, for $0\leq r \leq L$.  

From \eqref{eq:x|u1a} and \eqref{eq:EX}, the key ratio in \eqref{eq:symm_EX2} is
\be
\begin{split}
  \frac{\expec [X | \, U_1 =1]}{\expec X}  
   & = \sum_{r=0}^{L} {L \choose r} (M-1)^{L-r} \frac{P(U_2=U_1=1 | \, \mc{F}_{12}(r))} { e^{nR}\, (P(U_1=1))^2}  \\
 & \stackrel{(a)} \sim  1 +  \sum_{\alpha =  \frac{1}{L}, \ldots, \frac{L}{L}} {L \choose L\alpha} \frac{P(U_2=U_1=1 | \, \mc{F}_{12}(\alpha))}{ M^{L\alpha} \, (P(U_1=1))^2}  \\
& \stackrel{(b)}{=} 1 +  \sum_{\alpha =  \frac{1}{L}, \ldots, \frac{L}{L}} e^{n\Delta_\alpha}
\end{split}
\label{eq:alph_summands}
\ee
where $(a)$ is obtained by substituting $\alpha = \tfrac{r}{L}$ and $e^{nR}=M^L$. The notation $x_L\sim y_L$ means that $x_L/y_L \to 1$  as $L \to \infty$. The equality $(b)$ is from \cite[Appendix A]{RVGaussianML}, where it is shown that
\be
\Delta_{\alpha} \leq \frac{\kappa}{L}+  \frac{R}{b} \min\{ \alpha, \, \bar{\alpha}, \, \tfrac{\log2}{\log L}\}  - h(\alpha)
\label{eq:delalph_bound}
\ee
where $\kappa >0$ is a universal constant, and 
\be
h(\alpha) := \alpha R - \frac{1}{2} \log \left( \frac{1+ \alpha}{1 - \alpha(1-\frac{2D}{\rho^2})} \right).
\ee
The term $e^{n \Delta_\alpha}$ in \eqref{eq:alph_summands} may be interpreted as follows. Conditioned on $U_1=1$, i.e. $\beta(1)$ is a solution, the expected number of solutions that share $\alpha L$ common terms with $\beta(1)$ is $\sim e^{n \Delta_\alpha} \expec X$. Recall that we require the left side of \eqref{eq:alph_summands} to tend to $1$ as $n \to \infty$. Therefore, we need  $\Delta_\alpha <0$ for $\alpha =\tfrac{1}{L}, \ldots, \tfrac{L}{L}$.  From \eqref{eq:delalph_bound}, we need $h(\alpha)$ to be positive in order to guarantee that $\Delta_\alpha <0$. 

It can be shown \cite[Appendix A]{RVGaussianML} that  for $R> (1- D/\rho^2)$, the function $h(\alpha) = \alpha R  - g(\rho^2)$ is strictly positive in the interval $[\tfrac{1}{L}, \tfrac{L-1}{L}]$. Further, for all sufficiently large $L$  its minimum in the interval is attained at $\alpha=1/L$ where it equals
\be h( {1}/{L})  = \frac{1}{L} \left( R - (1- D/\rho^2)\right) + \frac{\kappa}{L^2}, \quad \kappa >0.\ee

Using this bound for $h(\alpha)$ in \eqref{eq:delalph_bound}, we find that if  $R> (1- D/\rho^2)$, and  
\be 
b > \frac{2.5 R + \frac{\kappa}{\log L}}{R - (1- D/\rho^2) + \frac{\kappa}{L}}, 
\ee
then the exponent $\Delta_\alpha$ in \eqref{eq:alph_summands} is strictly negative  for 
$\frac{1}{L} \leq \alpha \leq \frac{(L-1)}{L}$.  Consequently, the key ratio
$\frac{\expec [X | \, U_1 =1]}{\expec X}   \to 1$ as $n \to \infty$.

However, when $\frac{1}{2} \log \frac{\rho^2}{D} <  R < (1-\tfrac{D}{\rho^2})$,  it can be verified that  $h(\alpha)<0$  for $\alpha \in (0, \alpha^*)$ where $\alpha^* \in (0,1)$ is the solution to $h(\alpha)=0$.  Thus $\Delta_\alpha$ is \emph{positive} for $\alpha \in (0, \alpha^*)$ when $\frac{1}{2} \log \frac{\rho^2}{D} < R \leq (1-\tfrac{D}{\rho^2})$. Consequently,  \eqref{eq:alph_summands} implies that
\be
\frac{\expec [X | \, U_1 =1]}{\expec X} \sim  \sum_{\alpha} e^{n\Delta_\alpha} \to \infty \ \text{ as } \ n \to \infty, 
\label{eq:x|u1b}
\ee
so the second MoM fails.  

The reason for the failure of the second MoM in the regime $R < (1-\tfrac{D}{\rho^2})$ is due to the size-biasing induced by  conditioning on $U_1=1$.  Indeed, for any $s$, there are atypical realizations of the design matrix that yield a very large number of solutions.  The total probability of these matrices is small enough that $\expec X$ in not significantly affected by these realizations. However, conditioning on $\beta$ being a solution increases the probability that the realized design matrix is one that yields an unusually large number of solutions. 
At low rates, the conditional probability of the design matrix being atypical is large enough to make $\expec [ X| U_1=1] \gg \expec X$, causing the second MoM to fail.

The key to rectifying the second MoM failure is to show that $X(\beta) \approx \expec X$ {with high probability} \emph{although} $\expec [ X | U_1=1]  \gg \expec X$. We then apply the second MoM to count just the `good' solutions, i.e., solutions $\beta(i)$ for which $X \lvert_{U_i=1} \approx \expec X$. This approach was first used in the  work of Coja-Oghlan and Zdeborov\'{a} \cite{CojaZdeb12} on finding sharp thresholds for  two-coloring of random hypergraphs.

\subsection{Refining the second moment method}

Given that $\beta \in \mcb$ is a solution, for $\alpha =0, \tfrac{1}{L}, \ldots, \tfrac{L}{L}$ define  $X_{\alpha}(\beta)$ to  be the number of solutions that share $\alpha L$ non-zero terms with $\beta$. The \emph{total} number of solutions given that $\beta$ is a solution is
\begin{align}
X(\beta) & = \sum_{\alpha=0, \frac{1}{L}, \ldots, \frac{L}{L}} X_{\alpha}(\beta)
\end{align}
Using this notation,  we have  
\be
\begin{split}
& \frac{\expec [X | \, U_1 =1]}{\expec X} \stackrel{(a)}{=} \frac{\expec[X(\beta)]}{\expec X} \\
&  = \sum_{\alpha = 0, \frac{1}{L}, \ldots, \frac{L}{L}} \frac{\expec[X_{\alpha}(\beta)]}{\expec X} \ \stackrel{(b)}{\sim} \  1+  \sum_{\alpha = \frac{1}{L}, \ldots, \frac{L}{L}}  e^{n\Delta_\alpha},
\end{split}
\label{eq:exbeta}
\ee
where ($a$) holds because  the symmetry of the code construction allows us to condition on a generic $\beta \in \mcb$ being a solution; ($b$) follows from \eqref{eq:alph_summands}.

The key ingredient in the proof is the following lemma, which shows that $X_{\alpha}(\beta)$ is much smaller than $\expec X$ w.h.p  $\forall \alpha \in \{ \frac{1}{L}, \ldots, \frac{L}{L}\}$. In particular, $X_\alpha(\beta) \ll \expec X$ \emph{even} for $\alpha$ for which
\[ \frac{\expec[X_{\alpha}(\beta)]}{\expec X}  \sim e^{n\Delta_\alpha} \to  \infty \ \text{ as } n \to \infty.   \]

\begin{lemma} \cite[Lemma 4]{RVsparc_ml}
Let $R> \tfrac{1}{2}\log \frac{\rho^2}{D}$. If $\beta \in \mc{B}_{M,L}$ is a solution, then for sufficiently large $L$ 
\be 
P\left( X_\alpha(\beta)  \leq   L^{-5/2} \, \expec X, \   \text{ for } \tfrac{1}{L}\leq \alpha \leq  \tfrac{L-1}{L} \right) \geq 1- \eta
\ee
where
\be \eta = L^{-3.5 \left(\frac{b}{b_{min}(\rho^2/D)} - 1 \right)}.  \label{eq:eta_def} \ee
The function $b_{min}(.)$ is defined in \eqref{eq:bmin_def}.
\label{lem:good}
\end{lemma}

We refer the reader to \cite{RVsparc_ml} for the proof.  The probability measure in  Lemma \ref{lem:good} is the conditional distribution on the space of design matrices $A$ given that $\beta$ is a solution.

\begin{definition}
For $\e >0$,  call a solution $\beta$ ``$\e$-good"  if
\be \sum_{\alpha = \frac{1}{L}, \ldots, \frac{L}{L}}X_\alpha (\beta) < \e \, {\expec X}. \label{eq:good_def} \ee
\label{def:good}
\end{definition}
Since we have fixed $\tbfs= (\rho, \ldots, \rho)$, whether a solution $\beta$ is $\e$-good or not is determined by the design matrix. Lemma \ref{lem:good} guarantees that w.h.p. any solution $\beta$ will be $\e$-good, i.e., if $\beta$ is a solution, w.h.p. the design matrix is such that  the number of solutions sharing any common terms with $\beta$ is less $\e \expec[X]$.

The key to proving  the asymptotic result in \eqref{eq:asymp_result} is  to apply the second MoM only to $\e$-good solutions.
 Fix $\epsilon = L^{-1.5}$. For $i=1,\ldots, e^{nR}$, define the indicator random variables
 \be
V_i  = \left\{
\begin{array}{ll}
1 & \text{ if } \abs{A \beta(i) - \tbfs}^2 \leq D \ and  \  \beta(i) \text{ is $\e$-good} ,\\
0 & \text{ otherwise}.
\end{array} \right.
\label{eq:vi_def}
\ee
The number of $\e$-good solutions, denoted by $X_g$, is given by
\be
X_g = V_1 + V_2 + \ldots + V_{e^{nR}}.
\label{eq:xg_count}
\ee
We will apply the second MoM to $X_g$ to show that $P(X_g > 0) \to 1$ as $n \to \infty$.   We have
\be
P(X_g > 0) \geq \frac{(\expec X_g)^2}{\expec [ X_g^2]} \, = \, \frac{\expec X_g}{\expec [ X_g | \, V_1 =1]}
\label{eq:2mom_v}
\ee
where the second equality is obtained by writing  $\expec [ X_g^2] = (\expec X_g) \expec [ X_g | \, V_1 =1]$, similar to \eqref{eq:symm_EX2}.

\begin{lemma}
a) $\expec X_g \geq (1- \eta) \expec X$, where $\eta$ is defined in \eqref{eq:eta_def}.

b) $\expec [ X_g | \, V_1 =1] \leq (1 + L^{-0.5}) \expec X$.
\label{lem:Xg}
\end{lemma}
\begin{proof}
Due to the symmetry of the code construction, we have
\be
\begin{split}
\expec X_g &= e^{nR} P(V_1=1) \stackrel{(a)}{=} e^{nR} P(U_1=1) P(V_1=1|U_1=1) \\
&  = \expec X \cdot P(  \beta(1) \text{ is $\e$-good } \mid \beta(1) \text{ is a solution}  ).
\end{split}
\label{eq:pvu}
\ee
In \eqref{eq:pvu},  $(a)$ follows from the definitions of $V_i$ in \eqref{eq:vi_def} and $U_i$ in \eqref{eq:ui_def}.
Given that $\beta(1)$  is a solution, Lemma \ref{lem:good} shows that
\be
 \sum_{\alpha =  \frac{1}{L}, \ldots, \frac{L}{L}} X_{\alpha}(\beta(1)) < (\expec X )  L^{-1.5}.
\label{eq:b1good}
\ee
with probability at least $1 - \eta$. As $\e= L^{-1.5}$,   $\beta(1)$ is $\e$-good according to Definition \ref{def:good}  if \eqref{eq:b1good} is satisfied. Thus $\expec X_g$ in \eqref{eq:pvu} can be lower bounded as
\be
\expec X_g \geq  (1 - \eta) \expec X. 
\ee
For part (b),  first observe that the total number of solutions $X$ is an upper bound for the number of $\e$-good solutions $X_g$. Therefore
\be
\expec [ X_g | \, V_1 =1] \leq \expec [ X | \, V_1 =1]. 
\label{eq:exg_v1}
\ee
  Given that $\beta(1)$ is an $\e$-good solution, the expected number of solutions can be expressed as
  \be
  \begin{split}
 & \expec [ X | \, V_1 =1]  \\
  & =   \expec[X_0 (\beta(1)) \mid  V_1=1]   + \expec[\sum_{\alpha = \frac{1}{L}, \ldots, \frac{L}{L}} \hspace{-4pt} X_\alpha (\beta(1)) \mid  V_1=1].
  \end{split}
  \label{eq:ex_v1_split}
  \ee
  There are $(M-1)^L$ codewords that share no common terms with $\beta(1)$, and are thus independent of the event $V_1=1$.
  \be
  \begin{split}
 &  \expec [X_0 (\beta(1)) \mid V_1=1]  = \expec[X_0(\beta(1))]  = (M-1)^L \, P(\abs{\tbfs - A \beta}^2 \leq D)  \\
  & \leq M^L \, P(\abs{\tbfs - A \beta}^2 \leq D)  = \expec X.
  \end{split}
  \label{eq:ex0_bound}
  \ee
  Next,  note that conditioned on $\beta(1)$ being an $\e$-good solution (i.e., $V_1=1$),
  \be
  \sum_{\alpha = \frac{1}{L}, \ldots, \frac{L}{L}} \hspace{-4pt} X_\alpha (\beta(1)) < \e \, \expec X
  \label{eq:exalph_bound}
  \ee
  \emph{with certainty}. This follows from the definition of $\e$-good  in \eqref{eq:good_def}.  Using \eqref{eq:ex0_bound} and   \eqref{eq:exalph_bound} in \eqref{eq:ex_v1_split}, we conclude that
  \be
  \expec [ X | \, V_1 =1]  < (1 + \e)\expec X.
  \label{eq:x_v1_bnd}
  \ee
  Combining \eqref{eq:x_v1_bnd} with \eqref{eq:exg_v1} completes the proof of Lemma \ref{lem:Xg}.
  \end{proof}
Using Lemma \ref{lem:Xg} in \eqref{eq:2mom_v}, we obtain
\be
\begin{split}
& P(X_g > 0) \geq  \frac{\expec X_g}{\expec [ X_g | \, V_1 =1]} \geq \frac{(1-\eta)}{1+\e}   =  \frac{1-L^{-3.5 (\frac{b}{b_{min}(\rho^2/D)} - 1)}}{1+L^{-3/2}},
\end{split}
\label{eq:pgx0_lb}
\ee
where the last equality is obtained by using the definition of $\eta$ in \eqref{eq:eta_def} and $\e=L^{-0.5}$. Hence the probability of the existence of at least one good solution tends to $1$ as $L \to \infty$. Therefore $P(X >0)$ in \eqref{eq:asymp_result} also tends to one.

\subsection{A non-asymptotic bound for $P(X=0)$ }
We now prove the result \eqref{eq:ld_bnd}  by obtaining a non-asymptotic bound for $P(X_g=0)$. In contrast to \eqref{eq:pgx0_lb} which proves that $P(X_g=0)$ decays polynomially in $L$, we will use Suen's inequality to show that this probability decays \emph{super-exponentially} in $L$.

We begin with some definitions.
\begin{definition}
[Dependency Graphs  \cite{JansonBook}] Let $\{V_i\}_{i \in \mc{I}}$ be a family of random variables (defined on a common probability space). A dependency graph for $\{V_i\}$ is any graph $\Gamma$ with vertex set $V(\Gamma)= \mc{I}$ whose set of edges satisfies the following property: if $A$ and $B$ are two disjoint subsets of  $\mc{I}$ such that there are no edges with one vertex in $A$ and the other in $B$, then the families $\{V_i\}_{i \in A}$ and $\{V_i\}_{i \in B}$
are independent.
\end{definition}

\begin{remark}
\cite[Example $1.5$, p.11]{JansonBook}
Suppose $\{Y_\alpha\}_{\alpha \in \mc{A}}$ is a family of independent random variables, and each $V_i, i\in \mc{I}$ is a function of the variables
$\{Y_\alpha\}_{\alpha \in A_i}$ for some subset $A_i \subseteq \mc{A}$. Then the graph with vertex set $\mc{I}$ and edge set
$\{ij : A_i \cap A_j \neq \emptyset\}$ is a dependency graph for  $\{U_i\}_{i \in \mc{I}}$.
\label{fact:depgraph_ex}
\end{remark}

In our setting, we fix $\e =L^{-3/2}$, let $V_i$ be the  indicator the random variable defined in \eqref{eq:vi_def}. Note that  $V_i$ is one if and only if  $\beta(i)$ is an $\e$-good solution. The set of codewords that share at least one common term with $\beta(i)$ are the ones that play a role in determining whether $\beta(i)$ is an $\e$-good solution or not. Hence, the graph $\Gamma$ with vertex set $V(\Gamma) = \{1,\ldots,e^{nR} \}$  and edge set $e(\Gamma)$ given by
\ben
\begin{split}  
&  \{ ij:  i\neq j \text{ and the codewords } \beta(i), \beta(j) \\
& \quad  \text { share at least one common term} \} 
 \end{split} \een 
 is a dependency graph for the family
$\{ V_i \}_{i=1}^{e^{nR}}$. 

 For a given codeword $\beta(i)$, there are ${L \choose r} (M-1)^{L-r}$ other codewords that have exactly $r$ terms in common with $\beta(i)$, for $0 \leq r \leq (L-1)$.   Therefore each vertex in the dependency graph for the family $\{V_i\}_{i=1}^{e^{nR}}$ is connected to
\[ \sum_{r=1}^{L-1} {L \choose r} (M-1)^{L-r} = M^L - 1 - (M-1)^L \]
other vertices.

\begin{proposition}[Suen's Inequality \cite{JansonBook}] Let $V_i \sim \text{Bern}(p_i),  i\in \mc{I}$, be a finite family of Bernoulli random variables having a dependency graph $\Gamma$.  Write $i \sim j$ if $ij$ is an edge in $\Gamma$. Define
\ben
\begin{split}
\lambda  = \sum_{i\in \mc{I}} \expec V_i,
 \  \, \Delta = \frac{1}{2} \sum_{i \in \mc{I}} \sum_{j \sim i} \expec(V_i V_j),
\  \, \delta =  \max_{i\in \mc{I}} \sum_{k \sim i}  \expec V_k.
\end{split}
\een
Then
\be P\left(\sum_{i \in \mc{I}} V_i =0\right) \leq \exp\left(-\min \left\{ \frac{\lambda}{2}, \frac{\lambda}{6\delta}, \frac{\lambda^2}{8\Delta}  \right\} \right).
 \label{eq:suens_ineq} \ee
\end{proposition}
We apply Suen's inequality  with the dependency graph specified above for $\{V_i\}_{i=1}^{e^{nR}}$ to compute an upper bound for $P(X_g =0)$, where $X_g= \sum_{i=1}^{e^{nR}} V_i$ is the total number of $\e$-good solutions for $\e =L^{-3/2}$.

\textbf{First Term $\frac{\lambda}{2}$}: We have
\be
\begin{split}
& \lambda = \sum_{i =1}^{e^{nR}} \expec V_i = \expec X_g   \geq   \expec X \,  (1 - \eta).
 \end{split}
\label{eq:lambda_expand}
\ee
where the last inequality follows from Lemma \ref{lem:good}, with $\eta$ defined in \eqref{eq:eta_def}. Using the expression from \eqref{eq:EX} for the expected number of solutions $\expec X$, we have
\be
\lambda \geq (1 - \eta) \frac{\kappa}{\sqrt{n}} e^{n (R - \frac{1}{2} \log \frac{\rho^2}{D})},
\label{eq:lambda_lb1}
\ee
where $\kappa >0$ is a universal  constant. For $b > b_{min}(\rho^2/D)$,  \eqref{eq:eta_def} implies  that $\eta$ approaches $1$ with growing $L$.

\textbf{Second term ${\lambda}/(6\delta)$}: Due to the symmetry of the code construction, we have
\begin{align}
 \delta   =  \max_{i\in \{1, \ldots, e^{nR} \}} \sum_{k \sim i}  P\left(V_k  =1 \right)  
 & = \sum_{k \sim i}  P\left(V_k =1  \right) \quad \forall i \in \{1, \ldots, e^{nR} \} \nonumber \\
%
& = \left(M^L - 1 - (M-1)^L\right) P\left(V_1 =1 \right).
\end{align}
Combining this together with the fact that $\lambda =  M^L \, P(V_1 =1)$,
we obtain
\be
\frac{\lambda}{\delta}  = \frac{M^L}{M^L - 1 - (M-1)^L} = \frac{1}{1- L^{-bL} - (1- L^{-b})^L},
\label{eq:lamb_over_del}
\ee
where the second equality is obtained by substituting $M=L^b$. Using a Taylor series bound for the denominator of \eqref{eq:lamb_over_del} (see \cite[Sec. V]{RVGaussianML} for details) yields the following lower bound for sufficiently large $L$:
\be
\frac{\lambda}{\delta} \geq \frac{L^{b-1}}{2}.
\label{eq:lamb_Del_lb}
\ee

\textbf{Third Term $\lambda^2/(8\Delta)$}: We have
\be
\begin{split}
 \Delta  = {\frac{1}{2} \sum_{i=1}^{M^L} \sum_{j \sim i} \expec\left[V_i V_j \right] }   &  =   {\frac{1}{2} \sum_{i=1}^{M^L} P(V_i=1) \sum_{j \sim i}  P(V_j=1 \mid V_i=1)}  \\
& \stackrel{(a)}{=}  \frac{1}{2} \,  \expec X_g \sum_{j \sim 1}  P(V_j=1 \mid V_1=1)  \\
& =   \frac{1}{2} \,  \expec X_g  \,  \expec\Big [ \sum_{j \sim 1} \mathbf{1}\{V_j=1\}  \mid V_1=1\Big ] \\
&  \stackrel{(b)}{\leq} \frac{1}{2} \,  \expec X_g  \, \expec\Bigg [\sum_{\alpha = \frac{1}{L}, \ldots, \frac{L-1}{L}} \hspace{-4pt} X_{\alpha} (\beta(1)) \mid  V_1=1 \Bigg].
\end{split}
\label{eq:Delta_expand}
\ee
In \eqref{eq:Delta_expand}, $(a)$ holds because of the symmetry of the code construction. The inequality $(b)$ is obtained as follows. The number of $\e$-good solutions that share common terms with $\beta(1)$ is bounded above by the total number of solutions sharing common terms with $\beta(1)$.  The latter quantity can be expressed as the sum of the number of solutions sharing  exactly $\alpha L$ common terms with $\beta(1)$, for $\alpha \in \{\tfrac{1}{L}, \ldots, \tfrac{L-1}{L} \}$.

Conditioned on $V_1=1$, i.e., the event that $\beta(1)$ is a $\e$-good solution, the total number of solutions that share common terms with $\beta(1)$ is bounded by $\e \, \expec X$. Therefore, from \eqref{eq:Delta_expand} we have
\be
\begin{split}
& \Delta \leq   \frac{1}{2} \expec X_g  \, \expec\Bigg [\sum_{\alpha = \frac{1}{L}, \ldots, \frac{L-1}{L}} \hspace{-4pt} X_{\alpha} (\beta(1)) \mid  V_1=1 \Bigg]  \\
& \leq \frac{1}{2} \left(\expec X_g \right) (L^{-3/2}\,  \expec X) \leq \frac{L^{-3/2}}{2} (\expec X)^2,
\end{split}
\label{eq:Delta_ub}
\ee
where we have used $\e = L^{-3/2}$, and the fact that $X_g \leq X$. Combining \eqref{eq:Delta_ub} and \eqref{eq:lambda_expand}, we obtain
\be
\frac{\lambda^2}{8 \Delta} \geq \frac{ (1-\eta)^2 (\expec X)^2}{4 L^{-3/2}  (\expec X)^2} \geq \kappa L^{3/2},
\label{eq:lamb_Del2_lb}
\ee
where $\kappa$ is a strictly positive constant.

\textbf{Applying Suen's inequality}:  Using the lower bounds obtained in \eqref{eq:lambda_lb1},  \eqref{eq:lamb_Del_lb}, and  \eqref{eq:lamb_Del2_lb}  in  \eqref{eq:suens_ineq}, we obtain
\be
\begin{split}
& P\left( \sum_{i=1}^{e^{nR}} V_i \right)  \leq \exp \left( - \kappa \, \min \left\{ e^{n (R - \frac{1}{2} \log \frac{\rho^2}{D} -\frac{\log n}{2 n})}, \, L^{b-1}, \, L^{3/2} \right\} \right),
\end{split}
\label{eq:sum_vi_0}
\ee
where $\kappa$ is a positive constant. Recalling from \eqref{eq:rel_nL} that $L = \Theta( \tfrac{n}{\log n})$ and  $R > \frac{1}{2} \ln \frac{\rho^2}{D}$, we see that for $b >2$,
\be P\left( \sum_{i=1}^{e^{nR}} V_i \right) \leq \exp \left( - \kappa n^{1+c} \right),  \label{eq:sum_vi_bound} \ee  
Note that the condition $b >  b_{min}(\rho^2/D)$  is also  required for $\eta$ in Lemma \ref{lem:good} to go to $0$ with growing $L$.

Using \eqref{eq:sum_vi_bound} in \eqref{eq:err_bound},  we conclude that for any $\gamma^2 \in (\sigma^2, D^{e^2R})$ the probability of excess distortion can be bounded as  
\be
\begin{split}
P_{e,n} & \leq P(\abs{\bfs}^2 \geq \gamma^2) +  \max_{\rho^2 \in (D, \gamma^2)} P(\mc{E}(\tilde{\bfs}) \mid \abs{\tilde{\bfs}}^2 = \rho^2) \\
& \leq P(\abs{\bfs}^2 \geq \gamma^2)  + \exp(-\kappa n^{1+c}),
\end{split}
\label{eq:Pen_exp_bound}
\ee
provided the parameter $b$ satisfies
\be
b > \max_{\rho^2 \in (D, \gamma^2)} \max \left\{  2, \, b_{min}\left( \rho^2/D \right)  \right \}.
\label{eq:b_bound}
\ee
It can be verified from the definition in \eqref{eq:bmin_def} that $b_{min}(x)$ is strictly increasing in $x \in (1, e^{2R})$.  Therefore, the maximum on the RHS of \eqref{eq:b_bound} is  bounded by $\max \left\{  2, \, b_{min}\left( \gamma^2/D \right)  \right \}$. Choosing $b$ to be larger than this value will guarantee that \eqref{eq:Pen_exp_bound} holds. This completes the proof of the theorem.

\chapter{Computationally Efficient Encoding}  \label{chap:comp_eff_enc}
   
   In this chapter, we discuss an efficient SPARC encoder for lossy compression with squared-error distortion. The encoding algorithm is based on successive cancellation: in each iteration, one column from a section of $A$ is chosen to be part of the codeword. The column is chosen based on a test statistic that measures the correlation of each column in the section with a residual  vector.

 For any ergodic source with variance $\sigma^2$, it is shown that the encoding algorithm attains the optimal Gaussian distortion-rate function $D^*(R) = \sigma^2e^{-2R}$, for any rate $R>0$. Furthermore, for any fixed distortion level above $D^*(R)$, the probability of excess distortion decays exponentially in the block length $n$. We note that for finite alphabet memoryless sources, several coding techniques have been proposed to approach the rate-distortion bound with computationally feasible encoding and decoding \cite{Kont99,GuptaVerduWeiss,KontGioran,JalaliWeiss, GuptaVerdu09, WainManeva10,Polarrd,aref2015approaching}.
 
 We first give a heuristic derivation of the encoding algorithm and then state the main result (Theorem  \ref{thm:rd_feasible}), a large deviations bound on the excess distortion probability. We also present numerical results to illustrate the empirical compression performance of the algorithm.

 \emph{Notation}: As in the previous chapter, for any vector $x$ we write $\abs{x}$ to denote $\norm{x}/\sqrt{n}$. We also write $\langle x,y \rangle$ for the Euclidean inner product between vectors $x,y \in \reals^n$.

   \section{Computationally efficient encoding algorithm} \label{sec:alg_desc}

   Consider a source sequence $s \in \reals^n$ generated by an ergodic source with zero mean and variance $\sigma^2$. 
   The SPARC is defined via an $n \times ML$ design matrix $A$ with entries drawn i.i.d. $\mathcal{N}(0, 1/n)$. The codebook consists of all vectors $A \beta$ such that $\beta \in \mcb$. The non-entry of $\beta$ in section $\ell$ is set to 
\be
c_\ell =  \sqrt{ 2 (\ln M) \sigma^2 \left( 1- \frac{2R}{L} \right)^{\ell-1}}, \quad  \ell  \in [L].
\label{eq:ci_def}
\ee

The encoding algorithm is intialized with $r_0 = \bfs$, and consists of $L$ steps, defined as follows.

\emph{Step} $\ell$, $\ell=1, \ldots, L$: Pick
\be  m_\ell =  \underset{j: \ (\ell-1)M < \; j \; \leq \ell M }{\operatorname{argmax}}
\left\langle \sqrt{n} A_j,  \frac{r_{\ell-1}}{\norm{r_{\ell-1}}} \right\rangle. \label{eq:max_stepi} \ee
Set
\begin{equation}
r_{\ell} = r_{\ell-1} - c_\ell A_{m_\ell},
\label{eq:gen_residue}
\end{equation}
where $c_\ell$ is given by \eqref{eq:ci_def}.

The codeword $\hat{\beta}$  has  non-zero values in positions $m_\ell, \ 1 \leq \ell \leq L$. The value of the non-zero in section $\ell$ given by $c_\ell$.

The algorithm chooses  the non-zero locations $\{ m_\ell \}$  in a greedy manner  (section by section) to minimize the norm of the residual $r_\ell$. In the next section, we give a heuristic derivation of the algorithm, which also explains the choice of coefficients in \eqref{eq:ci_def}.

\paragraph{Computational complexity} There are $L$ stages in the algorithm, where each stage involves computing $M$ inner products followed by finding the maximum among them. The complexity therefore scales as $O(nML)$; the number of operations per source sample is $O(ML)$.   If we choose $M=L^b$ for some $b>0$,  then $L=\Theta\left( \frac{n}{\log n}\right)$, and the  per-sample complexity is   $O\left( n/ \log n \right)^{b+1}$.

When we have several source sequences to be encoded in succession, the encoder can have a  pipelined architecture with $L$  modules.  The first module computes the inner product of the source sequence  with each column in the first section of $A$ and determines the maximum; the second module computes the inner product of the first-step residual with each column in the second section of $A$, and so on. Each module has $M$ parallel units, with each unit consisting of a multiplier and an accumulator to compute an inner product in a pipelined fashion. After an initial delay of $L$ source sequences, all the modules work simultaneously. This encoder architecture requires computational space (memory) of the order $nLM$ and has constant computation time per source symbol.

\section{Heuristic derivation of the algorithm} \label{subsec:heur}

We now  present a non-rigorous analysis of the encoding algorithm based on the following observations.
\begin{enumerate}

\item For $1 \leq j \leq ML$,  by standard concentration of measure arguments, $\abs{A_j}^2$ is close to $1$ for large $n$.

\item Similarly, for an ergodic source $\abs{s}^2$ is close  to $\sigma^2$ for large $n$.

\item For random variables $X_1,X_2\ldots, X_M \sim_{i.i.d.} \mc{N}(0,1)$, the maximum $\max\{ X_1,\ldots, X_M\}$  concentrates on $\sqrt{2\ln M}$ for large $M$ \cite{DavidNagaraja}.
\end{enumerate}
The deviations of these quantities from their typical values above are precisely characterized in  the proof of the main result  (Section \ref{sec:thm_feasible_proof}).

We begin with the following lemma about projections of  standard normal vectors. \begin{lemma}
Let $A_1, \ldots, A_N \in \reals^n$ be $N$ mutually independent random vectors with i.i.d. $\mc{N}(0,1/n)$ entries. 
Then, for any unit norm random vector $r \in \reals^n$ which is independent of the collection $\{A_j\}_{j=1}^N$, the inner products
\[ T_j  :=  \left\langle   \sqrt{n} A_j,  r \right\rangle, \quad j=1,\ldots,N \]
are i.i.d. $\mc{N}(0,1)$ random variables that are independent of $r$.
\label{lem:inner_prod}
\end{lemma}
The lemma is a straightforward consequence of the rotational invariance of the distribution of a  standard normal vector. A proof can be found in \cite[Appendix I]{RVGaussFeasible}.

\emph{Step} $1$: Consider the statistic
\be T^{(1)}_{j} \triangleq \left\langle  \sqrt{n} A_j, \frac{r_{0}}{\norm{r_{0}}} \right\rangle,  \quad  1 \leq  j  \leq M. \ee
Since  $r_0 = s$ is independent of each $A_j$,  by Lemma \ref{lem:inner_prod},  the random variables $T^{(1)}_{j}, \, 1 \leq j \leq M$ are i.i.d. $N(0,1)$. Hence
\be
\max_{1 \leq j \leq M} T^{(1)}_{j} = \left\langle  \sqrt{n} A_{m_1}, \frac{r_{0}}{\norm{r_{0}}}  \right\rangle \approx \sqrt{2 \log M}.
\label{eq:max_T1j}
\ee
The normalized norm of the residual $r_1 = r_0 - c_1 A_{m_1}$ is
\be
\begin{split}
\abs{r_1}^2  & = \abs{r_0}^2 + \frac{c_1^2}{n} \abs{A_{m_1}}^2  - \frac{2 c_1 \norm{r_0}}{n}
\left\langle A_{m_1}, \frac{r_0}{\norm{r_0}} \right\rangle\\
& \stackrel{(a)}{\approx} \abs{r_0}^2  +  \frac{c_1^2}{n}- \frac{2 c_1}{n} \frac{\norm{r_0}}{\sqrt{n}} \sqrt{2\log M}  \\
& \stackrel{(b)}{\approx} \sigma^2 + \frac{c_1^2}{n}  - \frac{2 c_1 \sigma}{n}\sqrt{2\log M}  \stackrel{(c)}{=} \sigma^2\left( 1 - \frac{2R}{L}\right).
\end{split}
\label{eq:R0_R1}
\ee
Here $(a)$ and $(b)$ follow from \eqref{eq:max_T1j} and the observations listed at the beginning of this section, while $(c)$ follows by substituting for $c_1$ from \eqref{eq:ci_def} and using $n= L \log M/R$.

\emph{Step $\ell$, $\ell=2,\ldots,L$}: We show  that if
 $\abs{r_{\ell-1}}^2 \approx \sigma^2\left( 1 - \frac{2R}{L}\right)^{\ell-1}$,
then
\be \abs{r_{\ell}}^2 \approx \sigma^2\left( 1 - \frac{2R}{L}\right)^{\ell}.  \label{eq:Res_i}\ee
We already showed that \eqref{eq:Res_i} is true for $\ell=1$.

For each $j \in \{ (\ell-1)M+1, \ldots, \ell M \}$, consider the statistic
\be T^{(\ell)}_{j} \triangleq \left\langle \sqrt{n} A_j,  \frac{r_{\ell-1}}{\norm{r_{\ell-1}}} \right\rangle.   \label{eq:stat_Tij}\ee
Note that $r_{\ell-1}$ is {independent} of $A_j$ because $r_{\ell-1}$ is a \emph{function} of the source sequence $s$ and the columns $\{A_j\}, \ 1\leq j \leq (\ell-1)M$, which are all independent of  $A_j$ for $j \in \{ (\ell-1)M+1, \ldots, \ell M \}$.  Therefore, by Lemma \ref{lem:inner_prod}, the $T^{(\ell)}_{j}$'s are i.i.d. $\mc{N}(0,1)$ random variables for $j \in \{ (\ell-1)M+1, \ldots, \ell M\}$. Hence, we have
\be
\max_{(\ell-1)M+1 \, \leq j \, \leq \ell M} \; T^{(\ell)}_{j} = \left\langle  \sqrt{n} A_{m_\ell}, \frac{r_{\ell-1}}{\norm{r_{\ell-1}}} \right\rangle \approx \sqrt{2 \log M}.
\label{eq:max_Tij}
\ee
From the expression for $r_\ell$ in \eqref{eq:gen_residue}, we have
\be
\begin{split}
&\abs{r_\ell}^2   = \abs{r_{\ell-1}}^2 + \frac{c_\ell^2}{n} \abs{A_{m_\ell}}^2  - \frac{2 c_\ell}{n} \frac{\norm{r_{\ell-1}}}{\sqrt{n}}
\left\langle \sqrt{n} A_{m_\ell}, \frac{r_{\ell-1}}{\norm{r_{\ell-1}}} \right\rangle\\
& \stackrel{(a)}{\approx} \abs{r_{\ell-1}}^2  + \frac{c_\ell^2}{n}- \frac{2 c_\ell \abs{r_{\ell-1}} }{n}\sqrt{2\log M} \\
& \stackrel{(b)}{\approx} \sigma^2\left( 1 - \frac{2R}{L}\right)^{\ell-1} + \frac{c_\ell^2}{n} - \frac{2 c_\ell \sigma}{n}  \left( 1 - \frac{2R}{L}\right)^{(\ell-1)/2} \sqrt{2\log M}   \\
&   \stackrel{(c)}{=} \sigma^2\left( 1 - \frac{2R}{L}\right)^\ell.
\end{split}
\label{eq:Ri1_Ri}
\ee
For $(a)$ and $(b)$ we have used \eqref{eq:max_Tij} and the induction assumption on $\abs{r_{\ell-1}}$. The equality $(c)$ is obtained by substituting for $c_\ell$ from \eqref{eq:ci_def} and for $n$ from \eqref{eq:ml_nR}. It can be verified that the chosen value of $c_\ell$ minimizes the third line in \eqref{eq:Ri1_Ri}.

Therefore, when the algorithm terminates the final residual satisfies 
\be
\begin{split}
\abs{r_L}^2 = \abs{\bfs - A\hat{\beta}}^2 & \approx  \sigma^2\left( 1 - \frac{2R}{L}\right)^L  {\leq} \ \sigma^2 e^{-2R}
\end{split}
\ee
where we have used the inequality $(1+x) \leq e^{x}$ for  $x \in \mathbb{R}$.

Thus the encoding algorithm picks a codeword $\hat{\beta}$ that yields squared-error distortion approximately equal to $\sigma^2 e^{-2R}$, the Gaussian distortion-rate function at rate $R$. The heuristic analysis above is made rigorous (in the proof of Theorem \ref{thm:rd_feasible}) by bounding the deviation of the residual distortion each stage from its typical value.

\section{Main result} \label{sec:main_result}

\begin{theorem}
Consider a length $n$  source sequence $\bfs$ generated by an ergodic source with mean $0$ and variance $\sigma^2$.
Let $\delta_0, \delta_1, \delta_2$ be any positive constants such that
\be \Delta \triangleq \delta_0 + 5R (\delta_1 + \delta_2) < \frac{1}{2}. \label{eq:del0del1del2}\ee
Let $A$ be an $n \times ML$ design matrix  with i.i.d. $\mc{N}(0,1/n)$ entries and $M,L$ satisfying \eqref{eq:ml_nR}. With the SPARC defined by $A$, the proposed encoding algorithm produces a codeword $A\hat{\beta}$ that satisfies the following for  sufficiently large $M,L$.
\be P\left( \; \abs{\bfs - A\hat{\beta}}^2 \;  >  \; \sigma^2 e^{-2R}(1 + e^R \Delta)^2  \; \right) < p_0+ p_1 + p_2  \label{eq:p_err_ub} \ee
where
\begin{equation}
\begin{split}
p_0  & = P\left( \;  \left| \frac{\abs{\bfs}}{\sigma} -1 \right| > \delta_0 \; \right),  \quad  
p_1  = 2ML \exp\left( -n{\delta_1^2}/{8}\right), \\
p_2 & = \left( \frac{8 \log M}{M^{2\delta_2}} \right)^{L}.
\end{split}
\label{eq:p0p1p2}
\end{equation}
\label{thm:rd_feasible}
\end{theorem}
\begin{remark}
For a given rate $R$, Theorem \ref{thm:rd_feasible} guarantees that with high probability, the  proposed  encoder   achieves distortion close to $D^*(R)= \sigma^2 e^{-2R}$ for any ergodic sources with variance $\sigma^2$. This complements the result in Theorem \ref{thm:err_exp_rd} for minimum-distance encoding.
\end{remark}
\begin{corollary}
Let  $\{ \mc{S}_n\}_{n \geq 1}$ be a sequence of rate $R$ SPARCs, indexed by block length $n$, with  $M = L^b$, for $b >0$. Then, for an  i.i.d. $\mc{N}(0, \sigma^2)$ source, the sequence $\{ \mc{S}_n\}_{n \geq 1}$ attains the optimal distortion-rate function$D^*(R) = \sigma^2 e^{-2R}$ with the proposed encoder. Furthermore, for any fixed distortion-level above $D^*(R)$, the probability of excess distortion decays exponentially with the block length $n$ for sufficiently large $n$.
\label{corr:Gauss}
\end{corollary}
\begin{proof}
For a fixed distortion-level $\sigma^2 e^{-2R} + \gamma$ with $\gamma >0$, we can find  $\Delta >0$ such that
$\sigma^2 e^{-2R} + \gamma = \sigma^2 e^{-2R}(1 + e^R \Delta)^2$. Equivalently, $\Delta >0$ satisfies
\be
 \gamma = \sigma^2 \Delta^2 + 2 \Delta e^R \sigma^2.
\label{eq:gam_Delta}
\ee
Without loss of generality, we may assume that $\gamma$ is small enough that  $\Delta$ satisfying \eqref{eq:gam_Delta} lies in the interval $(0, \tfrac{1}{2})$.   For positive constants $\delta_0, \delta_1, \delta_2$ chosen to satisfy \eqref{eq:del0del1del2}, Theorem \ref{thm:rd_feasible} implies that
\be
P\left( \abs{\bfs - A\hat{\beta}}^2 \;  >  \; \sigma^2 e^{-2R} + \gamma \right) < p_0 + p_1 + p_2.
\label{eq:excess_gam}
\ee
 We now obtain upper bounds for $p_0, p_1, p_2$.
 
 For an i.i.d. $\mc{N}(0, \sigma^2)$ source, $\|S\|^2/\sigma^2$ is a $\chi^2_n$ random variable. A standard Chernoff bound yields 
 \be p_0  <  2 \exp(- 3n \delta_0^2/4). \label{eq:p0_bound} \ee
Since $ML=L^{b+1}$  grows polynomially in $n$, the term $p_1$ in \eqref{eq:p0p1p2} can be expressed as 
\be
p_1 = \exp\left(-n\left(\tfrac{\delta_1^2}{8} - O( \tfrac{\log n}{n}) \right)\right).
 \label{eq:p1_bound}
\ee
Finally using $M=L^b= \Theta((n/\log n)^{b})$, we have
\be
 p_2 = \exp\left( -n \left( 2 \delta_2 R - O\left( \tfrac{ \log \log n}{\log n} \right)\right) \right).
 \label{eq:p2_bound}
\ee
 Using \eqref{eq:p0_bound}, \eqref{eq:p1_bound} and \eqref{eq:p2_bound} in \eqref{eq:excess_gam}, we conclude that for any fixed distortion-level $D^*(R) + \gamma$, the probability of excess distortion decays exponentially in $n$ when $n$ is sufficiently large.
\end{proof}

\subsection{Gap from $D^*(R)$}  \label{sec:gapDR}

For a fixed $R$, to achieve distortions close to the optimal distortion-rate function $D^*(R)=\sigma^2 e^{-2R}$, we need $p_0,p_1,p_2$ to all go to $0$. Ergodicity of the source ensures that that $p_0 \to 0$ as $n \to \infty$ (at a rate depending only on the source distribution). For $p_2$ to tend to  $0$ with growing $L$, from \eqref{eq:p0p1p2}  we require that
$ {M^{2\delta_2}} > {8 \log M}$.
Or,
\be
\delta_2 > \frac{\log \log M}{2 \log M} + \frac{\log 8}{2 \log M}.
\label{eq:delta_min}
\ee
To approach $D^*(R)$, we need $n, L, M$ to all go to $\infty$ while satisfying \eqref{eq:ml_nR}: $n,L$ need to be large
for the probability of error in  \eqref{eq:p0p1p2} to be small, while $M$ needs to be large in order to allow $\delta_2$ to be small according to \eqref{eq:delta_min}.

When $M=L^b$,  both $L, M$ grow polynomially in $n$, and \eqref{eq:delta_min} implies that the gap from the optimal distortion $D^*(R)$ is  $\Theta\left( \frac{\log \log n}{\log n} \right)$. On the other hand, if we choose $M = \kappa \log n$ for $\kappa >0$, we have $L = \frac{nR}{\log (\kappa \log n)}$. In this case,  the gap $\delta_2$ from  \eqref{eq:delta_min} is approximately $\frac{ \log \log  \log n}{ \log \log n}$, i.e., the convergence to $D^*(R)$ with $n$ is much slower. However, the per-sample computational complexity is $\Theta \left( \frac{n \log n}{\log \log n}\right)$,  lower than the previous case, where the per-sample complexity was $\Theta\left(({n}/{\log n})^{b+1}\right)$.

At the other extreme,  $L=1, M=e^{nR}$ reduces to the Shannon-style random codebook with. In this case, the SPARC consists of only one section and the proposed algorithm is essentially minimum-distance encoding. The computational complexity is $O(e^{nR})$, while the gap $\delta_2$ from \eqref{eq:delta_min} is approximately $\frac{\log n}{n}$. The gap $\Delta$ from $D^*(R)$ is now dominated by $\delta_0$ and $\delta_1$ which are $\Theta(1/\sqrt{n})$, consistent with the results in \cite{SakrisonFin, IngberKochman, KostinaV12}.\footnote{ For $L=1$,  the factor $ML$ that multiplies the exponential term in $p_2$ can be eliminated via a sharper analysis.}

An interesting direction for future work is to design encoding algorithms with faster convergence to $D^*(R)$ while still having complexity that is polynomial in $n$.

\subsection{Successive refinement interpretation}

The  encoding algorithm may be interpreted in terms of successive refinement source coding \cite{EqCover91, Rimoldi94}. We can think of each section of the design matrix $A$ as a lossy codebook of rate $R/L$. For each section $\ell$, $i=1, \ldots, L$, the residual $r_{\ell-1}$ acts as the `source' sequence, and the algorithm attempts to find the column \emph{within} the section that minimizes the distortion. The distortion after section $\ell$  is the variance of the residual $r_\ell$; this residual  acts as the source sequence for section $\ell-1$.  Recall  that the minimum mean-squared distortion achievable with a Gaussian codebook  at rate $R/L$ is  \cite{Lapidoth97}
\be
\begin{split}
D^*_\ell & = \abs{r_{\ell-1}}^2 \exp(-2R/L)   \approx  \abs{r_{\ell-1}}^2 \left( 1 - \frac{2R}{L} \right), \   \text{ for } R/L \ll 1. 
\end{split}
\label{eq:Di_star}
\ee
This minimum distortion can be attained with a codebook with elements chosen $\sim_{i.i.d.} \mc{N}(0, \abs{r_{\ell-1}}^2 - D^*_\ell)$. From \eqref{eq:ci_def}, recall that the codeword variance in section $i$ of the codebook is
\be
 c_\ell^2 = \frac{2R \sigma^2}{L} \left( 1- \frac{2R}{L} \right)^{\ell-1} \approx \abs{r_{\ell-1}}^2 - D^*_\ell,
\ee
    where the approximate equality follows from \eqref{eq:Di_star} and \eqref{eq:Res_i}. Therefore, the typical value of the distortion in Section $i$ is close to $D^*_\ell$ since the algorithm is equivalent to minimum-distance encoding within each section. However, since the rate $R/L$ is infinitesimal, the deviations from $D^*_\ell$  in each section can be significant.  Despite this,  when the number of sections $L$ is large, Theorem \ref{thm:rd_feasible} guarantees that the {final} distortion $\abs{r_L^2}$ is close to the typical value $\sigma^2 e^{-2R}$. 

A similar successive refinement approach was used in \cite{no2016rateless} to construct a lossy compression scheme that shares some similarities with the successive cancellation encoder.

\section{Simulation results}

\begin{figure}
\begin{center}
\includegraphics[width=4.5in]{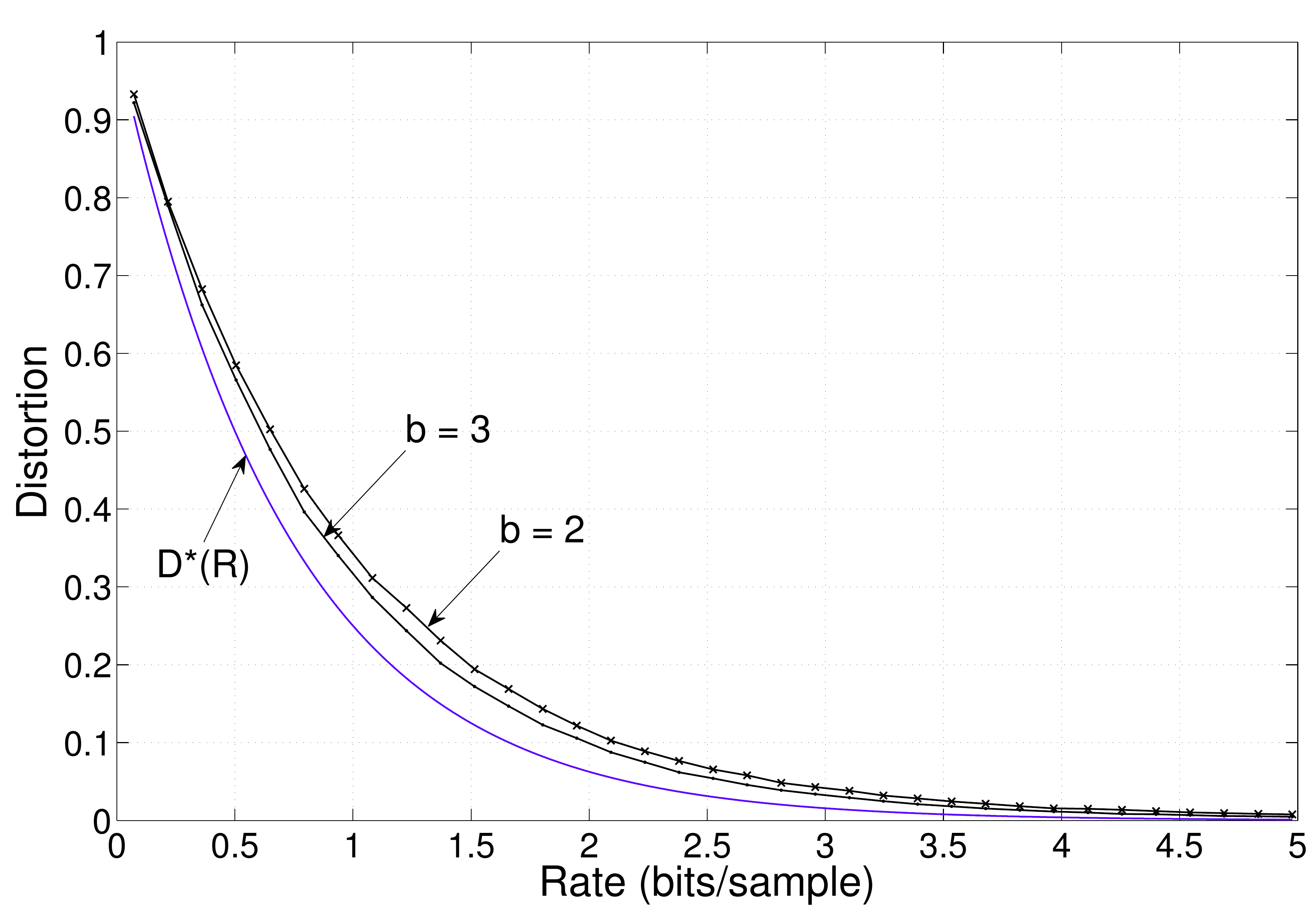}
\includegraphics[width=4.49in]{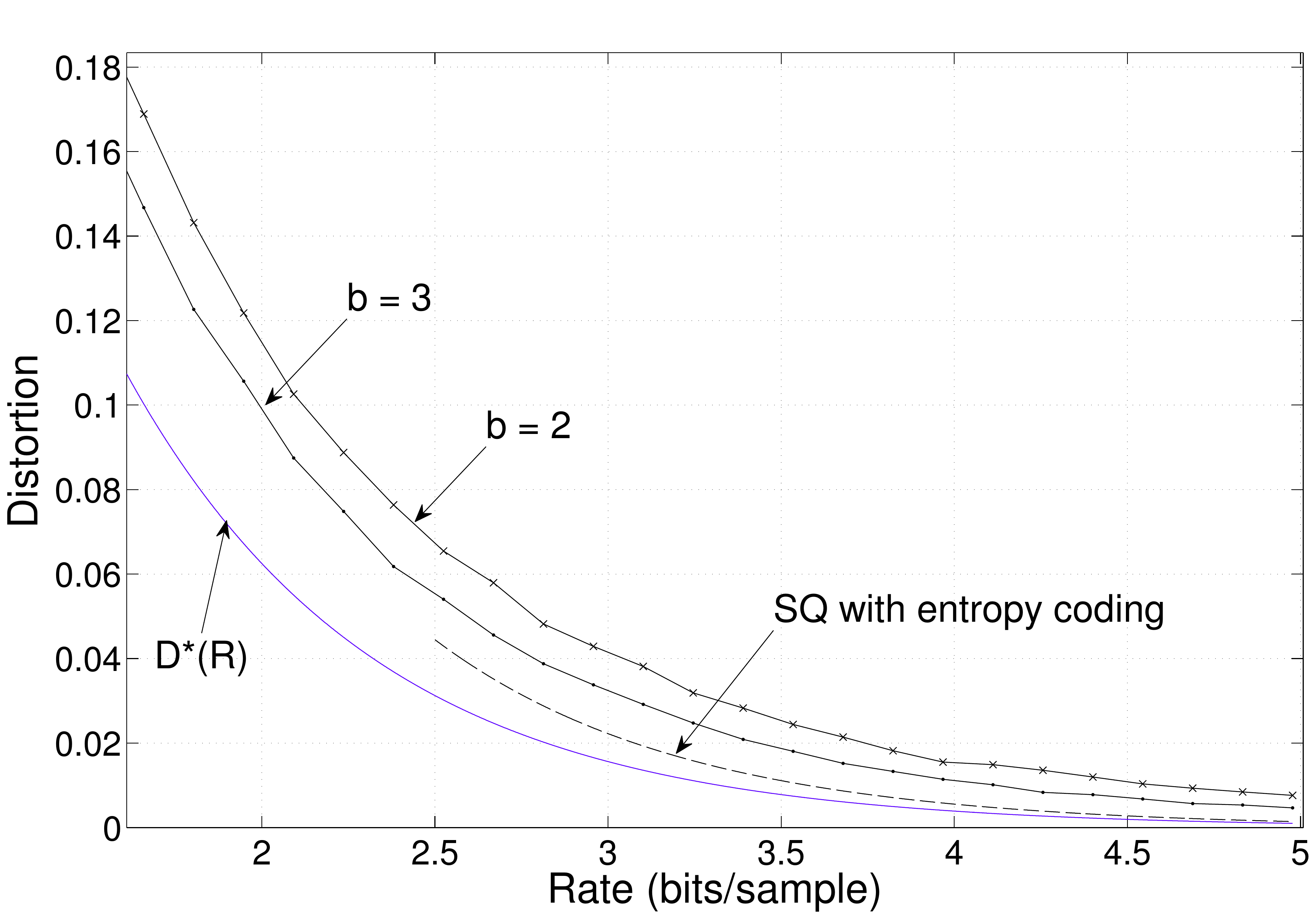}
\caption{\small{Top: Average distortion of the proposed encoder for i.i.d $\mc{N}(0,1)$ source. The design matrix has dimension $n \times ML$ with $M=L^b$. The distortion-rate performance is shown for $b=2$ and $b=3$ along with $D^*(R)=e^{-2R}$.
Bottom: Focusing on the higher rates. The dashed line is the high-rate approximation for the distortion-rate function of an optimal entropy-coded scalar quantizer. }}
\label{fig:sparc_perf}
\end{center}
\end{figure}

In this section, we examine the empirical rate-distortion performance  of the encoder via numerical simulations.  The top graph in Fig. \ref{fig:sparc_perf} shows the performance on a unit variance i.i.d Gaussian source. The dictionary dimension is $n \times M L$ with $M=L^b$.  The curves show the average distortion at various rates for $b=2$ and $b=3$. The average was obtained from $70$ random trials at each rate. Following convention, rates are plotted in bits rather than nats. The value of $L$ was increased with rate in order to keep the total computational complexity ($\propto nL^{b+1}$) similar across different rates. Recall from \eqref{eq:ml_nR} that the block length is determined by
\[ n = \frac{b L \log L}{R}. \]
For example, for the rates  $1.082, 2.092, 3.102$ and $4.112$ bits/sample, $L$ was chosen to be $46,66,81$ and $97$, respectively.   The corresponding values for the block length are $n=705,573,497,468$ for $b=3$, and $n=470,382,331,312$ for $b=2$. The graph shows the reduction in distortion obtained by increasing $b$ from $2$ to $3$. This reduction comes at the expense of an increase in computational complexity by a factor of $L$.
Simulations were also performed for a unit variance Laplacian source. The resulting distortion-rate curve was virtually identical to Fig. \ref{fig:sparc_perf}, which is consistent with Theorem \ref{thm:rd_feasible}.

For the simulations, a slightly modified version of the algorithm in Section \ref{sec:alg_desc} was used: the column selected in each iteration was based on minimum distance from the residual, rather than on maximum correlation as in \eqref{eq:max_stepi}. That is, 
\be
\begin{split}
 m_\ell  & =  \underset{j: \ (\ell-1)M < \; j \; \leq \ell M }{\operatorname{argmin}} \norm{r_{\ell-1} - c_\ell A_j }^2. 
\end{split}
\label{eq:mod_col_sel}
\ee
Though the two rules are similar for large $n$ (since $\| A_j \| \approx 1$ for all $j$), we found the distance-based rule to give slightly better empirical performance. 

Gish and Pierce \cite{GishPierce68} showed that uniform quantizers with entropy coding are nearly optimal at high rates and that their distortion for a unit variance source is well-approximated by $\frac{\pi e}{6} e^{-2R}$. ($R$ is the entropy of the quantizer in nats.) The bottom graph of Fig. \ref{fig:sparc_perf} zooms in on the higher rates and shows the above high-rate approximation for the distortion of an optimal entropy-coded scalar quantizer (EC-SQ). Recall from \eqref{eq:delta_min} that the distortion gap from $D^*(R)$ is of the order of \footnote{The constants in Theorem \ref{thm:rd_feasible} are not optimized, so the theorem does not give a very precise estimate of the excess distortion in the high-rate, low-distortion regime.}
 \[ \delta_2 \approx \frac{\log \log M}{2 \log M} = \frac{\log b  + \log \log L}{2 b \log L}, \] which is comparable to the optimal $D^*(R)= e^{-2R}$ in the high-rate region. (In fact, $\delta_2$ is larger than $D^*(R)$ at rates greater than $3$ bits for the values of $L$ and $b$ we have used.) This explains the large  {ratio} of the empirical distortion to $D^*(R)$ at higher rates. 
 
  In summary, the proposed encoder has good empirical performance, especially at low to moderate rates even with modest values of $L$ and $b$. At high rates, there are a few other compression schemes including EC-SQs and the shape-gain quantizer of \cite{HamkZ02} whose empirical rate-distortion performance is close to optimal (see \cite[Table III]{HamkZ02}).

\section{Proof of Theorem \ref{thm:rd_feasible}}  \label{sec:thm_feasible_proof}

The proof involves analyzing the deviation from the typical values of the residual distortion at each step of the encoding algorithm.  In particular, we have to deal with atypicality concerning the source sequence, the design matrix, and the maximum  computed in each step of the algorithm.

We introduce some notation to capture the deviations from the typical values. Define $\Delta_0$ via 
\be
\abs{s}^2 = \abs{r_0}^2 = \sigma^2(1+\Delta_0)^2.
\ee
The deviation of the norm of the residual at stage $i=1, \ldots, L$ from its typical value is captured by $\Delta_i$, defined via
\be
\abs{r_i}^2 = \sigma^2 \left( 1- \frac{2R}{L} \right)^i (1+ \Delta_i)^2.
\label{eq:Deli_def}
\ee

The deviation in  the norm of $A_{m_i}$ (the column  chosen in step $i$) is captured by $\gamma_i$, defiend as 
\be
\abs{A_{m_i}}^2 = 1+ \gamma_i, \quad i=1,\ldots, L.
\label{eq:Ami_dev}
\ee
Recall that the statistics $T^{(i)}_j$ defined  in \eqref{eq:stat_Tij} are i.i.d $\mc{N}(0,1)$  for $j \in \{(i-1)M+1, \ldots, iM \}$.
We write
\be
\begin{split}
\max_{(i-1)M+1 \leq \, j \, \leq iM} \, T^{(i)}_{j} & = \left\langle A_{m_i}, \frac{r_{i-1}}{\norm{r_{i-1}}} \right\rangle \\
& = \sqrt{2 \log M}(1+ \epsilon_i), \quad i=1, \ldots, L.
\end{split}
\ee
The $\epsilon_i$ measure the deviations of the maximum from $\sqrt{2\log M}$ in each step.

With this notation, using the expression for $r_i$ from \eqref{eq:gen_residue} we  have
\be
\begin{split}
& \abs{r_i}^2  =   \sigma^2\left(1 - \frac{2R}{L}\right)^{i}(1+\Delta_i)^2 \\
& = \abs{r_{i-1}}^2 + c_i^2 \abs{A_{m_i}}^2  - \frac{2 c_i \norm{r_{i-1}}}{n}
\left\langle A_{m_i}, \frac{r_{i-1}}{\norm{r_{i-1}}} \right\rangle\\
&= \sigma^2\left(1 - \frac{2R}{L}\right)^{i-1} (1 + \Delta_{i-1})^2 + c_i^2 (1+ \gamma_i) \\
& \quad - {2c_i \sigma \left(1 - \frac{2R}{L}\right)^{\frac{i-1}{2}} (1+\Delta_{i-1})} \sqrt{\frac{2\log M}{n}} (1+\epsilon_{i}) \\
&= \sigma^2\left(1 - \frac{2R}{L}\right)^{i} \Bigg( (1+\Delta_{i-1})^2   +  \frac{\frac{2R}{L}}{1 - \frac{2R}{L}}( \Delta_{i-1}^2 + \gamma_i -2\epsilon_i(1+\Delta_{i-1})) \Bigg).
\end{split}
\label{eq:expand_R1}
\ee
From \eqref{eq:expand_R1}, we obtain 
\be
\begin{split}
 (1+\Delta_i)^2   = (1+\Delta_{i-1})^2 + \frac{\frac{2R}{L}}{1 - \frac{2R}{L}}( \Delta_{i-1}^2 + \gamma_i -2\epsilon_i(1+\Delta_{i-1}) ), \ \  i \in [L].
\end{split}
\label{eq:Deli_Deli1}
\ee
The goal is to bound the final distortion  given by
\be \abs{r_L}^2 = \sigma^2 \left( 1 - \frac{2R}{L}\right)^L (1 + \Delta_L)^2 .\ee
We would like to find an upper bound for $(1 + \Delta_L)^2$ that holds under an event whose probability is close to $1$.
Accordingly, define $\mc{A}$ as the event where \emph{all} of the following hold:
\begin{enumerate}
\item $\abs{\Delta_0} < \delta_0$,
\item $\sum_{i=1}^L \frac{\abs{\gamma_i}}{L}< \delta_1$,
\item $\sum_{i=1}^L \frac{\abs{\epsilon_i}}{L} < \delta_2$,
\end{enumerate}
for $\delta_0, \delta_1, \delta_2$ that satisfy \eqref{eq:del0del1del2}. We upper bound the probability of the event $\mc{A}^c$  using the following  large deviations bounds.

\begin{lemma}
For $\delta \in (0,1]$,  $P\left( \frac{1}{L} \sum_{i=1}^L \abs{\gamma_i}  > \delta \right) < 2 ML \; \exp\left( -n{\delta^2}/{8} \right).$
\label{lem:gamma_i}
\end{lemma}

\begin{lemma} \label{lem:eps_i}
For $\delta>0$, $P\left( \frac{1}{L} \sum_{i=1}^L \abs{\epsilon_i}  > \delta \right) <  \left( \frac{M^{2\delta}}{8 \log M} \right)^{-L}$.
\end{lemma}

The proofs of these lemmas can be found in Appendix II and III of \cite{RVGaussFeasible}, respectively.  

Using these lemmas, we have
\be
P(\mc{A}^c) < p_0 + p_1 + p_2
\label{eq:PAc}
\ee
where $p_0, p_1, p_2$ are given by \eqref{eq:p0p1p2}.
The remainder of the proof consists of obtaining a bound for  $(1 + \Delta_L)^2$  under the condition that $\mc{A}$ holds.


\begin{lemma}
For all sufficiently large $L$, when $\mc{A}$ holds  we have
\be
\Delta_{i} \geq \Delta_0  - \frac{4R}{1-2R/L} \left( \sum_{j=1}^{i} \frac{\abs{\gamma_j} + \abs{\epsilon_j}}{L}\right), \quad i=1, \ldots,L.
\label{eq:Deli_ind}
\ee
In particular, $\Delta_i > -\frac{1}{2}, \quad i=1, \ldots ,L$
\label{lem:Deli_LB}
\end{lemma}

\begin{proof}
We first show that $\Delta_i > -\frac{1}{2}$ follows from \eqref{eq:Deli_ind}. Indeed, \eqref{eq:Deli_ind} implies that
\be
\begin{split}
\Delta_{i} & \geq \Delta_0  - \frac{4R}{1-2R/L} \left( \sum_{j=1}^{i} \frac{\abs{\gamma_j} + \abs{\epsilon_j}}{L}\right) \stackrel{(a)}{>} - \delta_0 - {5R}\left(\delta_1  + \delta_2 \right) \stackrel{(b)}{>}  -\frac{1}{2} \label{eq:final_deli}
\end{split}
\ee
where $(a)$ is obtained from the conditions of $\mc{A}$ while $(b)$ holds due to \eqref{eq:del0del1del2}.

The statement \eqref{eq:Deli_ind} trivially holds for $i=0$. Towards induction, assume \eqref{eq:Deli_ind} holds for $i-1$ for some $i \in \{1, \ldots, L\}$. From \eqref{eq:Deli_Deli1}, we obtain
\begin{align}
& (1+\Delta_i)^2 = (1+\Delta_{i-1})^2 + \frac{2R/L}{1 - 2R/L}( \Delta_{i-1}^2 +  \gamma_i - 2\epsilon_i(1+\Delta_{i-1}) ) \nonumber \\
& \geq (1+\Delta_{i-1})^2  -  \frac{2R/L}{1 - 2R/L}(\abs{\gamma_i} + 2\abs{\epsilon_i}(1+\Delta_{i-1}) ).
\label{eq:st1}
\end{align}
For $L$ large enough, the right side above is positive and we therefore have
\begin{align}
& (1 + \Delta_i)  \geq (1+ \Delta_{i-1})
\left[ 1 - \frac{2R/L}{1 - 2R/L} \left[ \frac{\abs{\gamma_i}}{(1+\Delta_{i-1})^2} + \frac{2\abs{\epsilon_i}}{1+\Delta_{i-1}} \right)  \right]^{\frac{1}{2}} \nonumber \\
& \geq   1 + \Delta_{i-1} - \frac{2R/L}{1 - 2R/L} \left( \frac{\abs{\gamma_i}}{(1+\Delta_{i-1})} +  {2\abs{\epsilon_i}} \right),
\label{eq:st2}
\end{align}
where the second  inequality  is obtained using $\sqrt{1-x} \geq 1-x$  for $x\in(0,1)$.  We therefore have
\be
\begin{split}
\Delta_i & \geq \Delta_{i-1} - \frac{2R/L}{1-2R/L} \left( \frac{\abs{\gamma_i}}{(1+\Delta_{i-1})} +  {2\abs{\epsilon_{i}}} \right) \\
& \stackrel{(a)}{\geq}  \Delta_{i-1} - \frac{2R/L}{1-2R/L} \left( 2{\abs{\gamma_i}} +  {2\abs{\epsilon_{i}}} \right) \\
& \stackrel{(b)}{\geq}  \Delta_0  - \frac{4R}{1-2R/L} \left( \sum_{j=1}^{i-1} \frac{\abs{\gamma_j} + \abs{\epsilon_j}}{L}\right)  - \frac{4R/L}{1-2R/L}( \abs{\gamma_i} + \abs{\epsilon_i}).
\end{split}
\label{eq:Del1_Del0_rel}
\ee
In the chain above, $(a)$ holds because $\Delta_{i-1} > \frac{1}{2}$, a consequence of the induction hypothesis as shown in \eqref{eq:final_deli}. $(b)$ is obtained by using the induction hypothesis for $\Delta_{i-1}$.
\end{proof}


\begin{lemma}
When $\mc{A}$ is true and $L$ is large enough that Lemma \ref{lem:Deli_LB} holds,
\be  \abs{\Delta_i} \leq \abs{\Delta_0} w^i + \frac{4R/L}{1 - 2R/L} \sum_{j=1}^i w^{i-j} (\abs{\gamma_j} + \abs{\epsilon_j}) \ee
for $i=1,\ldots, L$, where $w= \left( 1 + \frac{R/L}{1-2R/L}\right)$.
\label{lem:Ri_UB}
\end{lemma}
\begin{proof}
We prove the lemma by induction. For $i=1$, we have from \eqref{eq:Deli_Deli1}
\be
\begin{split}
& (1+\Delta_1)^2  = (1+\Delta_0)^2 + \frac{\tfrac{2R}{L}}{1 - \tfrac{2R}{L}}(\Delta_0^2 + \gamma_1 -2\epsilon_1(1+\Delta_0)) \\
& = (1 + \abs{\Delta_0})^2
\Bigg[  1 +   \frac{\tfrac{2R}{L}}{1 - \tfrac{2R}{L}}\left( \frac{\Delta_0^2}{(1 +\abs{\Delta_0})^2}  + \frac{\abs{\gamma_1}}{(1+\abs{\Delta_0})^2} + \frac{2\abs{\epsilon_1}}{(1 + \abs{\Delta_0})} \right)   \Bigg].
\end{split}
\label{eq:ub_step1}
\ee
Therefore,
\be
\begin{split}
 1 + \Delta_1 & \leq (1 + \abs{\Delta_0})\Bigg[  1 +   \frac{ \frac{2R}{L}}{1 - \frac{2R}{L}}\left( \frac{\Delta_0^2}{(1 +\abs{\Delta_0})^2}  + \frac{\abs{\gamma_1}}{(1+\abs{\Delta_0})^2} + \frac{2\abs{\epsilon_1}}{(1 + \abs{\Delta_0})} \right)   \Bigg]^{\frac{1}{2}} \\
& \leq (1 + \abs{\Delta_0})\Bigg[  1 +   \frac{ \frac{R}{L}}{1 - \frac{2R}{L}}\left( \frac{\Delta_0^2}{(1 +\abs{\Delta_0})^2}  + \frac{\abs{\gamma_1}}{(1+\abs{\Delta_0})^2} + \frac{2\abs{\epsilon_1}}{(1 + \abs{\Delta_0})} \right)   \Bigg]
\end{split}
\label{eq:ub_step2}
\ee
where we have used the inequality $\sqrt{1+x} \leq 1 + \frac{x}{2}$ for $x>0$. We therefore have
\begin{align}
\Delta_1 & \leq \abs{\Delta_0} + \frac{R/L}{1 - 2R/L}\left( \frac{\Delta_0^2}{(1 +\abs{\Delta_0})} + \frac{\abs{\gamma_1}}{(1+\abs{\Delta_0})} + {2\abs{\epsilon_1}} \right) \nonumber \\
& \stackrel{(a)}{\leq}  \abs{\Delta_0} + \frac{R/L}{1 - 2R/L} ( \abs{\Delta_0} + \abs{\gamma_1} + 2\abs{\epsilon_1}) \nonumber \\
& \leq \abs{\Delta_0}\left( 1 +  \frac{R/L}{1 - 2R/L} \right) + \frac{2R/L}{1 - 2R/L}(\abs{\gamma_1} +  \abs{\epsilon_1}),
\label{eq:Del1UB}
\end{align}
where $(a)$ is obtained using $\abs{\Delta_0}/(1 + \abs{\Delta_0}) < 1$. From Lemma \ref{lem:Deli_LB}, we have
\be
\begin{split}
\Delta_{1} & \geq \Delta_0  - \frac{4R/L}{1-2R/L} \left(\abs{\gamma_1} + \abs{\epsilon_1}\right) \\
&  \geq  - \abs{\Delta_0}  - \frac{4R/L}{1-2R/L} \left(\abs{\gamma_1} + \abs{\epsilon_1}\right).
\end{split}
\label{eq:Del1LB}
\ee
Combining \eqref{eq:Del1UB} and \eqref{eq:Del1LB}, we obtain
\be
\abs{\Delta_1} \leq \abs{\Delta_0}\left( 1 +  \frac{R/L}{1 - 2R/L} \right) + \frac{4R/L}{1 - 2R/L}(\abs{\gamma_1} + \abs{\epsilon_1}).
\label{eq:AbsDel1UB}
\ee
This completes the proof for $i=1$. 

Towards induction, assume that the lemma holds for $i-1$. From \eqref{eq:Deli_Deli1}, we obtain
\be
\begin{split}
(1+\Delta_i)^2  &  \leq 1   + \Delta_{i-1}^2 + 2\abs{\Delta_{i-1}}  \\
& \  +  \frac{2R/L}{1 - 2R/L}( \Delta_{i-1}^2 + \abs{\gamma_i} + 2\abs{\epsilon_i}(1 + \abs{\Delta_{i-1}}) ).
\end{split}
\ee
Using arguments identical to those in \eqref{eq:ub_step1}--\eqref{eq:Del1UB}, we get
\be
\Delta_i \leq \abs{\Delta_{i-1}}\left( 1 +  \frac{R/L}{1 - 2R/L} \right) + \frac{2R/L}{1 - 2R/L}(\abs{\gamma_i} + \abs{\epsilon_i}).
\label{eq:DeliUB}
\ee

From the proof of Lemma \ref{lem:Deli_LB} (see \eqref{eq:Del1_Del0_rel}), we have
\be
\begin{split}
\Delta_{i} & \geq  \Delta_{i-1} - \frac{4R/L}{1-2R/L} \left( {\abs{\gamma_i}} +  {\abs{\epsilon_{i}}} \right) \\
&  \geq -\abs{\Delta_{i-1}} - \frac{4R/L}{1-2R/L} \left( {\abs{\gamma_i}} +  {\abs{\epsilon_{i}}} \right).
\end{split}
\label{eq:DeliLB}
\ee
Combining \eqref{eq:DeliUB} and \eqref{eq:DeliLB}, we obtain
\be
\abs{\Delta_i} \leq \abs{\Delta_{i-1}}\left( 1 +  \frac{R/L}{1 - 2R/L} \right) + \frac{4R/L}{1 - 2R/L}(\abs{\gamma_i} + \abs{\epsilon_i}).
\label{eq:AbsDeliUB}
\ee
Using the induction hypothesis to bound $\abs{\Delta_{i-1}}$ in \eqref{eq:AbsDeliUB}, we obtain
\ben
\begin{split}
\abs{\Delta_i}  \leq & \left(\abs{\Delta_0} w^{i-1}  +  \frac{4R/L}{1 - 2R/L} \sum_{j=1}^{i-1} w^{i-1-j} (\abs{\gamma_j} + \abs{\epsilon_j}) \right)  \\
& \cdot \left(1 +  \frac{R/L}{1 - 2R/L} \right)  +   \frac{4R/L}{1 - 2R/L}(\abs{\gamma_i} + \abs{\epsilon_i}) \\
& = \abs{\Delta_0} w^i + \frac{4R/L}{1 - 2R/L} \sum_{j=1}^i w^{i-j} (\abs{\gamma_j} + \abs{\epsilon_j}),
\end{split}
\een
as required.
\end{proof}

Lemma \ref{lem:Ri_UB} implies that when $\mc{A}$ holds and $L$ is sufficiently large,
\be
\begin{split}
\abs{\Delta_L} & \leq  \abs{\Delta_0} w^L + \frac{4R/L}{1 - 2R/L} \sum_{j=1}^L w^{L-j} (\abs{\gamma_j} + \abs{\epsilon_j}) \\
& \leq w^L \left[ \abs{\Delta_0} + \frac{4R}{(1-2R/L)w}  \left(\sum_{j=1}^L \frac{\abs{\gamma_j}}{L}  +  \sum_{j=1}^L  \frac{\abs{\epsilon_j}}{L} \right)  \right] \\
& \stackrel{(a)}{\leq}  w^L \left[  {\delta_0} + \frac{4R}{(1-R/L)} (\delta_1 + \delta_2) \right] \\
& \stackrel{(b)}{\leq} \exp\left(\frac{R}{1 - 2R/L}\right) \left[  \delta_0 + \frac{4R}{(1-R/L)} (\delta_1 + \delta_2)  \right] \\
& \leq e^R \left(\delta_0 + 5R(\delta_1 + \delta_2) \right) \quad \text{ for  large enough } L.
\end{split}
\label{eq:Del_L_UB}
\ee
In the above chain, $(a)$ is true because $\mc{A}$ holds, and $(b)$ is obtained by applying the inequality $1 + x \leq e^x$ with $x=\frac{R/L}{1 - 2R/L}$.

Hence when $\mc{A}$ holds and $L$ is sufficiently large, the distortion can be bounded as
\be
\begin{split}
\abs{R_L}^2  = \sigma^2 e^{-2R}(1 + \Delta_L)^2 & \leq \sigma^2 e^{-2R} (1 + \abs{\Delta_L})^2 \\
&  \stackrel{(c)}{\leq} \sigma^2 e^{-2R} (1 + e^R {\Delta})^2 
\end{split}
\label{eq:finalRL}
\ee
where $(c)$ follows from \eqref{eq:Del_L_UB} by defining $\Delta = \delta_0 + 5R(\delta_1 + \delta_2)$. Combining \eqref{eq:finalRL} with \eqref{eq:PAc}
completes the proof of the theorem.

%
%
\part{Multiuser Communication and Compression with SPARCs}
\chapter{Broadcast and Multiple-access Channels}  \label{chap:bcmac}

 In the final part of the monograph, we discuss the use of SPARCs for multiuser channel and source coding models.  It is well known \cite{elGKBook} that the optimal rate regions for several multiuser channel and source coding problems can be achieved using the following ingredients: i)  rate-optimal point-to-point source and channel codes, and ii) combining or splitting these point-to-point codes via superposition or random binning.  In this chapter, we show  how superposition  coding can be implemented using SPARCs for Gaussian broadcast and multiple-access channels.  In the next chapter, we describe how random binning can be implemented using SPARCs.
 
  All rates within the capacity region of the Gaussian broadcast and multiple access channels can be achieved by combining codes designed for point-to-point Gaussian channels \cite{coverT12,elGKBook}. Therefore the SPARC construction for point-to-point channels, where codewords are defined as the superposition of columns of a matrix, can be easily extended  to these multiuser channels.

\section{The Gaussian broadcast channel}

The $K$-user AWGN broadcast channel has a single transmitter and $K$ output sequences, one for each user. The input sequence $x=(x_1, \ldots, x_n)$ transmitted over the broadcast channel  has to satisfy an average power constraint: $\frac{1}{n}\sum_{j}x_j^2 \leq P$.  The channel output sequence of user $i \in [K]$ is denoted by 
$y^{(i)}=(y^{(i)}_1, \ldots, y^{(i)}_n)$, where the $j$th output symbol is produced as 
$y^{(i)}_j=x_j + w^{(i)}_j$, for $j \in [n]$.  The noise variables $(w^{(i)}_j)_{j \in [n]}$  are i.i.d. $\sim \normal(0, \sigma_i^2)$, where $\sigma_i^2$ is the noise variance of the $i$th receiver.

We will focus on the two-user broadcast channel for simplicity, although the coding schemes can be extended to $K >2$ users in a straightforward way. Throughout, we will assume $\sigma_1^2 < \sigma_2^2$, i.e., the noise at the first receiver has a  lower variance than the noise at the second. 

If we denote the rates of the two users by $R_1$ and $R_2$, then the capacity region \cite[Chapter 5]{elGKBook}
 is the union of all rate pairs $(R_1, R_2)$ over $\alpha \in [0,1]$ which satisfy 
\begin{align}
R_1 &  \leq \frac{1}{2} \log \left( 1 + \frac{\alpha P}{\sigma_1^2} \right), \label{eq:R1BC}\\
R_2 &  \leq \frac{1}{2} \log \left( 1 + \frac{(1-\alpha)P}{\alpha P+ \sigma_2^2} \right).  \label{eq:R2BC}
\end{align}
 It is well known that any rate pair within the capacity region can be achieved using superposition coding \cite{cover1972broadcast}. In superposition coding, the transmitted codeword $x$ is generated as  the sum of two independent codewords $x^{(1)}, x^{(2)}$,   with powers $\alpha P, (1-\alpha)P$, drawn from codebooks of size $2^{nR_1}$ and $2^{nR_2}$, respectively. 
 
Receiver $2$  has to decode $x^{(2)}$ from its output sequence $y^{(2)}= x^{(2)}+ x^{(1)} +w^{(2)}$. Treating  $x^{(1)}$ as interference (with average power $\alpha P$), receiver $2$   has an effective point-to-point channel with signal to noise ratio  
$\frac{(1-\alpha)P}{\alpha_P+ \sigma_2^2}$.  If receiver $2$ can reliably decode $x^{(2)}$,  receiver $1$ will be also able to first decode $x^{(2)}$ (with high probability) from $y^{(1)}= x^{(2)}+ x^{(1)} +w^{(1)}$. This is because  $\sigma_1^2 \leq \sigma_2^2$.  After subtracting the decoded $x^{(2)}$ from $y^{(1)}$, receiver 1 has an effective point-to-point channel with $\snr= \frac{\alpha P}{\sigma^2}$ to decoder $x^{(1)}$.

We now implement this superposition coding scheme with SPARCs.

\section{SPARCs for the Gaussian broadcast channel} \label{sec:SPARC_BC}

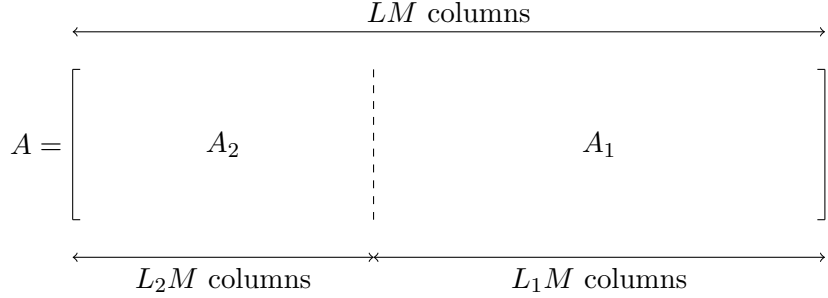
\begin{figure}
    \centering
    \begin{tikzpicture}
        \draw (0, 0) node[anchor=east] {$A=$};
        \draw (0, -1) -- (0, 1);
        \draw (0, -1) -- (0.1, -1);
        \draw (0, 1) -- (0.1, 1);
        \draw[dashed] (4, 1) -- (4, -1);
        \draw (10, -1) -- (10, 1);
        \draw (10, -1) -- (9.9, -1);
        \draw (10, 1) -- (9.9, 1);

        \draw (2, 0) node {$A_2$};
        \draw (7, 0) node {$A_1$};

        \draw[<->] (0, 1.5) -- (10, 1.5);
        \draw (5, 1.5) node[anchor=south] {$LM$ columns};

        \draw[<->] (0, -1.5) -- (4, -1.5);
        \draw[<->] (4, -1.5) -- (10, -1.5);
        \draw (2, -1.5) node[anchor=north] {$L_2M$ columns};
        \draw (7, -1.5) node[anchor=north] {$L_1M$ columns};
    \end{tikzpicture}
    \caption{Division of the SPARC design matrix $A$ for two users in the Gaussian broadcast channel. The second user is allocated the first $L_2$ sections, and the first user is allocated the remaining $L_1$ section.}
\label{fig:mu-bc-a-split}
\end{figure}

Fix rates $R_1, R_2$ that lie within the capacity region \eqref{eq:R1BC}--\eqref{eq:R2BC}.  The two users' codebooks  are defined via  SPARC design matrices $A_1$ and $A_2$ with parameters $(n, M, L_1)$ and $(n, M, L_2)$, respectively.  The parameters are chosen such that 
\be
nR_1 = L_1 \log M, \qquad nR_2 = L_2 \log M.
\label{eq:bcsparc_relations}
\ee
The entries of $A_1, A_2$ are chosen $\sim_{\text{i.i.d}} \mc{N}(0,1/n)$.

By concatenating the two design matrices, we obtain a  SPARC defined by $A=[A_2 \ A_1]$ with $L=L_1+L_2$ sections. This combined SPARC, shown in Figure \ref{fig:mu-bc-a-split}, has rate $R_1+R_2$, and from \eqref{eq:bcsparc_relations} we see that its parameters satisfy 
\[ n(R_1+R_2) = L \log M. \]

\paragraph{Power allocation} \label{page:palloc} 
The non-zero coefficients in the sections of users $1$ and $2$ are set to $\{ \sqrt{nP_{1\ell}}\}_{\ell \in [L_1]}$ and $\{\sqrt{nP_{2\ell}}\}_{\ell \in [L_2]}$, respectively. 
For optimal (ML) decoding, we use a flat power allocation, i.e.,  
$P_{1\ell} = \sqrt{\frac{(1-\alpha)P}{L_1}}$ and $P_{2\ell}= \sqrt{\frac{\alpha P}{L_2}}$, respectively.  For AMP decoding, the two power allocations $\{P_{1,\ell}\}_{\ell \in [L_1]}$ and $\{P_{2,\ell}\}_{\ell \in [L_2]}$ are chosen in the same way as for point-to-point SPARCs. For example, one could use the power allocation determined by the  iterative algorithm in Chapter \ref{chap:emp_perf} using the  parameters $(L_i, R_i, P_i)$ for user $i \in \{1, 2 \}$.

\paragraph{Encoding} The message of each user $i \in \{1,2\}$ mapped to a message vector $\beta^{(i)} \in \mc{B}_{M,L_i}$. The concatenated message vector is denoted by $\beta  \in \mcb$. The transmitted codeword is 
\[ x = A \beta = [A_1 \quad A_2] \begin{bmatrix}  \beta^{(1)}  
\\ \beta^{(2)} \end{bmatrix} = A_1 \beta^{(1)} + A_2 \beta^{(2)}.  
\]

\paragraph{Optimal decoding}   Receiver $2$ (with the higher noise variance) decodes
\be    
 \hat{\beta}^{(2)}_{\textsf{opt}} = \argmin_{\hat{\beta}^{(2)} \in \mc{B}_{M, L_2} }  \, \norm{y^{(2)}-A_2\hat{\beta}^{(2)}  }^2.  
\label{eq:rx2_opt}
\ee
Receiver $1$ first decodes the concatenated message vector $\beta$ as 
\be   
\hat{\beta}_{\textsf{opt}} = \argmin_{\hat{\beta} \in \mc{B}_{M, L} }  \, \norm{y^{(1)}-A\hat{\beta}  }^2, 
\label{eq:rx1_opt}
\ee
 and then reconstructs $\beta^{(1)}_{\textsf{opt}}$ by taking the last $L_1$ sections.

\paragraph{AMP decoding} Receiver $2$  decodes $ \hat{\beta}^{(2)}$ by running a standard SPARC AMP decoding
routine (as described  in Section \ref{sec:AMP_dec}), using the design matrix $A_2$, i.e., the first $L_2$ columns in $A$. 

Receiver $1$  decodes $\hat{\beta}$ via AMP decoding on the concatenated design matrix $A$, with the combined power allocation $\{ P_\ell \}_{\ell \in [L]}$ given by 
\begin{equation*}
    P_\ell =
\begin{cases}
    P_{2,\ell} & \quad \ell\le L_2\\
    P_{1,\ell-L_2} & \quad  L_2 < \ell \leq L= L_1+L_2. 
    \end{cases}
\end{equation*}
The last $L_1$ sections of  $\hat{\beta}$ represent $\hat{\beta}^{(1)}$.

\section{Bounds on error performance}

\subsection{Optimal decoding}

As seen from \eqref{eq:rx2_opt} and \eqref{eq:rx1_opt}, each receiver uses a point-to-point SPARC decoder, using design matrix $A_2$  for receiver $2$ and $A$ for  receiver $1$. Therefore, Theorem \ref{thm:MLresult} can be directly applied to obtain bounds 
on the probability of excess section error rate at each receiver with optimal decoding. 
\begin{theorem}
Let $M = L^{\msf{a}}$, with $\msf{a}$ such that $\msf{a} \geq \max\{ \msf{a}^*_L(\snr), \, \msf{a}^*_{L_2}(\snr) \}$, where $\msf{a}^*_L(\snr)$ is defined in \eqref{rem:al}. Let $(R_1, R_2)$ be a rate pair within the capacity region given by   \eqref{eq:R1BC}--\eqref{eq:R2BC}. Let 
\be \Delta =   \frac{1}{2} \log \left( 1 + \frac{ P}{\sigma_1^2} \right) - (R_1+R_2), \quad \Delta_2 =  \frac{1}{2} \log \left( 1 + \frac{(1-\alpha)P}{\alpha P+ \sigma_2^2} \right) - R_2 \label{eq:DelDelta2} \ee 
be strictly positive distances (of  $(R_1+R_2)$ and $R_2$, respectively) from  a  point on the boundary parametrized by some $\alpha \in [0,1]$. Then with optimal decoding, for any $\e_1,\e_2 >0$ the section error rates $\Esec^1$ and $\Esec^2$ at the two receivers satisfy
\be
\label{eq:P_SER_ML_BC}
\pr\left( \Esec^1 \geq \e_1 \right) = e^{-nE_1(\e_1, \Delta)}, \quad \pr\left( \Esec^2 \geq \e_2 \right) = e^{-nE_2(\e_2, \Delta_2 )}
\ee
where 
\[
E_1(\e_1 , \Delta) \geq h(\e_1, \Delta) - \frac{\log 2L}{n},  \quad   E_2(\e_2 , \Delta_2) \geq h(\e_2, \Delta_2) - \frac{\log 2L_2}{n},
\]
where $h(\cdot, \cdot)$ is defined in \eqref{eq:hedel_def}.
\end{theorem}
\begin{proof}
The result follows by applying Theorem \ref{thm:MLresult} to decoder $2$ which performs point-to-point decoding on a rate $R_2$ SPARC defined by $A_2$, and to decoder $1$ which decodes a rate $R_1 + R_2$ SPARC of defined by $A=[A_1 \ A_2]$.
\end{proof}

As in Proposition \ref{thm:msg_error}, one can bound the probabilities of message error, i.e., $P(\hat{\beta}^{(1)} \neq \beta^{(1)})$ and $P(\hat{\beta}^{(2)} \neq \beta^{(2)})$, by using an outer Reed-Solomon code. For any $\e>0$ and rate pair $(R_1, R_2)$ inside the capacity region,  by using an outer RS code of rate $(1-2\e)$ for each of the component SPARCs, one  obtains a code with rates $(R_1(1-2\e), R_2(1-2\e)) $ with message error probabilities of the two users bounded by $ e^{-nE_1(\e, \Delta)}$ and  $e^{-nE_2(\e, \Delta_2 )}$, respectively.

\subsection{AMP decoding}

For AMP decoding, we can apply Theorem \ref{thm:main_amp_perf} to obtain a bound on the probability of excess section error rate with an exponentially decaying allocation for each user. That is, with $P_1= \alpha P$ and $P_2 = (1-\alpha)P$, the power allocation for design matrix $A_1$ is  
\be P_{1\ell} = \kappa_1 \left(1 + \frac{\alpha P}{\sigma_1^2} \right)^{-\ell/L_1}, \quad \ell \in [L_1]. \label{eq:user1bc} \ee
The power allocation for $A_2$ is 
\be P_{2, \ell} = \kappa_2 \left(1 + \frac{(1-\alpha) P}{ \alpha P + \sigma_2^2} \right)^{-\ell/L_2}, \quad \ell \in [L_2]. \label{eq:user2bc} \ee
Here $\kappa_1, \kappa_2$ are normalizing constants chosen to satisfy the two power constraints. 
\begin{theorem}
Fix any rate pair $(R_1, R_2)$ within the capacity region  \eqref{eq:R1BC}--\eqref{eq:R2BC}. Consider a broadcast SPARC  defined by design matrix $A=[A_1 \, A_2]$  with parameters $n, M, L_1, L_2$,  that satisfy \eqref{eq:bcsparc_relations}, and  a power allocation given by \eqref{eq:user1bc} and \eqref{eq:user2bc}.  

Fix $\e >0$. Then for sufficiently large $L_1, L_2,M$, the section error rate of the AMP decoder satisfies
\be
\begin{split}
& P \left( \mc{E}_{sec}^1  >  \e \right) \leq K_{T} \exp\left\{\frac{-\kappa_{T} L}{(\log M)^{2T-1}}   
\left(\frac{\e  \sigma_1^2 \mc{C}}{2} - P f_R(M) \right)^2 \right\}, \\ 
& P \left( \mc{E}_{sec}^2  >  \e \right) \leq K_{T} \exp\left\{\frac{-\kappa_{T} L_2}{(\log M)^{2T-1}}   
\left(\frac{\e  \sigma_2^2 \mc{C}_2}{2} - P f_R(M) \right)^2 \right\},
\end{split}  
\ee
where $\mc{C} = \frac{1}{2}\log(1+ \frac{P}{\sigma_1^2})$,  $\mc{C}_2 = \frac{1}{2}\log(1+ \frac{(1-\alpha)P}{ \alpha P+\sigma_2^2})$, $T$ is the maximum number of iterations of the two AMP decoders, and $\kappa_T, K_T$ are defined in Theorem  \ref{thm:main_amp_perf}.  Furthermore, $T$ is inversely proportional to the minimum of $\Delta, \Delta_2$, which are defined in \eqref{eq:DelDelta2}.
\end{theorem}
\begin{proof}
It can be verified (using Lemma \ref{lem:conv_expec}) that the  state evolution  recursion \eqref{eq:tau_def} predicts reliable decoding in the large system limit for any rate pair within the capacity region and the specified power allocation. The result then follows using arguments very similar to  those used to prove Theorem \ref{thm:main_amp_perf}.
\end{proof}

\section{Simulation results}

\begin{figure}[p]
    \centering
    \begin{tabular}{ll}
    \includegraphics[width=0.5\textwidth]{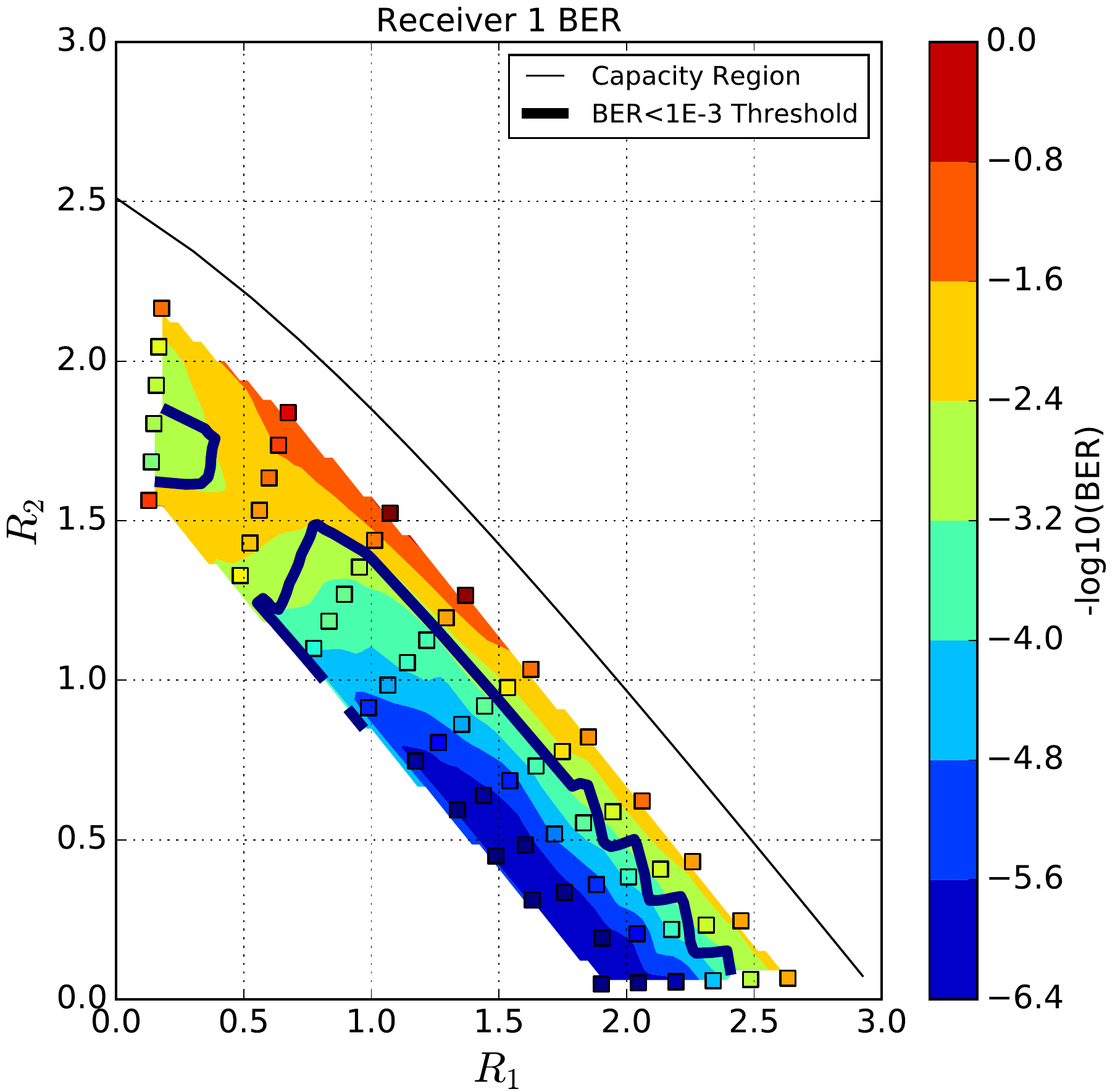}
        &
    \includegraphics[width=0.5\textwidth]{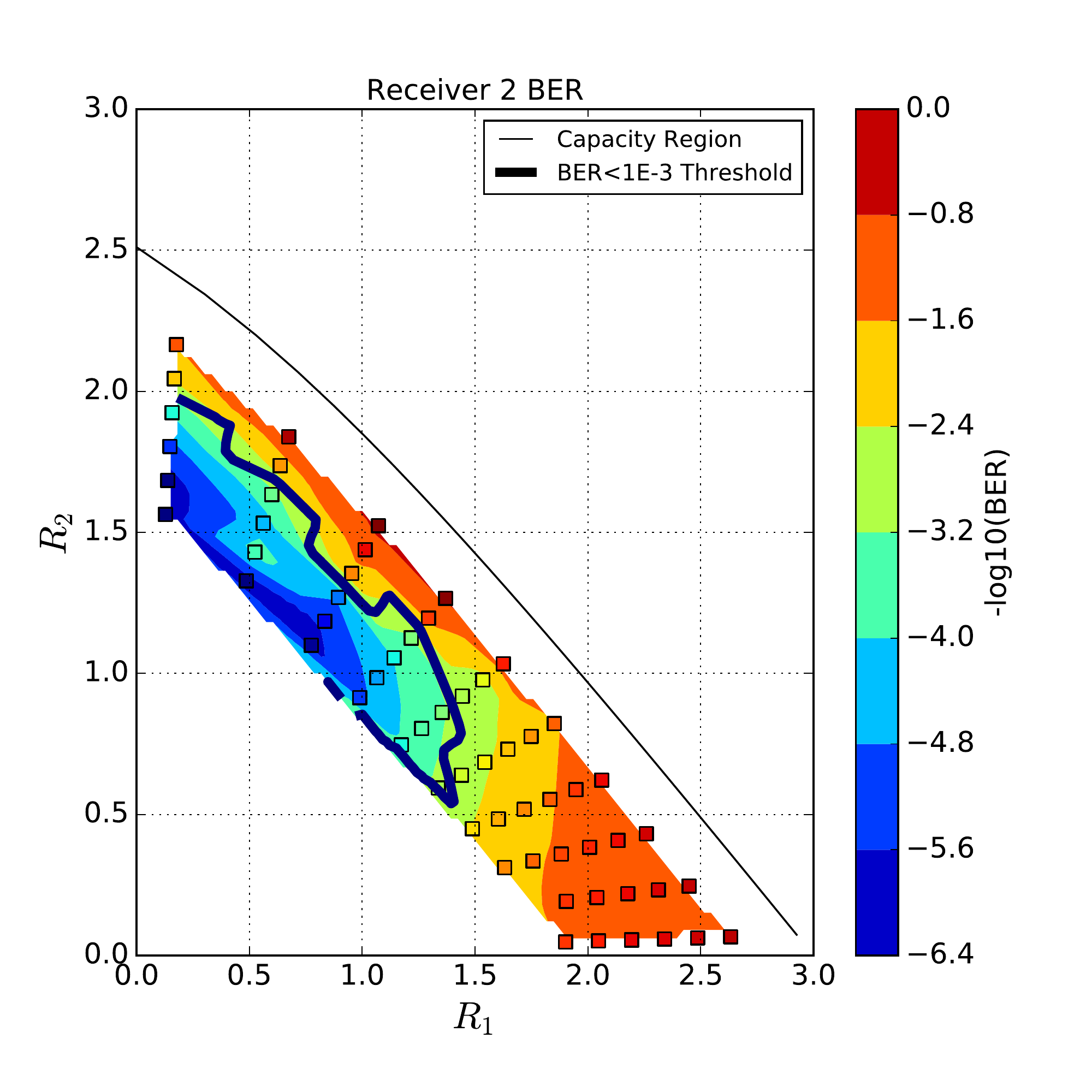}
    \end{tabular}
    \includegraphics[width=0.5\textwidth]{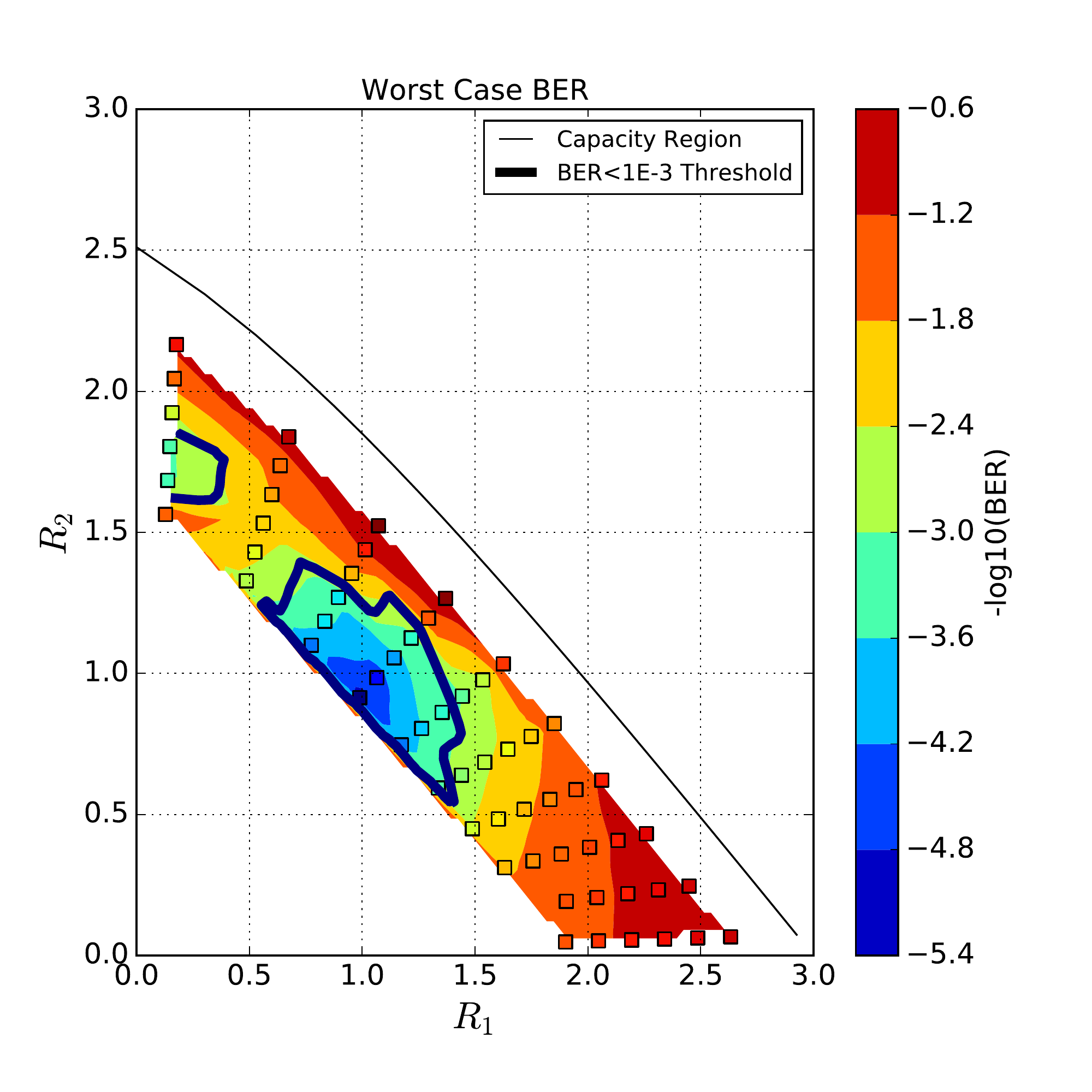}
    \caption{Bit error rates for the broadcast channel with AMP decoding, with contour
    indicating the boundary of the regions where the error rate is below
    $10^{-3}$.
    Each displayed point is the average of approximately 3000 trials.
    $P=63$, $\sigma_1^2=1.0$, $\sigma_2^2=2.0$, $M=512$, $n=4095$.}
\label{fig:bc-results}
\end{figure}

We now discuss the empirical performance of SPARCs for the Gaussian broadcast channel with AMP decoding, considering a setup with $P=63$, $\sigma_1^2=1$, and $\sigma_2^2=2$.  Each operating point is set as follows: fix $\alpha\in[0,1]$, which specifies
the balance of power between the two receivers. The point on the boundary of the capacity region corresponding to this $\alpha$ is
$\mc{C}_1=\frac{1}{2}\log(1+\frac{\alpha P}{\sigma_1^2})$ and
$\mc{C}_2=\frac{1}{2}\log(1+\frac{1-\alpha}{\alpha P + \sigma_2^2})$. Fix
$\gamma\in[0,1]$ and set $R_1=\gamma\mc{C}_1$ and $R_2=\gamma\mc{C}_2$. The parameter
$\gamma$  determines the back-off from the boundary point
$(\mc{C}_1,\mc{C}_2)$ of the capacity region.

For the SPARC design matrix, we first fix $M=512$ and $n=4095$. The parameters $L_1, L_2$ are then determined by the rate pair as   $L_1 = n R_1 / \log M$, and $L_2 = n R_2 / \log M$. With $P_1=\alpha P$ and $P_2 = (1-\alpha)P$,  the value of the non-zero coefficient in each section is set using the iterative power allocation algorithm as discussed on p. \pageref{page:palloc}.

The encoding and decoding operations are then performed as described in Section \ref{sec:SPARC_BC}. Figure~\ref{fig:bc-results} shows the bit-error rate performance in three charts. The  two charts on the top show the bit error rate
performance achieved when only considering the first and second receiver,
while the third chart shows the worst case of those two. Additionally, a
contour is plotted showing the boundaries of regions where the bit error
rate is found to be below $10^{-3}$.

When considering each receiver independently, we observe that the bit error rates
are reasonably low when that receiver's rate is above $0.5$, which is in
accordance with the results from point-to-point channels.  However, because
each receiver suffers badly once its rate goes below $0.5$, the worst-case
error is only better than $10^{-3}$ in a small number of cases, relatively far
from the capacity boundary and only near equal power balance where
$\alpha=0.5$.

Degradation of decoding performance at low rates is already observed in the simpler 
point-to-point case, but in the broadcast setup there is an additional factor contributing to
the performance. For the experimental setup as described, $M$ is fixed to the
same value for both users. As we saw in Chapter~\ref{chap:emp_perf}, for a
particular channel set-up, there is an optimal $M$, above and below which
performance can rapidly decrease. Therefore, when the rates for the two
receivers in the BC setup differ substantially, so too will the optimal $M$.
The gap between each receiver's optimum $M$ and the chosen value will lead to
performance degradation, as observed. It is conceptually possible to run the
AMP decoder with differing values for $M$ in different sections, which would
allow each receiver to operate on an optimal $M$, but this has not been
explored.

\section{The Gaussian multiple-access channel }

The $K$-user Gaussian multiple-access channel (MAC) has $K$ transmitters and a single receiver. The  codeword transmitted by user $i \in [K]$, denoted by
$x^{(i)}=(x^{(i)}_1, \ldots, x^{(i)}_n)$, has to satisfy an average power constraint $P_i$. The channel output sequence is $y =\sum_{i=1}^K  \, x^{(i)}  +  w$, where $w \in \mathbb{R}^n$ is $\sim_{\text{i.i.d.}} \normal(0, \sigma^2)$.

We will focus on the two-user MAC for simplicity, although the coding schemes can be extended to $K > 2$ users in a straightforward manner.  Denoting the rates of the two users by $R_1$ and $R_2$, the capacity region is the set of all rate pairs $(R_1, R_2)$  which satisfy \cite[Chapter 4]{elGKBook} 
\begin{align}
R_1 &  \leq \frac{1}{2} \log \left( 1 + \frac{P_1}{\sigma^2} \right), \label{eq:R1mac}\\
R_2 &  \leq \frac{1}{2} \log \left( 1 + \frac{P_2}{\sigma^2} \right),   \label{eq:R2mac}  \\
R_1 +  R_2 &  \leq \frac{1}{2} \log \left( 1 + \frac{P_1+P_2}{\sigma^2} \right).   \label{eq:R1R2mac} 
\end{align}

Any rate pair $(R_1, R_2)$ within the capacity region can be achieved using Shannon-style random coding, with independent codebooks for each of the two users.  Here the  entries of the two codebooks are generated $\sim_{i.i.d.} \normal(0,P_1)$ and $\sim_{i.i.d.} \normal(0,P_2)$, respectively.  Each user transmits a codeword from their own  codebook, and the receiver attempts to recover  the two codewords from $y = x^{(1)} +  x^{(2)} + w$  via \emph{joint}  maximum-likelihood  decoding.

We now show how an efficient SPARC coding scheme can be used to communicate reliably at any pair of rates within the capacity region.

\section{SPARCs for the Gaussian multiple-access channel}

Consider a rate pair $(R_1, R_2)$  within the MAC capacity region \eqref{eq:R1mac}--\eqref{eq:R1R2mac}.  The SPARC construction is similar to that of the broadcast channel. The two user's codebooks  are defined via  SPARC design matrices $A_1$ and $A_2$ with parameters $(n, M, L_1)$ and $(n, M, L_2)$, respectively.  The parameters are chosen such that 
\be
nR_1 = L_1 \log M, \qquad nR_2 = L_2 \log M.
\label{eq:maccsparc_relations}
\ee

\paragraph{Power allocation}  The non-zero coefficient in section  $\ell \in [L_1]$ of $A_1$  is
$\sqrt{nP_{1\ell}}$, while that in  section  $\ell \in [L_2]$ of $A_2$ is $\sqrt{nP_{2\ell}}$. With optimal (ML) decoding,  each transmitter  can use a flat power allocation across sections, i.e., $P_{1\ell} =P_1/L$ and $P_{2\ell} =P_2/L$ for $\ell \in [L]$.  

 For  AMP decoding, we design a power allocation for the combined SPARC defined by the design matrix $A=[A_1 \, A_2]$, and then partition the allocation between the two transmitters. Designing a combined allocation facilitates effective joint decoding  of both the messages by the receiver. The details of designing such an allocation are  given in the next section.

\paragraph{Encoding} Each transmitter $i \in \{1,2\}$  first maps its message  to a message vector $\beta^{(i)} \in \mc{B}_{M,L_i}$, and then generates its codeword $x^{(i)} = A_i \beta^{(i)}$. The channel output sequence at the receiver is 
\be
y = A_1 x^{(1)} \, + \, A_2 x^{(2 )} \, + \, w = [A_1 \quad A_2] \begin{bmatrix}  \beta^{(1)}  
\\ \beta^{(2)} \end{bmatrix}  \, + \, w.   \label{eq:macop}
\ee

\paragraph{Optimal decoding}   The receiver jointly decodes  the two message vectors as 
\be    
 [ \hat{\beta}^{(1)}_{\textsf{opt}}, \, \hat{\beta}^{(2)}_{\textsf{opt}} ] \ = \,  \argmin_{\hat{\beta}^{(1)} \in \mc{B}_{M, L_1}, \, \hat{\beta}^{(2)} \in \mc{B}_{M, L_2} }  \, \norm{y^{(2)} - A_1\hat{\beta}^{(1)} -A_2\hat{\beta}^{(2)}  }^2.  
\label{eq:rx2_opt_mac}
\ee
Writing $L= L_1 + L_2$, $A= [A_1 \, A_2]$, and 
$\beta = \begin{bmatrix}  \beta^{(1)}  
\\ \beta^{(2)} \end{bmatrix}$,
an equivalent representation of the optimal decoder  is 
\be
\hat{\beta}_{\textsf{opt}} = \, \argmin_{\hat{\beta} \in \mcb} \, \norm{y - A \hat{\beta}}^2.
\ee
The first $L_1$ sections of $\hat{\beta}_{\textsf{opt}}$  form $\hat{\beta}^{(1)}_{\textsf{opt}}$, while the remaining sections form $\hat{\beta}^{(2)}_{\textsf{opt}}$.

\paragraph{AMP decoding} The receiver  runs a standard SPARC AMP decoding
routine (as described  in Section \ref{sec:AMP_dec}) on the concatenated design matrix $A = [A_1 \, A_2]$. The first $L_1$ sections of the decoded message vector constitute $\hat{\beta}^{(1)}$, and the next $L_2$ constitute $\hat{\beta}^{(2)}$.

\section{Power allocation for AMP decoding} \label{sec:MAC_PA}

As there is a single combined decoder, we first construct an overall power allocation with total power $P=P_1 +P_2$ for the concatenated SPARC $A = [A_1 \, A_2]$, which has $L=L_1+L_2$ sections and rate $R=R_1+R_2$. 

 This  overall power allocation  $\{  P_\ell \}_{\ell \in [L]}$ is constructed
 using the iterative technique described in Chapter \ref{chap:emp_perf} for point-to-point channels. Now, this power allocation must now partitioned between the two transmitters. We will use the fact that $\{  P_\ell \}_{\ell \in [L]}$  in non-increasing in $\ell$.
 
Transmitter $i \in \{1, 2 \}$ must be allocated precisely
$L_i$ sections, whose total power must be no more than $P_i$. Additionally,
we would like for any section errors to be fairly distributed between transmitters, so that
each transmitter experiences approximately the same error rate. We know from
Chapter~\ref{chap:emp_perf} that the majority of errors occur in sections towards
the end of the power allocation, where the power per section is lower and many sections
share the same power (the \emph{flat} region). Therefore we would like each
user to have the same proportion of their allocation be flat. To summarize,
the requirements for the power allocation are:

\begin{enumerate}
    \item Of the $L=L_1+L_2$ sections, we must allocate any $L_1$ sections to
        user 1, and the remaining $L_2$ to user 2.
    \item The sections allocated to user 1 must sum to no more than $P_1$,
        and those for user 2 to no more than $P_2$.
    \item Once the  conditions above are met, choose the solution which
        divides any equal power sections more equally between the two users.
\end{enumerate}

\begin{figure}[p]
    \captionsetup[subfigure]{justification=centering}
    \begin{subfigure}[b]{0.9\textwidth}
    \centering
    \includegraphics[width=0.45\textwidth]{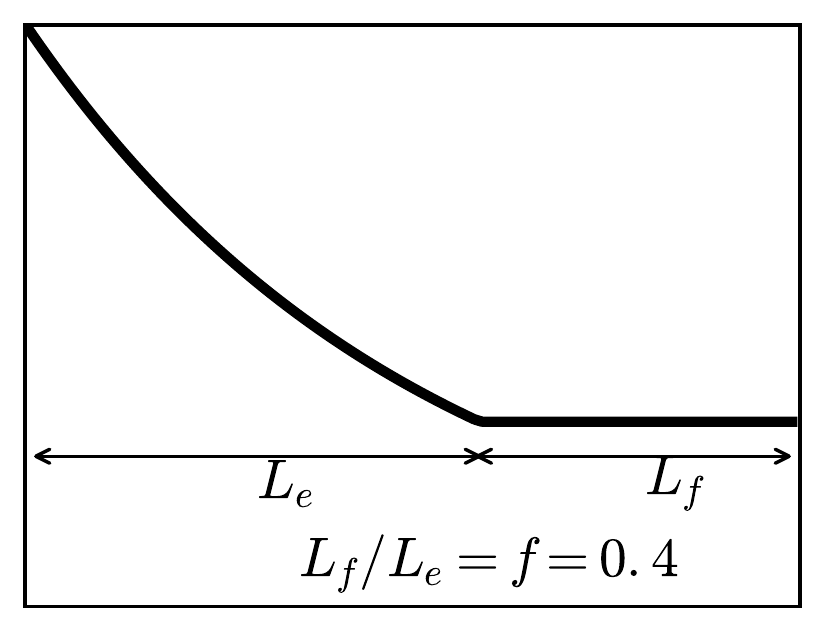}
    \caption{The overall power allocation to partition.The ratio of exponential
        to flat sections is denoted $f$, here with $f=0.4$.}
    \end{subfigure}

    \begin{subfigure}[b]{0.9\textwidth}
    \centering
    \includegraphics[width=0.45\textwidth]{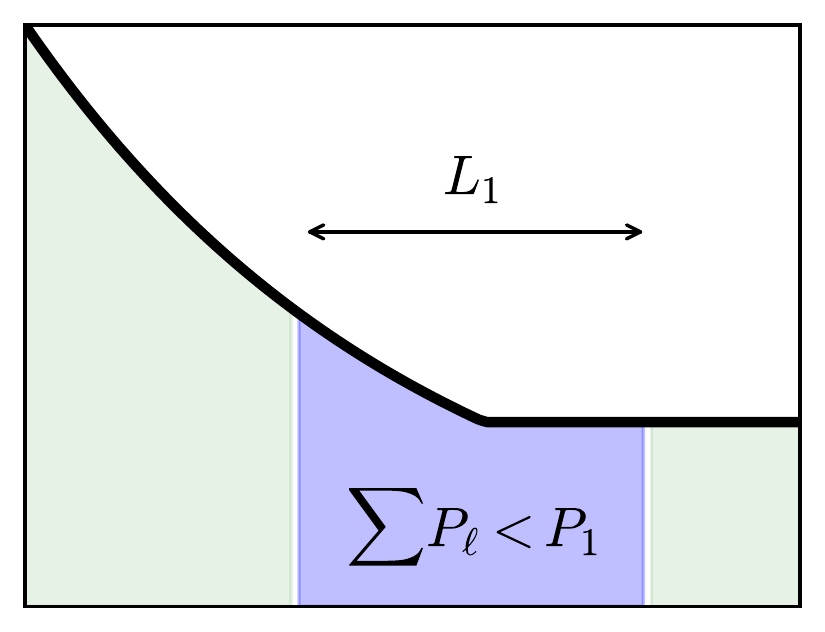}
    \includegraphics[width=0.45\textwidth]{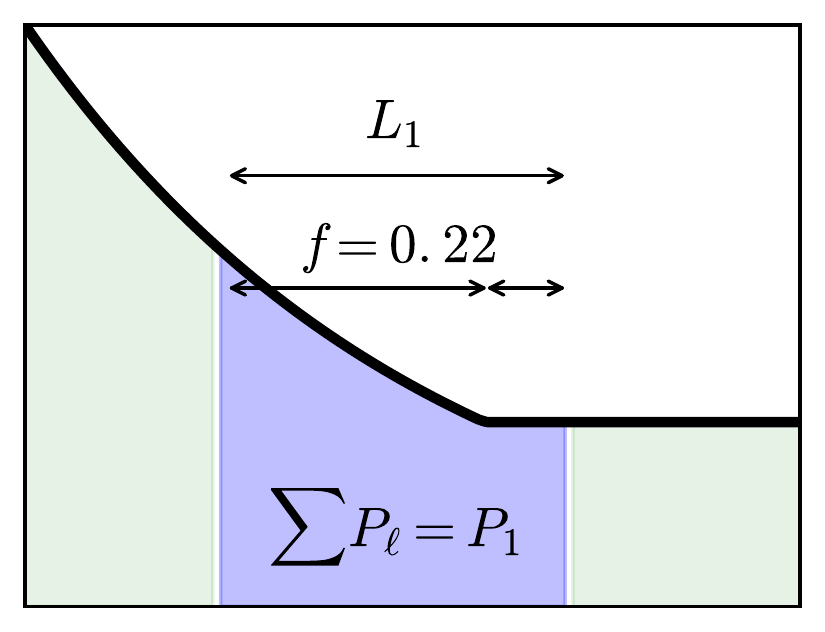}
    \caption{
        We first consider a bracket of size $L_1$, shown in blue.
        Our first attempt (on the left) contains too little power, so we move it
        leftwards until the contained power reaches $P_1$, shown on the right.
        At this position, $f=0.22$.
    }
    \end{subfigure}

    \begin{subfigure}[b]{0.9\textwidth}
    \centering
    \includegraphics[width=0.45\textwidth]{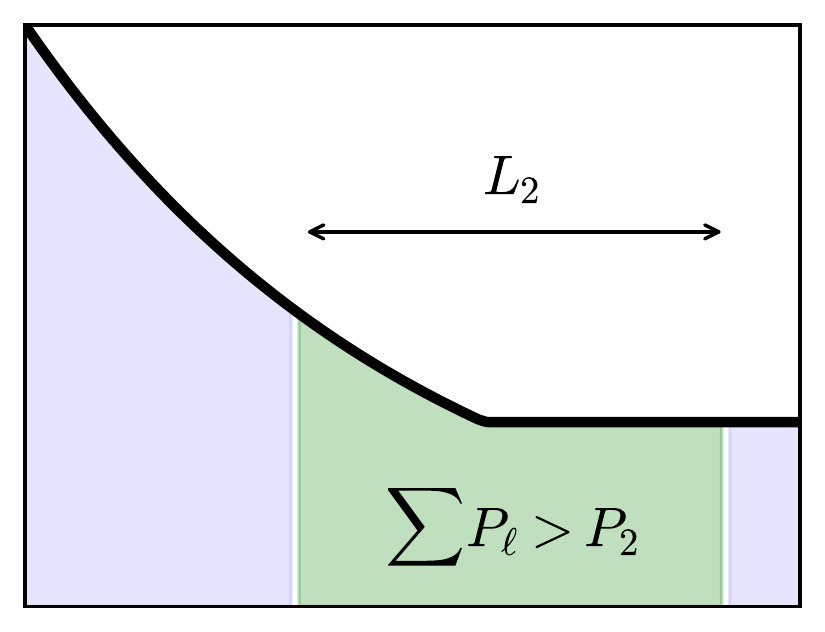}
    \includegraphics[width=0.45\textwidth]{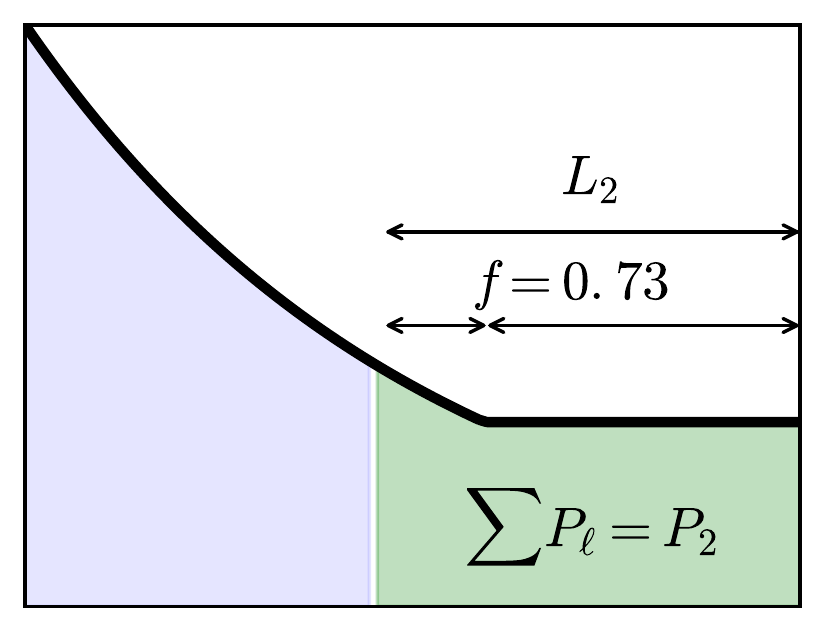}
    \caption{
        Next we consider a bracket of size $L_2$, shown in green.
        The first attempt contains too much power, so we move it
        rightwards until the contained power is $P_2$. In this position $f=0.73$.
        Since the $L_1$ bracket has an $f$ closer to that of the original
        allocation, we use that bracket for partitioning.
    }
    \end{subfigure}
\caption{Example of the power allocation partitioning strategy. Sections highlighted in blue are considered for user 1,
    and those in green for user 2.}
\label{fig:mu-mac-pa-cartoon}
\end{figure}

 The strategy for partitioning the power allocation is as follows. We locate a bracket of size either
$L_1$ or $L_2$ sections inside $\{P_\ell\}_{\ell\in[L]}$ such that its sum is
as close as possible to, without exceeding, $P_1$ or $P_2$ respectively, and
allocate the coefficients within the bracket to transmitter 1 or transmitter 2
as appropriate. The remaining sections on either side of the bracket are
allocated to the other transmitter.  The choice of bracket size (and thus of
which transmitter is allocated the bracketed coefficients) is determined by
which option gives the closest to optimal division of the coefficients from the
flat section.  This strategy is illustrated graphically in
Figure~\ref{fig:mu-mac-pa-cartoon}.

\begin{remark}
The partitioning method above can be applied to any overall power allocation $\{P_\ell\}_{\ell\in[L]}$, including the exponentially decaying one where $P _\ell \propto e^{-2 \mc{C} \ell/L}$ for $\ell \in [L]$. (Here $\mc{C}$ is the sum capacity $\frac{1}{2}\log(1 + P/\sigma^2)$, where $P=P_1+P_2$.) The decoding analysis is identical to that of a point-to-point AWGN channel with power constraint $P$ operating at rate $R=R_1 + R_2$. Therefore the AMP decoding result of Theorem \ref{thm:main_amp_perf} can be directly applied, and \eqref{eq:pezero}  bounds the probability of  the \emph{sum} of the section error rates of the two users exceeding some $\e >0$. This establishes that all rate pairs within the capacity region are achievable with an efficient AMP decoder.
\end{remark}

\section{Simulation results}

\begin{figure}[p]
    \centering
    \begin{tabular}{ll}
    \includegraphics[width=0.5\textwidth]{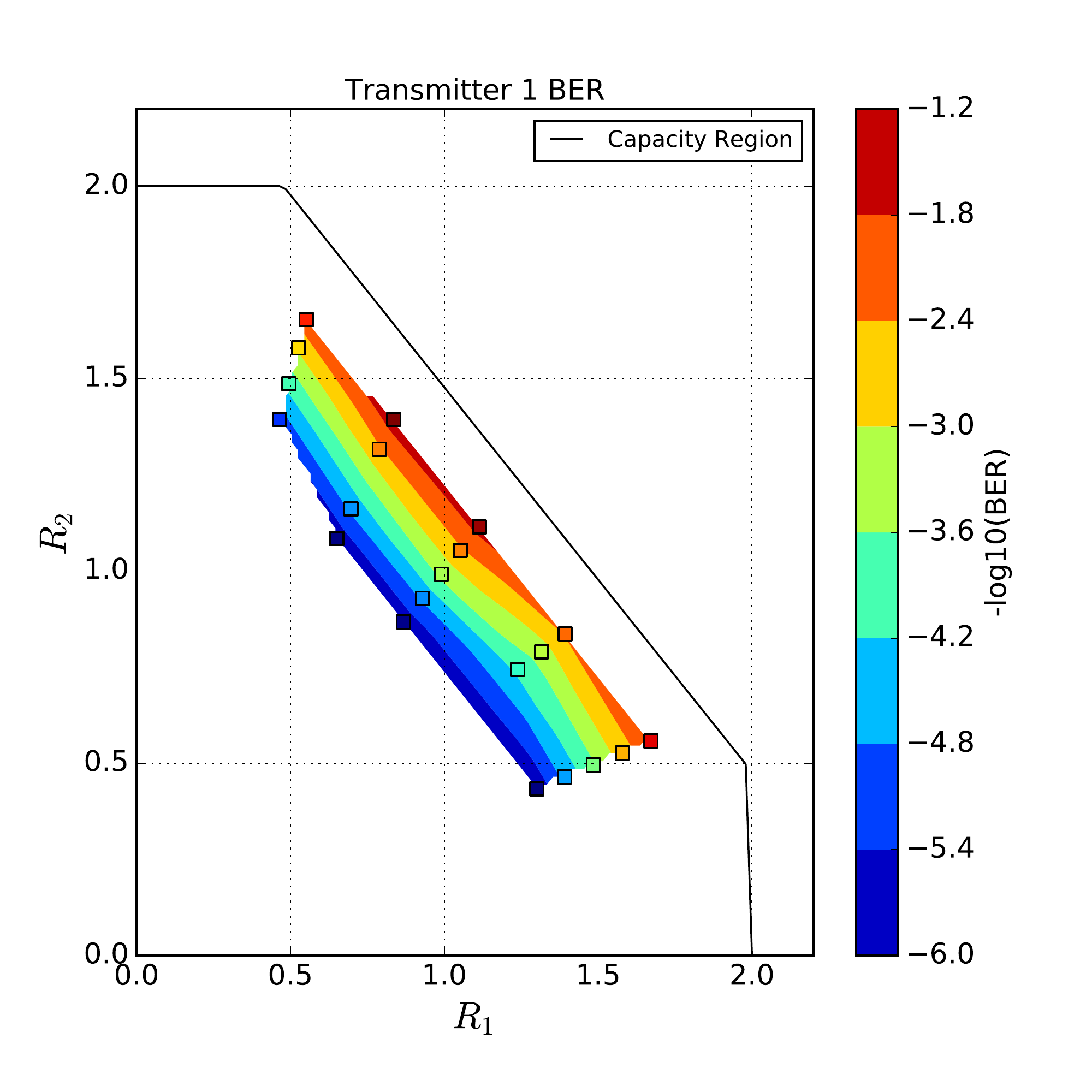}
        &
    \includegraphics[width=0.5\textwidth]{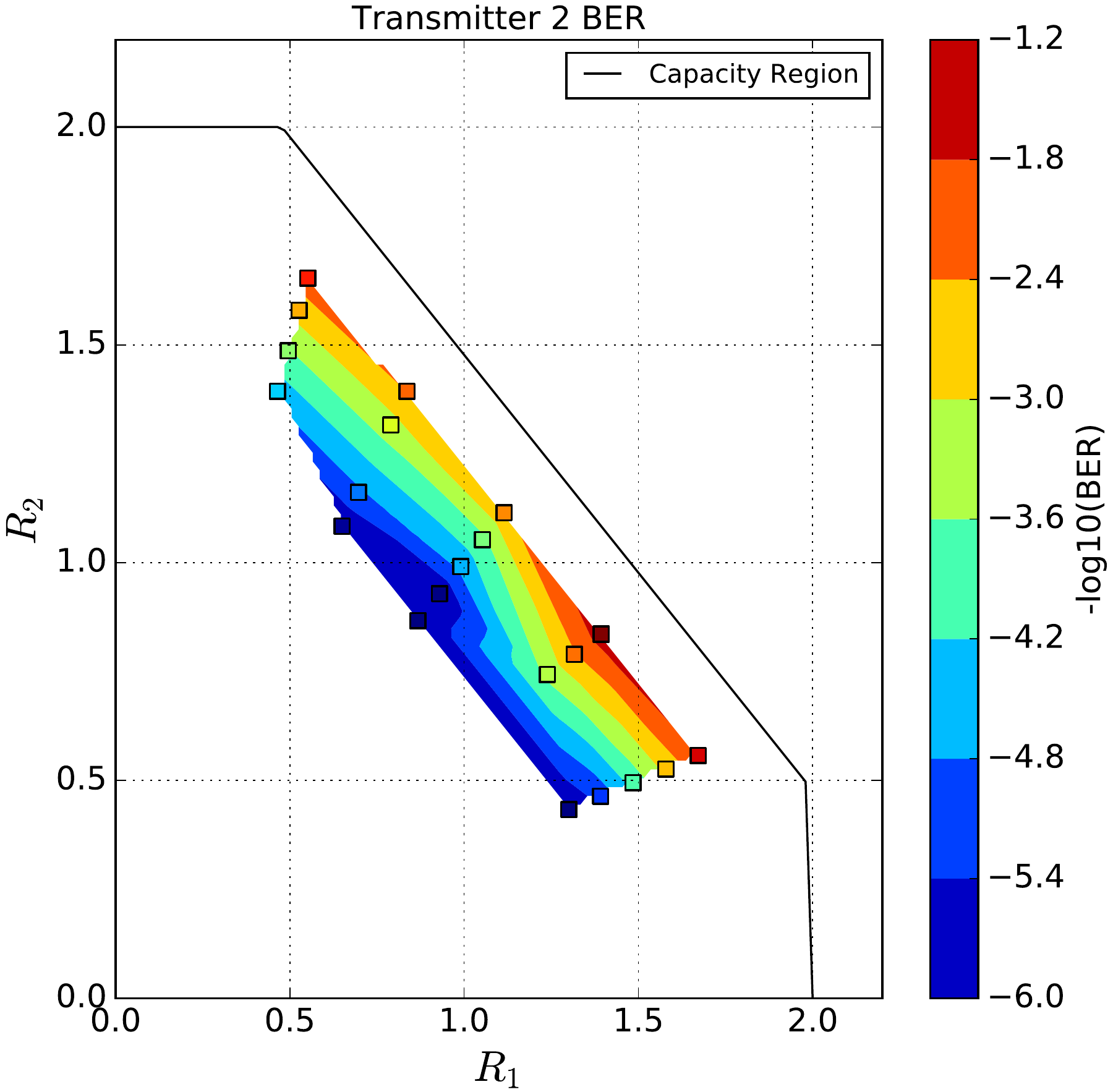}
    \end{tabular}
    \includegraphics[width=0.5\textwidth]{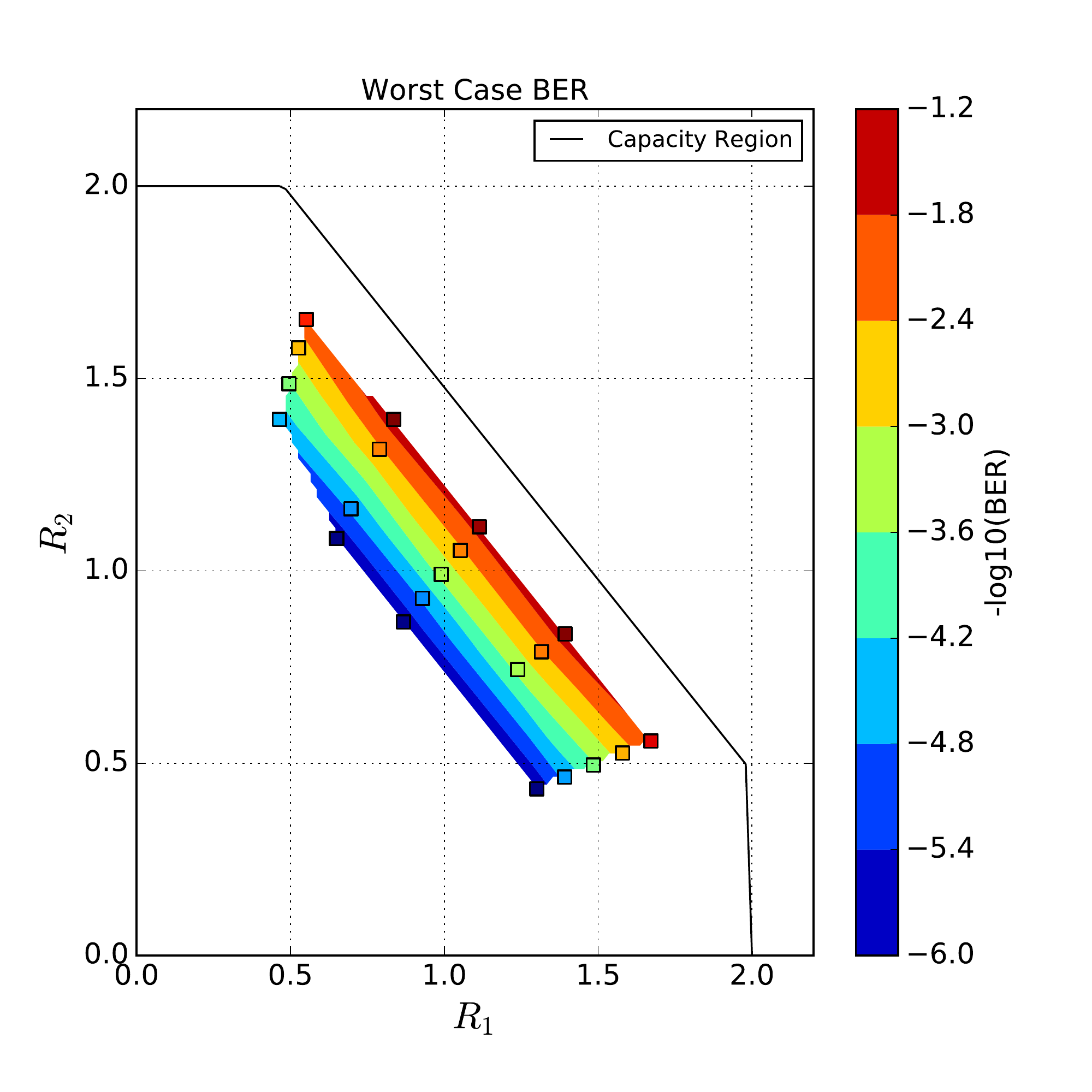}
    \caption{Bit error rates for multiple access channel simulations.
    Each displayed point is the average of approximately 30000 trials.
             $P_1=15$, $P_2=15$, $\sigma^2=1$, $M=512$, $n=4095$. }
    \label{fig:mac-results}
\end{figure}

We now discuss the empirical performance of SPARCs with AMP decoding for the Gaussian MAC, considering  a symmetric setup with $P_1=P_2=15$, and $\sigma^2=1$. The sum-capacity is $\mc{C} = \frac{1}{2} \log \left( 1+ \frac{P_1+P_2}{\sigma^2} \right)$  Each operating point is set as follows.  We fix
$\gamma\in[0,1]$ and then set $R_1+R_2=R=\gamma\mc{C}$. The parameter $\gamma$ represents the backoff from the sum rate limit. Next fix
$\alpha\in[0,1]$ which sets the share of this sum rate for each transmitter as 
$R_1=\alpha R$ and $R_2=(1-\alpha)R$.

For the SPARC design matrix,  first fix $M=512$ and $n=4095$. The parameters $L_1, L_2$ are then determined by the rate pair:   $L_1 = n R_1 / \log M$, and $L_2 = n R_2 / \log M$.  A power allocation is then designed using the sum rate $R$ and $L=L_1+L_2$, and partitioned according to the strategy described in the previous section.

Figure ~\ref{fig:mac-results} shows the bit error rates with AMP decoding for different values of $(R_1, R_2)$. The boundary of the capacity region is also shown.  After performing AMP decoding on the combined SPARC $A$, the number of sections decoded in error is counted separately for the first $L_1$ and then the last $L_2$ sections, and used to report section error rates for each transmitter. We also report the worst case section error rate between the two users.

As the experimental set-up is equivalent to a single point-to-point channel
with the same sum power and sum rate and power allocation, we expect to obtain
similar bit error performance to that scenario. The results support this, with
good bit error rates even reasonably close to capacity at all points in the
rate region. The worst case bit error rate is also close to each individual user's bit error rate,
showing that the power allocation partition is generally successful at
ensuring equal error rates between users.

%
\chapter{Communication and Compression  with Side Information}  \label{chap:sideinfo}
 
 Random binning is a key ingredient of optimal coding schemes for several models in multiuser information theory. In a scheme with binning, a large codebook is partitioned into equal-sized codebooks of smaller rate.   In this chapter, we demonstrate how random binning can be efficiently implemented using SPARCs for two canonical multiuser models: lossy compression with side information at the decoder (the Wyner-Ziv model \cite{WynerZiv}), and channel coding with state information at the encoder (the Gelfand-Pinsker model \cite{GelfandPinsker,CostaDP}). We will consider the Gaussian versions of these models, and aim to construct SPARC-based coding schemes  that achieve rates close to the information-theoretic optimum.   The proposed binning construction can be extended to other multiuser models such as  lossy distributed source coding where the best known achievable rates use a random binning strategy.

\begin{figure}[t]
    \captionsetup[subfigure]{justification=centering}
    \begin{subfigure}[b]{0.9\textwidth}
    \centering
    \includegraphics[width=0.7\textwidth]{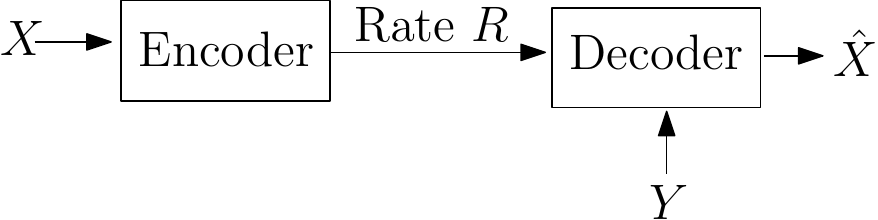}
    \vspace{10pt}
    \caption{Compressing $X$ with decoder side-information $Y$. \\ }
    \label{fig:wz}
    \end{subfigure}

       \vspace{20pt}
        
    \begin{subfigure}[b]{0.9\textwidth}
    \centering
    \includegraphics[width=0.7\textwidth]{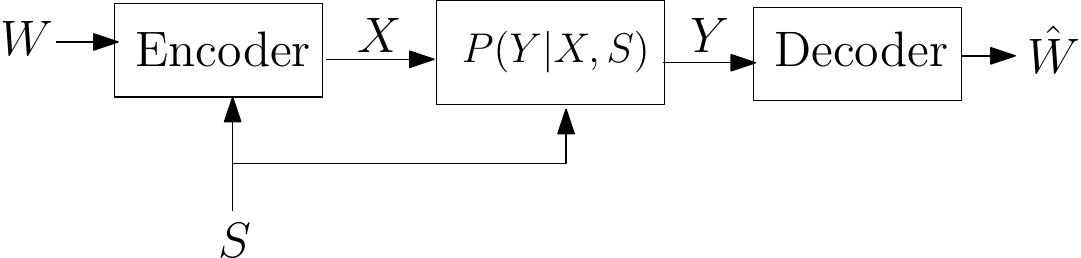}
        \vspace{10pt}
    \caption{Communicating a message $W$ over channel $P(Y|X S)$ with state $S$ known at the encoder.}
    \label{fig:gp}
    \end{subfigure}
   \label{fig:wz_gp}
\end{figure}  

  In the  Wyner-Ziv compression problem (Figure \ref{fig:wz}), given a rate $R >0$, the goal is to design a coding scheme that optimally uses the decoder side information $Y$  to reconstruct the source with minimal distortion. This model has been used in  distributed video coding \cite{Girod2,Prism07}.  In the Gelfand-Pinsker channel coding problem (Figure \ref{fig:gp}),  the encoder knows the  channel state $S$ non-causally (at the beginning of communication), while the decoder only receives the channel output $Y$.  The goal is to design a scheme that optimally uses the channel state information available at the encoder to achieve the maximal rate.  The Gelfand-Pinsker problem is the channel coding dual  of the Wyner-Ziv problem \cite{PradhanDuality,WornellDuality}, and has been used in multi-antenna communication \cite{SpencerCommMag}, digital watermarking \cite{MoulinK05,ChenW01} and steganography \cite{SulliStegan}.

   We will consider the Gaussian versions of the Wyner-Ziv and the Gelfand-Pinsker models, where the side information and the additive noise are independent Gaussian random variables.  An interesting feature of the Gaussian Wyner-Ziv and Gelfand-Pinsker problems is that the optimal rate in each case is the same as the setting where the side/state  information is available to both the encoder and the decoder \cite{WynerZiv,CostaDP}. 
   
   \paragraph{Previous code constructions} Several practical coding schemes have been proposed for the Wyner-Ziv problem, e.g.,  \cite{Discus03,GirodCoding1, XiongCoding1}. Recent constructions based on polar codes are the  first computationally efficient schemes that are provably rate-optimal  \cite{Polarrd,Polarwzgp}. However, the polar coding  constructions are only applicable to problems where the source and side-information distributions are discrete and symmetric.    For the Wyner-Ziv problem with continuous source and side-information distributions, elegant coding schemes such as  those based on lattices \cite{Zamir02, ErezLZ05} have been proposed. But these generally have exponential encoding and decoding complexity. Constructions using nested lattice codes have also been used for  the Gaussian Gelfand-Pinsker model (dubbed `writing on dirty paper')  \cite{Zamir02,ErezSZ05,ChenW01}. Computationally efficient code designs for this problem have been proposed in several works, e.g., \cite{XiongDPCoding,EreztenBrink}.

\section{Binning with SPARCs} \label{sec:sparc_bin}

We now describe the SPARC binning construction, which is applied to the Wyner-Ziv and Gelfand-Pinsker problems in the next two sections.  For any pair of rates $R_1, R$ with $R_1 > R$, we would like to divide a rate $R_1$ SPARC with block length $n$ into $e^{nR}$ bins. Each bin corresponds to a rate $(R_1-R)$ SPARC with the same block length $n$. This is done as follows. 

\begin{figure}[t]
\centering
\includegraphics[width = 0.9\textwidth]{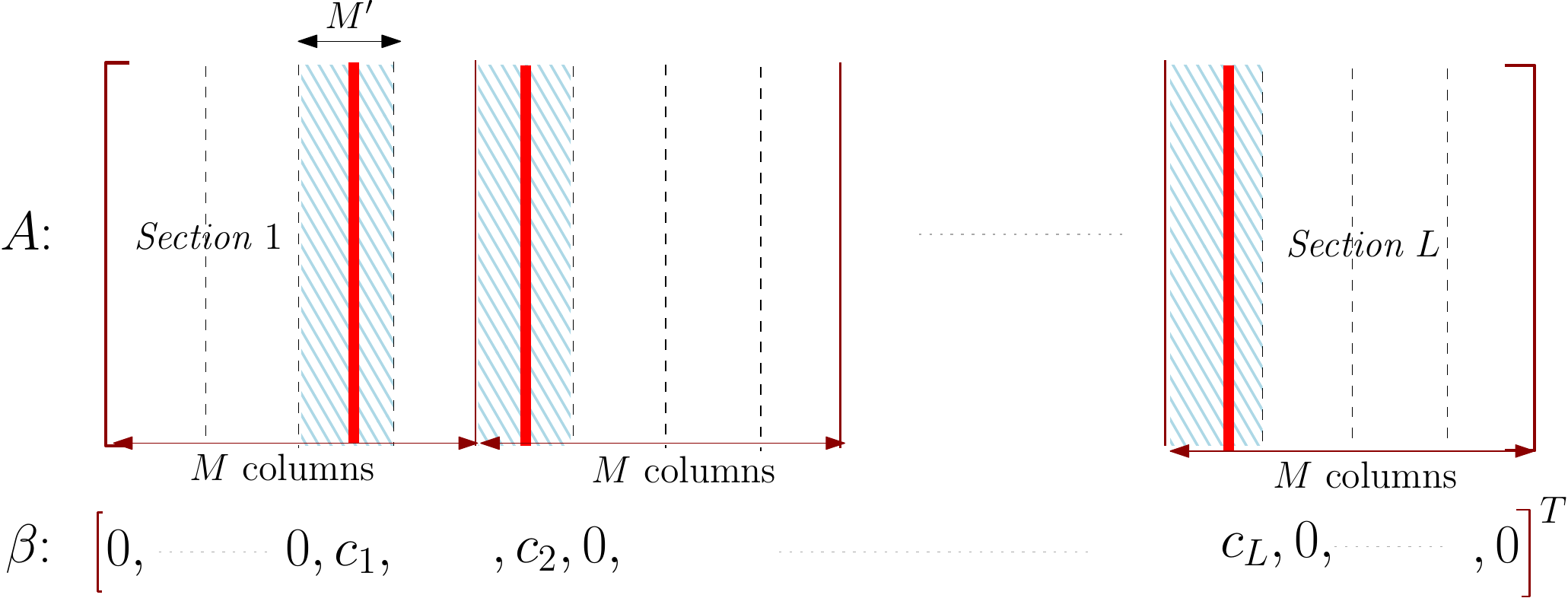}
\caption{\small Codebook binning using SPARCs}
\label{fig:sparc_binning}
\end{figure}

Fix  the parameters $M, L, n$  of  the design matrix $A$ such that \be L \log M = nR_1.  \label{eq:LMR1_def} \ee  As shown in Figure \ref{fig:sparc_binning}, divide each section of $A$ into sub-sections consisting of $M'$ columns each. Each bin is a smaller SPARC defined via a \emph{sub-matrix} of $A$ formed by picking one {sub-section} from each section. For example, the collection of shaded sub-sections in the figure together forms one bin.  If $M'$ is chosen such that 
\be
n(R_1 -R) = L \log M',
\label{eq:mpr_def}
\ee
then each sub-matrix defines a  bin that is a rate $(R_1-R)$ SPARC with parameters $(n, L, M')$. Since we have $(M/M')$ sub-section choices in each of the $L$ sections, the total number of bins that can be chosen is $\left(M/M'\right)^L$.
Combining \eqref{eq:LMR1_def} and \eqref{eq:mpr_def}, we have
\be
L \log \frac{M}{M'} =n R, \quad \text{ or } \quad \left( \frac{M}{M'} \right)^L = e^{nR}
\ee
Therefore, the number of bins is $e^{nR}$, as required. The binning structure mimics that of the SPARC codebook, where two codewords may share up to $(L-1)$  columns.  Similarly, two bins may share as many as $(L-1)$ sub-sections of $A$. 

\section{Wyner-Ziv coding with SPARCs} \label{sec:wz_sparcs}

Consider the model in Figure \ref{fig:wz}, with $X \sim \mc{N}(0, \sigma^2)$ being an i.i.d. Gaussian source to be compressed with mean-squared distortion $D$. The decoder side-information $Y$ is noisy version of $X$  and is related to $X$ by $Y=X + Z$, where $Z \sim \mc{N}(0,N)$ is independent of $X$.  Let the target distortion be $D < \text{Var}(X|Y)$. (Distortions greater than this value can be achieved with zero rate, by simply estimating $X$ from $Y$.) 

Let $x,y \in \reals^n$ be the source and side information sequences drawn i.i.d. according to the distributions of $X$ and  $Y$, respectively.  The sequence $y$ is available at the decoder non-causally.  If $y$ were available at the encoder as well, then an optimal strategy is to compress $(y-x)$ with a rate high enough to ensure reconstruction of $x$ within distortion $D$. The minimum rate required for this strategy is $\frac{1}{2} \log \frac{\text{Var}(X|Y)}{D}$ nats/sample. The result of Wyner and Ziv  \cite{WynerZiv} shows that this rate is achievable even when $y$ is available only at  the decoder.

We first review the main ideas in the Wyner-Ziv random coding scheme. Define an auxiliary random variable $U$  jointly distributed with $X$ according to
\be \label{eq:uxv} 
U = X  + V, 
\ee 
where $V \sim \mc{N}(0, Q)$ is independent of $X$.  We choose  
\be 
Q = 
\left( \frac{1}{D} - \frac{1}{\text{Var}(X|Y)} \right)^{-1}. 
\label{eq:Qdef0}
\ee 
One can invert this test channel to write
\be \label{eq:uxv_rev} 
X = \frac{\sigma^2}{\sigma^2+Q} \, U  + V' = U' + V', 
\ee 
where $V' \sim \mc{N}(0, \frac{\sigma^2 Q}{\sigma^2+Q})$ is independent of 
$U' \sim \mc{N}(0, \frac{\sigma^4}{\sigma^2+Q})$.

In the Wyner-Ziv scheme \cite{WynerZiv}, $x$ is first quantized to a codeword $u'$ within distortion $\frac{\sigma^2 Q}{\sigma^2+Q}$, using a rate-distortion codebook of rate $R_1$. The sequences in this codebook, say $\{ u'(1), \ldots, u'(e^{nR_1}) \}$,  are drawn $\sim_{i.i.d.} \mc{N}(0, \frac{\sigma^4}{\sigma^2+Q})$.
 Shannon's rate-distortion theorem guarantees that with high probability  a codeword will be found within the required distortion if 
\be
R_1 >  I(X; U') = I(X;U) = \frac{1}{2} \log \left(1 + \frac{\sigma^2}{Q} \right),
\label{eq:R1_cond}
\ee
where the mutual information has been calculated using the test channel  in  \eqref{eq:uxv_rev}. The index corresponding to the chosen codeword $u'$ is not sent to the decoder directly. Instead, the size $e^{nR_1}$ rate-distortion codebook is divided into $e^{nR}$ equal-sized bins, with a uniformly random assignment of sequences to bins. The encoder only sends  the index of the bin containing $u'$, which requires a rate of $R$ nats/sample .

 The decoder's task is to recover $u'$ using the bin index and the side information $y$. This is equivalent to a \emph{channel decoding} problem, where the effective channel is 
 \be Y=X +Z  = U' + V' +Z.  \label{eq:eff_channel} \ee
 The $\snr$ of the effective channel is 
 \be
 \frac{ \text{Var}(U')}{\text{Var}(V') + \text{Var}(Z)} =  \frac{\sigma^4}{\sigma^2 Q + (\sigma^2 + Q)N}.
 \ee
 The decoder will succeed with high probability if the number of sequences within the bin, $e^{n(R_1-R)}$,  is less than $e^{nI(U';Y)}$. We therefore need
\be  R_1 -R  < I(U';Y) = \frac{1}{2} \log \left( 1 +  \frac{\sigma^4}{\sigma^2 Q + (\sigma^2 + Q)N} \right). \label{eq:R2_cond}
\ee

Combining the conditions \eqref{eq:R1_cond} and \eqref{eq:R2_cond}, we conclude that both the encoding and the decoding steps are successful with high probability if 
\be
R > I(X;U') - I(Y;U') = \frac{1}{2} \log \frac{\text{Var}(X|Y)}{D},
\label{eq:Rwz_cond}
\ee
where the last equality follows by substituting the value for $Q$ from \eqref{eq:Qdef0}. For any $R$ satisfying \eqref{eq:Rwz_cond}, $R_1 > R$ can be chosen to satisfy both \eqref{eq:R1_cond} and \eqref{eq:R2_cond}.

The final step at the decoder is produce the reconstructed sequence $\hat{x}$ as the MMSE estimate of $x$ given $u'$ and $y$:
\be
\hat{x} = \left(\frac{1}{Q} + \frac{1}{\sigma^2} + \frac{1}{N} \right)^{-1} \left( \frac{u'}{Q} + \frac{y}{N}  \right).
\ee
Since the sequence triple $(x,y, u')$ is jointly typical according to the joint distribution of the random variables $(X,Y, U')$,  the expected distortion $\expec \abs{x-\hat{x}}^2$ can be made arbitrarily close to $D$ for sufficiently large $n$.

We now implement the coding scheme using a  SPARC with optimal (least-squares) encoding and decoding.  For a given target distortion $D$, choose $R > \frac{1}{2} \log \frac{\text{Var}(X|Y)}{D}$, and choose $R_1$ such that 
\be 
\frac{1}{2} \log \left(1 + \frac{\sigma^2}{Q} \right) < R_1 <  R +  \frac{1}{2} \log \left( 1 +  \frac{\sigma^4}{\sigma^2 Q + (\sigma^2 + Q)N} \right),
\label{eq:R1_lohi}
\ee
so that both \eqref{eq:R1_cond} and \eqref{eq:R2_cond} are satisfied. 

Fix block length $n$, and consider a SPARC defined via an $n \times ML$ design matrix $A$ with $M=L^b$ and $bL \log L = nR_1$, where $b$ is greater than the minimum value specified by Theorem \ref{thm:err_exp_rd}. The entries of $A$ are $\sim_{i.i.d.} \mc{N}(0, \frac{1}{n})$, and the non-zero entries each $\beta \in \mcb$ are all set to $\sqrt{\frac{n}{L} \frac{\sigma^4}{(\sigma^2 +Q)}}$.

 With $M'$ be determined by $M'^L = e^{n(R_1 -R)}$, partition each section of $A$ into sub-sections of  $M'$ columns each, as shown in Figure \ref{fig:sparc_binning}.

\emph{Encoding}: The encoder determines the SPARC codeword that is closest in Euclidean distance to the source sequence $x$. Let 
\be  \beta^* = \underset{\beta \in \mcb}{\argmin} \ \norm{x- A \beta}^2. \label{eq:sparc_wz_enc}\ee
The encoder sends the decoder a message $W  \equiv (p_1,\ldots,p_L)$, where $p_i \in \{1, \ldots, \frac{M}{M'} \}$  indicates the subsection in the $i$th section of $A$ where
$\beta^*$ contains a non-zero element.

\emph{Decoding}: Let $A_{\textsf{bin}(W)}$ be  the $n \times M'L$ sub-matrix corresponding to the subsections of $A$ corresponding to  the bin index $W$.  Recalling from Section \ref{sec:sparc_bin} that $A_{{\sf{bin}}(W)}$ defines an $n \times M'L$ SPARC design matrix, the decoder determines 
\be
\hat{\beta} = \underset{\beta \in \mc{B}_{M',L}}{\argmin} \ \norm{y -  A_{{\sf{bin}}(W)}\beta}^2.
\ee
Letting $\hat{u} = {A}_{{\sf{bin}}(W)}\hat{\beta}$,  the source sequence is reconstructed as
\be \hat{x}  = \left(\frac{1}{Q} + \frac{1}{\sigma^2} + \frac{1}{N} \right)^{-1} \left( \frac{\hat{u}}{Q} + \frac{y}{N}  \right).  \ee

\emph{Analysis}: The SPARC rate-distortion result in Theorem \ref{thm:err_exp_rd} guarantees that the encoding will succeed with high probability. That is, for $R_1$ exceeding the lower bound in \eqref{eq:R1_lohi}, the probability of the event 
\[ 
\left\{ \frac{1}{n} \norm{x- A \beta^*}^2 > \frac{\sigma^2 Q}{\sigma^2+Q}  \right\}
\]
is exponentially small in $n$. On the decoder side, the effective channel is given by  \eqref{eq:eff_channel}, with the task being to recover the codeword $u'$ with $v'+z$ treated as a noise sequence. If we assume that  $v' + z$ is distributed as  
$\sim_{i.i.d.} \mc{N}(0, \text{Var}(V') + \sigma^2)$, then the SPARC channel coding result Theorem \ref{thm:MLresult} guarantees reliable decoding with exponentially small error probability. However, this assumption is not true at finite code lengths.   Indeed, recall that $v=(x-u')$ is the distortion (quantization noise) incurred at the encoder.  Therefore, $v'$ may be dependent on the codeword $u'$; moreover, its distribution will not be exactly Gaussian. 

Though we expect that the distribution of $(v' + z)$ will be asymptotically independent of $u'$ and converge to Gaussian, a careful analysis is needed to rigorously prove that SPARCs achieve the Gaussian Wyner-Ziv rate-distortion bound. This remains an open question for future work. 

\section{Gelfand-Pinsker coding with SPARCs} \label{sec:gp_sparcs}

Consider the model in Figure \ref{fig:gp}, with the channel law $P(Y|X,S)$ given by 
\be
Y=X+ S + Z.
\ee
The state variable $S \sim \mc{N}(0, \sigma^2_s)$ is independent of the additive noise $Z \sim \mc{N}(0,N)$. There is an average power constraint $P$ on the input sequence $x \in \reals^n$. The state sequence $s \in \reals^n$  is  $\sim_{i.i.d.} \mc{N}(0, \sigma^2_s)$, and  is known non-causally at the encoder.

Due to the power constraint, the encoder cannot simply cancel out the effect of the  state sequence $s$ using the codeword.   In Costa's capacity-achieving scheme \cite{CostaDP},  one part of the state sequence $s$ is used by the encoder to produce the input sequence $x$, and the remaining part  is treated as noise. This is done as follows.

Define an auxiliary random variable $U$ as
\be
U = X+ \alpha S,
\label{eq:UdefAlph}
\ee
where $X \sim \mc{N}(0, P)$ is independent of $S \sim \mc{N}(0, \sigma^2_s)$, and
\be
\alpha = \frac{P}{P+N}.
\ee
Inverting the test channel $U$, we write 
\be
S =  \frac{\alpha \sigma_s^2}{P + \alpha^2 \sigma_s^2} U + X' = U' +X'
\label{eq:Sx'_test_ch}
\ee
where $U' \sim \mc{N}(0,  \frac{\alpha^2 \sigma_s^4}{P + \alpha^2 \sigma_s^2}  )$ and $X' \sim \mc{N}(0, \frac{P \sigma_s^2}{ P + \alpha^2 \sigma_s^2})$ are independent. 

In Costa's scheme \cite{CostaDP}, we first construct a random codebook of rate $R_1$,  with sequences $\{ u'(1), \ldots $ $, u'(e^{nR_1}) \}$  drawn $\sim_{i.i.d.} \mc{N}(0, \frac{\alpha^2 \sigma_s^4}{P + \alpha^2 \sigma_s^2} )$. The  codebook is partitioned into $e^{nR}$ bins, with each bin containing $e^{n(R_1 - R)}$ codewords. 

 To transmit message $W \in \{1, \ldots, e^{nR} \}$, the encoder finds a  codeword $u'$ \emph{inside bin $W$} of the codebook that is within distortion $\frac{P \sigma_s^2}{ P + \alpha^2 \sigma_s^2}$ of the state sequence $s$. From Shannon's rate-distortion theorem, such a codeword will be found with high probability if
\be
R_1 -  R > I(S; U') = \frac{1}{2} \log \left( 1  + \frac{ \alpha^2 \sigma_s^2}{P} \right).
\label{eq:R_gpcond}
\ee

The transmitted sequence $x$ is then determined as 
\be
x = \left(  \frac{P + \alpha^2 \sigma_s^2}{\alpha \sigma_s^2} \right) u'  - \alpha  s.
\label{eq:xtran_GP}
\ee
Since the empirical joint distribution of $(u',s)$ is close to that specified by the test channel \eqref{eq:Sx'_test_ch}, it can be verified that the second moment of $x$ will be close to $P$ with high probability. 

Using the test channel in \eqref{eq:Sx'_test_ch}, the channel input-output relationship can be expressed as 
\be
\begin{split}
Y  & = X + S +  Z  \\ 
& =   \left(  \frac{P + \alpha^2 \sigma_s^2}{\alpha \sigma_s^2} \right) U'  + (1- \alpha)  S + Z \\
& = \left( \frac{P + \alpha^2 \sigma_s^2}{\alpha \sigma_s^2} + 1- \alpha \right)U' + (1- \alpha) X' + Z,
\end{split}
\label{eq:eff_testch}
\ee
where $U', X', Z'$ are mutually independent. 

The decoder's task is to determine the codeword $u'$ from the output sequence $y$.  The index of the bin containing the decoded codeword gives the reconstructed message $\hat{W}$.  Assuming that the encoding has been successful, the empirical joint distribution of $(u',y)$ will be close to that of the effective  channel in \eqref{eq:eff_testch}. The effective signal to noise ratio of this channel is  
\be
\snr_{\sf{eff}} =  \frac{ \frac{(P + \alpha \sigma_s^2)^2}{\alpha^2 \sigma_s^4} \text{Var}(U')}{(1-\alpha)^2\text{Var}(X') + \text{Var}(Z)}
= \frac{(P + \alpha \sigma_s^2)^2}{ (1-\alpha)^2 P \sigma_s^2 + N(P + \alpha^2 \sigma_s^2)}.
\ee
The decoding step will be successful with high probability if 
\be
\begin{split}
R_1 < I(U'; Y) & = \frac{1}{2}\log(1 + \snr_{\sf{eff}})  \\ 
& = \frac{1}{2} \log \left( \frac{(P + \alpha^2 \sigma_s^2 )(\sigma_s^2 +P+N)}{ (1-\alpha)^2 P \sigma_s^2 + N(P + \alpha^2 \sigma_s^2)}  \right).
\end{split}
\label{eq:R1_gpcond}
\ee

Combining \eqref{eq:R_gpcond} and \eqref{eq:R1_gpcond}, we conclude that both the encoding and the decoding steps are successful with high probability if 
\be
R < I(S; U') - I(U';Y) = \frac{1}{2} \log \left( 1 + \frac{P}{N} \right),
\label{eq:Ropt_gp}
\ee
where the last equality follows by substituting $\alpha = \frac{P}{P+N}$ and simplifying.  For any $R$ satisfying \eqref{eq:Ropt_gp}, $R_1$ should be chosen large enough to satisfy \eqref{eq:R1_gpcond}.

We now implement the above coding scheme using a SPARC with optimal (minimum distance) encoding and decoding. Fix block length $n$, and consider a SPARC defined via an $n \times ML$ design matrix $A$ with $M=L^b$ and $bL \log L = nR_1$, where $b$ is greater than the minimum value specified by Theorem \ref{thm:err_exp_rd}. The entries of $A$ are $\sim_{i.i.d.} \mc{N}(0, \frac{1}{n})$, and the non-zero entries each $\beta \in \mcb$ are all set to $\sqrt{\frac{n}{L} \frac{\alpha^2 \sigma_s^4}{(P + \alpha^2 \sigma_s^2)}}$.

 With $M'$ be determined by $M'^L = e^{n(R_1 -R)}$, partition each section of $A$ into sub-sections of  $M'$ columns each, as shown in Figure \ref{fig:sparc_binning}.  This defines $e^{nR}$ bins, one for each message.

\emph{Encoding}: To transmit message $W \in [e^{nR}]$, the encoder determines the SPARC codeword within bin $W$ that is closest to the state sequence $s$.  Denoting by $A_{{\sf{bin}}(W)}$ the sub-matrix of $A$ corresponding to bin $W$, the encoder computes
\[ \beta^* = \underset{\beta \in \mathcal{B}_{M',L} }{\argmin} \ \norm{s- A_{{\sf{bin}}(W)} \beta}^2. \]
Following \eqref{eq:xtran_GP}, the transmitted sequence is 
\[
x=  \left(  \frac{P + \alpha^2 \sigma_s^2}{\alpha \sigma_s^2} \right) A_{{\sf{bin}}(W)} \beta^*  - \alpha  s.
\]

\emph{Decoding}: The decoder determines the codeword in the big rate $R_1$ SPARC closest to the output sequence $y$. It computes
\be
\hat{\beta} = \underset{\beta \in \mc{B}_{M,L}}{\argmin} \ \norm{y -  A\beta}^2.
\ee
The index of the bin containing $\hat{\beta}$ is the decoded message $\hat{W}$.

\emph{Analysis}: The SPARC rate-distortion result in Theorem \ref{thm:err_exp_rd} guarantees that encoding will succeed with high probability provided $R_1-R$ satisfies \eqref{eq:R_gpcond}.  However the analysis of the decoder is challenging, for reasons similar to those for Wyner-Ziv SPARCs. In the effective channel \eqref{eq:eff_testch}, we cannot assume that  the  noise $(1- \alpha)x' +z$ is exactly i.i.d. Gaussian  and independent of the codeword $u'$.  This is because $x'=s-u'$, the quantization noise incurred at the encoder, cannot be assumed to be independent of $u'$ and Gaussian. 

In summary, a rigorous analysis of the SPARC coding schemes for the Gaussian Gelfand-Pinsker and Wyner-Ziv models remains open. Another important open problem is to construct feasible SPARC coding schemes that achieve the Shannon limits for these models. The challenges in designing such schemes, and some ideas to address them, are discussed in the final chapter.

%
%
\chapter{Open Problems and Further Directions}  \label{chap:summary}   

In the first nine chapters, we described how sparse regression codes can be used for channel coding, lossy source coding, and  multi-terminal versions of these problems. 
We conclude the monograph with a discussion of  open questions and directions for further work.

\section{Channel coding with SPARCs}

\paragraph{Gap from capacity} For fixed value of decoding error probability, how fast can the gap from capacity  $(\mc{C} - R)$ shrink with growing block length $n$?   With optimal decoding,  the result in Chapter \ref{chap:AWGN_opt}  implies  that the gap from capacity for SPARCs is of  $O(\frac{1}{\sqrt{n}})$, which is order-optimal (from the results in \cite{ppv10}). In contrast, the  feasible decoders  in Chapter \ref{chap:AWGN_eff} all have much larger gap from capacity: even with optimized power allocations, the gap to capacity is no smaller than $O(\frac{\log \log n}{\log n})$. 

A key open question is: can one achieve polynomial gap to capacity using SPARCs with feasible decoding?  For polar codes over binary input symmetric channels, This question was recently answered in the affirmative \cite{GuruXia15polar}.  Achieving a polynomial gap from capacity for SPARCs may require new constructions such as spatially coupled SPARCs with power allocation as well as new decoding algorithms.

\paragraph{Analysis of sub-Gaussian and Hadamard-based SPARCs} In Chapter \ref{chap:AWGN_opt}, we analyzed the optimal decoder for SPARCs defined via i.i.d. Gaussian and Bernoulli dictionaries. An interesting direction is to generalize these results to dictionaries with arbitrary i.i.d. sub-Gaussian random variables. The key part of the analysis in the Gaussian case involves controlling the moment generating function of the difference between the  squares of two Gaussian random variables. Extending the analysis to sub-Gaussian dictionaries  would involve reworking the proof  of Proposition \ref{prop:nonasympML} to replace the steps using Gaussian-specific results  with the appropriate results for sub-Gaussians. 

A more difficult open question is to analyze SPARC dictionaries defined via partial  Hadamard or Fourier matrices.  As described in Chapter \ref{chap:emp_perf}, these structured matrices significantly reduce decoding and storage complexity, and are therefore important for practical implementation of SPARCs. 

The analysis of feasible SPARC decoders with non-Gaussian dictionaries is more challenging than that of   optimal decoding, and remains open even for  Bernoulli dictionaries.

\paragraph{Coded modulation using SPARCs}  The empirical results in Section \ref{sec:ldpc-outer} show that concatenating an outer LPDC code with an inner SPARC can produce a steep drop in error rates, when the $\snr$ exceeds a threshold. This waterfall behaviour was obtained using an off-the-shelf LDPC code. It appears likely that one can further improve the error performance (i.e., obtain a smaller threshold) by jointly optimizing the design of the outer LDPC code and inner SPARC using an EXIT chart analysis, similar to \cite{ten2004design}.

\paragraph{SPARCs for general channel models} The recent work of Barbier et al. \cite{barbITW16,BarbierDM17} extends the spatially coupled SPARC construction to general memoryless channels. The idea is to apply a symbol-by-symbol mapping to each Gaussian SC-SPARC codeword $A\beta$ to produce a codebook  whose input alphabet and distribution are tailored to the channel. The message vector $\beta$ is recovered from the channel output sequence using generalized AMP (GAMP) \cite{Rangan11} decoding.  The fixed points of the state evolution recursion  are analyzed using the potential function method in \cite{BarbierDM17}. This analysis predicts that the GAMP decoder is asymptotically capacity achieving for a general memoryless channel. 

The results in \cite{BarbierDM17} suggest two directions for further work. One is to extend the analysis of spatially coupled SPARCs to general memoryless channels with GAMP decoding, and rigorously prove that the coding scheme is capacity achieving. It would also be interesting to study the empirical performance of the SPARC coding scheme over commonly studied memoryless channels such as the binary symmetric channel, and compare the performance with that of capacity achieving codes such as  LDPC and polar codes. 

Another interesting direction is to generalize SPARC coding schemes originally designed for AWGN channels to  Gaussian fading channels and MIMO channels, which are important in wireless communication \cite{tse2005fundamentals,goldsmith2005wireless}.

\section{Lossy compression with SPARCs}

\paragraph{Gap from optimal rate-distortion function} With optimal encoding, the results in Chapter \ref{chap:comp_opt}  imply that for a given distortion level $D$,  the  SPARC rate $R$ should be $O(\frac{1}{\sqrt{n}})$ higher than $R^*(D)$ (the optimal Gaussian rate-distortion function) in order to ensure a fixed probability of excess distortion. In contrast, the successive cancellation encoder described in Chapter  \ref{chap:comp_eff_enc} can only achieve rates that are $O(\frac{\log \log n}{\log n })$ above $R^*(D)$ (see Sec. \ref{sec:gapDR}). 

A key open problem is  to design a feasible SPARC encoder with a gap from $R^*(D)$ that is of smaller order.  One idea is to investigate algorithms  that process multiple sections at a time and use soft-decision updating, instead of encoding one section at a time. Another idea to improve the gap from $R^*(D)$ is to  use spatially coupled SPARCs for lossy compression. Doing so would also yield  coding  schemes with spatially coupled SPARCs for multiuser models that require random binning.

\paragraph{Sub-Gaussian and structured dictionaries}  The compression performance of sub-Gaussian or Hadamard-based SPARC design matrices is empirically very similar to that of i.i.d. Gaussian design matrices. 
As in the channel coding case, Hadamard-based designs significantly reduce both the encoding complexity and the memory required.  However, there are no theoretical results for compression with non-Gaussian dictionaries. It would of interest to establish  performance guarantees for compression with such dictionaries, both with optimal encoding and with successive cancellation encoding.

\paragraph{Lossy compression of general sources} We expect that the results for SPARC compression of i.i.d. Gaussian sources in Chapters \ref{chap:comp_opt} and \ref{chap:comp_eff_enc} can be extended to Gaussian sources with memory (e.g., Gauss-Markov sources), using the spectral representation of the source distribution. More broadly, can one use SPARC-like constructions to compare finite alphabet sources, e.g., binary sources with Hamming distortion?

\section{Multi-terminal coding schemes with SPARCs}

\paragraph{Performance guarantees with optimal encoding} A key open problem is to provide a rigorous proof that the SPARC coding schemes in Chapter \ref{chap:sideinfo} for the  Wyner-Ziv and Gelfand-Pinsker problems achieve the Shannon limits. As discussed in Sections \ref{sec:wz_sparcs} and \ref{sec:gp_sparcs}, the main challenge in the analysis of both schemes is that  part of the effective noise at the decoder is quantization noise, which cannot be assumed to be  Gaussian and independent of the codeword.

\paragraph{Performance guarantees with feasible encoding}  The  coding schemes for the Wyner-Ziv and Gelfand-Pinsker problems each involve a quantization operation at the encoder and a channel decoding operation at the decoder. Both operations are to be performed using the same overall SPARC. The key challenge in constructing a rate-optimal  feasible coding scheme based on power allocation is that the optimal allocations  for the quantization and the channel decoding parts are different (as the rates for the two parts are different). Since the overall SPARC must have a single power allocation, we cannot simultaneously ensure that the power allocation is optimal for both the quantization and  channel decoding operations. One idea to construct rate-optimal feasible coding schemes is to use spatially coupled SPARCs, which circumvent the need for power allocation. This will require first designing a lossy compression scheme using SC-SPARCs.

\paragraph{Coding schemes for general Gaussian multiuser models}   The best-known rates for many canonical models in multiuser information theory are achieved by coding schemes that use superposition coding and/or random binning. Since we have shown how to implement both these operations using SPARCs, one can design SPARC-based schemes for a variety of  Gaussian multi-terminal problems such as the interference channel,  distributed lossy compression, and multiple descriptions coding. In addition to the theoretical question of establishing rigorous  performance guarantees for such schemes,  an interesting direction for empirical work is to provide design guidelines to optimize the error performance at finite block lengths. 

SPARC-based coding schemes have recently been used  for unsourced random access communication over the AWGN channel \cite{fengler2019sparcs}. Since SPARCs are built on the principle of superposition coding, we expect that they will be  promising candidates for  new variants of multiple-access or broadcast communication involving a large number of nodes.

%
%
%

\backmatter  

\begin{ack}
This monograph grew out of work done in collaboration with several people, to whom we are very grateful.
 In particular, we thank Sanghee Cho, Adam Greig, Kuan Hsieh,  Antony Joseph, Cynthia Rush, and Tuhin
Sarkar. Special thanks to Cynthia Rush and Kuan Hsieh for proof-reading parts of this manuscript. We also thank the two anonymous reviewers for several helpful comments. This work was supported in part by the National Science
Foundation under Grant CCF-1217023, by a Marie Curie
Career Integration Grant (Grant Agreement No. 631489), and by EPRSC Grant EP/N013999/1.
\end{ack}

\bibliographystyle{abbrv}
\bibliography{fnt_sparc}
\addcontentsline{toc}{section}{References}

\end{document}